\documentclass[10pt, book]{amsart}
\usepackage{amsmath}
\usepackage{latexsym}
\usepackage[all]{xy}
\usepackage{color}
\newtheorem{teorema}{Theorem}[section]
\newtheorem{definicion}[teorema]{Definition}
\include{amslatex}
\newtheorem{proposicion}[teorema]{Proposition}

\newtheorem{lema}[teorema]{Lemma}
\newtheorem{corolario}[teorema]{Corollary}

\newcommand{\Diff}[0]{{\textup{Diff}}}
\usepackage[margin=1.4in]{geometry}
\newtheorem{comentario}[teorema]{Remark}

\newtheorem{ejemplo}[teorema]{Example}
\numberwithin{equation}{section}
\begin{document}
\begin{title}[Emergent Quantum Mechanics and emergent gravity]
{Foundations of a theory of emergent quantum mechanics and emergent classical gravity}
\end{title}
\maketitle
\begin{center}
\author{Ricardo Gallego Torrom\'e}
\end{center}
\bigskip
\bigskip
\thispagestyle{empty}
\newpage

\thispagestyle{empty}

\begin{center}
\Large\emph{Este trabajo está dedicado a mis padres y a mi familia.}
\end{center}
\thispagestyle{empty}

\newpage

 \section*{Preface}\pagenumbering{roman}
One of the most attractive and relevant  problems in the foundation of physics is to find a complete theory consistently embracing, maybe in the adequate limits, quantum theory and gravity. This is what experts call, in an ample sense, a {\it theory of quantum gravity}, a name that, however,  could indicate a misleading meaning. It is difficult to categorize to which class of research is quantum gravity. From one side, the problem raises interpretational issues. But from another side, the standard approaches to quantum gravity are what we could call {\it explicit quantum theories of gravitation}, where quantum mechanics is applied to models of spacetime as it is applied to gauge theories, as if gravitation was another interaction as the others. According to this view of quantum gravity, nothing prevents to speak of {\it spacetimes superpositions} or the quantum particle associated with gravitation.

However, it is unclear if such an attitude is the correct one to follow. This is because the conceptual problems that it faces, among them the problem of time in quantum gravity and the problem of general covariance for some of the approaches. The problem of time is a reminiscence that quantum dynamics requires of a notion of time to be formulated, while general relativity promotes the block universe picture. The problem of compatibility of quantum theory with the diffeomorphism invariance of general relativistic theories of gravity reflects on the problem between the notion of superposition of spacetimes and how spacetime is seen in a general relativistic way, where the concept of identification of points between different manifolds is incompatible with general covariance \cite{Einstein et al., Penrose 2005}.
Moreover, the phenomenological  quantum mechanical models and quantum theory of fields are constructed making use of an underlying spacetime arena, flat Lorentzian spacetimes of diverse dimensions, usually two dimensional models, three dimensional models or the usual four dimensional Minkowski spacetime. However, the use of a fixed spacetime background is in conflict with the dynamical picture of spacetime provided by general relativity.

The foundations of quantum theory are a topic of intense debate. The foundations or interpretational issues arise mainly because quantum theory does not provide a spacetime picture of the quantum phenomena, specially for the phenomena involving quantum non-locality. Even if we accept that such phenomena are not reproducible by a classical, local dynamics (according to usual interpretations of Bell's theorems and other fundamental results), quantum mechanics does not offer a mechanism to make them compatible with the spirit of relativity in a spacetime description, although it offers a phenomenological description of physical processes.

Partially motivated by understanding the quantum phenomena, the author has been developing a new research program whose goal is to supersede the current quantum description of nature in a way that quantum correlations find a geometric interpretation and mechanism for quantum non-locality, quantum interference and for entanglement. This new point of view point on quantum mechanics embraces a novel view on the nature of gravity as an emergent phenomenon, which it turns out to be present only in the {\it classical regime} of the new fundamental dynamics.  The interpretation of gravity as a phenomenological dynamical regime of the hypothesized fundamental dynamics is achieved as an unavoidable consequence of the mathematical consistency of the theory. It is in this sense that our proposal is potentially an unifying framework for quantum mechanics and gravity, but in a way that gravity is not quantized.

The theory in this work has evolved from a first intuition in the summer of 2001 and fall 2002, after reading Chern's brief account of Finsler geometry \cite{Chern}. The idea was that natural information loss dynamics could be the origin of the probabilistic phenomenlogy of quantum system, in a similar way as, if one averages a Finsler metric to obtain a Riemannian metrics, there is a loss of information. Therefore, the use of Finslerian type methods were among the first methods in the author's attempt to formulate a deterministic theory beneath quantum mechanics \cite{Ricardo05b,Ricardo06}. Since that first intuition and attempt, many changes, corrections, implications and development of methods have occurred. This has led to the application other mathematical theories and methods besides Finslerian methods and also the recognition that more general schemes (fiber integration rather than average of Finsler structures and recently, a categorical view) as the natural arena for the theory of emergent quantum mechanics.

There are several mathematical theories involved in the present formulation of our theory of emergent quantum mechanics. Elements of the general theory of dynamical systems, statistical mechanics, differential geometry and measure theory are used in the construction of the theory or in the development of arguments. The result of its applications is a program which contains the germs for a new theory of physics based on deterministic, local dynamical systems at a fundamental scale, and where the principles of a classical theory of gravity and the principles of quantum mechanics are recovered in the form of an harmonic unification.

 The fundamental assumption is the existence of a dynamics at a fundamental scale which is local and deterministic in a particular mathematical framework. This seems contradictory with the well-known consequences of the experimental violation of Bell's inequalities. However, the circumvention of this problem relies on the fact that one needs to go beyond usual spacetime formulation of dynamics, to a picture where quantum systems are compose and complex, with a dynamics not directly living in the spacetime, but in a configuration space that, although related with spacetime, it differs ostensively from it. Furthermore, an essential ingredient of the theory is a two-time dynamics. It is the un-recognition of the two-time description, what makes the effective description of quantum correlation by quantum mechanics particularly counterintuitive. The quantum description of quantum systems and the gravitational interaction emerge from such a fundamental dynamics. While quantum mechanical description of physical systems is associated with an effective coarse grained description of the dynamics, gravity and {\it quantum state reduction} appear as two complementary aspects of the theory, consequences of the application of concentration of measure to the fundamental dynamical systems.

The theory developed is also related with the foundations of the High Arithmetics. The theory not only reveals a simple but fundamental relation between the classical theory of congruences and the proposed fundamental dynamics, but it also reveals formal properties that could turn to be of importance to understand the distribution of prime numbers through a resolution of the Hilbert-Polya approach to Riemann's conjecture.
\\
\bigskip
\bigskip

Ricardo Gallego Torrom\'e$\,\,\,\,\,\quad\quad\quad \quad\quad \quad\quad\quad\quad\quad\quad\quad $ Trieste, \today.
 \newpage
 {\small \tableofcontents{}}
 \thispagestyle{empty}

\newpage
\pagenumbering{arabic}
\section{\LARGE{Introduction}}\label{Introduction}
\bigskip
\bigskip
\subsection{Motivations for a new approach to the foundations of quantum mechanics}
  Since the advent of general relativity and its direct geometric generalizations, spacetime is described by a Lorentzian structure subjected to the Einstein's field equations, characterized by being affected and affecting the dynamics of matter and fields defined geometrically on such spacetime arena. This is in sharp contrast with the usual description of physical phenomena offered by quantum mechanics, where the back-ground spacetime structure is fixed. Being general relativity and quantum mechanics theories aimed to be of universal validity and due to such differences in the way they describe physical processes, one should expect to find predictions associated with general relativity in direct conflict with predictions of quantum mechanical models and viceversa. Such predictions appear where the gravitational effects are strong enough or in situations where induced effects of quantum mechanics on gravitational systems are detectable.

Therefore, the conceptual and technical confrontation between quantum mechanics and general relativity must be superseded by a new conceptually and mathematically consistent unified theory. It is generally believe that such a theory must be a quantum mechanical theory of the gravitational interaction, or that gravity must be described by a quantum model. There are several research programs aimed to solve the incompatibility problem between general relativity and quantum theory. String theory, loop quantum gravity are being developed for a long time. However, despite these programs are based upon very attractive and deep insights, they still do not provide a full consistent picture of physical reality.

On the other hand, finding a consistent spacetime representation for the quantum measurement processes, the non-local quantum correlation phenomena and a geometric understanding of quantum entanglement is a source for mystery and thought for a century. They are remarkably difficult problems and it is not exaggerated to state that the {\it standard approaches} do not provide satisfactory solutions. Indeed, the quantum description of measurement processes poses an intrinsic problem, which is extended in the collapse models found in the literature. For instance, the {\it collapse postulate} of quantum mechanics involves the instantaneous reduction of the wave packet in each quantum measurement. Such instantaneous character of a fundamental process is also present in collapse models. The existence of instantaneous collapse processes is clearly against the spirit of the theories of special \cite{EinsteinPodolskiRosen} and general relativity \cite{Penrose 1996, Penrose 2005, Penrose 2014a, Ricardo 2022}, but also in opposition to the methodology of explaining physical phenomena by means of spacetime representations of physical processes without the intervention of spooky actions.

The quantum mechanical formalism provides itself an approach to attack some of the fundamental questions. Indeed, {\it quantum decoherence} has been proposed as a natural solution of the measurement problem \cite{Zurek2002}.  However, there are also doubts that the problem is settled totally by decoherence, since at the end of the day, one needs to use the collapse postulate to understand the actual realization of one of the possible alternatives. Detailed criticisms to decoherence along these lines can be found in \cite{Adler2003} and in \cite{Penrose 2005}, chapter 29. Furthermore, the theory of decoherence does not provide an spacetime or spacetime representation of the physical process.

Our point of view is that, physical processes are essentially {\it dynamical processes}. {\it Change}, either with respect to a time or with respect to the internal configuration of the system, is the driving ingredient of physical processes; a {\it model} is a construction that tries to be as faithful as possible of the physical reality and change process. In order to represent change, it is necessary to have a label, a language with what represent the system. We can call such a language a {\it geometric interpretation}, where  points are elements of the configuration space. The objective of a fundamental theory is to find the most appropriate category of configuration spaces and the allowed dynamics laws to describe changes at the supposed fundamental level.

The nature of the quantum correlations poses an additional difficulty to the above scheme if by a geometric representaion one means a spacetime representation, since the notion of quantum correlation is deeply related to the notion of non-locality, which is at the end of the day, a notion related to spacetime geometry. Non-locality is notorious in quantum correlation phenomenology, which makes rather difficult to marry the principles of relativity with the fundamental principles of relativistic physics, especially if we are asking for an interpretation of quantum correlations as a spacetime process.

In view of the problematic presented above, the formal unification of quantum mechanics and special relativity achieved by the relativistic quantum field theory excludes in the current formulation the description of quantum measurement processes and it is difficult to see how it can accommodate a spacetime or geometric description of quantum correlations and entanglement. Moreover, fundamental principles of general relativity are not implemented in current formulations of quantum field theory. This is the case of the principle of general covariance. General covariance protects the Einstein's view of classical physics that all physical measurements amount to the determination of spacetimes coincidences. This view is intimately related to the physical construction of the theory of relativity, both the special theory \cite{Einstein 1905a,Einstein et al.}, but more dramatically in the construction of general relativity  \cite{Einstein 1916, Einstein et al.,Pauli 1958}. However, the principle of uncertainty in quantum mechanics, as a mathematical consequence of the quantum formalism, rules out the possibility of a coincident measurement of momentum and position, since according to such principle, it is impossible to measure with arbitrary precision position and momentum at the same location. Therefore, in the setting of quantum theory, the spacetime coincidence principle enters in conflict with the hypothetical assignment of a coincidence for a momentum measurement. This conceptual conflict is reflected in the difficulty of accommodating the respective mathematical formalisms. Moreover, since general relativity supersedes special relativity as a description of the spacetime arena, the theory to be unified with quantum mechanics should be general relativity or an appropriate modification of it, a task deemed unnatural due to the conflict of the uncertainty principle and the coincidence principle.

Several ways of addressing this conundrum have been proposed. Theories of quantum gravity are some of them. Such attempts include loop quantum gravity models and non-commutative spacetimes. Quantum gravity models reject or limit the validity of Einstein's {\it coincidence postulate} for the case of measurements related to quantum phenomena. When the coincidence principle is fully rejected, then quantum gravity theories explore how general covariance can be restored without such a principle. However, the principle of coincidence could very well be a valid postulate from an epistemological point of view, if one accepts its applicability within the classical realm. Thus the renounce of the coincidence principle in other domains does not affect its applicability the classical domain of physics. In other words, there is no advantage rejecting the principle out of its domain of applicability, except if an alternative principle is formulated and found equally useful in an ampler domain.

The problems of the spacetime or geometric representations of quantum entanglement, quantum correlations and quantum measurement processes are related with the logical and mathematical structure of the quantum theory, with the realistic interpretation of the quantum physical processes  and with the origin of spacetime and gravity. It is difficult that an approach where gravity is quantized could address these questions. On the other hand, the $EP=EPR$ conjecture is a relevant example of approach aimed to understand quantum entanglement via an spacetime picture \cite{MaldacenaSusskind,Susskind 2014}. It represents entanglement between black holes as Einstein-Rosen bridges and speculates that this type of mechanism is the mechanism for generic quantum entanglement. However, although the approach is very suggestive, because it is based upon several conjectures on quantum gravity, $EP=EPR$ remains logically incomplete.

In order to obtain a spacetime or a geometric representation of the fundamental physical dynamics and fundamental systems, we adopt a different perspective beyond the quantum mechanical description and beyond the general relativistic description of physical phenomena. Attempts on this line are hidden variable models. But hidden variable models are constrained by experiments and by theoretical results, among the most serious one the Kochen-Specken theorem \cite{Kochen-Specken1967,Redhead,Isham 1995}. Moreover, the experimental violation of Bell type inequalities poses strong constraints on the type of hidden variable models. In particular, the violation of the inequalities implies that hidden variables either need to violate standard spacetime locality or need to violate statistical independence or both assumptions. Any of these modifications of the standard view constitute a radical conceptual change. Inevitably, one needs to accept such constraints.

The developments considered in this work turns around the nature of quantum correlations as the physical key phenomena to understand. We have applied techniques and theories from physics and mathematics with the aim to construct a geometric description, in the above sente, of quantum correlations as dynamical processes happening in a generalized form of geometric arena. Quantum correlations are at the heart of quantum mechanics. To understand quantum correlations and quantum entanglement geometrically has led to the author to consider and develop a dynamical theory where time is two-dimensional, although the two-dimensionality is not in a geometric spacetime sense. The notion of two-dimensional time appears more related with the models used in classical dynamics \cite{Arnold} than the geometric notions appearing in string theory \cite{Bars2001}. Two-dimensional time is part of the key concept to understand quantum non-locality as caused by a incomplete description of the dynamics. On the other hand, how two-dimensional time appears from the dynamical elements of the theory is a result of our theory. Our concept of two-dimensional time is a genuine new concept in physics with fundamental consequences for physics, philosophy and mathematics. Since space and time are not independent notions, a change in the notion of time must be accompanied by a change in the geometric arena of physical processes \cite{Ricardo06,Ricardo2017b}. This is also investigated in this work.

The attempt to understands the above questions has led to a theory where quantum theory is emergent from a deterministic, local fundamental dynamical systems, based upon the assumption of a deeper level of description of physical reality than the current quantum description. It is the case that, besides containing a resolution of the measurement problem in terms of a mechanism of natural spontaneous collapse and a geometric interpretation of quantum entanglement and non-locality, the theory contains the germs for a formal understanding of gravity as an emergent phenomenon. These two phenomena, natural spontaneous collapse and gravity, appear as suplementary aspects of the same dynamical regime of the fundamental dynamics: both phenomena are in some sense simultaneous and appear under the same circumstances. They are not the cause of each other, but two different aspects of a deeper process, based on the general geometric phenomena of {\it concentration of measure}.

\subsection{Hamilton-Randers Theory}
 A new avenue to address fundamental questions of interpretation of quantum mechanics is suggested under the common name of {\it emergent quantum mechanics}. These research programs share the point of view that there must exist an underlying more radical level of physical description from where quantum mechanics is obtained as an {\it effective description} \cite{Adler,Hooft2016,Ricardo05b,Ricardo06,Blasone2,Elze2003,Elze,Elze2009,Elze2009,Singh}.
Characteristics of several emergent quantum mechanics is that the degrees of freedom at the fundamental level evolve by a deterministic and local law \cite{Ricardo06, Groessing2013, Hooft, Smolin2012} and in some cases, that the emergent paradigms are build on a relation between this new and deeper level of description related with the quantum mechanical description  in an analogous way as thermodynamics is related with classical statistical mechanics \cite{FernandezIsidroPerea2013, FernandezIsidroVazquez2016}.

A fundamental difficulty in some of such deterministic approaches is that the associated Hamiltonian operators, being linear in the momentum operators, are not bounded from below. Therefore, in order to ensure the existence of stable minimal energy states, a natural requirement for the construction of viable quantum models of matter, a mechanism to stabilize the vacuum is necessary. One of the mechanism proposed in the literature involves a dissipative dynamics at the fundamental Planck scale \cite{Hooft}. It was suggested that the gravitational interaction plays an essential role as the origin of the information loss in dynamics. Therefore, gravity must be present at the level of the fundamental  scale. However, gravitational interaction could be a classical, emergent phenomenon too, absent at the fundamental scale  where it is assumed that the dynamics of the microscopic fundamental degrees of freedom takes place \cite{JacobsonI, Verlinde2011,Ricardo2015b,Padmanabhan2012a,Padmanabhan2012b,Padmanabhan2015,Ricardo2016,Ricardo 2024,SharmaSingh,SharmaSingh,DeSinghVarma2019,Singh}. If this is the case, it is unnatural to appeal from the beginning to gravity as the origin of the dissipation of information at the fundamental level of physical description. Instead, in the theory to be developed in this work, gravity  has an emergent nature, while the dissipation of information is substituted by the view that quantum mechanics is an incomplete description of the dynamics rather than caused by an objective dissipation mechanism.

One fundamental problem in the construction of deterministic quantum models for physics at the fundamental scale is to establish an unambiguous  relation between the degrees of freedom  at the fundamental scale and the degrees of freedom at the quantum scales. By quantum scales we mean not only the degrees of freedom of the standard model of particles or at unification quantum field gauge theories scales, but also atomic and hadronic scales, atomic scales and molecular scales, although it can be a hierarchy in several of such a scales. Although they should be seen as composed systems, the new concepts should be applicable to such {\it non-standard model scales} too, since the exhibit quantum phenomenology. Examples of  quantum models that can be described as deterministic models have been discussed in the literature. In particular, it has been shown that  massless non-interacting $4$-dimensional first quantized Dirac neutrino particle can be identified with a deterministic system as well as the free $(1+1)$-bosonic quantum field models and certain string models  \cite{Hooft2016}. Such examples illustrate that the description of  the dynamics of non-trivial quantum systems as deterministic dynamical systems is feasible. Based on the insights extracted from the examples cited above, attempts to construct a theory of emergent quantum mechanics are, for instance, the work of Hooft \cite{Hooft2016}.

The present work exposes the construction of a novel program on the foundations of physics. According to this program, present quantum mechanics is an effective emergent scheme of physical description. The program has evolved from first ideas exposed in deterministic systems as models for quantum mechanics \cite{Ricardo05b, Ricardo06}, but supersedes these previous investigations in several aspects. In the previous works \cite{Ricardo05b,Ricardo06} it was shown that any first order ordinary dynamical system with maximal speed and maximal proper acceleration can be described using Hilbert space theory as application of Koopman-von Neumann theory of dynamical systems \cite{Koopman1931,von Neumann, Koopman-von Neumann1932} to certain dynamical models called Hamilton-Randers models defined in co-tangent spaces of large dimensional configuration spaces. In this way, we observe that Hilbert space formalism can be applied quite universally to relevant physical dynamics.
In the present work it is shown how the fundamental elements of the quantum mechanical formalism can be derived from the underlying formalism of an specific type of deterministic dynamical models. It was also a surprise to discover that the mathematical theory developed to approach quantum mechanics as an effective, emergent scheme from an deeper underlying deterministic and local (in the corresponding configuration space) dynamics contains a dynamical regimen with strong formal similarities with classical gravity. In this way, the theory also proposes an emergent character for gravity \cite{Ricardo2015b,Ricardo2016,Ricardo 2024}, which appears as a consequence of the application of the mathematical theory of {\it concentration of measure} as it appears in probability theory and metric geometry \cite{Gromov,MilmanSchechtman2001,Talagrand} to high dimensional Hamilton-Randers dynamical systems. It turns out that under the assumptions of the theory, a concentration of measure induces a {\it natural spontaneous collapse} of the state to a classical state. Differently from spontaneous collapse models \cite{GhirardiGrassiRimini,GhirardiRiminiWeber} and gravity induced collapse models \cite{Diosi, Penrose 1996, Penrose 2005}, in our proposal the natural spontaneous collapse of the quantum state is not necessarily induced by the measurement device and happens spontaneously, it is associated with one regime of the fundamental dynamics.

In this context, the perceived reality by a conscious, macroscopic observer consists of a process of large number of successive classical emergence from the underlying fundamental dynamics. This emergent origin of the {\it macroscopic observables} provides a resolution of the measurement problem: at any instant of the time parameter used in the description of quantum or classical evolution, the value of any observable associated with a quantum system is well-defined. These succession of processes of emergence is such that they also describe successfully the characteristic discontinuous {\it quantum jumps}. Regarding the question about the domain of applicability of the Einstenian principle of spacetime coincidence, Hamilton-Randers theory suggests that it is valid in the concentration domain where values of observables are defined and hence, can be attached to the system by means of a measurement. There is no need to extend or modify the coincidence principle.

 The way how our theory avoids the direct application of the standard Bell-type inequalities depends on some of the assumptions made on the fundamental dynamics and also on the particular use of Koopman-von Neumann theory. The dynamical theory exposed in this work is build upon a new notion of time. Time appears as a two-dimensional parameter. The associated dynamics is a two-time dynamics, as we will discuss later. Two-time dynamical systems are not unusual in physics. They appear, for instance, in the investigation of classical dynamical systems \cite{Arnold}. However, in the case we consider, the two-dimensional nature of time is of fundamental relevance and will be essential in our interpretation of fundamental notions of quantum mechanical description as an effective descriptions, when one of the time dynamics projects out. This includes a new theory of quantum correlations and entanglement. In particular, the notion of two-time dynamics provides a natural way to speak of a local dynamics in the configuration space of a Hamilton-Randers dynamical system, that when projected in a spacetime description, is interpreted as having the non-local character of quantum mechanical description of the systme, since the fundamental dynamics happens along this additional, novel {\it internal time}.

The present state of our work is far from being in complete form. Several fundamental elements of the theory are not fully developed and others, that perhaps in an optimal presentation should appear as consequences, remain still on the sphere of  assumptions. Indeed, we have identified a large number of assumptions in the construction of the theory, indication that the current state  is still far from being put in the most natural terms.  Among the elements and problems that deserve further attention, let us mention the construction of models for the fundamental dynamics, the development of a model for the quantum measurement, a precise theory of emergent gravity in the form of a field theory and a development of an incipient relation between the fundamental dynamics and elements of number theory. Still, the development of the theory will bring fundamental tests of the assumptions and its consequences, and predictions different from the ones from the standard paradigm.
\subsection{Structure of the work}
The work is organized in four parts. In the first part, which is developed in {\it chapters} \ref{chapter on categorical approach} and \ref{chapter on Assumptions and General Theory}, an introduction to the fundamental assumptions of the theory is elaborated. In chapter \ref{chapter on categorical approach}, a categorical approach to the fundamental dynamics is discussed.

Then in chapter \ref{chapter on Assumptions and General Theory}, a general theory of dynamics, including a notion of non-reversible dynamics, is developed. The non-reversibility of the dynamical evolution is justified by showing that non-reversible dynamics is the generic situation. The introduction of the notion of quasi-metric as a mathematical implementation of non-reversibility of the dynamical laws at the fundamental scale is discussed. Then a recollection of the assumptions on the dynamics and nature  for the {\it fundamental degrees of freedom} of our models is presented. In particular, the properties associated with causality, determinism an local character  of the dynamical systems that will be investigated in this work are introduced. It is argued heuristically  that under  such assumptions, there must exist a maximal proper acceleration. The existence of a maximal proper acceleration is related with a modification of the Lorentzian geometry of the spacetime \cite{Ricardo2014}. The collection of assumptions in which our theory is based should not be taken as the basis for an axiomatic approach to our theory, but rather as a way to better implement changes and modifications in the future, by modification of the assumptions towards a more complete and consistent picture.

The second part of the work begins at {\it chapter} \ref{chapter on classical dynamics Hamilton Randers} and expands until {\it chapter} \ref{chapter of the Hilbert space structure}. It describes the fundamental mathematical structure  of Hamilton-Randers dynamical systems and the reconstruction of several fundamental notions of the mathematical formalism of quantum mechanics from Hamilton-Randers systems. {\it Chapter} \ref{chapter on classical dynamics Hamilton Randers}  provides a  complete exposition of generalized Hamilton spaces of Randers type. After a first step towards symmetrization of the dynamics, we  show the relation of  Hamilton geometry of Randers type  with a specific type of dynamical systems described by systems of first order ordinary differential equations. The dynamical models developed in chapter \ref{chapter on classical dynamics Hamilton Randers} are what we have called {\it Hamilton-Randers dynamical systems}. They are our models for the dynamics of the fundamental degrees of freedom at the Planck scale. It is also shown that such models are general covariant, that is, are geometric models.

 Furthermore, in {\it chapter} \ref{chapter on classical dynamics Hamilton Randers} the existence of maximal acceleration and maximal speed in Hamilton-Randers systems is proved. It is also discussed  the  notion of emergence of the macroscopic and quantum $\tau$-time parameters, the need of external time diffeomorphism invariance of the theory and the relation of the fundamental cycles with  prime numbers. As a consequence, a natural interpretation of the quantum mechanical energy-time quantum relation emerges.

 In {\it chapter} \ref{chapter on Koopman-von Neumann formulation}, the  formulation of Hamilton-Randers system by means of  Hilbert space  theory is described. This approach is an application of Koopman-von Neumann theory to certain class of dynamical systems defined in a tangent space of a high dimensional manifold as configuration space. This formalism allows to describe classical dynamical systems using quantum mechanical techniques. It is shown that the Koopman-von Neuwman formulation in {\it chapter} \ref{chapter on Koopman-von Neumann formulation} of Hamilton-Randers systems is general covariant.

  {\it Chapter} \ref{chapter of the Hilbert space structure} is devoted to the theoretical construction of generic quantum states from the original degrees of freedom of the underlying  Hamilton-Randers systems.  This is achieved as consequence of the application of Koopman-von Neumann for Hamilton-Randers systems discussed in {\it chapter} \ref{chapter on Koopman-von Neumann formulation}. The transition from the description in terms of  discrete degrees of freedom associated with Hamilton-Randers systems to continuous degrees of freedom associated with quantum wave functions is natural, if the difference between the fundamental scale and the quantum scale is large.  As a consequence of the theory, a constructive approach to the quantum wave function is introduced. The construction  admits a natural probabilistic interpretation and the associated quantum Hilbert space from the underlying Hamilton-Randers systems is discussed. In particular we show how a preliminary Born rule emerges. It is also shown how the quantum $\tau$-time diffeomorphism invariance emerges from the fundamental $t$-time diffeomorphism invariance discussed in {\it chapter} \ref{chapter on classical dynamics Hamilton Randers}.
  It is also shown the emergent character of the Heisenberg dynamical equations of quantum mechanics from the more fundamental Koopman-von Neumann formulation of the fundamental dynamics. Moreover, the relation between conserved quantities of the fundamental dynamics and conserved quantities of quantum mechanics are related. In a nutshell, the last ones are averaged version of the fundamental preserved quantities. All this suggest a functorial relation between the fundamental description and the quantum description, where the functor is based upon the averaging operations constructed in the theory.

The third  part of the work deals with the application of certain notions of metric geometry to the dynamics of Hamilton-Randers systems and its consequences in the form of a theory of measurement in quantum systems, a theory of emergent gravity and a theory of non-local quantum correlations. In {\it chapter} \ref{chapter on concentration of measure} we apply the theory of concentration of measure to introduce the notion of {\it natural spontaneous collapse} of the wave function. This theoretical process is similar to the spontaneous collapse of the quantum states that happen in collapse models \cite{GhirardiGrassiRimini,GhirardiRiminiWeber}. However, our notion of spontaneous collapse, its origin, mechanism , mathematical description and formulation are rather different from such  models. The main difference is that in our theory the {\it collapse of the state} happens spontaneously, independently of a possible interaction  with a measurement device and not induced by such physical processes. This interpretation is of fundamental importance for our picture, since it defines the emergence of many observable magnitudes.

In {\it chapter} \ref{Chapter on emergence of gravity} it is shown how the Hamiltonian of a Hamilton-Randers system can be decomposed in a $1$-Lipschitz continuous piece plus an additional non-Lipschitz piece. It is argued that the non-Lipschitz term should be associated with the matter Hamiltonian. This decomposition property is used to discuss a natural mechanism to bound from below the quantum Hamiltonian operator for matter in Hamilton-Randers systems. After this we discuss how a property that can be identified with the weak equivalence principle emerges in Hamilton-Randers theory by application of concentration of measure to a natural notion of {\it free falling system}. Since in {\it chapter} \ref{chapter on classical dynamics Hamilton Randers} and \ref{chapter on Koopman-von Neumann formulation} we showed the general covariance of the theory, we have a dynamical regime compatible with two fundamental assumptions of the relativistic theories of gravitation \cite{Ehlers Pirani Schild}.

In {\it chapter} \ref{chapter on properties of quantum mechanics} several fundamental issues  of quantum mechanics are discussed.  It is first considered the  quantum interference phenomena in the form of the {\it ideal quantum double slit experiment} \cite{Feynman}, first without considering the gravitational field and then followed by a discussion of the gravitationally induced quantum interference experiments \cite{ColellaOverhauser1974,ColellaOverhauserWerner1975,Sakurai}. Based upon this discussion and the form of how quantum jumps are described in chapter \ref{chapter on concentration of measure}, we show that classical emergent gravity must be essentially fluctuating.

Then we show how non-locality emerges in Hamilton-Randers theory as a consequence of the projection mechanism from the $2$-time description to the $1$-time description in the mathematical formalism. The interpretation of entangled states and a mechanism for the violation of Bell inequalities is discussed. Even if the complete description of the fundamental double dynamics is local in the sense of Hamilton-Randers dynamical systems, the formal projection 2-time $\to$ 1-time description can account of the non-local spacetime character of quantum entanglement and its consequences for EPR settings.

The fourth part of the work pursues additional applications and consequences of Hamilton-Randers theory
{\it chapter} \ref{Chapter On the notion of time} compiles part of our reflections on the notion of time according to Hamilton-Randers theory. It is argued the emergent character of external time parameters and how such a emergence defines an universal arrow of time.

{\it Chapter} \ref{chapter on Hamilton-Randers theory and number theory} explores the surprising relations between Hamilton-Randers theory and the Riemann hypothesis through the quantum approach to the Hilbert-Polya conjecture pioneered by Berry, Connes, Keating and others. We work on the direction that the conjectured fundamental dynamics beneath quantum mechanical systems is indeed deeply related with the Hilbert-Polya conjecture. It is also explore the relation between modular arithmetics and the dynamics of Hamilton-Randers systems.

In {\it chapter} \ref{Chapter on Discussion}, a comparison of Hamilton-Randers theory with other emergent quantum mechanical frameworks is succinctly discussed. Specifically, we discuss its relation with Bohmian mechanics and G.'t Hooft's theory cellular automaton interpretation of quantum mechanics. We briefly discuss the notion of emergent gravity in Hamilton-Randers theory with the proposal of E. Verlinde's theory of {\it entropic gravity}. We show how our theory is immune to several relevant criticisms of original Verlinde's theory \cite{Kobakhidze2011}, since in our theory, although emergent and classical,  gravity is not an {\it entropic force}.
We also discuss possible falsifiable tests of Hamilton-Randers theory.

Finally, several open problems of our approach are briefly discussed. The most pressing issues are to find concrete realizations of fundamental $U_t$ flow and of the relation between the non-Lipschitz piece of the Hamiltonian and conventional matter Hamiltonian. This is also related with the effective construction of quantum models as effective description of quantum mechanics.

\newpage

\section{\Large{Categorical theory of fundamental dynamical systems}}\label{chapter on categorical approach}
\bigskip
\bigskip
\subsection{Introduction}
The main assumption of our approach to emergence of quantum mechanics \cite{Ricardo05b, Ricardo06, Ricardo2015b, Ricardo2016} is the one on the existence of a deeper level in the structure of nature than the one assumed in the standard quantum mechanical level description. Such assumption poses an intrinsic limit in the direct knowledge acquirable by direct experimental methods. Because the deeper level of description conjectured by the theory the theory, a macroscopic observer does not have any physical system, quantum or classical, that can be used to probe the fundamental scale physics.

Aware of this difficulty, a methodology is need to tailor the research in emergent quantum mechanics. One aspect of such methodology consists on avoiding the concepts and elements directly motivated or learn from the investigation of large physical scales that do not have the highest degree of generality to be applied in the fundamental scale. This  { method} alienees  along Einstein's methodology of suppressing particular points of view in the formulation of general principles in physics and is also a strong support for the principle of general covariance as an heuristic principle for the formulation of physical laws \cite{Pauli 1958}.

We are not suggesting the substitution of the experimental or the empirical method by pure theoretical constructions. Indeed, the structures that one can developed have direct phenomenological or testable consequences. It is a similar situation to the confinement of quarks, that although they are real particles, the theory prevents them from being detected in free motion individually. However, the situation is novel for us, since the limits in knowledge posed by the emergent structure of a sub-quantum level of physics. Hence the unavoidable use of purely theoretical methods, that we should constrain to the maximal from a logical a nd methodological point of view.

Prior to proposing models for the dynamical systems at the fundamental scale and investigate then thoroughly, we develop a general theory of dynamical systems. The intention is not to build a comprehensive theory applicable to any dynamical system, but to underline the general structure of a meta-theory with views in its application as a guide in the construction of dynamical models for the elusive, hypothetical dynamical systems at the fundamental scale of nature. We believe that an adequate framework to develop the general theory can be found in the setting of category theory interpreted within set theory \cite{Mac Lane}.

\subsection{Categorical approach to the fundamental dynamics}
Fundamental dynamical systems refers to the hypothesised  physical systems underlying quantum mechanical systems.
 Our categorical approach to a meta-theory of fundamental dynamical systems is build on the following notions. A fundamental theory is a methodology to construct models for the fundamental dynamics beneath the quantum mechanical scale systems. A candidate to fundamental theory is identified with a category $Cat_{Fun}$. A fundamental dynamical model is partially determined by an object $\mathcal{O}$ of $Cat_{Fun}$. The morphisms between different objects of $Cat_{Fun}$ satisfy the property that allow for compositions: given $f_{ij}:\mathcal{O}_i\to \mathcal{O}_j$ and $f_{kj}:\mathcal{O}_j\to\mathcal{O}_k$, the composition of  morphisms $f_{kj}\circ f_{ji}:\mathcal{O}_i\to\mathcal{O}_k$ is defined when the codomain of $f_{ij}$ is a subset of the domain of $f_{kj}$. Furthermore,
 each object $\mathcal{O}$ of $Cat_{Fun}$ has associated a non-empty set of endomorphisms ${End}(\mathcal{O})$. The composition $f_{j}\circ \tilde{f}_{j}:\mathcal{O}\to \mathcal{O}$ of endomorphisms $f_{j}:\mathcal{O}\to \mathcal{O},\,\tilde{f}_{j}:\mathcal{O}\to \mathcal{O}$ is always well defined. Also, $End(\mathcal{O})$ contains the identity morphism $Id_j:\mathcal{O}_j\to \mathcal{O}_j,\,a\mapsto a$. Therefore, $End(\mathcal{O})$ with the composition of morphisms is a monoid for each object $\mathcal{O}$.

The composition law of morphisms is assumed to be associative in the sense that the relations
 \begin{align}
 f_{ij}\circ ( f_{jk}\circ f_{kl} )=\,(f_{ij}\circ f_{jk})\circ f_{kl},
 \label{associativity of morphisms}
 \end{align}
hold good whenever the compositions are defined. In particular, the composition of endomorphisms
\begin{align}
f_j\circ (f_k \circ f_l)=(f_j \circ f_k)\circ f_l
\label{associativity of endomorphisms}
\end{align}
holds always good, since the composition of endomorphisms is always defined.

The theory under consideration applies to dynamical systems interpreted within the framework of set theory.
We restrict our theories to the case when the morphims are functions between sets. We still denote them as morphisms of the fundamental category. The restriction to functions has its role to implement determinism, although it is not logically necessary. On the other hand, one of the reasons to implement deeterminism in a mathematical description of physical systems is that it is the simplest alternative of generic type.

Within set theory and when the composition of the morphisms are defined, associativity holds. By assuming the associativity property for the composition of morphsims, $(End(\mathcal{O}),\circ )$ is a monoid for each object $\mathcal{O}$. Associativity ensures that, if there is an identity morphism $\mathcal{I}d:\mathcal{O}\to \mathcal{O}$ such that the diagrams
 \begin{align}
\xymatrix{\mathcal{O}\ar[rd]^{g}\ar[d]^{Id}\\
\mathcal{O} \ar[r]^{g} & \mathcal{O}}
\label{identities}
\end{align}
for each $g$, then it is unique. This is a nice property for our theory, since having two different neutral elements leads to difficulties related with cyclic dynamics, as a direct argument shows. Furthermore, associativity will be very useful for the construction of multiple composition of morphism.

It is not required that the composition of endomorphisms to be commutative. Neither it is required that the endomorphisms have inverse.

\begin{ejemplo}
Let us consider theories such that the spacetime is the dynamical object, as it is the case of the $3+1$ description of physical systems in general relativity. Such theories can be formalized by considering the category $Cat_{Fun}$ as being the category ${\bf Smo}_4$ of smooth manifolds of dimension $4$ and then considering foliations consistent with diffeomorphism invariance. However, when adopting ${\bf Smo}_4$ as the category for the description of fundamental processes, further assumptions of mathematical and physical nature are necessary:
\begin{enumerate}
\item Einstein equivalence principle, implying the existence of a spacetime Lorentzian metric on each object of ${\bf Smo}_4$,

\item General covariance principle, excluding the existence of other fundamental metric structures,

\item General covariant field equations,

\item A dynamical law for point particles generalizing the inertial law compatible with the above principles,

\item Phenomenological constraints limiting the possibilities for the field equations. These constraints regulate the specification of the subcategory of ${\bf Smo}_4$ compatible with the principles.
\end{enumerate}
 Relaxing the amount of assumptions, implies a landscape theory.
\end{ejemplo}

In the construction of a given dynamical theory, it is necessary to introduce a parameter labeling the evolution. Hence we assume the existence of a {\it parameter space} or {\it time parameter}. The parameter variable is used to label the elements of sequences of morphisms. Such parameter spaces can be represented by subsets of the integer numbers or by subsets of the real numbers, for instance. However, other classes of parameter spaces can be used for this goal. As such, they lay first outside the category $Cat_{Fun}$, but it must be construable from set theory and $Cat_{Fun}$.

\subsubsection{Construction of the natural numbers and zero from set theory}Given a category, it is possible to construct the ordered set of natural numbers and the zero number $\mathbb{N}\cup\{0\}$ within the framework of set theory. We start considering a category and an object $\mathcal{O}$ of the category. Then we construct the following sets
\begin{align*}
\emptyset,\, \{\emptyset\},\,\{\emptyset,\{\emptyset\}\},\,\{\emptyset,\{\emptyset,\{\emptyset\}\}\},\,...\,.
\end{align*}
That, is the sequence of empty set, the subset of all the subsets of $\mathcal{O}$ consisting of the empty sets. These two mathematical objects are represented by symbols
\begin{align*}
& \emptyset \mapsto 0,\\
& \{\emptyset\} \mapsto 1 .
\end{align*}
Then $0<1$ and $1$ is the successor of $0$. The next is to assign couples $C_1$ formed by the empty subset of $\mathcal{O}$ and the subset of the power set of $\mathcal{O}$ composed by one element, the empty set. Then one defines the above sequence recursively and uses the natural sequences to denote the corresponding elements of the sequence as follows,
\begin{align*}
& \emptyset \mapsto 0,\\
& \{\emptyset\} \mapsto 1,\\
& \{\emptyset,\{\emptyset\}\}\mapsto 2,\\
& \{\emptyset,\{\emptyset,\{\emptyset\}\}\}\mapsto 3,\\
& ...\, .
\end{align*}
Thus the sequence $\{\emptyset,\, \{\emptyset\},\,\{\emptyset,\{\emptyset\}\},\,\{\emptyset,\{\emptyset,\{\emptyset\}\}\},\,...\}$ is identified with the natural numbers union with the zero number. Thus we call this sequence $\mathbb{N}\cup\{0\}$.

Being recursive, the above sequence is ordered. Thus, inducing the order $0<1<2<3<...$.

It is possible to define a sum operation in $\mathbb{N}\cup \{0\}$, just because it is recursive, if $C_n$, $C_m$ are given, then $C_n + C_M := C_{n+m}$.
\begin{proposicion}
$(\mathbb{N}\cup\{0\},+)$ is an ordered monoid.
\end{proposicion}

The requirement that the time parameter spaces should be endowed with an order relation is intrinsic in our theory, if not necessary. This requirement is a local condition on the parameter space. The assumption of partial order for the time parameter spaces serves to define the notion of {\it limits}. The above construction shows that category theory offers at least one natural parameter space, $\mathbb{N}\cup\{0\}$.

From $\mathcal{I}_N$, one can construct other parameter spaces as subsets of the extensions of the natural numbers. These parameter spaces are abstract, not linked to any specific physical system. There are at disposal as mathematical constructions for any theory based in $Cat_{Fun}$.
\subsubsection{Notion of fundamental dynamics}
\begin{definicion}
Let us consider an object $\mathcal{O}$ of a fundamental category $Cat_{Fun}$ and a collection $\mathcal{S}\subset End(\mathcal{O})$ of endomorphisms of $\mathcal{O}$. Let us also consider an ordered parameter space. A sequence $\hat{\mathcal{S}}$ is a map that associates to each element of $\mathcal{S}$ a value of the time parameter space. A {\it pre-dynamics} defined on the  model $\mathcal{O}$ with initial morphism $f_0:\mathcal{O}\to \mathcal{O}$ is a couple of sequences $(Dyn^-(f_0),Dyn^+(f_0))$ of endomorphisms of $\mathcal{O}$ with initial element $f_0$.
\label{predynamics}
\end{definicion}

The {\it time parameter} serves to label the elements of sequences. The time parameter space can be realized as subsets of the integer $\mathbb{Z}$ or as subsets of the reals $\mathbb{R}$. In the theory, there is not an absolute rule for choosing the time parameter space, although the requirement that time parameter space being endowed with an order relation is adopted. Note that this is a local condition, not a global one, useful when constructing limits.

A typical sequence of endomorphisms is of the form $\hat{\mathcal{S}}\equiv\,\{ f_{\alpha_0}, f_{\alpha_1}, f_{\alpha_2}, f_{\alpha_3},...\},$
where $\alpha_0 < \alpha_1<\alpha_2<...$ are elements of the parameter space. Note that the association from morphisms to parameter values is not necessarily injective, neither is necessary to be surjective. If the sequence is injective, then the pre-dynamics is {\it deterministic}; otherwise, the pre-dynamics is {\it non-deterministic}. Associativity and the assumption that the theory is based upon a fundamental category  within the set theory framework are fundamentally based to look for consistency of the models with determinism. Moreover, both sequences $(Dyn^-(f_0),Dyn^+(f_0))$  must be well-ordered.

In order that the elements of the sequences $(Dyn^-(f_0),Dyn^+(f_0))$ are ordered or partially ordered in our construction, the parameter space must be endowed with an order relation or a partial order relation. The elements of $\mathcal{S}$ acquire the order relation from the map $\hat{\mathcal{S}}:\mathcal{S}\to \mathcal{I}$ and the order in the parameter space $\mathcal{I}$. The parameter labeling the succession of endomorphisms is related with the notion of time parameters used in the dynamics. However, these two concepts do not fully coincide logically: the time parameter just needs to be {\it finer} than the parameter labelling the successions of $Dyn(f_0)$.

A {\it local cyclic pre-dynamics} on an object $\mathcal{O}$ is a pre-dynamics consisting of two morphisms sequences such that the composition of all the morphims of the second sequence $\beta$ after the composition of all morphisms of thirst sequence $\lambda$ is the identity, $\beta\circ \lambda =Id$. The associativity property of the composition of endomorphisms allows the possibility of well defined deterministic cyclic pre-dynamics. On the other hand, let us assume that for a given object $\mathcal{O}$ the composition of endomorphisms was not associative. Then the neutral element for the composition $Id:\mathcal{O}\to \mathcal{O}$ does not need to be unique, which implies that a possible inverse element of $\lambda$ does not need to be unique neither. Thus the conditions of cycles will be of the form $\beta_1 \circ \lambda=Id_1$, $\beta_2\circ \lambda =\,I_2$ and the second part of the local cycle dynamics is not well defined without further information, since there are two possibilities, $\beta_1$ and $\beta_2$ to complete the local cycle.

Let us consider two pre-dynamics,
\begin{align*}
\begin{cases}
& Dyn(f_0)=\,(Dyn^-(f_0),Dyn^+(f_0)),\\
& Dyn(f_{\alpha})=(Dyn^-(f_\alpha ), Dyn^+(f_\alpha ) )
\end{cases}
\end{align*}
 with initial endomorphisms $f_0$ and $f_\alpha\in\,Dyn(f_0)$ respectively and such that $Dyn^-(f_0)\cup\, Dyn^+(f_0)=\,Dyn^-(f_\alpha )\cup Dyn^+(f_\alpha ) )$. Then $Dyn(f_0)$ and $Dyn(f_\alpha)$ are equivalent. Thus, for example, if
\begin{align*}
Dyn^+(f_0)=\{f_0, f_{1}, f_{2} , f_{3},...\},\quad Dyn^-(f_0)=\{f_{0}, f_{-1}, f_{-2}, f_{-3},...\}
\end{align*}
 and
\begin{align*}
\widetilde{Dyn}^+(f_k)=\{f_k, f_{k+1}, f_{k+2} , f_{k+3},\cdot \cdot\cdot\},\quad \widetilde{Dyn}^-(f_k)=\{f_{k}, f_{k-1}, f_{k-2}, f_{k-3},...\},
\end{align*}
but the two sequences coincide as ordered sets, the pre-dynamics $(Dyn^-(f_0),Dyn^+(f_0))$ and $(\widetilde{Dyn}^-(f_k),\widetilde{Dyn}^+(f_k))$ are equivalent.

There is also a consistent condition for the definition of the morphisms. It can happen that a given endomorphism $f_j:\mathcal{O}\to \mathcal{O}$ of $Dyn^+(f_0)$ can be re-casted as a composition of two non-trivial morphisms $f_{j_2}\circ f_{j_1}=f_{j}$ with $codom (f_{j_2}\circ f_{j_1})=\,codom (f_j)$ and such that $f_{j_1},f_{j_2}\neq Id$. In this case, if otherwise the same, $Dyn^+(f_0)=\{f_0, f_{1}, f_{2} , f_{3},...,f_{j_1},f_{j_2},...,\}$ {\it is finer} than $\widetilde{Dyn}^+(f_0)=\{f_0, f_{1}, f_{2}, f_{3},...,f_{j},...,\}$. Analogous considerations apply to the composition of elements $Dyn^-(f_0)$. This construction suggests a natural extension of the equivalence relation discussed above. The equivalence class containing $Dyn(f_0) $ is denoted by $[Dyn(f_0) ]$.
\begin{definicion}
If given a morphism $f_\alpha\in \, Dyn(f_0)$, the condition $f_\alpha=\,f_\beta \circ g$ holds good for certain  $f_\beta,\,g\in Dyn(f_0)$  only if $g=Id$, then $Dyn(f_0)$ is the {\it finest pre-dynamics} containing $f_0$.
\end{definicion}

After the above preliminary considerations, we introduce our notion of {\it dynamics},
\begin{definicion}
 Given a fundamental category $Cat_{Fun}$, a {\it dynamics} $Dyn(\mathcal{O})$ on the object $\mathcal{O}$ is an equivalence class of pre-dynamics containing a finest pre-dynamics. A dynamics $Dyn(Cat_{Fun})$ in the category $Cat_{Fun}$ is a collection consisting of one dynamics for each object of $Cat_{Fun}$.
\label{definicion de evolucion}
 \end{definicion}
We have the following result:
 \begin{proposicion}
 Given an initial endomorphism $f_0$ of an object $\mathcal{O}$ of $Cat_{Fun}$ for a pre-dynamics $(Dyn^-(f_0),Dyn^+(f_0))$, there is at most an unique dynamics $Dyn(\mathcal{O})$ with initial morphism $f_0$.
 \end{proposicion}
 \begin{proof}
 Let us consider two pre-dynamics $Dyn({f_0})$ and $\widetilde{Dyn}(f_0)$ on $\mathcal{O}$ corresponding to the same dynamics $Dyn(\mathcal{O})$ with the same initial morphism $f_0$, but differing on the components $Dyn^+({f_0})$ and $\widetilde{Dyn}^+(f_0)$ such that $f_\alpha\in Dyn^+({f_0})$ and $f_\beta\in \widetilde{Dyn}^+(f_0)$, $f_\alpha\neq f_\beta$. This implies either that $f_\alpha =\,f_\beta \circ g$ or that $f_\beta =\,g'\circ f_\alpha$ for certain morphims $g$ and $g'$. This corresponds to the combined action of morphisms $f_\beta , g$ or $g', f_\alpha$. Thus one can construct a new pre-dynamics either by extending the original dynamics either by containing $g$ or $g'$. However, this is impossible because being the pre-dynamics initially associated with the finest one, the original pre-dynamics  $Dyn({f_0})$ and $\widetilde{Dyn}(f_0)$ are the finest sequences, except if $g=Id$.
 \end{proof}
 Note that the morphisms $g, g' :\mathcal{O}\to\mathcal{O}$ are not generally determined by the constructions.

It is not obvious that, given a pre-dynamics there is a dynamics as defined above.

The relevance of the notion of dynamics with respect to the pre-dynamics consists of its implication of determinism. Therefore, this is the property to be required, either by proving the existence of a dynamics or assuming it.

Note that the notion of dynamics can be applied to both, the case when the parameters hold the property of the intermediate point, as the rational $\mathbb{Q}$ or real $\mathbb{R}$ fields, or for parameter spaces where it is not possible to define such intermediate.

\subsection{Natural monoids of cummulant endomorphisms}
 Let $\mathcal{I}_N$ be the parameter space as defined above and $\mathcal{J}\subset \,\mathcal{I}_N$ an arbitrary subset. Let us consider the {\it cummulants} of endomorphisms defined from pre-dynamics on an object $\mathcal{O}_j$, that is, endomorphisms of the form
 \begin{align*}
 \begin{cases}
 & \,^+C(\mathcal{O}_j, \mathcal{J}):=\,\left\{\prod_{\alpha \in \mathcal{J} } f^j_\alpha ,\,f^j_\alpha\in \,Dyn^+(\mathcal{O}_j),\,\mathcal{J}\subset\mathcal{I}_N\right\},\\
 & \,^-C(\mathcal{O}_j, \mathcal{J}):=\,\left\{\prod_{\alpha \in \mathcal{J} } f^j_\alpha ,\,f^j_\alpha\in \,Dyn^-(\mathcal{O}_j),\,\mathcal{J}\subset\mathcal{I}_N\right\} .
 \end{cases}
 \end{align*}
 We can define a {\it convolution operation of cummulatives},
 \begin{align*}
 \star : ^+C(\mathcal{O}_j, \mathcal{I}_N)\times\, ^+C(\mathcal{O}_j, \mathcal{I}_N)\to \,^+C(\mathcal{O}_j, \mathcal{I}_N),\,\, (\,^+C(\mathcal{O}_j, \mathcal{I}_1),^+C(\mathcal{O}_j, \mathcal{I}_2))\mapsto \, ^+C(\mathcal{O}_j, \mathcal{I}_1\cup\,\mathcal{I}_2).
 \end{align*}
 and similarly for $^-C(\mathcal{O}_j, \mathcal{I})$.

Then it is direct the following result,
 \begin{proposicion}
 $(\,^+C(\mathcal{O}_j, \mathcal{I}_N),\star )$ and $(F\,^-C(\mathcal{O}_j, \mathcal{I}_N)),\star)$ are abelian monoids endowed with an order relation.
 \label{monoides de acumulacion}
 \end{proposicion}
 \begin{proof}
 The composition law is associative and it has a neutral element, when $\mathcal{J}=\,\emptyset$. By the definition, the composition law of cummulants is commutative.

 Since $\mathcal{I}_N$ is ordered, we can state that $\mathcal{I}_1 < \,\mathcal{I}_2 $ if the biggest element of $\mathcal{I}_1$ is smaller than the last element of $\mathcal{I}_2$. If they are equal, we check the previous one, etc... Then one has a criteria to decide when $\mathcal{I}_1 <\, \mathcal{I}_2$, $\,^+C(\mathcal{O}_j, \mathcal{I}_1) <\, ^+C(\mathcal{O}_j, \mathcal{I}_2)$.
 \end{proof}
 The monoids $(\,^+C(\mathcal{O}_j, \mathcal{I}_N),\star )$ and $(\,^-C(\mathcal{O}_j, \mathcal{I}_N)),\star)$ are almost totally recursive constructions, since only depend upon the elements of the category and the parameter space $\mathcal{I}_N$.

\subsection{Notion of asymmetric dynamics}
\begin{definicion}
Given a dynamics determined by the finest pre-dynamics $(Dyn^-(f_0),Dyn^+(f_0))$, {\it the inverted dynamics} is the equivalence class of pre-dynamics equivalent to $(Dyn^+(f_0),Dyn^-(f_0))$.
\label{inverted dynamics}
\end{definicion}
Let us consider the finest pre-dynamics $Dyn (f_0)=(Dyn^+(f_0),Dyn^+(f_0))$ such that for $Dyn^+(f_0)=\{f_0, f_1, f_2,...\}$ it holds that $Dyn^-(f_0)=Dyn^+(f_0)$. This is a symmetric dynamics. That a dynamics could be symmetric depends on the specific initial morphism $f_0$. In contrast, the following notion of asymmetric dynamics does not depend on the election of the initial morphism,
\begin{definicion}
Given a dynamics determined by the finest pre-dynamics $(Dyn^-(f_0),Dyn^+(f_0))$, a dynamics is {\it asymmetric} if there is an initial morphism $f_0$ such that $Dyn^+(f_0)\neq Dyn^-(f_0)$; it is {\it strong asymmetric} if it is asymmetric for any possible initial morphism $f_0$.
\label{asymmetric dynamics}
\end{definicion}
Given a pre-dynamics $(Dyn^-(f_0),Dyn^+(f_0))$ constituted by a large number of morphisms, it must be generically asymmetric. Thus, for each symmetric dynamics there are many more asymmetric dynamics with the same collection of morphisms, but rearranged in different way than the symmetric one.

\subsection{Recursive principle}
  A fundamental theory should be complete in the sense of capable to describe the elements of the theory with elements of the theory. As a way towards the formalization of this idea, we introduce the following:
\bigskip
\\
{\bf Recursive principle}. {\it The elements of a dynamical theory of fundamental processes must be defined recursively in terms relative to the dynamical changes associated with the dynamical systems described by the theory.}
\bigskip
\\
This principle should apply to all the elements of a dynamical theory: objects, dynamical law and parameter space. It must also apply to the construction of the time parameters.

From the categorical point of view, a dynamics is described by collections of endomorphisms $\{f_{j}:\mathcal{O}_j\to \mathcal{O}_j\}$ for each object $\mathcal{O}_j$ of the fundamental category $Cat_{Fun}$. Therefore, the principle is re-casted in the following form,
\bigskip
\\
 {\bf Recursive principle in categorical form}. {\it The elements of a dynamical theory of fundamental processes must be defined from the objects and the morphisms of the fundamental category $Cat_{Fun}$}.
\bigskip
\\
A category is determined by two type of fundamental elements: the objects of the category and the morphisms of the category. Therefore, according to this principle, the dynamical degrees of freedom, the dynamical law and the time parameter should be defined in terms of the objects and the morphisms of the category, within the framework of set theory.

Prior to develop the principle and show how it can be implemented, we discuss the notion of parameterized dynamics in the context of the category framework.
\subsection{Parameterized dynamics from the categorical point of view}
The indexes denoting the morphisms of the dynamics can be realized by subsets of the natural numbers  $\mathbb{N}=\,\{1,2,3,...\}$ together with zero. This parameter can be understood either as the domain set of labels for ordering the morphisms of the dynamics or as elements of an algebraic structure. There is a direct generalization of the construction, when we consider the indices as valuated in an ordered monoid $(mon,+,0)$, since the associativity and the existence of a neutral element properties hold for composition of morphisms of the objects of $Cat_{Fun}$ and the same must be required for the parameter space. The recursivity principle imposes that the monoids should be constructed from the elements of the category, modulo re-parameterizations.

\begin{definicion}
Let  $(mon,+,0)$ be an ordered monoid. A {\it parameterized dynamics} with starting point the endomorphism $f:\mathcal{O}\to \mathcal{O}$ is a functor
\begin{align}
\tilde{\mathcal{P}}_{mon}:Cat_{Fun}\times mon \to Cat_{Fun}
\label{dynamics on a monoid mon}
\end{align}
such that
\begin{enumerate}
\item The action on an object is trivial, $(\mathcal{O},t)\mapsto \mathcal{O}$.

\item The action on morphism is such that:
\begin{align}
\tilde{\mathcal{P}}_{mon}(f,\hat\lambda)\circ \tilde{\mathcal{P}}_{mon}(f,\hat{\beta})=\,\tilde{\mathcal{P}}_{mon}(f,\hat\lambda')\circ \tilde{\mathcal{P}}_{mon}(f,\hat{\beta}'),\quad \forall\quad \lambda+\beta=\lambda'+\beta',
\label{preservación de asociatividad}
\end{align}
for the specific morphisms
$\hat{\lambda}:mon\to mon,\,t\mapsto t+\lambda$, etc...
\end{enumerate}
The ordered monoid $(mon,+,0)$ is the {\it parameter space}.
\label{Definicion de dinamica parametrizada}
\end{definicion}

The significance of the first condition is that measuring time does not affect the dynamics. This is a definition of {\it ideal clock}.

The second condition states the compatibility of the functor $\tilde{\mathcal{P}}$ with the associativity law of composition of morphisms in $Cat_{Fun}$.

Let us consider another category $\mathcal{S}_{\mathcal{P}}$ that it will be identified with the category of time parameter spaces. From the previous discussion, $\mathcal{S}_\mathcal{P}$ is in first instance identified with the category of (partial) ordered monoids ${\bf OMon}$.
In this setting, the recursive principle is partially implemented by the application of {\it Proposition} \ref{monoides de acumulacion}, that defines a functor
\begin{align*}
\tilde{\mathcal{P}}:Cat_{Fun}\to \mathcal{S}_{\mathcal{P}}, \quad \mathcal{O}_j\mapsto\, ^+C(\mathcal{O}_j,\mathcal{I})
\end{align*}
where $\mathcal{S}_\mathcal{P}$ is a sub-category of ${\bf OMon}$.
and where the morphism $\tilde{\mathcal{P}}$ maps the morphism $f_{ji}:\mathcal{O}_i\to \mathcal{O}_j$ to the morphism $\tilde{\mathcal{P}} (f_{ji}):\,^+C(\mathcal{O}_i,\mathcal{I})\to \,^+C(\mathcal{O}_j,\mathcal{I})$ defined by the map on cummulants,
\begin{align*}
\prod_{\alpha\in J} f^{\alpha}_i\mapsto \prod_{\alpha\in J}f^\alpha_j.
\end{align*}
The morphism $\tilde{\mathcal{P}} (f_{ji})$ does not depend upon the specific characteristics of $f_{ji}$, but only on the origin object and target object. In this sense, $\tilde{\mathcal{P}} (f_{ji})$ is a {\it constant morphism}.

 The parameter space $\mathcal{I}$ is leaved un-specified, but it is required to be constructed from the elements of the category $Cat_{Fun}$, in particular, it must be constructible using the elements of each object $\mathcal{O}$. Moreover, it is assumed the same for the model based upon the object $\mathcal{O}_i$ and for the model bades upon $\mathcal{O}_j$.

 The functor $\tilde{\mathcal{P}}$ assigns to each object $\mathcal{O}$ of $Cat_{Fun}$ an object $\mathcal{I}$ of $\mathcal{S}_{\mathcal{P}}$ and to each morphism $f_{ij}:\mathcal{O}_i\to \mathcal{O}_j$ a morphism $\tilde{\mathcal{P}} (f_{ji}): \tilde{\mathcal{P}}(\mathcal{O}_i)\to\tilde{\mathcal{P}}(\mathcal{O}_j)$ of the parameter category $\mathcal{S}_{\mathcal{P}}$.
$\tilde{\mathcal{P}}$ is not surjective in the sense that time parameters that could be used in the mathematical description of a given dynamics are not in the codomain of $\tilde{\mathcal{P}}$, neither it is a {\it full functor}, since morphisms of an object $mon$ in $\mathcal{S}_\mathcal{P}$ will not be taken into account by the morphisms of the given objects of $Cat_{Fun}$. In order to construct a full surjective parametrization functor we consider the category $\mathcal{S}_{\mathcal{P}}/Aut$ obtained by quotient each object $\mathcal{I}$ of $\mathcal{S}_{\mathcal{P}}$ by the set of order preserving automorphisms $Aut(\mathcal{I})$. The corresponding morphisms of  $\mathcal{S}_{\mathcal{P}}/Aut$ are the  morphisms between quotients $\mathcal{I}/Aut(\mathcal{I})$, that is, the quotient morphisms. From the construction of the category $\mathcal{S}_{\mathcal{P}}/Aut$ in this way we have the following result,
\begin{proposicion}
Let us consider a sub-category $\mathcal{S}_\mathcal{P}$ of ${\bf OMon}$ and the functor $\tilde{\mathcal{P}}:Cat_{Fun}\to \mathcal{S}_{\mathcal{P}}$.
Then the quotient functor
 \begin{align}
{\mathcal{P}}:Cat_{Fun}\to \mathcal{S}_{\mathcal{P}}/Aut,
\end{align}
is full.
\label{parameterization full functor 2}
\end{proposicion}
Proposition \ref{parameterization full functor 2} is a manifestation of the recursive principle, which is further fulfilled by assuming that $\mathcal{P}$ is surjective and by the characterization of evolution as given by  sequences of endomorphisms of objects. According to this point of view, not only one can associate a clock to each possible fundamental system represented by objects $\mathcal{O}$ in $Cat_{Fun}$, as it is read from the assumption on the existence of the functor $\tilde{\mathcal{P}}$, but also any imaginable time parameter is realized in this way, up to re-parametrization, by means of the functor $\mathcal{P}$.

  The characterization of the category of time parameters $\mathcal{S}_\mathcal{P}$ is intimately related with the assumptions concerning the dynamical laws and $Cat_{Fun}$, because the use of time parameterizations should not impose extra conditions on the dynamics. From the properties already discussed, it follows that each of the objects of $\mathcal{S}_\mathcal{P}$ that serves as time parameter space must be endowed at least with one algebraic binary relation, that furnishes each of them with the structure of monoid. Such binary relations constitute an integral part in the definition of dynamics, while the associative law in the objects of $\mathcal{S}_\mathcal{P}$ serves to be consistent with the associativity conditions for the morphisms,
  \begin{align*}
  \mathcal{I}_j\circ (\mathcal{I}_k \circ \mathcal{I}_l)=(\mathcal{I}_j \circ \mathcal{I}_k)\circ \mathcal{I}_l
    \end{align*}
  without imposing conditions on the original category $Cat_{Fun}$.

In order to implement further the recursive principle, the parameters used to label the sequences of the dynamics need to be defined in terms of the elements of the category $Cat_{Fun}$ or in terms of elements of categories constructed from $Cat_{Fun}$. One procedure to achieve this goal is the following. One can consider the convolution operators on morphisms. These operations define a category $\mathcal{F}$ whose objects are the monoids $(F^+(\mathcal{O}_j),\star )$ consisting on convolutions of endomorphisms given by the sequence of the dynamics and the morphisms are the induced morphisms between $(F^+(\mathcal{O}_j),\star )$ and $(F^+(\mathcal{O}_k),\star )$ induced from the morphims $f_{jk}:\mathcal{O}_j\to\mathcal{O}_k$ of $Cat_{Fun}$, modulo automorphisms of each monoid $(F^+(\mathcal{O}_j),\star )$.
Since the use of re-parameterizations can be seen as a form of abstraction in the time labeling of the dynamics, the formal category of parameters is the quotient category $\mathcal{F}/Aut $, the quotient with respect to the automorphisms of each monoid $(F^+_j,\star)$.

The identification of the notion of time in the categorical theory of dynamics involves two steps. First, $\mathcal{S}_\mathcal{P}$  is identified with $\mathcal{F}$. Second, time parameters are restricted to be modelled as codomains of the functor ${\mathcal{P}}:Cat_{Fun}\to \mathcal{F}/Aut$. In this way, the categorical recursive principle is full filled.
\subsection{Evolution of the category}
The notion of parameterized dynamics can be extended in natural way to the category itself.
\begin{definicion}
Given a functor $\Phi:Cat_{Fun}\to Cat_{Fun}$ and an partial ordered monoid $mon$, a {\it generalized parameterized dynamics} along $\Phi$ parameterized on $mon$ is a functor
\begin{align}
\tilde{\mathcal{P}}_{mon}:Cat_{Fun}\times mon \to Cat_{Fun}
\label{generalized dynamics along F on a monoid mon}
\end{align}
such that:
\begin{itemize}

\item The action on objects of $Cat_{Fun}$ is covariant in the sense of $\Phi$, $\tilde{\mathcal{P}}_{mon}(\mathcal{O},\lambda)=\, \Phi (\mathcal{O})$,

\item  The action on morphism is such that:
\begin{align}
\tilde{\mathcal{P}}_{mon}(f,\hat\lambda)\circ \tilde{\mathcal{P}}_{mon}(\Phi(f),\hat{\beta})=\,\tilde{\mathcal{P}}_{mon}(f,\hat\lambda')\circ \tilde{\mathcal{P}}_{mon}(\Phi(f),\hat{\beta}'),\quad \forall\quad \lambda+\beta=\lambda'+\beta'
\label{preservación de asociatividad}
\end{align}
holds good.
\end{itemize}
\label{generalized dynamics}
\end{definicion}

The notion of generalized parameterized dynamics accommodates better to the view of general covariance, since also the law of evolution is exposed to evolution. The categorical point of view becomes in this context a natural framework, since the mere concept of dynamical model, identified with the notion of dynamics on a object of the category, needs to be transcended to consider the full changes in the category due to the evolution and measured by compatibility with $\Phi$.

We observe in {\it Definition} \eqref{generalized dynamics} the germ of the notion of {\it evolution of the dynamical law}. Such evolution of the dynamical law is driven by the functor $\Phi:Cat_{Fun}\to Cat_{Fun}$.

\subsection{Binary composition operation for categories}
Given a category $Cat_{Fun}$, a category of parameters $\mathcal{S}_\mathcal{P}$ and a parameterized dynamics ${\mathcal{P}}_{mon}$ on each of the objects $mon$ of $\mathcal{S}_{\mathcal{P}}$, there is a binary operation $*:Cat_{Fun}\times Cat_{Fun}\to Cat_{Fun}$ such that the following diagram commutes,
 \begin{align}
\xymatrix{Cat_{Fun}\times Cat_{Fun}\ar[d]^{Id \times \mathcal{P}}\ar[rd]^{*}\\
Cat_{Fun} \times \mathcal{S}_{\mathcal{P}}/Aut\, \ar[r]^{\hat{\Phi}} & Cat_{Fun}}
\label{Composition law in Cat}
\end{align}
The functor $\hat{\Phi}$ is defined in the following way: for each  $\mathcal{O}$ of $Cat_{Fun}$ and for each time parameter space $I$, an object of $\mathcal{S}_{\mathcal{P}}$, one has that $\hat{\Phi}(\mathcal{O}, I)=\mathcal{O}$, a definition that it is keep consistent under reparameterizations.

The action on morphisms of the product $Cat_{Fun}\times \mathcal{S}_{\mathcal{P}}$ is constructed in the following way. Let us consider $\tilde{\mathcal{P}}(O_i)$, $\tilde{\mathcal{P}}(O_j)$ as objects in $\mathcal{S}_{\mathcal{P}}$. Therefore, we need to consider two generic objects $I_\alpha, I_\beta$  of $\mathcal{S}_{\mathcal{P}}$ and $\lambda_{\alpha\beta}:I_\alpha \to I_\beta$ a morphism. Then the action of the functor  $\hat{\Phi}(f_{ij},\lambda_{\alpha\beta})$ on product morphisms is
\begin{align*}
\tilde{P}_{\tilde{\mathcal{P}}(O_i)}(f_0,t)\to \tilde{P}_{\tilde{\mathcal{P}}(O_j)}(f_0,\lambda_{\alpha\beta }(t)) .
\end{align*}
By the relation \eqref{preservación de asociatividad} and since $\lambda_{\alpha\beta}$ is a morphism, $ \hat{\Phi}(f_{ij},\lambda_{\alpha\beta}):\mathcal{O}_i\to\mathcal{O}_j$ is a morphism and it is independent of the election of $t\in I_\alpha$.

An analogous construction follows from the notion of generalized parameterized dynamics.

\subsection{Conclusion}
A categorical framework for fundamental dynamics has been developed. In the theory, the nature and role of the time parameters has been discussed. There are two type of time parameters. The first are the once used to choose the labels of the sequences of morphisms defining the dynamics. Such parameters are abstract mathematical parameters, not linked to any physical phenomena and in the model, build from the underlying set theory. The second type of parameters are linked with specific physical systems. It is discussed how time parameter must be naturally emerge from the principles of the theory and the elements of the category. This is in concordance with our recurrence principle.

An initial parameter space with an inherited order relation is need, in order to label the elements of the dynamics. Such parameter are defined by the first type discussed above. Besides such a parameter, the recurrence principle is totally fulfilled in the theory, since all the objects of the dynamics are re-written in terms of the elements of the dynamics.

The theory developed has a general character. It implies that also the laws of physics must evolve with time. Thus the categorical approach justifies a change dynamical law of physics. Second, the development of the theory implies deterministic laws for the fundamental dynamics.
\newpage

\section{\LARGE{General theory of fundamental dynamical systems}}\label{chapter on Assumptions and General Theory}
\bigskip
\bigskip
\subsection{Introduction}
In this {\it chapter} the elements and fundamental assumptions and principles that determine the hypothesized fundamental dynamics beneath quantum mechanics are discussed. Such assumptions and principles will guide the search for the dynamical systems that we conjecture describe physical systems at a deeper level than the quantum mechanical description. Our exposition is formal, but does not pretend to be axiomatic. We do not fully discuss the logical dependence between the different assumptions and structures. However, we clarify several relations among them and their  mathematical formalization. In this way, the theory is open to further formal treatment and to accommodate additional or alternative developments to the ones expressed currently in the present form.

The theory of dynamical systems that we are going to developed (Hamilton-Randers Theory) is partially characterized by the following facts. First, a relevant feature of the theory of Hamilton-Randers models is the emergent interpretation of many fundamental notions in physics, among them the notion of time parameter as they appear in the dynamical description of quantum systems and classical dynamical systems, the notion of spacetime and the emergence of gravity from of a class of sub-quantum deterministic and the notion of inertial mass, as a measure of the complexity of the fundamental dynamics. Indeed, in our theory, the quantum mechanical description of physical systems and processes appear as a consequence of a coarse grained, effective mathematical description of the fundamental dynamical systems. Physical systems describing an electron, a photon and any other quantum system, will be described in Hamilton-Randers theory as complex dynamical systems with many degrees of freedom. In such a setting, it is impossible for a physical macroscopic observer and by the use of available experimental or phenomenological methods, to follow the details of the fundamental dynamics at the fundamental scale of physics. The limitations are even stronger than in the analogous case of thermodynamical systems, where it is impossible to follow the fine details of the atomic or molecular dynamics too due to the complexity of the systems. This is because in the case on hands, for the task of obtaining detailed information of a sub-quantum level description, an observer is compelled to use macroscopic or quantum process and the corresponding dynamical systems, which are too gross to probe the detailed content of the fundamental scale. Hence the impossibility. Moreover, our scheme of three levels of description of reality (macroscopic, quantum and sub-quantum levels) and our position in the extreme macroscopic case implies that no further deeper description levels with the possibility of being probed experimentally are possible. Therefore, one very important epistemological consequence of our theory is the limits on the achievable scientific knowledge. Beyond direct experimental falsifiability, only formal, mathematical further insights could be achievable. The novelty of this consequence comes from the point that is a total general result of the structure of the theory.

That there is a common description of quantum and sub-quantum dynamical systems is reinforced by the considerations of Koopman-von Neumann theory of dynamical systems \cite{Koopman1931,von Neumann,ReedSimonI} as a common language for both types of dynamical systems. Our claim on the naturalness description of the quantum theory relies partially on the relevance of Koopman-von Neumann theory as a common description of certain classical systems and quantum systems. The technic of applying Hilbert spacetime techniques appears in the work of G. 't Hooft \cite{Hooft2016} in the context of quantum mechanics and quantum field theory, but it was initially developed in the 30's of last century in the context of the theory of ergodic systems. Indeed, ergodicity is a property that, from a formal point of view, is also present at the  sub-quantum level of description.

Another fundamental ingredient in the derivation of quantum mechanic as the coarse grained description of a fundamental dynamics is the concept of two-time dimensional dynamics. The notion of two-dimensional time as conceived in this work is a novel concept. In our view, time has {\it an emergent component}. Such emergent component of time is identified with the type of standard time parameters used in the description of quantum dynamics or macroscopic dynamical systems. The emergence of one of the components of time is an aspect of the emergent nature of the phenomenological description of reality by means of quantum models. Note that this could suggest a contradiction with the timeless description usual in relativistic theories. But let us remark that in essence, the dynamical systems that we will consider are re-parametrization invariant.

The second component of time corresponds to  generic time parameter required to the description of the fundamental dynamics. This type of time parameters and the previous one, do not coincide. They are independent time parameters from each other.

A third ingredient in our theory  is the essential irreversibility of the dynamics at the fundamental level. Therefore, there is the need to reconcile such non-reversibility with the notion of double dynamics and re-parametrization invariance, but it will be shown how the emergence of a reversible dynamics associated with quantum time evolution is originated from the non-reversibility of the fundamental dynamics by means of a symmetrization of the dynamics, a first step towards an averaging of the dynamics.

In this chapter, we will develop a general theory of dynamics where the above ingredients are realized.
\subsection{On time, reversibility and non-reversibility in a theory of fundamental dynamics}
If a theory of dynamics is associated with time evolution, then it is natural to fix first the notion of time. The nature of time is probably one of the most deep mysteries in physics and philosophy with many related questions concerning the nature of time. In contemporary physics, the problem of the problem on the nature of time appears when one tries to embrace quantum mechanics and general relativity in an unifying framework. Very often in such a context, the question on the ontological nature of time is answered in a negative way, under the grounds that spacetime diffeomorphism invariance symmetry of general relativistic theories of gravitation leads to a theory where time is dismissed of any physical significance, according to the  {\it bulk spacetime picture of physics}. But this line of thought is probably answering the question from the beginning, at least since the geometric work of Minkowski \cite{Minkowski 1908}, such diffeomorphism invariance is settle in a framework where the spacetime description is viewed as the fundamental geometric framework for the description of physical phenomena. However, It could well happen that at the levels of the fundamental scale, the notions on which this orthodox point of view is maintained are not adequate. Therefore, we should consider the problem of time from a more general perspective than the bulk spacetime picture offers, suggesting the need for an ampler frame for dynamical systems. At the end, taken in a more weak perspective, the theory of relativity could perfectly accommodate the picture of the non-existence of a privileged time parameter, instead of no existence of a physical time in the sense of physical evolution, in the sense of evolution as a creative process.

Closely related to these questions is the reversible/non-reversible character of the dynamics at the fundamental scale. Although the standard model of elementary particles slightly violates {\it time reversal invariance}, the dynamics of the model is reversible. The electromagnetic sector and the strong sectors are $\mathcal{T}$-inversion invariant\footnote{Despite the $H$-theorem discussed in Scattering Theory. See for instance \cite{Weinberg1995}, section 3.6 or our discussion below.}, while the weak sector has a $CP$-violation piece, that by the CPT-theorem, implies $\mathcal{T}$-violation. However, we understand non-reversibility as the quality by which new structures are either created or annihilated  during the dynamical process without a symmetric time inversion counterpart. This kind of non-reversibility is difficult to achieve in an unitary quantum dynamical model, since in general, the probabilities for {\it inverse processes} are different from zero, allowing the formation of structures.
Consistent with this reversible point of view is the fact that general relativity provides a timeless description of physical processes. According to the standard point of view on general relativity, there is no fundamental notion of physical time, events are represented by points of a four spacetime manifold ${M}_4$ and the physical observables are insensitive to active time diffeomorphism transformations of the manifold. Thus the question of reversibility/non-reversibility dilutes in such a picture of Nature.

The need to address these questions forces us to have convenient notions for  {\it dynamical evolution}, {\it fundamental scale} and then for {\it reversible}/{\it non-reversible dynamics}. Indeed, the above arguments support the view that the fundamental laws of physics should be reversible and if one assumes the contrary hypothesis, then there is the need to explain from non-reversible models at the fundamental scale the absence of a {\it fundamental  arrow of time} in such highly successful reversible dynamical models and theories. This implies to provide a mechanism for the emergence of reversibility of the standard model dynamics. This problem is the converse to the  problem of deriving the second law of thermodynamics or a physical thermodynamical arrow of time from reversible microscopic dynamical laws.

Despite the above arguments in favour of a fundamental reversibility as a property of a fundamental theory, we think that  the reversibility/non-reversibility question of the dynamical laws at fundamental scales is still open. In order to motivate our point of view, let us remark that if the fundamental dynamical systems are constructed from fundamental assumptions, then the reversibility/non-reversibility properties of the evolution must be considered first, since depending on the assumptions that one adopts, the mathematical structure of the theory shapes its physical content.

Therefore, there is a need of a theory of dynamics for fundamental systems where the above question can be settle down and, if the fundamental dynamics is non-reversible, then it is necessary a mechanism that recovers the reversible character that dynamical theories have at the quantum and classical scale  of description, namely, at the level where unitary quantum mechanics and general relativity theories applied. The method that we will propose is the {\it time symmetrization of the dynamics}.

\subsubsection{Time parameter in a general theory of dynamics}

Let us consider the question of the reversibility or non-reversibility of the dynamical laws of fundamental systems in the categorical framework developed in {\it chapter} \ref{chapter on categorical approach}. It is natural to assume the existence of a {\it conjugate dynamics}. A conjugate dynamics can be constructed if for each  value $t$ of the time parameter, there is an opposite $-t$ such that $t+(-t)=0$. Given a parameterized dynamics, the conjugate dynamics refers to the evolution in the reverse direction of application of the automorphisms, in a way that corresponds to the opposite values of the corresponding time parameters. In terms of pre-dynamics sequence pairs $(Dyn^-(f_0),Dyn^+(f_0))$, the first pair $Dyn^-(f_0)$ is labeled by opposite values of the time parameter with respect to the secod pair $Dyn^+(f_0)$.

In order to full-fill these formal requirements, we assume that the time parameters are objects of the category of groups ${\bf Grp}$ or a reasonable sub-category of it.

 Furthermore, the definition of incremental quotient, a very suitable tool to define the notion of {\it change of rate}, can only be implemented if the time parameter space is endowed with an independent, second operation with the possibility of defining the inverses of certain elements. Therefore, the time parameters should take values in a number field $\mathbb{K}$. Therefore, the number field  $(\mathbb{K},+,\cdot)$ is endowed with an additive group operation
 $+:\mathbb{K}\times \mathbb{K}\to \mathbb{K}$ on $\mathbb{K}$ and a multiplicative operation
  $\cdot:\mathbb{K}\times\mathbb{K}\to \mathbb{K}$ that determines a multiplicative group on $\mathbb{K}^*=\,\mathbb{K}\setminus\{0\}$.

Motivated by this reasoning,  the time  parameters that we will consider are subsets $J\subset\,\mathbb{K}$ subjected to the following additional restrictions:
 \begin{enumerate}

 \item For the domain of definition of $J$, the binary operation of addition $+:\mathbb{K}\times \mathbb{K}\to \mathbb{K}$ is required to be inherited by $J$ in the form of the operation $+:J\times J\to \mathbb{K}$. This requirement is necessary if the dynamics is associated with a flow composition law.

 \item The product of elements in $J\subset\,\mathbb{K}$ must be well defined. This condition is necessary to be able to consider non-linear expressions of quantities depending on the time parameter.

 \item  Existence of the inverse elements $(t_2-t_1)^{-1}\in\,\mathbb{K}$ for elements $t_2,t_1\in \,J$ close enough in the sense of a quasi-metric function is required, except when $t_2=t_1$. This condition is necessary in order to define a notion of incremental quotients or derivative operation limits.

 \end{enumerate}

 It is the last requirement that implies the need that $\mathbb{K}$ to be a number field instead than an algebraic ring. Therefore, the time parameters are objects in the category of fields ${\bf Fld}$ and the parameter changes are associated with morphisms of ${\bf Fld}$; the change in parameters, are morphisms of ${\bf Fld}$.

 Furthermore, the categorical formulation of the dynamics furthermore implies to assume the existence of a parametrization functor on ${\bf Fld}$.

 \subsubsection{Further conditions on the time parameters}

In order to accommodate incremental quotients, further  properties need to be incorporated. In particular, it is necessary to endow $J\subset \,\mathbb{K}$ with a notion of {\it sufficiently close elements}. This requirement can be satisfied in at least two ways. The first way is by considering that $\mathbb{K}$ is endowed with a pre-order compatible with the product (in reality, with the inverse of the product) law. In this case, it is possible to define the notion of sequences of quotients with decreasing but positive denominators of the form $(t_2-t_1)^{-1}\in\,\mathbb{K}$. Quotient limits are defined by considering all the well-defined sequences
\begin{align*}
\frac{1}{(\Delta t)_n}\,\left(\Psi(t+(\Delta t)_n)-\Psi(t)\right)
\end{align*}
 of the form $(\Delta t)^{-1}_n > (\Delta t)^{-1}_{n+1} >...(\Delta t)^{-1}_{n + a} > 0$ with increasing $a\in\,\mathbb{N}$.

The second method proposed is to endow $\mathbb{K}$ with  a {\it quasi-metric structure}, that is, a function satisfying the metric axioms  except the symmetric axiom \cite{Wilson1931, Javaloyes et al.}. For instance, a non-reversible Finsler metric structure determines a quasi-metric function \cite{BaoChernShen}.
 In our general setting, a quasi-metric function is a function of the form
\begin{align}
d_{\mathbb{K}}:\mathbb{K}\times \mathbb{K}\to \mathbb{K}'
\label{distance function in the field body}
\end{align}
where, in order to be able to formulate the triangle inequality for $d_{\mathbb{K}}$, the number field  $\mathbb{K}'$ must be endowed with an order relation, but it is not necessarily an ordered field\footnote{The relevance of these two concepts of ordered field and field endowed with an order relation has been highlighted in \cite{Ricardo quotient rings}, considering a relevant example of field for us.}. Then one defines the notion of incremental quotient as a metric limit of quotients.

Let us remark that in this setting, It is not required for the field $\mathbb{K}$ to be endowed with an order relation and it is not required that $\mathbb{K}'$ to be the real number field $\mathbb{R}$, as is usually taken.

Every time interval $J\subset\,\mathbb{K}$ inherits a distance function $d:J\times J\to \mathbb{K}'$.

We may also consider number fields endowed with metric functions \cite{Ricardo general theory of dynamics, Ricardo quotient rings}. Indeed, given a quasi-metric function one can symmetrized to obtain a metric function and with it to define a metric topology in $\mathbb{K}$ and other relevant spaces for the dynamics.

Given a quasi-metric structure, there is a natural pre-order determined by the condition $x\leq y$ iff $d_{\mathbb{K}}(0,x)\leq \, d_{\mathbb{K}}(0,y)$, providing the resources to formulate the first definition of incremental quotient discussed above. One can economize the construction by considering $\mathbb{K} \cong \mathbb{K}'$, since the first one can be embedded of a pre-order relation induced from the distance function, which is enough to construct a sensible notion of quotient limit. In this way, the possibility to define incremental quotients by means of quasi-metric functions \eqref{distance function in the field body} appears as a more general procedure in two different directions than the method of defining time with a pre-order.

Given a quasi-metric structure endowed in $\mathbb{K}$, there are two possibilities to define incremental quotients,
\begin{itemize}

\item There is a non-zero minimal distance $d_{\mathbb{K}min}>0$ such that if $t_1\, \neq t_2$, then $d_{\mathbb{K}}(t_1,t_2)\geq d_{\mathbb{K}min}>\,0$. In this case, the incremental quotient of a function $\psi:J\to \mathcal{E}$ is defined by
\begin{align*}
\frac{d\psi}{dt}:=\,\frac{1}{t_2-t_1}\,\left(\psi(t_2)-\,\psi(t_1)\right)
\end{align*}
such that $d_{\mathbb{K}}(t_2,t_1)=\,d_{\mathbb{K}min}$, when the expression is well defined (unique).

\item In the case when $d_{\mathbb{K}min}=0$, as for instance it happens in the fields $\mathbb{Q}$, $\mathbb{R}$ or $\mathbb{C}$, the incremental quotient of a map $\psi:J\to \mathcal{E}$ is defined by the expression
\begin{align*}
\frac{d\psi}{dt}=\lim_{d_{\mathbb{K}}(t_2,t_1)\to\, 0}\,\frac{1}{t_2 -\,t_1}\,\left(\psi(t_2)-\psi(t_1)\right),
\end{align*}
when the limit is well defined.

\end{itemize}
Both possibilities can be unified in the form of a single notion, by defining the incremental quotient by the expression
\begin{align}
\frac{d\psi}{dt}:=\,\lim_{d_{\mathbb{K}}(t_2,t_1)\to \,d_{\mathbb{K}min}}\,\frac{1}{t_2 -\,t_1}\,\left(\psi(t_2)-\psi(t_1)\right),
\label{incremental quotient}
\end{align}
where $d_{\mathbb{K}min}=0$ for continuous number fields. It is assumed that the space where $\psi(t)$ is defined allows to take combinations of the form  $\psi(t_2)-\psi(t_1)$ and to take the incremental limits.

Let $\mathbb{K}$ be a number field endowed with a distance function $d_\mathbb{K}$. If the minimal distance $d_{\mathbb{K}min}\neq 0$, then $t_2-t_1$ is considered to be  {\it small} if $d_{\mathbb{K}}(t_2,t_1)=\,d_{\mathbb{K}min}$. If $d_{\mathbb{K}min}=0$, then $t_2-t_1$ is small if for all practical purposes when evaluation limits, the difference $(t_2-t_1)$ can be approximated by $0$.
\begin{comentario}
Let us remark at this point that we are considering the formal relation $(t_2-t_1)\equiv 0$ in the sense that this relation is used when evaluating certain limits, for instance, when evaluating incremental quotients. Another option to interpret such relation could be within the framework of fields containing infinitesimal elements, for instance in non-linear analysis \cite{Robinson}. For the moment, we will adopt the first interpretation.

Also, the above discussion makes natural to assume that the number field $\mathbb{K}$ must be endowed with a distance function $d_{\mathbb{K}}$ and that such distance function must be continuous in the topology of $\mathbb{K}$. Otherwise, the increasing quotients \eqref{incremental quotient} could depend upon the sequence of elements $t_2$ in the neighborhood of $t_1$ in an un-natural way, that is, not because the object $\psi(t)$ is discontinuous, but because of the particularities of the distance function $d_\mathbb{K}$.
\label{cometarion on standard and non-standard analysis}
\end{comentario}

Let us remark the following points:
\begin{itemize}
\item  The field $\mathbb{K}$ does not need to be complete. This is because to define incremental quotients one does not substitute the value of $\Delta t$ by the corresponding limit. It is however necessary to speak of limits, as we have seen above. Therefore, we assume the case that certain particular limits exist.

\item The number field $\mathbb{K}$ does not need to be commutative. However, we will consider $\mathbb{K}$ to be commutative, in order to eliminate ambiguities in the definition of products as the one involved in the definitions of incremental quotients.

\end{itemize}

The above considerations suggest that a general notion of time parameter suitable for our purposes is provided by the following
\begin{definicion}
Let $(\mathbb{K},+,\cdot,d_\mathbb{K})$ be a number field equipped with a distance function $d_\mathbb{K}:\mathbb{K}\times \mathbb{K}\to \mathbb{K}'$, where $\mathbb{K}'$ is a field equipped with a pre-order relation. A time parameter is a subset $J\subset \,\mathbb{K}$ such that
\begin{enumerate}
\item $(J,+)$ is a proper sub-set of $(\mathbb{K},+)$,

\item For any $t_1\in J$, there are elements $t_2\in\,J$ such that  $(t_2-t_1)\in\,\mathbb{K}$ is small.
\end{enumerate}
If $(J,+)$ is a sub-group of $\mathbb{K}$ we say that the time parameter with values in $J$ is {\it algebraically close}.
\label{definicion of time parameter}
 \end{definicion}

Consider a field automorphism $\theta:\mathbb{K}\to \mathbb{K}$. A change of parameter is a restriction $\theta|_J:J\to \mathbb{K}$. It is direct that the image $\theta(J)\subset \mathbb{K}$ must be consistent with the definition of group law, the possibility to combine the new parameter with other variables to provide non-linear expressions and allow for a definition of incremental quotient. Also, in the case that the time parameter is complete, the addition operation in $J$ and in $\theta(J)$ must be defined for all the elements of $J$. In the complete case, $(J,+)$ must be a sub-group of $(\mathbb{K},+)$.

This notion of change of parameter is intimately related with the notion of parameterized dynamics.

\begin{ejemplo}
Examples of time parameters arise when $\mathbb{K}$ is the field of real numbers $\mathbb{R}$. But the field number $\mathbb{K}$ could be a discrete field such as the field of rational numbers $\mathbb{Q}$ or an algebraic extension of $\mathbb{Q}$. In both cases, there is a well defined notion of incremental quotient.
\end{ejemplo}

\begin{ejemplo}
An example of finite number field that can be used to define time parameters for dynamics is the prime field,
\begin{align*}
\mathbb{Z}/p\mathbb{Z}:=\{[0],[1],[2],[3],...,[p-1]\}
\end{align*}
of classes $[k]$ module $p$ with $p$ a prime number. In $\mathbb{Z}/p\mathbb{Z}$ there is a natural distance function $d_{\mathbb{Z}/p\mathbb{Z}}:\mathbb{Z}/p\mathbb{Z}\times \mathbb{Z}/p\mathbb{Z}\to \,\mathbb{R}$ defined by the expression
\begin{align}
d_p([n],[m]):=\,|n_0-m_0|,\,n_0\in\,[n],\,m_0\in\,[m],\,0\leq n_0,m_0\leq p-1.
\label{distance function in Zp}
\end{align}
The prime field $\mathbb{Z}/p\mathbb{Z}$ endowed with this and other distances \cite{Ricardo quotient rings}.

 The following properties are easily proven:
 \begin{enumerate}

  \item The induced topology in $\mathbb{Z}/p\mathbb{Z}$ from the distance function $d_p$ and the discrete topology of $\mathbb{Z}/p\mathbb{Z}$ coincide.

\item The minimal distance function for $d_p$ is $d_{\mathbb{Z}/p\mathbb{Z} min}=1$.

\item By an induction argument that the possible subsects $J\subset \mathbb{Z}/p\mathbb{Z}$ that can serve as time parameters according to definition \ref{definicion of time parameter} coincide with $\mathbb{Z}/p\mathbb{Z}$ itself. Effectively, if $[t_1]\in\,J\subset\,\mathbb{Z}/p\mathbb{Z}$, then it must be (by point 2. in definition \ref{definicion of time parameter}) another $[t_2]\in J$ such that $[t_2]=\,[t_1]+\,[1]=\,[t_2+1]$. Since $card (\mathbb{Z}/p\mathbb{Z} )=\,p$  is finite, then it follows the result by repeating the argument.

\item $\mathbb{Z}/p\mathbb{Z}$ with the distance topology induced from $d_p$ is Haussdorff separable. For this, given two points $[k_1]\neq \,[k_2]\in\,\mathbb{Z}/p\mathbb{Z}$ it is enough to consider the balls
    \begin{align*}
    B([k_i],1/4):=\,\{[k]\in\,\mathbb{Z}/p\mathbb{Z}\,s.t.\,d_p([k],[k_i])<1/4\},\,i=1,2
    \end{align*}
    of radii $1/4$.
    Then $B([k_1],1/4)=\,\{[k_1]\}$, $B([k_2],1/4)=\{[k_2]\}$  and are such that $B([k_1],1/4)\cap B([k_2],1/4)=\emptyset$.

\item $\mathbb{Z}/p\mathbb{Z}$ can be endowed with an order relation,
  \begin{align*}
 [n] <\,[m]\,\textrm{iff}\,\,\, n_0\,<m_0,\,\,\textrm{with}\, \,n_0\in\,[n],\,m_0\in\,[m],\,0\leq n_0,m_0\leq p-1.
 \end{align*}

\end{enumerate}
From these properties it follows that $\mathbb{Z}/p\mathbb{Z}$ can serve as the set where a time parameter can take values.
\end{ejemplo}

\subsection{Configuration space and associated dynamical objects}
The second ingredient in the specification of a dynamics is the {\it mathematical objects} that  change with time due to the dynamics. In order to specify this concept, we start introducing a general notion of {\it configuration space} suitable for our purposes.

It is required that the configuration space $\mathcal{M}$ to be equipped at least with a topological structure. The existence of a notion of continuity in $\mathcal{M}$ allows to consider continuous dynamical laws. Continuity of the dynamical law is an useful condition to establish cause-effect relations between different points of the same orbit of the evolution. Let us assume that a given quantity $E(t)$ does not evolve continuously on time. In situations where it is difficult to follow the details of the dynamical evolution, one could instead consider differences on measurements of a given quantity $E$, namely, quantities of the form $E_2-\,E_1$, where the time dependence has been erased. Then a large change of the form $E_2-E_1$ can be associated either to a short time evolution $t_2/t_1 \approx 1$ through a local causal explanation, or to different values at different times $t_2/t_1>>1$  through a global causal explanation. Therefore, with a non-continuous law, that does not restrict the amount of change in evolution of quantities due to small changes of the time parameter $t$, it is more difficult to discriminate a local causation from a global causation of a large change $E_2-E_1$.
The absence of continuity in the dynamical law does not lead to a contradiction, but the identification of a global cause with $t_2>>t_1$ of a large change $E_1\to E_2$ with $E_2>> E_1$ is a much more difficult problem than with a continuous law.

We will consider a general theory of dynamics in the category of topological spaces and topological maps and  where the dynamical laws are continuous maps.

\begin{definicion}
The {\it configuration space} $\mathcal{M}$  of a dynamical system is a topological space whose elements describe the state of the system.
\label{configuration space}
\end{definicion}

The configuration space $\mathcal{M}$ can be either a discrete set or a continuous set; the time parameters can be discrete or continuous. We will consider the discrete/continuous character of $\mathcal{M}$ later when we discuss a specific class of dynamical models for our theory.

Our notion of configuration space applies to classical dynamical systems, where $\mathcal{M}$ is a classical configuration space. It can also be applied to quantum dynamical systems. For a pure quantum system, $\mathcal{M}$ is a projective Hilbert space $\mathcal{H}$ and the state of the system is described by elements of  $\mathcal{H}$.

Every dynamical system has associated a sub-domain of elements of $\mathcal{M}$ at each time of the dynamical evolution. Such elements are called points of the configuration space. The orbit of an evolution is the aggregate of all points of a given evolution.

\subsubsection{Sheave theory point of view} There are other physical properties associated with the system, that can be described in terms of fields defined over the configuration space $\mathcal{M}$. It is necessary to formalize the type of mathematical structures over $\mathcal{M}$ that lead to define dynamical fields such that continuous laws for the time evolution can be formulated. This is achieved by considering sheaves over $\mathcal{M}$. The notion of sheaf offers the natural setting to speak of fields over $\mathcal{M}$ as sections \cite{Godement, Hirzebruch1978}. Sheaf theory has multiple encounters in gauge theory and theoretical physics. Hence it is not a surprise to find useful in the formulation of a general theory of dynamics.

 Let us consider a $\mathcal{A}$-sheaf $(\mathcal{E},\pi_\mathcal{E},\mathcal{M})$, where $\pi_\mathcal{E}:\mathcal{E}\to \mathcal{M}$ is continuous. The composition operations of the algebraic structures of the stalks $\mathcal{A}_u =\,\pi^{-1}(u)$  are continuous.
 Typical algebraic structures for the stalks $\mathcal{A}_u$ to be considered are $\mathbb{K}$-vector fields of finite dimension and algebraic geometric constructions. A section of a sheaf is a continuous map $E:\mathcal{M}\to \,\mathcal{E}$ such that $\pi_\mathcal{E}\circ E =\,Id_\mathcal{M}$. Then we propose the following notion,
 \begin{definicion}
 A field $E$ is a section of a sheaf $\pi_\mathcal{E}:\mathcal{E}\to \mathcal{M}$.
 \end{definicion}
 For example, for an abelian sheaf, where the stalks are abelian groups, the zero section $x\mapsto 0_u\in \,\mathcal{A}_u$ is a continuous section. Any small continuous deformation of the zero section is also a continuous section.

 The set of sections of a sheaf $\mathcal{E}$ is denoted by $\Gamma\,\mathcal{E}$. The algebraic operations on the stalks induce operations on sections, defined pointwise in a natural way.

 In particular, there is a well defined notion of incremental quotient for sections of a sheaf.

\subsection{Notion of dynamics} The third ingredient that we need in the formulation of a general theory of dynamics is a notion of dynamical law compatible with the above notions of time parameters and dynamical objects. We suggest the following definition of dynamics,
\begin{definicion}
Given a configuration space $\mathcal{M}$, a number field  $(\mathbb{K},+,\cdot)$ and a time parameter $J\subset \mathbb{K}$,
a {\it local dynamics} or {\it flow} is a map
\begin{align*}
\Phi:J\times \,\mathcal{M}\to \mathcal{M},\quad (t,u)\mapsto \Phi_t(u)
\end{align*}
continuous in the product topology such that
\begin{itemize}
\item The following group composition condition holds:
\begin{align}
\Phi(t_1,u)\circ \Phi(t_2,u)=\,\Phi(t_1+t_2,u),\quad t_1,\,t_2, \,t_1+t_2\,\in \,J
\label{composition law for Phi}
\end{align}
whenever both sides are defined.
\item
The condition
\begin{align}
\Phi(0,u)=\,u
\end{align}
 holds for every $u\in\,\mathcal{M}$, where $0$ is the neutral element of the sum, $+:\mathbb{K}\times\mathbb{K}\to \mathbb{K}$.

\end{itemize}
\label{definitionofdynamics}
 \end{definicion}
 The relation \eqref{composition law for Phi} is very often  re-written in the theory of dynamical systems in the form
  $\Phi_{t_1+t_2}(x)=\,\Phi_{t_1}\circ \Phi_{t_2}$, where $\Phi_{t}=\Phi(t,\cdot)$.
 One can compare this notion of dynamics with standard notions of dynamical systems, for instance as in  \cite{ArnoldAvez} or as in \cite{Sternberg}, chapter 12.

 The term local dynamics refers to the fact that the outcome of the evolution depends pointwise on $\mathcal{M}$.

 {\it Definition} \ref{definitionofdynamics} deviates from usual notions of dynamical system \cite{ArnoldAvez, Sternberg, Chicone}. We do not require conditions on the existence of a measure on the configuration space $\mathcal{M}$. In contrast, we make emphasis on the character and properties of the time parameter $t\in \,J$ and its reflections on the dynamics.

 In definition \ref{definitionofdynamics} there is no apparent need for the time parameter $J$ being a subset of a number field $\mathbb{K}$, being sufficient that the addition operation $+:J\times J\to J$ is well-defined. But in order to consider incremental quotients as given by the expressions of the form \eqref{incremental quotient}  of fields defined over $\mathbb{M}$, it is required that $J$ is a subset of a number field $\mathbb{K}$ to secure that the expressions $ (\delta t)^{-1}=(t_2-t_1)^{-1}$ are defined for $t_2-t_1$ for sufficiently small but non-zero elements. Incremental quotients are useful as a measure of change. Indeed, let us consider a quantity $E$ such that it changes {\it very little} with a small amount of time parameter $\delta t$. The exact definition of {\it very little} here is not of relevance, because for any meaningful definition, the change produced in the incremental quotient could be significatively large, hence, easier to measure or detect. Thus incremental quotients has the potential advantage of magnifier changes in situations where they are small.

 If $\mathbb{K}$ is endowed with an order relation, then there is an induced order relation in $J$. In this case, time ordered sequences from $A\in \mathcal{M}$ to $B\in\mathcal{M}$ along the dynamics $\Phi$ can be defined and a chronological order can be attached to the evolution from $A$ to $B$ by means of $\Phi$. This is the typical situation in general relativistic spacetimes without closed causal curves, but even for classical spacetimes  that do not have a global time ordering as in G\"odel-type solutions, there is a time directionality\footnote{For the distinction among these notions, see for instance \cite{Earman1994}.}.

  If the number field $\mathbb{K}$ is not ordered, then there is no notion of local time ordering for the dynamics $\Phi$. In such a case, there is no notion of global time ordering as it appears in relativistic spacetime models. From the general point of view discussed until now, the number field $\mathbb{K}$ that appears in the notion of dynamics does not need to be ordered, but the number field $\mathbb{K}'$ where the distance function $d_\mathbb{K}:\mathbb{K}\times\,\mathbb{K}\to \,\mathbb{K}'$ takes values, must be ordered.

\begin{definicion}
Standard fundamental notions of dynamical systems apply.
The local dynamics $\Phi:J\times \mathcal{M}\to \mathcal{M}$ is complete if the time parameter $J$ is a sub-field (not necessarily proper) of $\mathbb{K}$.; it is transitive if for every $(x,y)\in \,\mathcal{M}\times \,\mathcal{M}$ there is a $t\in\,\mathbb{K}$ such that $\Phi(t,x)=y$. The dynamics is filling if for every $(x,y)\in\,\mathcal{M}\times\,\mathcal{M}$ there is a $t\in\,J$ such that $\Phi (t,x)=y$.
\label{complete and transitive dynamics}
\end{definicion}
The following result is direct,
 \begin{proposicion}
 Let $\Phi$ be a complete local dynamics.
The transformations $\{\Phi_t\}_{t\in\,J}$ according to {\it definition} \ref{definitionofdynamics} determines a group of transformations of $\mathcal{M}$.
\label{group of transformations}
\end{proposicion}

 The extension of the notion of local dynamics to the evolution of mathematical objects defined over the configuration space $\mathcal{M}$  can be achieved in the following way. Let us consider a $\mathbb{K}$-module sheaf $\pi_{\mathcal{E}}:\mathcal{E}\to \mathcal{M}$ and $\Phi:J\times \mathcal{M}\to \mathcal{M}$ a dynamics. Let $\varphi:\mathbb{K}\times \mathbb{K}\to \mathbb{K}$ be a  $\mathbb{K}$-isomorphism,
 \begin{align*}
 &\varphi(t_1+t_2)=\,\varphi(t_1)+\,\varphi(t_2),\\
 & \varphi(t_1 \cdot \,t_2)=\,\varphi(t_1)\cdot \,\varphi(t_2).
 \end{align*}
 The simplest case is to consider the identity map $\varphi= Id_\mathbb{K}$.
The continuous map $\Phi_{\mathcal{E}}:\mathbb{K}\times  \mathcal{E}\to \mathcal{E}$ is such that the diagram
 \begin{align}
\xymatrix{\mathbb{K}\times \mathcal{M} \ar[r]^{\Phi}   & \mathcal{M}\\
\mathbb{K}\times \mathcal{E} \ar[r]^{\Phi_{\mathcal{E}}} \ar[u]^{\varphi\times \pi_{\mathcal{E}}}  & \mathcal{E} \ar[u]^{\pi_{\mathcal{E}}} }
\label{commutativity of the diagram on E}
\end{align}
commutes. The map $\Phi_\mathcal{E}$ such  that the diagram \eqref{commutativity of the diagram on E} commutes is the induced dynamics in the sheaf $\mathcal{E}$.
 \begin{proposicion} If $\Phi$ is a local dynamics on $\mathcal{M}$ and $\Phi_\mathcal{E}$ is the induced dynamics on the sheaf $\mathcal{E}$, then for any section $E\in\,\Gamma \mathcal{E}$ there is an open  neighborhood $N\subset \mathcal{E}$ such that
\begin{align}
\Phi_\mathcal{E}(t_1+t_2,\cdot)=\,\Phi_\mathcal{E}(\varphi(t_2),\Phi_\mathcal{E}(\varphi(t_1),\cdot))
\label{composition Phi for E}
\end{align}
holds good on $N$, whenever $t_1,t_2,t_1+t_2 \in J\subset \mathbb{K}$.
\label{proposicion on law group for the dynamics}
\end{proposicion}
\begin{proof}
   The commutativity of the diagram \eqref{commutativity of the diagram on E} and the homomorphism law of the dynamics for $\Phi$, eq. \eqref{composition law for Phi}, implies
\begin{align*}
\pi_\mathcal{E}\circ \Phi_\mathcal{E}(t_1+t_2,e_x)& =\,\Phi(\varphi(t_1+t_2),u)\\
& =\,\Phi(\varphi(t_1)+\varphi(t_2),u)\\
& = \Phi(\varphi(t_1),\Phi(\varphi(t_2),u))\\
& =\,\Phi(\varphi(t_1),\Phi(\varphi(t_2),\pi_\mathcal{E}(e_u)))\\
& =\,\Phi(\varphi(t_1),\pi_\mathcal{E}\circ \Phi_\mathcal{E}(t_2,e_u))\\
& =\,\pi_\mathcal{E}\circ \Phi_\mathcal{E}(t_1,\Phi_\mathcal{E}(t_2,e_u)),
\end{align*}
for every element $e_u$ of the stalk $E_u$.
Since the restriction of  $\pi_\mathcal{E}$ in some open neighborhood $N\in \mathcal{E}$ is an homeomorphism, then it follows the relation \eqref{composition Phi for E}.
\end{proof}

Given a sheaf $\pi_\mathcal{E}:\mathcal{E}\to \mathcal{M}$, the dynamics of sections $E\in \,\Gamma \mathcal{E}\mapsto \,\Phi_\mathcal{E}(t,E)$ is determined by the section $\widetilde{E}\in\,\Gamma \mathcal{E}$ such that the   diagram
 \begin{align}
\xymatrix{\mathbb{K}\times \mathcal{M} \ar[r]^{\Phi}   & \mathcal{M}\\
\mathbb{K}\times \mathcal{E} \ar[r]^{\Phi_{\mathcal{E}}} \ar[u]^{\varphi{-1}\times E}  & \mathcal{E} \ar[u]^{\widetilde{E}}}
\label{commutativity of the diagram on E 2}
\end{align}
is commutative and where $\Phi_\mathcal{E}(t,E):=\,\Phi_\mathcal{E}(t,E(u))$ at $E(u)\in A_u\,\pi^{-1}_\mathcal{E}(u)$.

Given the sheaf $\pi_\mathcal{E}:\mathcal{E}\to \mathcal{M}$, the dynamics $\Phi_\mathcal{E}$ is not necessarily a morphism. When $\pi_\mathcal{E}:\mathcal{E}\to \mathcal{M}$ is a $\mathbb{K}$-vector sheaf, $\Phi_\mathcal{E}$ does not need to be linear.

\subsection{Notion of two-dimensional time dynamics}
 For the dynamical systems that we will consider in this work, the following notion is of special relevance,
 \begin{definicion}Let $\mathbb{K}_i,\,i=1,2$ be number fields and $J_i\subset \,\mathbb{K}_i$.
  A $2$-time local dynamics is continuous a map
   \begin{align*}
\Phi:J_1\times J_2\,\times \mathcal{M}\to \mathcal{M},\quad (t, \tau, u)\mapsto \Phi_{(t,\tau)}(u)
\end{align*}
in the product topology such that
\begin{itemize}
\item
The morphism condition
\begin{align*}
\Phi_{(t_1,\tau_1)}\circ \Phi_{(t_2,\tau_2)}=\,\Phi_{(t_1+\,t_2,\tau_1+\,\tau_2)}
\end{align*}
holds good,
where $t_1,t_2\in\,J_1$, $\tau_1,\tau_2\in\,J_2$,  $J_i$ are sub-sets time parameters of the number fields $\mathbb{K}_i$ and $u\in \mathcal{M}$.
\item
The condition $\Phi_0(x)=\,x$ holds for every $x\in\,\mathcal{M}$, where the zero element is $0=(0_1,0_2)\in \,\mathbb{K}_1\times \mathbb{K}_2$ is the product of  neutral elements of the sum operation for $\mathbb{K}_1$ and $\mathbb{K}_2$.
\end{itemize}
\label{two time dynamics}
\end{definicion}
Analogous results to {\it Propositions} \ref{group of transformations} and {\it Proposition} \ref{proposicion on law group for the dynamics} hold for two-dimensional time parameter dynamics.
Most of the notions of dynamics discussed above for a $1$-time dynamics can be generalized in the context of dynamics with a $2$-time parameter. The existence of a continuous distance function can be taken over the product $J_1\times J_2$ with the product topology. However, the notion of ordered number field cannot take over the product $J_1\times J_2$.

Note that in {\it Definition} \ref{two time dynamics}, the number fields $\mathbb{K}_1$ and $\mathbb{K}_2$ are not related to each other. However,
an alternative definition of two-dimensional dynamics is a dynamics based upon the concept of a $2$-dimensional time parameter. For example, when $J$ is a $2$-dimensional subset of the field of complex numbers $\mathbb{C}$.  In this case, the field $\mathbb{K}$ and the subfield $J$ must be two dimensional. Most of the above notions of dynamics can also be constructed using a two-dimensional time parameter. again, a meaningful notion of time ordering is certainly non-trivial to be implemented.
\subsection{On time re-parametrization covariance}
Our notion of dynamics assumes the existence of time parameters, which are subsets $J$ of a number field $\mathbb{K}$. Given that these time parameters can be defined arbitrarily and given that they lack of observational or phenomenological interpretation and that there is no macroscopic observer attached to such parameters, it is natural to expect that in a consistent description of the theory, physical quantities must be independent of the choice of the time parameter for description of the fundamental dynamics. By this we mean that physical quantities are equivalent classes of mathematical objects $[\Upsilon_0]$ which are covariantly defined in the following sense: for every time re-parametrization $\theta:J\to \mathbb{K}$ there is at least two representatives $\Upsilon,\widetilde{\Upsilon}\in\,[\psi_0]$ such that $\widetilde{\Upsilon}(\theta (t))=\,\Upsilon(t)$ for $t\in\,J$. Then it is said that $\Upsilon$ and $\widetilde{\Upsilon}$ are equivalent. Therefore, although for the description of the dynamics the introduction of time parameters is necessary, the theory of fundamental dynamics must be constrained by time re-parametrization invariance: the dynamics must be covariant with respect to time parameter diffeomorphisms. Indeed, the dynamics of Hamilton-Randers models that we will developed is invariant under spacetime diffeomorphisms, that imply time reparametrization invariance.
\subsection{Notions of local reversible and local non-reversible dynamics}
\begin{definicion}
 The {\it time conjugated dynamics} associated with the dynamics $\Phi:J\times \,\mathcal{M}\to \mathcal{M}$ with $J\subset \mathbb{K}$ is a map
 \begin{align*}
 \Phi^c:J\times \,\mathcal{M}\to \mathcal{M}
 \end{align*}
 continuous in the product topology
  such that if $\Phi(t,A)=B$, then it must hold that $\Phi^c(t,B)=A$,
 for $(A,B)\in\, \mathcal{M}\times \,\mathcal{M}$.
 \label{conjugate dynamics}
\end{definicion}
Since $t_1+t_2=\,t_2+t_1$ for any pair of elements $t_1,t_2\in\mathbb{K}$, it follows that if $\Phi$ is a dynamics, then $\Phi^c $ is also a dynamics and both are defined using the same time parameter field $\mathbb{K}$.

Note that the idempotent property
\begin{align}
\left(\Phi^c\right)^c=\,\Phi
\label{idempotent property of the conjugate}
\end{align}
holds good.

An analogous construction can be applied to the associated dynamics $\Phi_\mathcal{E}$ acting on sections of the sheaf $\pi_\mathcal{E}:\mathcal{E}\to \mathcal{M}$.

\begin{definicion}
Let  $\Phi_{\mathcal{E}}:J\times \Gamma\mathcal{E}\to\,\Gamma\mathcal{E}$  be a dynamics, where $J\subset\,\mathbb{K}$ and $\mathcal{E}$ is a $\mathbb{K}$-vector space sheaf $\pi_{\mathcal{E}}:\mathcal{E}\to \,\mathcal{M}$. The conjugate dynamics is a map $\Phi^c_{\mathcal{E}}:J\times \Gamma\mathcal{E}\to\,\Gamma\mathcal{E}$ such that if $\Phi_\mathcal{E}(t,E_1)=E_2$, then $\Phi^c_{\mathcal{E}}(t,E_2)=\,E_1$ for every $E_1,E_2$ elements of $\mathcal{E}$.
\end{definicion}

For the conjugate dynamics $\Phi^c_\mathcal{E}$, the idempotent property
\begin{align}
\left(\Phi^c_\mathcal{E}\right)^c=\,\Phi_\mathcal{E}
\label{idempotent property of the conjugate 2}
\end{align}
holds good.

 In the category of topological spaces and continuous functions, a first candidate for the notion of reversible dynamics $\Phi$ could be to satisfy the condition
\begin{align*}
\lim_{t\to 0}\Phi(t,E)=\,\lim_{t\to 0} \Phi^c(t,E)
\end{align*}
for all $E\in\,\mathcal{E}$.
But this condition is always satisfied in the category of topological spaces and topological maps, if $\Phi$ (and hence $\Phi^c$) are dynamics, since
\begin{align*}
\lim_{t\to 0}\Phi(t,E)=\,\lim_{t\to 0} \Phi^c(t,E)=E.
\end{align*}
If one takes the difference between the values of the dynamics and conjugate of the dynamics as the fundamental criteria for non-reversibility, we found the condition
 \begin{align*}
\Delta:\Gamma\mathcal{E}\to \mathbb{K},\,E\mapsto &\lim_{t\to 0}\,\left(\Omega\circ \Phi_\mathcal{E}(t,E)-\,\Omega\circ \Phi^c_\mathcal{E}(t,E)\right)\\
& =\, \lim_{t\to 0}\,\left(\Omega\circ\Phi_\mathcal{E}(0,E)-\,\Omega\circ \Phi^c_\mathcal{E}(0,E)\right)\\
& =\, \Omega(E)-\,\Omega(E)=\,0,
 \end{align*}
since $\Phi_\mathcal{E}(0,E)=\,\Phi^c_\mathcal{E}(0,E)=\,E$. This result holds for any continuous induced dynamics $\Phi_\mathcal{E}$. But the category of topological spaces with continuous functions as a maps is the natural category where to formulate our mathematical models. By the arguments discussed above, continuity is an essential ingredient for determinism in the context that we are condisdering, and for the construction of deterministic models. Hence considering $\Delta$ only  does not allow to define a notion of local non-reversible law in the category of topological spaces.

Instead, we propose a notion of local irreversibility based upon the following construction. If $t_1, t_2\in J$ and are such that $d_{\mathbb{K}min}=\,d_\mathbb{K}(t_1,t_2)=\,\|t_1-t_2\|$.
 \begin{definicion}
 Let ${\Phi}:J\times \mathcal{M}\to \,\mathcal{M}$ be a dynamics over $\mathcal{M}$. The dynamics $\Phi$ is non-reversible if there is a sheaf $\pi_\mathcal{E}:\mathcal{E}\to \mathcal{M}$ and a function
 $\Omega:\,\mathcal{\mathcal{E}}\to \mathbb{K}$  such that for the induced dynamics $\Phi_{\mathcal{E}}:J\times\Gamma \mathcal{E}\to \Gamma\mathcal{E}$, the relation
 \begin{align}
 \Xi_\Omega:\Gamma\mathcal{E}\to \mathbb{K},\,E\mapsto\lim_{\|t\|\to d_{\mathbb{K}min}}\,\frac{1}{t}\,\left(\Omega\circ\Phi_\mathcal{E}(t,E)-\,\Omega\circ \Phi^c_\mathcal{E}(t,E)\right)\neq 0
 \label{definitionofnonreversibilityfunction}
 \end{align}
 holds good for all $E \in\, \Gamma\,\mathcal{E}$.

 A dynamics which is not non-reversible in the above sense will be called reversible dynamics;
a function $\Xi$ for which the condition \eqref{definitionofnonreversibilityfunction} holds will be called a {\it non-reversibility function}.
 \label{nonreversibledynamics}
 \end{definicion}
When $\mathcal{E}$ is equipped with a measure, the non-reversibility condition  \ref{definitionofnonreversibilityfunction} can be formulated for almost all $E\in\,\mathcal{E}$, that is, for all subsets in $E$ except possible subsets of measure zero.

For the above notion of non-reversibility to be well defined, it is a requirement that the number field $\mathbb{K}$ and the configuration space $\mathcal{M}$ must allow to define a notion limit $t\to\,d_{\mathbb{K}min}$ as it appears in the expression \eqref{definitionofnonreversibilityfunction} and also as it appears in the notion of incremental quotient limit, given by the expression \eqref{incremental quotient}. For example, the field of real numbers $\mathbb{R}$ and the discrete field of rational numbers $\mathbb{Q}$ have well defined notions of  the limit $t\to\,0$. Other examples is the case of the field of complex numbers $\mathbb{C}$. An example where these limits fail to be defined are the prime field $\mathbb {Z}/p\mathbb{Z}$ for $p$. In this case, the notion of limit when $t\to d_{\mathbb{K}min}=1$ must be well defined.

 If a dynamics $\Phi$ is non-reversible, Then there is a sheaf $\mathcal{E}$ over the configuration space $\mathcal{M}$ where a function $\Omega$ implies that the non-reversibility function $\Xi_\Omega$ is different from zero. The converse of these of these conditions characterize a reversible dynamics such that for any sheaf $\mathcal{E}$ and any function $\Omega$ as above, the non-reversible function $\Xi_\Omega$ is identically zero. Furthermore, for the definition of non-reversible dynamics given above, it is  theoretically easier to check if a given dynamics is non-reversible than to check if it is reversible, because in order to prove the first possibility first, it is enough to find a function  $\Omega:\mathcal{\mathcal{E}}\to \mathbb{K}$ such that the relation \eqref{definitionofnonreversibilityfunction} is satisfied, while for the case of a reversible dynamics, one needs to check that for all such functions $\Omega$, the non-reversibility $\Xi_\Omega$ is identically zero.

Let us consider the conditions by which $\Xi_\Omega\equiv 0$. Assume that the functions  $\Omega:\Gamma\mathcal{E}\to \mathbb{K}$ are in some sense determined by the given models, smooth in both entries. Such smoothness condition can be stated as the formal Taylor expressions
\begin{align*}
&\Omega\circ\Phi_\mathcal{E}(t,E)=\,\Omega \circ \Phi_\mathcal{E}(0,E)+\,t\,\Omega'\star d\Phi_\mathcal{E}(0,E)+\,\mathcal{O}(t^2),\\
&\Omega\circ \Phi^c_\mathcal{E}(t,E)=\,\Omega\circ\Phi^c_\mathcal{E}(0,E))+\,t\,\Omega'\star d\Phi^c_\mathcal{E}(0,E)+\,\mathcal{O}(t^2),\\
\end{align*}
where
\begin{align*}
\Omega'\star d\Phi_\mathcal{E}(0,E) & :=\,\frac{d{\Omega}(u)}{d u}\large|_{u=\Phi_\mathcal{E}(0,E)}\star \,\frac{d\Phi_{\mathcal{E}}(t,E)}{dt}\large|_{t=d_{\mathbb{K}min}}\\
& =\,\frac{d{\Omega}(u)}{d u}\large|_{u=E}\star \,\frac{d\Phi_{\mathcal{E}}(t,E)}{dt}\large|_{t=d_{\mathbb{K}min}},
\end{align*}
and
\begin{align*}
\Omega'\star d\Phi^c_\mathcal{E}(0,E) & :=\,\frac{d{\Omega}(u)}{d u}\large|_{u=\Phi^c_\mathcal{E}(0,E)}\star \,\frac{d\Phi^c_{\mathcal{E}}(t,E)}{dt}\large|_{t=d_{\mathbb{K}min}}\\
& =\,\frac{d{\Omega}(u)}{d u}\large|_{u=E}\star \,\frac{d\Phi^c_{\mathcal{E}}(t,E)}{dt}\large|_{t=d_{\mathbb{K}min}}.
\end{align*}
The $\star$-pairing is the natural pairing induced from $\Omega$ operating on $\Phi_\mathcal{E}(t,E)$ by the natural composition $\Omega\circ \,\Phi_\mathcal{E}$.
Therefore,  $\Xi_\Omega$ can be re-written formally as
\begin{align*}
\Xi_\Omega(E)& =\,\lim_{\|t\|\to d_{\mathbb{K}min}}\,\frac{1}{t}\,\left(t\,\Omega'\star \,d\Phi_\mathcal{E}(0,E)-\,t\,\Omega'\star \,d\Phi_\mathcal{E}(0,E))\right)\\
& =\,\Omega'\star\,d\Phi_\mathcal{E}(0,E)-\,\Omega'\star\,d\Phi^c_{\mathcal{E}}(0,E).
\end{align*}
Taking into account the above expressions, we have
\begin{align*}
\Xi_\Omega(E) =\frac{d{\Omega}(u)}{d u}\large|_{u=E}\star \,\frac{d\Phi_{\mathcal{E}}(t,E)}{dt}\large|_{t=d_{\mathbb{K}min}}
-\,\frac{d{\Omega}(u)}{d u}\large|_{u=E}\star \,\frac{d\Phi^c_{\mathcal{E}}(t,E)}{dt}\large|_{t=d_{\mathbb{K}min}},
\end{align*}
which is in principle, different from zero. This shows that the criteria to decide when a dynamics is non-reversible is well-defined in the category of topological spaces. One only needs to consider a map $\Omega$ which is in certain sense differentiable and calculate the above expression.

 Let us consider a non-reversible dynamics $\Phi:J\times \mathcal{M}\to \mathcal{M}$ such that for an associated dynamics $\Phi_\mathcal{E}$, there is a non-zero reversibility function $\Xi_\Omega\neq 0$. Because of the algebraic structures of the stalk $E_u$, it is possible to define the map
 \begin{align}
 Sym \Phi_\mathcal{E}:J\times \Gamma\mathcal{E}\to \,\Gamma \mathcal{E},\quad (t,E)\mapsto \,\frac{1}{2}\,\left(\Phi_\mathcal{E}(t,E)+\,\Phi^c_\mathcal{E}(t,E)\right).
 \label{symmetrization of the dynamics}
 \end{align}
For $Sym \Phi_\mathcal{E}$ the property \eqref{composition Phi for E} holds good. Furthermore, $\Xi_\Omega =0$ for $Sym \Phi_\mathcal{E}$, indicating that the operation of symmetrization in \eqref{symmetrization of the dynamics} is a form of reducing non-reversible dynamics $\Phi_\mathcal{E}$ to reversible dynamics $Sym \Phi_\mathcal{E}$. Repeating this procedure for any induced dynamics $\Phi_\mathcal{E}$ on each sheaf $\pi_\mathcal{E}:\mathcal{E}\to \mathcal{M}$, we can assume the existence of an induced dynamics $Sym\Phi$ that by construction is reversible. When such dynamics $Sym \Phi:J\times \mathcal{M}\to \mathcal{M}$ exists, it will be a reversible dynamics, that we call the time symmetrized dynamics.
\\
{\bf General form of the reversibility condition}.  If the $\star$-pairing is linear, then the non-reversibility function $\Xi$ is identically zero if and only if
\begin{align}
\Omega'\,\star \left(\,d\Phi_\mathcal{E}(0,E)\,- d\Phi^c_\mathcal{E}(0,E)\right) =0
\label{characterization of reversible dynamics}
\end{align}
holds good.  If the $\star$-pairing is in appropriate sense invertible, the relation \eqref{characterization of reversible dynamics} can be re-written formally as a necessary condition that depends only on the dynamical law,
\begin{align}
\frac{d\Phi_\mathcal{E}(t,E)}{dt}\Big|_{t=d_{\mathbb{K}min}}-\, \frac{d\Phi^c_\mathcal{E}(t,E)}{dt}\Big|_{t=d_{\mathbb{K}min}}\equiv \,0.
\label{characterization of reversible dynamics 2}
\end{align}
Either the condition \eqref{characterization of reversible dynamics} or the condition \eqref{characterization of reversible dynamics 2} can be taken as the necessary and sufficient condition for reversibility of a local dynamics.
\bigskip
\\
{\bf Notion of non-reversible dynamics in configuration spaces endowed with a measure}. For spaces endowed with a measure, the function $\Xi_\Omega:\Gamma\mathcal{E}\to \mathbb{K}$ can be non-zero except in a sub-set of measure zero. The relevant point is that,  for a given measure on $\mathcal{E}$, $\Xi_\Omega$ is non-zero almost everywhere during the evolution. If there is no such a function $\Omega:\Gamma\mathcal{E}\to \mathbb{K}$ for a dense subset in an open domain $\mathcal{U}_A\subset \mathcal{M}$ containing $A$, then one can say that the dynamics is reversible locally.

However,  if one speaks of strict non-reversible laws or strict reversible laws, namely, dynamical laws which are non-reversible (resp. reversible) in the whole configuration space $\mathcal{E}$, one can avoid the introduction of a measure in $\mathcal{E}$ as we did in our definition \ref{nonreversibledynamics}. This is the simplest way of introducing our notion of non-reversible/reversible local dynamics and we attach our treatment to such notion.

\subsection{Non-reversible dynamics and the second principle of thermodynamics}
In order to introduce the notion of thermodynamical system in our theory, we consider the following cartesian product of topological spaces,
\begin{align}
\widetilde{\mathcal{M}}=\,\prod^N_{k=1}\,\mathcal{M}_k,
\label{thermodynamical product space}
\end{align}
where each of the spaces $\mathcal{M}_k$ is by assumption the configuration space of a given dynamical system and $N$ is a large natural number. By large integer we mean that the following asymptotic characterization holds good: for any map $P:\widetilde{\mathcal{M}}\to \mathbb{K}$ that depends upon $N$,  then it must hold that
\begin{align}
 P[N]=\,P[N-1]+{o}(N^\delta),
 \label{asymptotic property}
 \end{align}
  with $\delta >0$. That is, we assume the condition
\begin{align*}
\lim_{N\to \,+\infty} \frac{P[N]-P[N-1]}{N^\delta}\to 0.
\end{align*}
A topological space $\widetilde{\mathcal{M}}$ with this asymptotic property is what we can appropriately call a {\it thermodynamical space}, since it embraces the interpretation of thermodynamical systems as the ones where it is possible to define local intensive and extensive functions of the whole system where fluctuations due to the detailed structure of the system can be with great approximation disregarded (see for instance \cite{Kondepudi Prigogine 2015}, chapter 15). For thermodynamical spaces, one can speak of thermodynamical sub-system as an embeddings  $\widetilde{\mathcal{M}}'\hookrightarrow \,\widetilde{\mathcal{M}}$ for which the asymptotic conditions \eqref{asymptotic property} holds. Furthermore, since $\widetilde{\mathcal{M}}$ is by definition a large product space, there is an statistical framework for statistical interpretations of the maps $\mathcal{P}:\widetilde{\mathcal{M}}\to \mathbb{K}$.
\begin{ejemplo}
Let $\mathcal{M}=\prod^N \mathcal{M}_k$ with each configuration space of the form $N_k \cong \mathcal{M}_k$. Then each of the spaces $\mathcal{M}_k$ is not a thermodynamical space, since if $\mathcal{M}\cong \mathcal{M}_k$, then $N=1$ and the above asymptotic property does not hold. This in agreement to the idea that a system composed by a single individual system is not a thermodynamical system.
\end{ejemplo}

Let $\widetilde{\mathcal{M}}_c$ be a classical phase space and $J\subset \mathbb{R}$ open.
The entropy in the classical equilibrium thermodynamical theory is a map of the form
\begin{align}
\Lambda_c:J\times \widetilde{\mathcal{M}}_c\to \mathbb{R},
\label{entropy function}
\end{align}
 such that
 \begin{enumerate}
 \item The function $\Lambda$ is extensive: for two thermodynamically different classical thermodynamical spaces $\widetilde{\mathcal{M}}_{1c}\hookrightarrow \mathcal{M}_c$, $\widetilde{\mathcal{M}}_{2c}\hookrightarrow \,\mathcal{M}_c$ corresponding to two sub-systems of the thermodynamical system $\mathcal{M}_c$,  then it must follow that
 \begin{align*}
 \Lambda_c\left(t,(u_1,u_2)\right)\geq \,\Lambda_c\left(t,u_1\right)+\,\Lambda_c\left(t,u_2\right).
 \end{align*}

 \item For any thermodynamical system, it is non-decreasing with time,
 \begin{align*}
\frac{d}{d t} \Lambda_c\left(t,u\right)\geq 0.
 \end{align*}

 \item For any thermodynamical subsystem described by $\widetilde{\mathcal{N}}_c$ subset of $\widetilde{\mathcal{M}}_c$, it is non-decreasing with time.

  \end{enumerate}
  This properties does not fully characterize the entropy function, but will serve for our purposes and they are  consistent with the properties of Boltzmann's classical and quantum $H$-function \cite{Tolman}.

  Let us proceed to provide a generalized form of the second principle of thermodynamics for general dynamical systems.
  We first extend the above  notion of entropy function by considering entropy functions as extensive maps of the form
  \begin{align*}
  \Lambda:J\times \widetilde{\mathcal{M}}\to \mathbb{K},
  \end{align*}
 where $\mathbb{K}$ is an ordered number field and $\widetilde{\mathcal{M}}$ is the product space of the form \eqref{thermodynamical product space} and such that the above properties $(1)-(3)$ of the classical entropy hold good for the function $\Lambda$. This notion of generalized entropy  can be applied to  {\it local entropy densities} $\Lambda_i$. The notion of local internal entropy density appear in the theory of linear non-equilibrium thermodynamics and are described for instance in \cite{Kondepudi Prigogine 2015}, chapter 15.
In this general context of thermodynamical systems and generalized form of entropy functions, the second principle of thermodynamics can be stated as follows:
\bigskip
\\
{\bf Generalized Second Principle of Thermodynamics}: {\it The dynamical change in the state describing the evolution of a thermodynamical system is such that the production of internal local density entropy function $\Lambda_i(u)$ for each of the $i=1,...,r$ thermodynamical sub-systems do not decrease during the time evolution.}
\\

In order to compare the  notion of local non-reversible dynamical law as understood in definition  \ref{nonreversibledynamics} with the notion of physical evolution according to the non-decrease of generalized entropy \eqref{entropy function} the second principle as a form of non-reversibility evolution, the first step is to discuss the relations between the relevant notions, namely, the non-reversibility function \eqref{definitionofnonreversibilityfunction} and the possible specification of the metric function generalized entropy function. Let us consider the section $E\in\,\Gamma \mathcal{E}$ of the sheaf $\pi_\mathcal{E}:\mathcal{E}\to \mathbb{K}$, $\Phi_\mathcal{E}:J\times \Gamma \mathcal{E}\to \mathcal{E}$ the induced local dynamics and $\Omega:\Gamma \mathcal{E}\to \mathbb{K}$ a scalar function such that $\Xi_\Omega$ is locally positive. It is very suggestive that the correspondence between the generalized entropy and the notion of non-reversible law should be of the form
\begin{align*}
\frac{d}{dt}\large|_{t=d_{\mathbb{K}min}}\Lambda (t,u)& \equiv \,\frac{d}{dt}\large|_{t=d_{\mathbb{K}min}}\left(\Omega\,\circ \Phi_\mathcal{E}(t, E (u))\,-\Omega\,\circ \Phi^c_\mathcal{E}(t,\, E (u))\right)\\
& =\Xi_\Omega(E(u)).
\end{align*}
In principle, this relation depends on the section $E:\mathcal{M}\to \mathcal{E}$ and the scalar map $\Omega$, revealing a different entropy function for each choice of the pair $(E,\Omega)$. For instance, the notion of local thermodynamical equilibrium states, given by the extremal condition  $\frac{d}{dt}\large|_{t=d_{\mathbb{K}min}}\Lambda (t,u)=0$, for a parameter $t\in J_0$ such that the equilibrium is found for $t=d_{\mathbb{K}min}$. But such equilibrium condition is not consistent with the dependence of $\Xi_\Omega$ on the section $E$ and also on the map $\Omega$. Also, equilibrium conditions of asymptotic maximum entropy will depend upon the section $E$ and the map $\Omega$, for a given dynamics $\Phi_\mathcal{E}$. But if the section $E$ is fixed by a physical principle, the relation between the formalism depends upon the choice of the function $\Omega$ only up to  a constant. As we will see below, this will be the case of quantum dynamics. In general, additional mathematical structure in $\Gamma \mathcal{E}$ is necessary to fix the section $E$.

There are two different notions of irreversibility. The first one is related with the non-reversible dynamics according to definition \ref{nonreversibledynamics}. Such a definition will be used for the dynamics of fundamental systems that we will investigate in this work as models for quantum systems beyond quantum mechanics. This non-symmetric property of the dynamics will be implemented in the mathematical formalism that we will developed in this work. The second one, is the usual notion of irreversibility based upon the notion of thermodynamical system and generalized entropy.
The  distinction between these two notions of irreversibility is important, since the notion of non-reversible dynamics can be applied to individual systems described by points of the configuration space $\mathcal{M}$, whereas  thermodynamical relations between thermodynamic variables are not justified at such level.

\subsection{Examples of reversible and non-reversible dynamics}
 An interesting and relevant example of non-reversible dynamics is the following,
 \begin{ejemplo}
 Let $(M,F)$ be a Finsler space   \cite{BaoChernShen}, where $M$ is an $m$-dimensional manifold and $\mathbb{K}$ is the field of real numbers $\mathbb{R}$. For a point $A$ and for a variable point $X$ infinitesimally close to $A$, let us consider the limits
 \begin{align*}
  &\lim_{s\to 0}\,\frac{1}{s}\,\Omega(\Phi_s(X))=\,\frac{d}{ds}\big|_{s=0}\,\int^{X(s)}_A\,F^2(\gamma,\dot{\gamma})\,dt,\quad\\
   &\lim_{s\to 0}\,\frac{1}{s}\,\Omega(\Phi^c_s(X))=\,\frac{d}{ds}\big|_{s=0}\,\int^A_{X(s)} \,F^2(\tilde{\gamma},-\dot{\tilde{\gamma}})\,dt.
 \end{align*}
  The dynamics $\Phi$ is given by the geodesic flow of $F$. The parameterized curves $\gamma$ and $\tilde{\gamma}$ are geodesics,
 where $\gamma$ is a curve joining $A$ and $X$ and realizing the minimal length for curves joining $X$ and $A$ (and analogously for the inverted geodesic $\tilde{\gamma}$).  The local existence of such geodesics is guaranteed by Whitehead theorem \cite{Whitehead1932}. However, these pair of geodesics are not related as they are in Riemannian geometry by a relation of the form $\tilde{\gamma}(s)=\,\gamma(1-s),\,s\in [0,1]$, since in general the Finsler metric $F$ is not necessarily a {\it reversible Finsler metric}.

  Let us  consider the expression
  \begin{align*}
 \Xi_\Omega(A) & =\,\lim_{s\to 0}\,\frac{1}{s}\,\Omega(\Phi_t(X))-\,\lim_{s\to 0}\,\frac{1}{s}\,\Omega(\Phi^c_t(X)) \\
 & = \,\frac{d}{ds}\big|_{s=0}\,\left\{ \int^{X(s)}_A\,\left(F^2(\gamma,\dot{\gamma})-\, F^2(\tilde{\gamma},-\dot{\tilde{\gamma}})\right)\,dt\right\},
 \end{align*}
 It follows that
 \begin{align}
 \lim_{X\to A}\, \Omega(X)-\Omega^c(X^c)=\,F^2(A,V)-\, F^2(A,-V)
 \label{FmenosF}
 \end{align}
 along a given smooth $\mathcal{C}^1$ curve $\gamma:I\to M$ passing through $A$ and $B$ and such that
 \begin{align*}
 A=\,\gamma(0)=\,\lim_{X\to A}\,\tilde{\gamma}, \quad V=\,\dot{\gamma}(0)=\,-\lim_{X\to A}\,\dot{\tilde{\gamma}}.
 \end{align*}
 For a generic Finsler metric, the limit \eqref{FmenosF} is different from zero. This is also true for geodesics and since the geodesics are determined locally by the initial conditions, via \cite{Whitehead1932}, we have that the geodesic dynamics $\Phi$ is in this case non-reversible with
 \begin{align}
  \Xi_\Omega(A)=\,F^2(A,V)-\, F^2(A,-V),
 \end{align}
which is non-zero for almost all $(A,V)\in TM$.

This example is related with the notion of {\it reversibility function} in Finsler geometry \cite{Rademacher}, which is a measure of the non-reversibility of a Finsler metric. However, we consider that the functionals defining the dynamics are given in terms of $F^2$ instead of $F$. The associated geodesics are parameterized geodesics and the corresponding Hamiltonian is smooth on the whole tangent space $TM$.
 \label{ejemplononreversibledynamics}
 \end{ejemplo}
\begin{ejemplo}
 The construction of {\it Example} \ref{ejemplononreversibledynamics} in the case when $(M,F)$ is a Riemannian structure provides an example of reversible dynamics. In this case the dynamics is the geodesic flow and the formal limit
  \begin{align*}
  \lim_{X\to A}\frac{1}{d(A,X)}(\Omega(X)-\Omega(X))=0
    \end{align*}
    along $\Phi$ holds good.
    This is also true for the more general case of {\it reversible Finsler metrics}.
\end{ejemplo}
 \begin{ejemplo}
 Let us consider a quantum system, where pure states are described by elements $|\Psi\rangle$ of a Hilbert space $\mathcal{H}$. In the Schr\"odinger picture of dynamics\footnote{Let us note that the Schr\"odinger picture of dynamics can also applied to quantum field models, as discussed for instance by Hatfield \cite{Hatfield}.}, the evolution law is given by
 \begin{align}
 \Phi_\mathcal{H}(t,|\Psi (0)\rangle ) =\,\mathcal{U}(t)|\Psi (0)\rangle,
 \label{Dynamical law in quantum mechanics}
 \end{align}
where $\mathcal{U}(t)$ is the evolution operator \cite{Dirac1958}. The configuration space is the manifold of real numbers $\mathbb{R}$. We consider the sheaf over $\mathbb{R}$ with stalk the Hilbert space $\mathcal{H}$. Thus the sheaf we consider is of the form $\mathcal{E}\equiv\,\mathcal{H}\times \mathbb{R}$. The sections of the sheaf correspond to {\it arbitrary time evolutions} $t\mapsto |\Psi(t)\rangle$, not only Schr\"odinger's time evolutions. Let us consider first the function $\Omega$ is defined by the relation
\begin{align}
\Omega:\Gamma (\mathcal{H}\times \mathbb{R})\to \mathbb{C},\,\, |\Psi(0)\rangle\mapsto |\langle \Psi(0)|\Phi_\mathcal{H}(t,|\Psi (0)\rangle )\rangle  |^2=\,|\langle\Psi(0)|\mathcal{U}(t)|\Psi (0)\rangle |^2.
\label{Omegaquantum}
\end{align}
If the Hamiltonian $\hat{H}$ of the evolution is hermitian, then the conjugate dynamics $\Phi^c_\mathcal{H}$ is given by the expression
\begin{align*}
 \Phi^c_\mathcal{H}(t,|\Psi (0)\rangle )=\,\mathcal{U}^\dag (t)|\Psi (0)\rangle.
\end{align*}
Adopting these conventions, $\Xi_\Omega$ as defined in \eqref{definitionofnonreversibilityfunction} can be computed to be  trivially zero,
\begin{align*}
&\lim_{t\to d_{\mathbb{K}min}}\,\frac{1}{t}\,\left(|\langle \Psi(0)|\Phi_\mathcal{H}(t,|\Psi (0)\rangle )\rangle  |^2 -\,|\langle \Psi(0)|\Phi^c_\mathcal{H}(t,|\Psi (0)\rangle )\rangle  |^2\right)\\
& =\,\lim_{t\to d_{\mathbb{K}min}}\,\frac{1}{t}\,\left(|\langle\Psi(0)|\mathcal{U}(t)|\Psi (0)\rangle|^2 -\,|\langle\Psi(0)|\mathcal{U}^\dag(t)|\Psi (0)\rangle|^2\right)=0.
\end{align*}
Therefore, if we choose $\Omega$ as given by \eqref{Omegaquantum}, then unitary quantum mechanical evolution implies that the function $\Xi_\Omega= 0$. This choice for $\Omega$ is natural, since \eqref{Omegaquantum} is the probability transition for the possible evolution from the initial state $|\Psi (0)\rangle$ to itself by the $\mathcal{U}(t)$ evolution.

That $\Xi_\Omega =0$ as discussed above, does not imply that the dynamics $ \Phi_\mathcal{H}$ is reversible, since it could happen that the choice of another function $\widetilde{\Omega}$ is such that $\Xi_{\widetilde{\Omega}}\neq 0$. Indeed, let us consider instead the function
\begin{align}
\widetilde\Omega:\Gamma (\mathcal{H}\times \mathbb{R})\to \mathbb{C},\quad  |\Psi(0)\rangle\mapsto \langle \Psi(0)|\Phi_\mathcal{H}(t,|\Psi (0)\rangle )\rangle  =\,\langle\Psi(0)|\mathcal{U}(t)|\Psi (0)\rangle .
\label{Omegaquantum2}
\end{align}
Then it turns out that
\begin{align}
\Xi_{\widetilde{\Omega}}=\,2\,\imath\,\,\langle \Psi(0)|\hat{H}|\Psi(0)\rangle,
\end{align}
which is in general non-zero. Therefore, the dynamics of a quantum system is in general non-reversible.
\end{ejemplo}
\begin{ejemplo}The non-reversibility of quantum processes is usually related with the notion of entropy. There are several notions of entropy in quantum theory, but the one relevant for us here is the notion that emerges in scattering theory. Indeed,
a version of the $H$-theorem for unitary quantum dynamics can be found in the textbook from S. Weinberg, \cite{Weinberg1995} section 6.6. The entropy function is defined as
\begin{align}
S\,:=\,-\int \,d\alpha\,P_\alpha\,\ln\,(P_\alpha/c_\alpha),
\label{H-theorem}
\end{align}
where in scattering theory, $P_\alpha \,d\alpha$ is the probability to find the state in a volume $d\alpha$ along the quantum state $|\Psi_\alpha\rangle$ and $c_\alpha$ is a normalization constant.
The change of entropy can be determined  using scattering theory,
\begin{align}
\nonumber -\frac{d}{dt}\left\{\int \,d\alpha\,P_\alpha\,\ln\,(P_\alpha/c_\alpha)\right\}=\, -\int d\alpha\,&\int d\beta\,\left(1+\,\ln (P_\alpha/c_\alpha)\right)\\
&\left( P_\beta\,\frac{d\Gamma(\beta\to \alpha)}{d\alpha}-\,P_\alpha\,\frac{d\Gamma(\alpha\to\beta)}{d\beta}\right).
\label{derivative S}
\end{align}
As a consequence of the unitarity of the $S$-matrix, it can be shown that the change with time of entropy \eqref{H-theorem} is not decreasing.

 Although appealing, one can cast some doubts that $dS/dt$ can be interpreted as  a function of the form $\Xi_\Omega$. The first is on the comparison between the formal expressions \eqref{definitionofnonreversibilityfunction} for $\Xi_\Omega$ and \eqref{derivative S} for the derivative $dS/dt$. The difficulty in this formal identification relies on the factor $\left(1+\,\ln (P_\alpha/c_\alpha)\right)$, which is a short of non-symmetric.

  The second difficulty on the formal identification of \eqref{definitionofnonreversibilityfunction} with the derivative \eqref{derivative S} is at the interpretational level. This is because the derivation of the $H$-theorem presented by Weinberg has two foreign aspects to the spirit of the theory of non-reversible dynamics described above. The first is that to interpret the transitions amplitudes $\frac{d\Gamma(\alpha\to\beta)}{d\beta}$, etc... as measurable decay rates, a foreign macroscopic arrow of time must be included. That is, an oriented macroscopic time parameter is introduced in the formalism associated with the macroscopic experimental setting. Also, the interpretation of the relation \eqref{H-theorem} is for an ensemble of identical particles, indicating that beneath this {\it H-theorem} there is indeed an statistical interpretation and cannot be applied to the detailed dynamics of an individual quantum system.

Despite these dissimilarities, we can think that a reformulation of the $S$-entropy as it appears in scattering theory can be reformulated more symmetrically, such that it can be cast in the form of a suitable non-reversibility function.
\end{ejemplo}
\begin{ejemplo}
Let us consider the case of physical systems where the weak sector of the electroweak interaction of the Standard Model of particle physics is involved. The weak interaction slightly violates the $CP$-symmetry, as it is demonstrated in experiments measuring the decay rates of the $K^0$-$\bar{K^0}$ systems. Assuming that the $CPT$-theorem of relativistic quantum field theory holds \cite{Weinberg1995}, since such experiments show a violation of the $CP$-symmetry, as shown by measurements of the respective decay rates, then the $T$-symmetry invariance of the $S$-matrix must appear (slightly) violated in such experiments. The decay rates of the $K^0$-$\bar{K^0}$ systems show that the transition amplitudes for processes after the action of the $T$-inversion operation could be different than the direct amplitudes (prior to the action of the $T$-inversion operation).

 However, this non-reversibility character of the weak interaction is an indirect one and does not correspond to our notion of non-reversible dynamical law. In particular, it relies in an external element of the dynamics, namely, the existence of an external observer with a notion of macroscopic time. These elements are foreign to our notion of non-reversible dynamics, definition \ref{nonreversibledynamics}.
 \end{ejemplo}

\subsection{Time arrow associated with a non-reversible dynamics}
 The notion of reversible and non-reversible dynamical law that we are considering is attached to the possibility of defining a function $\Omega:\Gamma\mathcal{E}\to \mathbb{K}$ such that the property \eqref{definitionofnonreversibilityfunction} holds good. If one such function $\Omega$ is found, then there is a dynamical arrow of time defined by the following criteria:
 \begin{definicion}
 Given a non-reversible dynamics $\Phi_\mathcal{E}:J\times \Gamma\mathcal{E}\to \Gamma\mathcal{E}$ and a function $\Omega:\Gamma \mathcal{E}\to \mathbb{K}$ such that the corresponding non-reversibility $\Xi_\Omega$ is different from zero, a {\it global dynamical arrow of time} is a global choice of a sheaf section  $E\in\,\Gamma\mathcal{E}$ such that $\Xi_\Omega$ is non-zero.
 \end{definicion}

 Given a non-reversible local dynamics as in {\it definition} \ref{nonreversibledynamics}, the sign of the function \eqref{definitionofnonreversibilityfunction} is well defined at least on local open domains of $J\times \mathcal{E}$. Therefore, when the above definition holds at a local level, then one can speak of a {\it local dynamical arrow of time}.
 On the other hand, the existence of a dynamical arrow of time with constant sign defined in the whole configuration space $\mathcal{M}$ is a non-trivial requirement. For example, in the case of Finsler structures, it is not guaranteed that the sign in the difference \eqref{FmenosF} is keep constant on the whole tangent manifold $TM$. However, the difference \eqref{FmenosF} is a continuous function on $TM$. Hence one can reduce the domain of definition of the local dynamical time arrow to the open set where the difference is positive.

In case there is a measure defined, the relevant fact is that the non-reversibility function $\Xi_\Omega$ is different from zero for many evolutions.

 The notion of fundamental time arrow as discussed above immediately raises the problem of the coincidence or disagreement with the entropic or thermodynamical time arrow, indicated by evolutions leading to non-decrease of entropy. For a general local dynamics, a notion of time arrow based upon the non-reversible dynamics and the notion of time arrow based upon the increase of an entropy function will in general not coincide, since the former can change direction after crossing the condition $\Xi_\Omega =0$. In the case when $\dim (\mathbb{K})=1$, if $\Xi_\Omega  \neq 0$,  the {\it turning points} are the sections where $\Xi_\Omega =0$. Such domains correspond to points where the time arrow associated with a non-reversible local dynamics can change direction with respect to the time arrow based on the evolution of the internal entropy functions, which is always non-decreasing for arbitrary systems and for arbitrary thermodynamical sub-systems (see for instance \cite{Kondepudi Prigogine 2015}, section 3.4).

  Another difference between the dynamical arrow of time and the thermodynamical arrow of time appears if $\dim(\mathbb{K})> 1$. Then there is no a prescribed way to associate the arrow of time of a non-reversible dynamics with the arrow of time of an entropy function.

  Because the same arguments as above, we have the following
  \begin{proposicion}
  For any local dynamical law $\Phi_\mathcal{E}:J\times \Gamma\mathcal{E}\to \Gamma\mathcal{E}$ where $J\subset \,\mathbb{K}$, if there is an entropy function defined on $\mathcal{E}$, then the dynamical arrow of time associated with $\Phi_\mathcal{E}$ and a non-zero non-reversibility function $\Xi_\Omega$ coincides locally with the entropic arrow of time.
  \end{proposicion}

The above mentioned derivation of a version of $H$-theorem in scattering theory shows how a general unitary quantum local dynamics provides  a mathematical entropy function \eqref{H-theorem} and the corresponding arrow of time, which does not relies on Born's approximation \cite{Weinberg1995}. Despite the interpretation issues and the lack of a perfectly close interpretation of the derivative as a non-reversibility function, this argument can be extended to other deterministic dynamics via unitary Koopman-von Neumann theory \cite{Koopman1931, von Neumann}.

\subsection{Non-reversibility of the fundamental dynamics}
There are three independent arguments in favour of the non-reversibility of local dynamics as candidates for the fundamental dynamics.
\bigskip
\\
{\bf Non-reversibility is generic}. The first argument refers to the generality of non-reversible metrics respect to reversible dynamics. Given any configuration space $\mathcal{M}$, the collection of non-reversible dynamics as discussed above is larger and contains the set of reversible ones. Let us consider a reversible dynamics $\Phi_r$. Then for any function $\Omega:\Gamma\mathcal{E}\to\,\mathbb{K}$ we have
     \begin{align}
  \lim_{t\to d_{\mathbb{K}min}}\,\frac{1}{t}\,\left(\Omega(\Phi_r(t,A))-\Omega(\Phi^c_r(t,A))\right)=0.
  \label{reversibility condition}
    \end{align}
    But such condition is easy to be spoiled. Almost any deformation $\widetilde{\Phi}$ of a reversible dynamics $\Phi_r$ will admit a function $\widetilde{\Omega}$ such that
     \begin{align*}
      \lim_{t\to d_{\mathbb{K}min}}\,\frac{1}{t}\,\left(\widetilde{\Omega}(\widetilde{\Phi}_\mathcal{E}(t,A))-\widetilde{\Omega}(\widetilde{\Phi}^c_\mathcal{E}(t,A))\right)\neq 0.
     \end{align*}
This argument not only shows that it is more natural to consider non-reversible dynamics, but also that it is more stable the condition of non-reversibility than reversibility under perturbations of the dynamics.

In several situations, including the one that we shall consider in Hamilton-Randers theory in this work, given a non-reversible dynamics, an associated reversible dynamics can be constructed by a process of {\it time symmetrization} as discussed before. Hence a reversible dynamics can be seen as a class of equivalence of non-reversible dynamics. The notion of symmetrization  can be applied to both, continuous and discrete  dynamics.  Such a process is information loss, that is, for many non-reversible dynamics, there is an unique reversible dynamics. Thus the reversible dynamics obtained by symmetrization from a non-reversible dynamics hides an intrinsic non-reversibility in the same process of symmetrization.

If we further assume that any reversible dynamics can be obtained after a time symmetrization  from a (non-unique) non-reversible dynamics, then it is clear that the category of non-reversible dynamics is a natural choice to formulate a fundamental dynamics, since all the operations that one can perform in the reversible dynamics can be obtained from a slightly different operations in the non-reversible version of the dynamical system.

 We shall make these constructions explicit later, when we introduce the notion of Hamilton-Randers dynamical system and the time symmetrization operation in that context. Indeed, this view of non-reversible nature dressed by reversible laws obtained by process of symmetrization is on the core of our philosophy of emergence of physical description.
\bigskip
\\
{\bf Finsler type structures as models for non-reversible dynamical systems}. A second argument supporting the non-reversible character of a fundamental dynamics, at least if continuous models are considered, is based on Finsler geometry. One of the important characteristics of Finsler geometry is its ubiquity on the category of differentiable manifolds and differentiable maps. Finsler structures are natural objects in the sense that they can be defined on any manifold ${M}$ with enough regularity and with some few additional natural conditions on the manifold topology (Hausdorff and paracompact manifolds). This can be seen clearly for metrics with Euclidean signature, where a Finsler structure of Randers type \cite{Randers,BaoChernShen}. A Randers type metric is a small perturbation of a Riemannian structure. Any Haussdorf, paracompact manifold admits a Riemannian structure \cite{Warner}. If  additional conditions to ensure the existence of globally defined bounded and smooth vector fields are imposed on $M$, then Randers metrics can be constructed globally on the manifold. The conditions for the existence of such a vector field are rather weak.

In order to extend this argument to the case of interest for us, namely, spacetimes as described in {\it chapter 3}, one needs to have an extension of the averaging operation for Finslerian type structure of indefinite signatures. This is also an open problem in the theory of Finsler spacetimes in general\footnote{Because the structure of the spaces that we will consider as products of Lorentzian manifolds, it is only necessary to metrics with Lorentzian signature.}. The case of spacetime structures of Randers type is specially relevant for our theory of chapter 3. A Lorentzian metric is linearly perturbed to obtain a Randers type metric \cite{Randers} by the introduction  of a small vector field in the metric structure of the spacetime\footnote{This notion of Lorentzian Randers spacetime has been discussed and critizised in \cite{Ricardo Randers-Lorentz}. We still denote them as Finslerian spacetimes.} the deformation vector field is time-like, then there is defined an associated Finsler  metric by a process similar to the one that associates to a Lorentzian metric a Riemannian metric \cite{Hawking Ellis 1973}.
\subsection{Quasi-metric structures}
Following  the above arguments in favour of the non-reversibility property in fundamental dynamical models and taking into account what we have learn from the Finslerian Example \ref{ejemplononreversibledynamics}, it is natural to investigate the possibilities to endow the configuration space $\mathcal{M}$ with a {\it quasi-metric structure} associated with the dynamics. Our notion of quasi-metric structure is  a direct generalization of the analogous notion found in the literature \cite{Javaloyes et al.,Wilson1931} and is defined as follows:
\begin{definicion} Let ${\bf T}$ be a set and $\mathbb{K}$ an ordered number field.
A quasi-metric is a map $\varrho:{\bf T}\times {\bf T}\to \mathbb{K}$ such that
    \begin{enumerate}
    \item $\varrho(u,v)\geq 0$, for each $u,v\in\, {\bf T}$,
    \item $\varrho(u,v)=0$, iff $u=v\in \,{\bf T}$,
    \item $ \varrho(u,w)\leq \varrho(u,v)+\,\varrho(v,w)$, for each $u,v,w\in \,{\bf T}$.
    \end{enumerate}
\label{quasi-metric}
\end{definicion}
The main conceptual difference between a quasi-metric and a topological metric function  is that in the case of a quasi-metric, the
symmetry condition of metrics has been dropped out. Apart from this, we have adopted the more general case where the number field $\mathbb{K}$ is a general ordered number field.

The notion of quasi-metric $\varrho$ is linked with the non-reversibility of $1$-dimensional dynamical laws  through its potential relation with the two possible directions of evolution. When a dynamics $\Phi:\mathcal{M}\times \mathbb{K}\to \,\mathcal{M}$ is given and the evolution has a geometric interpretation (as in the example of Finsler geometry), then it will be defined a quasi-metric function $\varrho: \mathcal{M}\times \mathcal{M}\to \mathbb{K}$, where  $\varrho(A,B)$ will be  associated with the evolution from $A$ to $B$ following the dynamics $\Phi:\mathbb{K}\times \mathcal{M}\to \mathcal{M}$, while  $\varrho(B,A)$ will be associated with the time conjugate dynamics $\Phi^c$.

The above generic idea can be pursue in the following way. By a geometric dynamics, let us understood a filling dynamics $\Phi:J\,\times \mathcal{M}\to \mathcal{M}$.
Then for every $x,y \in \,\mathcal{M}$, the condition $\Phi(t_a,x)=y$ has at least a solution for a finite $t_a\in\,\mathbb{K}$.  Let us consider the map
\begin{align}
\varrho_\Phi:\mathcal{M}\times\,\mathcal{M}\to \,\mathbb{K},\quad (x,y)\mapsto \,\min \{t_a\in\,\mathbb{K},\,\,s.t.\,\,\varrho_\Phi (t_a,x)=\,y\}.
\label{varrho de phi}
\end{align}
\begin{proposicion} Let $\Phi:J\times \mathcal{M}\to \,\mathcal{M}$ be a mixing dynamics.
Then map $\varrho_\Phi :\mathcal{M}\times \mathcal{M}\to \,\mathbb{K}$ defined by \eqref{varrho de phi} is a quasi-metric.
\label{Proposicion varrho Phi}
\end{proposicion}
Since the dynamics $\Phi$ is mixing the map $\varrho_\phi$ is defined for each pair of points of $\mathcal{M}$.

If $\Phi (t_a,x)=y$, then it happens that $\Phi^c(t_a,y)=\,x$. But in general, $\Phi^c \neq \,\Phi$ and therefore,  $\Phi (t_a,y)\neq \,x$, that is translate to the condition of non-symmetricity of the associated metric $\varrho_\Phi$.

This notion of quasi-metric depends on the definition of the time parameter $J\subset \,\mathbb{K}$. In this sense, $\varrho_\Phi$ does not have a geometric character, except if the parameter used has a geometric character, as it could be, associated with a Finslerian norm in the case of geodesic dynamics in Finsler geometry.

From a  quasi-metric defined on ${\bf T}$ one can construct a genuine topological metric on ${\bf T}$ by a process of symmetrization,
\begin{align}
\varrho_+ :{\bf T}\times {\bf T}\to \mathbb{K}, \quad (u,v)\mapsto \frac{1}{2}\,\left(\varrho(u,v)+\,\varrho(v,u)\right).
\label{symmetrizedmetric}
\end{align}
Apart from the quasi-metric axioms $(1)$ to $(3)$ in {\it definition} \ref{quasi-metric}, $\varrho_+$ is symmetric.
There is also associated the skew-symmetric function,
\begin{align}
\varrho_- :{\bf T}\times {\bf T}\to \mathbb{K}, \quad (u,v)\mapsto \frac{1}{2}\,\left(\varrho(u,v)-\,\varrho(v,u)\right).
\label{asymmetrizedmetric}
\end{align}
$\varrho_-$ has the following immediate properties:
\begin{itemize}
\item $\varrho_- (u,v)$ can be positive (bigger than $=0$, negative (lesser than $0$) or zero.
\item $\varrho_- (u,v)= 0$ if and only if $\varrho(u,v)=\varrho (v,u)$, with $u,v\in\,{\bf T}$.
\item $\varrho_-(u,u)=0$, for all $u\in {\bf T}$.
\item It holds that $\varrho_- \,\leq \varrho_+ .$
\item The following inequality holds:
\begin{align}
\varrho_-(u,v)-\varrho_-(v,w)\,\leq \varrho_+(u,v) +\varrho_+(v,w),\quad \forall \,u,v,w\in {\bf T}.
\end{align}
\end{itemize}
Conversely, if two structures $\varrho_+$ and $\varrho_-$ are given with the above properties and relation, a quasi-metric $\rho$ is determined jointly by $\rho_+$ and $\rho_-$ by the relation
\begin{align}
\varrho =\,\varrho_+ +\,\varrho_- .
\label{relatiorhorhomasrhomenos}
\end{align}

Given  a quasi-metric structure, one can define  topologies on ${\bf T}$ by considering the topological basis composed by {\it forward or backward open balls}. In general, the two topologies do not coincide. Another topology is determined by considering the basis for the topology composed by the intersection of open and forward balls with the same center and same radius. This third topology is finer that the forward and backward topologies.

A quasi-isometry is an homeomorphism $\theta :{\bf T}\to {\bf T}$ that preserves the triangular function
\begin{align}
T:{\bf T}\times {\bf T}\times {\bf T}\to \mathbb{K},\quad (x,y,z)\mapsto \,\varrho(x,y)+\varrho(y,z)-\varrho(x,z).
\label{triangular function}
\end{align}
This is a direct generalization of the notion of almost-isometry discussed in  \cite{Javaloyes et al.}. Also, in terms of the triangular function \eqref{triangular function}, one can characterize minimal length oriented curves.

Given the topology induced by $\rho$ on ${\bf T}$, one can consider the corresponding notion of compact set.

\subsection{Causal structure}
A further element that we need to consider for the general formulation our dynamical systems is the existence of causal structures. Causal structures, in an ample sense that we will formalize here, refers to the minimal mathematical structure of causation.
\begin{definicion}
Let us consider the configuration space $\mathcal{M}$. A causal structure is a collection
\begin{align*}
\mathcal{C}_{\mathcal{M}}:=\left\{\mathcal{C}_u,u\in\tilde{\mathcal{M}}\subset\,\mathcal{M}\right\}
\end{align*}
of subsects $\mathcal{C}_u\subset \,\mathcal{M}$ such that the relation
\begin{center}
$u_1 \sim u_2$ if $u_2\in\mathcal{C}_{u_1}$,
\end{center}
meets the following properties:
\begin{itemize}
\item {\bf Symmetry}: If $u_1 \sim u_2$, then $u_2\sim u_1$.

\item {\bf Reflexive}: For every $u\in \tilde{\mathcal{M}}\subset \,\mathcal{M}$, $u\sim  u$.

\item {\bf Transitive}: If $u_1\sim u_2$ and $u_2\sim u_3$, then $u_1 \sim u_3$.

\item {\bf Diamond property}: For every $u_1,u_2\in \,\tilde{\mathcal{M}}$, $\mathcal{C}_{u_1}\cap \mathcal{C}_{u_2}$ is compact.
\end{itemize}
\label{definicion estructura causal}
\end{definicion}

\begin{definicion}
 Given $u\in \mathcal{M}$, the set $\mathcal{C}_u$ is the causal cone associated with $u$. If $v\in\,\mathcal{C}_u$, then $u$ and $v$ are causally connected.
 \end{definicion}
 It is clear that if $u$ is causally connected with $v$, then $v$ is causally connected with $u$, by the symmetric property in the relation defined by \ref{definicion estructura causal}.

Our notion is directly inspired by the concept of causal cone of Lorentzian geometry. However, {\it definition} \ref{definicion estructura causal} applies also to theories where the configuration manifold $\mathcal{M}$ is discrete.
Note that in the definition of quasi-metric, an ordered field $\mathbb{K}$ is required, while in our {\it definition} \ref{definicion estructura causal}, only a topology defining the compact sets is required. In particular, there is no notion of time or time ordering in our definition of causal structure.
\begin{ejemplo}
An example of causal structure is the case when $\mathcal{M}=\,M_4$ is a smooth manifold endowed with a Lorentzian metric and $\mathcal{C}_{\mathcal{M}}$ is the corresponding causal structure determined by the light cones. In general relativity, the term {\it light cone} have attached two meanings: one as a subset of the spacetime manifold $M_4$, the second as subset of the tangent space $TM_4$. The corresponding meaning of {\it definition} \ref{definicion estructura causal} is analogous  to the first one attached in general relativity.
\end{ejemplo}

The characteristic {\it property} in {\it definition} \ref{definicion estructura causal} is the {\it Diamond property}. The physical reason to adopt such assumption is based upon the corresponding property of global hyperbolic spacetimes \cite{Hawking Ellis 1973}. Furthermore, it suggests that the effects at a given point originated at another point, are determined by a compact set. This is a way to express finiteness in the possible actions for systems located at a given point of the configuration space.

Given a compact configuration space $\mathcal{M}$, there is always defined a causal structure if we make $\mathcal{C}_u=\,\mathcal{M}$. This structure corresponds to the situation where all the points of the configuration space are causally connected.
 In the case that  $\mathcal{C}_{u}\neq \mathcal{M}$, then the causal structure is non-trivial. Obstructions will make $\mathcal{C}_u=\,\mathcal{M}$ in general not possible, implying a non-trivial causal  structure.

The possibility of transfer information in the form of a physical signal from a point $x_1$ to a point $x_2$ of a spacetime manifold can only happen if they are causally connected.  The extension of this notion to a general configuration space $\mathcal{M}$ can be stated in the following grounds: the influence on the evolution of a point $A\in\mathcal{M}$ could depend only upon the points that are causally connected with $A$. It does not necessarily depends on all such a points.

\subsection{Assumptions for the fundamental dynamical systems}
The theory to be developed in this work defends that many of the concepts and notions usually encountered in classical field theory and quantum mechanics have an emergent nature. They can be useful to describe our experiences, but are not linked with the natural notions of would be present in a fundamental description of physical reality beyond quantum mechanics. Therefore, the problem of identifying a set of fundamental principles and notions arises. The problem is considerably and has no way to be resolved without a large dosis of speculation and luck, since the human sphere of experiences does not continues directly such fundamental sphere. Quantum level of description is in-between. Therefore, an  attempts to construct such a speculative vision pass by re-covering quantum description in a consistent way.

The assumptions discussed below will be considered and implemented in our proposed theory of fundamental dynamical systems describing the physics at the conjectured fundamental scale. Rather than a set of formal axioms, the assumptions must be considered as requirements that constrain the mathematical framework for the dynamical systems. Basically, the assumptions  constrain the dynamical systems to be causal, deterministic, local  (in certain sense that will be explained later) and non-reversible.

 The assumptions are categorized in three different groups, according to the level of description, apart from {\it Assumption 0}, regarding the existence of a deeper level of physical description than the offered by quantum mechanics. Such fundamental scale is deeper in the sense that the quantum mechanical description of physical observable processes is reconstructed as an emergent, effective description from fundamental degrees of freedom. As such, the fundamental scale assumption it does not refer to an energy scale, or time scale, which are also emergent concepts.
\begin{itemize}
\item  {\bf The fundamental scale assumption}.
\bigskip
\\
A.0. There is a fundamental scale in physical reality which deeper than any scale associated with quantum mechanical system or classical field theory system.
\end{itemize}
Here the term {\it scale} refers a limit of a new regime for physical reality. This could correspond to a fundamental scale in length or time, but not necessarily in terms of a fundamental scale in mass. The term {\it quantum mechanical system} refers to systems where the use of quantum dynamics is mandatory. For example, elementary particle systems, quantum field models, nuclear, atomic and molecular systems. The term {\it classical field theories} refers mainly, to classical theories of gravity, but also effective classical field theories.
\begin{itemize}
\item {\bf Assumptions on the metric, measure and topological structures associated with the fundamental dynamical systems.}
\bigskip
\\
A.1. There is a {\it topological space model}  $M_4$ which is the arena where macroscopic observers can locate events.
\bigskip
\\
A.2. There is a topological {\it configuration space} $\mathcal{M}$ which is endowed with a  quasi-metric structure $(\mathcal{M},\varrho)$ as in {\it definition} \eqref{quasi-metric} and also endowed with a probability measure $\mu_P$.
\bigskip
\\
\item {\bf Assumptions on the ontological structure of the fundamental dynamical models.}
The {\it physical degrees of freedom} are by definition the degrees of freedom represented in the Hamiltonian of a fundamental dynamical system. At the ontological level we assume the following hypotheses:
 \bigskip
 \\
 A.3. The physical sub-quantum degrees of freedom are identical, undistinguishable degrees of freedom composed by two {\it sub-quantum atoms}. Since the sub-quantum degrees of freedom are composed. We call them {\it sub-quantum molecules}.
\bigskip
\\
A.4. There is a natural minimal coordinate scale $L_{min}$. It is the universal  minimal coordinate difference between the coordinates associated with events when defined respect to ideal, instantaneous frames co-moving with a given sub-quantum molecule.
\end{itemize}
 Assumption A.4 relates the ontological content of the theory with the spacetime structure.
\begin{itemize}
\item {\bf Assumptions on the fundamental dynamics.} We understand by dynamics a map as given by {\it Definition} \ref{definitionofdynamics}. We consider that the following assumptions restrict very much the laws for the fundamental dynamics.
\bigskip
\\
A.5. For each fundamental system, there is a fundamental Hamiltonian function that determines the full dynamics of the fundamental dynamical systems. The fundamental dynamics is cyclical or almost cyclical and can depend explicitly of the time parameter.
\bigskip
\\
A.6. The following {\it locality condition} holds: given a system  $\mathcal{S}\subset\,\mathcal{M}$ corresponding to a collection  of sub-quantum molecules, there is a smallest neighborhood with $\mathcal{S}\subset \,\mathcal{U}\,\subset \mathcal{M}$ such that for any  $\widetilde{\mathcal{U}}\supset \,\mathcal{U}$, the dynamical effect of any action of $\mathcal{U}$ or $\widetilde{\mathcal{U}}$ on $\mathcal{S}$ are the same.
\bigskip
\\
A.7. {\it Causal structure}: in the configuration space $\mathcal{M}$, there is an unique non-trivial causal structure as defined in {\it definition} \ref{definicion estructura causal}.
\bigskip
\\
 A.8. The fundamental flow at the fundamental scale is non-reversible in the sense of {\it definition} \ref{nonreversibledynamics}.
\bigskip
\\
A.9. The fundamental dynamics is sensitive to initial conditions and with the appearance of regular almost cyclic patterns in the long time evolution.
\bigskip
\\
A.10. The fundamental dynamics is general covariant: all the structures that appear in the theory are dynamical.
\end{itemize}
The nature and structure of this fundamental Hamiltonian is the main problem in our approach to the foundations of quantum dynamics. There are several candidates, ranging from geometric flows related with geodesic motion to other types of models.

 In this work, a class of dynamical systems fulfilling the above assumptions, namely, Hamilton-Randers systems, is investigated.  We will show that Hamilton-Randers systems are  suitable candidates to describe the physical systems corresponding to current quantum dynamical systems, but from the point of view of dynamical laws setting at a more fundamental description than quantum mechanics. From the dynamics of the Hamilton-Randers systems one can recover quantum mechanical description as a coarse grained description. Also,  most of the phenomenological properties of the quantum theory can be explained in terms of emergent  notions from the properties and characteristics of the fundamental dynamics. We will also show that Hamilton-Randers systems are consistent with the constrains on realistic models for quantum mechanics.
\subsection{Remarks on the assumptions}

Assumption A.0 is a fundamental assumption  in emergent frameworks for quantum mechanics. Even if the concept that we propose is not directly related with a energy or time scale, because these concepts are also emergent in our view and hardly can serve as a reference to describe the new scale, it is useful to translate the scale in terms of known physical notions. This must be rigourously taken as a extrapolation of such notions.
 This scale is set where the dynamics is described by a fundamental, underneath dynamics.

 The idea of a fundamental scale is not exclusive for theories of emergent quantum mechanics and indeed, it is ubiquitous in many approaches to quantum gravity. Other thing is the actual value is the fundamental scale. Usually, the fundamental scale of quantum gravity is possed approximately at the Planck scale. The Planck scale corresponds to the distance where the Schwarzschild radius equals the Compton length. But then the significance of such scale is associated with the (approximately) validity of both, quantum mechanics and general relativity. Since we are adopting a skeptical point about that assumption, in principle we do not fix fundamental scale to be the Planck scale. We will left open this issue until we have developed further our ideas. We will see that indeed this fundamental scale is related (but not identified) with the emergence of gravity as a classical, macroscopic phenomena.

In Assumption A.1, by  {\it event} we mean either an observation of a phenomena or a physical state which is potentially observable by a macroscopic observer.
 Assumption A.1 is suggested by the way physical theories are constructed. It is assumption hard to be avoided. At this stage, we do not make precise if $M_4$ refers to a discrete space or to a continuous space, as in general relativistic theories, or to a discrete spacetime, as in quantum spacetime theories as for instance Snyder quantum spacetime \cite{Snyder}. Indeed, this important question depends upon very deep and technical details.

Assumption A.2 on the existence of  a quasi-metric structure defined on the configuration manifold $\mathcal{M}$ is useful for our constructions for several reasons. First, note that if a quasi-metric structure is defined on $\mathcal{M}$, then there are well defined topological structure on $\mathcal{M}$ determined by the quasi-metric function $\varrho$. Also, the notions of Cauchy sequence and completeness associated with the quasi-metric $\varrho$ are well defined, although they are not symmetric notions. A similar situation happens for the analogous constructions in Finsler geometry \cite{BaoChernShen}.
Second, if there is quasi-metric as fundamental metric structure of the configuration space $\mathcal{M}$, then it could be associated with the postulated irreversible character of a fundamental dynamics, according to  Assumption A.8 by linking the dynamics with the  metric structure of the fundamental configuration space. This can be established as discussed previously by proposition \ref{Proposicion varrho Phi}.

 A relevant  example of quasi-metric in the category of smooth manifolds are Randers spaces, a class of geometric spaces originally introduced by G. Randers in an attempt to describe the irreversibility of physical evolution as a fundamental property of the spacetime arena  \cite{Randers}. Basically, they are linear perturbations on the norm function associated with a Lorentzian metric\footnote{While the theory of Randers spaces is very well formulated in the case of metrics with Euclidean signature, there are formal difficulties in the case of metrics with indefinite signature. A particular theory of Randers spaces with Lorentzian signature appears in \cite{Ricardo2010}.}. In the theory to be developed in this pages, the theory of {\it dual Randers spaces} is introduced and applied to models describing deterministic, local, causal and non-reversible dynamics at the fundamental scale in a continuous approximated description. This motivation is similar to the original insight of G. Randers, namely, to formalize a fundamental non-reversible dynamics in a geometric way.
Note that although the mathematical formulation of the theory that we present in this work will be developed in terms of continuous models, the degrees of freedom are defined by discrete sets. Hence the probability measure $\mu_P$ is discrete and determined by operations based on {\it counting degrees of freedom in a determined way}. However, this will not forbid us to consider continuous models as approximations to more precise discrete systems. We will discuss this point below.

The measure properties and the metric properties are logically independent. The fact that the metric and measure structures are logically separated is a distinctive characteristic of the mm-spaces, a point of view that was initiated in geometry by M. Gromov \cite{Gromov, Berger2002}. Such a category is the natural framework for the formulation of the theory of Hamilton-Randers dynamical systems.

 Let us consider assumption A.3. That the degrees of freedom at the fundamental scale are deterministic and localized is an assumption in direct confrontation with the probabilistic quantum mechanical point of view on physical systems.  However, in defense of our deviated point of view, we should admit  that very few is known with certainty about the dynamics at the fundamental scale. We also need to recognise that the possibility to have a deterministic description  is at least technically appealing, because of its simpler mathematical structure.

 Let us also remark that the association of two sub-quantum atoms to each sub-quantum molecule is not due to a conjectured new fundamental interaction. It corresponds to more a radical point of view concerning the dynamical structure of the fundamental systems and its natural predisposition to be described by {\it averages}.  For each sub-quantum molecule, one of the sub-quantum atoms evolves towards {\it the future}, while the companion evolves towards {\it the past}. The justification for this approximation relies on the emergence of the quantum description of quantum systems as an emergent coarse grained description dynamics from the fundamental dynamics. The process of emergence is an average method and in such averaging process, the first step is the averaging on the time oriented evolution of associated degrees of freedom, where one degree of freedom evolves toward future and an associated degree of freedom towards past.
 To consider pair of sub-quantum atoms is then a requirement and first step to re-organize the description of the fundamental dynamics in terms of an averaged, coarse grained description and degrees of freedom in the dynamical evolution laws.

  Furthermore, assumption A.3 pre-supposes an intrinsic irreversible evolution, since the effect of future and past oriented evolutions are not equivalent at such fundamental level.

 Assumption A.4 suggests the existence of special coordinate systems associated with sub-quantum molecules, namely, local coordinate systems associated with the sub-quantum degrees of freedom. By co-moving coordinate system one means a curve $\gamma:I\to M_4$ that has locally at some instance the same second order jet at $t=0\in I$ than the world line of the molecule. This  interpretation of local co-moving system is understood once the four dimensional spacetime arena is introduced in the theory as a concept where  to locate macroscopic observations. The co-moving coordinate systems associated with sub-quantum molecules does not have attached a direct observable meaning.

 Assumption A.4 on the existence of a minimal length scale is in agreement with several current approaches to the problem of quantum gravity. Assumption A.4 can be read as stating that the fundamental scale has associated a fundamental length. That the Planck scale is the fundamental scale has been argued in several ways in the literature, although ultimately it is a conjecture based upon extrapolations of laws towards domains of validity which is in principle unclear to hold.

The interpretation given of the coordinate systems introduced in Assumption A.4 is consistent with the theory of quantum spacetime introduced by H. Snyder \cite{Snyder}. This remaind us that there are algebraic frameworks where minimal coordinate length in the sense described above and associated with particular coordinate systems can be made consistent with the constrains imposed by causality and with local diffeomorphism transformations.

We would like to remark that Assumption A.4 provides a realization of Assumption A.0: both {\it assumptions} could be merged in one. However, in terms to keep a more general framework, we keep by now both {\it assumptions} as independent, for the reasons discussed in relation  with assumption $0$.

   In addition to the geometric flow introduced in Assumption A.5, there is  a dynamics for the fundamental degrees of freedom introduced in Assumption A.3 associated with the evolution respect to an {\it external time} parameter, that we call $\tau$-time parameters. If the $\tau$-time parameter  is discrete, this fundamental dynamics is described by deterministic, finite difference equations. If the dynamics is continuous, then the associated $\tau$-time parameter will also be continuous, as we shall consider in the Assumption A.5.bis described below.

The choice between local or non-local character for the interactions of the degrees of freedom at the fundamental scale is obviously of relevance. In favour of assuming locality, we should say that this option is geometrically appealing, since it is possible to have geometric representations of local interactions in a consistent way with causality, but a non-local interaction is far from being understood in a causal picture of dynamics. We think that a theory  aimed to be fundamental must be a local dynamical theory, since in any pretended fundamental theory, there is no room left for explaining non-locality in terms of a more fundamental level.

 However, deep conceptual difficulties accompany the search for local descriptions of a sub-quantum theory, as it is implied by the experimental violation of Bell type inequalities in very general situations. A possible explanation of how the violation of Bell inequalities can happen is discussed in this work. It requires the notion of {\it emergent contextuality}, to be explained in later {\it chapters}. This notion is consistent with the hypothesis and idea of {\it superdeterminism}, a possible way out to the violation of Bell inequalities by local realistic models recently considered in some extend by G. 't Hooft. However, superdeterminism is not an element in our description on how the non-local quantum correlation happen. It appears more as a consistent consequence than as a constructive principle for falsifiable theories and predictive models of real physical systems (ideally isolated). Indeed, we see superdeterminism in direct confrontation with the coincidence of the interpretation of the wave function describing the state of an individual quantum system and at the same time, an statistical interpretation. We see superdeterminism as a purely coincidence characteristic of the mathematical structure of our theory.

The almost cyclic property of the fundamental dynamics of assumption 5 is motivated by the need to describe complex processes as emergent phenomena. If the fundamental dynamics had fixed points or attractors as generic final states, it will be difficult to reproduce non-trivial quantum phenomena. On the other hand, if the fundamental dynamics is too irregular, it will be difficult to explain the structure of quantum systems. Hence, the option of an almost cyclic dynamics, that will be explained in this work, appears as a natural requirement for the fundamental dynamics.

Furthermore,  let us note that the parameters of time used to describe the fundamental dynamical systems do not necessarily coincide with the time parameters  used to describe quantum or classical dynamical processes. Therefore, the origin of the notion of macroscopic time parameters time raises, in relation with the nature of the fundamental sub-quantum dynamical processes. Concerning this issue, we make the following conjecture, which is of fundamental importance in the elaboration of our theory:
\bigskip
\\
{\it The time parameters used in the construction of dynamics in classical and quantum dynamics corresponds to cyclic fundamental processes of sub-quantum dynamical systems.}
\bigskip
\\
Thus the classical time parameters  have an emergent origin. In general relativistic models, this will imply that the full spacetime arena has also an emergent origin. However, for this interpretation to hold in plenitude, it is necessary that fundamental dynamics to be cyclic for some systems that can serve as quantum clocks.

Assumption A.7 does not necessarily implies the existence of a Lorentzian metric. For instance, if $\mathcal{M}$ is a discrete spacetime, then the cone structure is different than the usual light cone of a Lorentzian spacetime. In addition, we assume that the limit for the speed of the sub-quantum degrees of freedom is the speed of light in vacuum, since it is natural in the approximation when the configuration space $\mathcal{M}$ is continuous that  Assumption A.7 is interpreted as that the speed of each fundamental degree of freedom is bounded by the speed of light in vacuum. Otherwise, we will have to explain why there are two fundamental causal structures, the one at the fundamental level and the macroscopic one associated with  light cones.

The values of these two fundamental scales, the fundamental length scale and speed of light in vacuum, cannot be determined currently by theoretical arguments. We will leave the value of the length scale as given, hoping that a future form of the theory can provide a better explanation for them.

The part of the  assumption A.8 is one of the most radical departures from the usual assumptions in theoretical physics. However, we have show below that a non-reversible dynamics is more natural than a reversible one. In particular, we have discuss how quantum dynamical systems comes to be essentially irreversible, in contrast with the popular tail of reversibility. This is enough to justify our Assumption A.8. Also, this assumption is fundamental in the explanation of several fundamental concepts of physics (spacetime) as of emergent character.

The motivation for Assumption A.9 is two-fold. From one side, it is the generic case for a dynamics with many degrees of freedom and with non-linear interactions. On the other hand, if the fundamental dynamics is chaotic, then one should be able to apply such property to the emergent descriptions. This is specially relevant to explain the un-predictable values of quantum mechanical observable measurements.

Assumption A.10 is helpful to reduce the number of possible models and developments. General covariance in the form of absence of absolute structures is not only appealing philosophically, but also useful to constrain candidates for the fundamental dynamics. Furthermore, the current theories of gravitation are general covariant. Hence to have this property from the beginning will avoid to explain its emergence.

Let us remark that there will be more assumptions and decisions adopted during the development of the theory. Also, it is possible that the assumptions listed and slightly discussed above will remain as assumptions in the last form of the theory. I do not foresee mechanism to reduce them from more fundamental assumptions. I also consider them necessary, if the theory that we would like to built must offer a complete explanation of all the quantum phenomenology. Maybe the exception for this comments is assumption A.4, which use in the theory is rather limited and whose consequences can be derived in another ways. On the other hand, the assumption A.4 is natural within the general description proposed by {\it assumption A.3} and the discrete character of the fundamental degrees of freedom.

\subsection{The approximation from discrete to continuous dynamics}
We anticipate that, due to the way  in which the notion of {\it external time parameter} (denoted by $\tau$-time parameter) is introduced in our theory, such parameters must have a discrete nature. The fundamental degrees of freedom are described by discrete variables and their dynamics could in principle also be discrete (the number field $\mathbb{K}$ can be discrete, as discussed in previous sections). Furthermore, a discrete dynamics is compatible with the existence of a minimal inertial coordinate difference $L_{min}$, for instance through the notion of quantum spacetime of Snyder type \cite{Snyder}.

Despite the discreteness that the assumptions require for the physical systems at the fundamental scale, a more practical approach is developed in this work, where continuous models for the dynamics are used. This is motivated by the smallness of the  fundamental scale compared with any other scales appearing in physical systems  but also by  mathematical convenience, because the logical consistency with several of the mathematical theories applied and developed along the way.

In this continuous approximation to the dynamical theory of fundamental objects, where the $\tau$-time parameter is continuous instead of discrete, several of the assumptions should be amended or modified as follows:
\begin{itemize}
\item A.1.bis. There is a smooth manifold $\mathcal{M}_4$ which is the {\it model manifold} of spacetime events.

\item A.2.bis. There is a configuration smooth manifold $\mathcal{M}$ endowed with  a quasi-metric structure \eqref{quasi-metric}, which is at least $C^2$-smooth on $\mathcal{M}$.

\item A.4.bis, which is identical to A.4.

\item A.5.bis. The dynamical law of the fundamental degrees of freedom respect to the  $\tau$-time parameter is deterministic and given by a system of first order ordinary differential equations.
\end{itemize}
We remark that $\mathcal{M}_4$, and the configuration manifold $\mathcal{M}$ cannot be arbitrary from each other. From a epistemological point of view this must be the case.

In the continuous approximation the rest of the assumptions remain formally the same than in the original formulation. However, because the different categories (smooth manifold category versus discrete topological spaces category), the implementation and the techniques that we can use are different than in the discrete case. In the continuous limit,  the assumption of the existence of minimal length must be interpreted as a theoretical constraint on the underlying geometric structure. There are also indications that in the continuous limit  such constraint is necessary.

Finally, let us mention that all our observations are linked to macroscopic or to quantum systems that can be represented consistently in a $4$-dimensional spacetime manifold. This should be motivation enough to consider a manifold structure as a convenient arena to represent quantum and classical systems.

\subsection{The need of a maximal proper acceleration}
  A direct consequence of the {\it assumptions} A.3, A.4, A.6 and A.7 is the existence of a maximal universal proper acceleration for sub-quantum atoms and sub-quantum molecules. In order to show how the maximal acceleration arises, I will follow an adaptation of an heuristic argument developed in \cite{Ricardo2014} for the existence of a maximal proper acceleration in certain physical models.

  The terms and fundamental notions of the following derivation are taken from classical point dynamics.
  Let us consider the simplified situation when the spacetime is a smooth manifold endowed with a metric of indefinite signature $\eta\equiv diag (1,-1,-1,-1)$. As a consequence of the assumption A.4,  there is a lower bound for the difference between coordinates of the fundamental degrees of freedom in any instantaneous inertial system.  For any elementary work $\delta\mathcal{W}$ the relation
\begin{align*}
\delta\,\mathcal{W}:=\vec{F}\cdot \delta \,\vec{L}=\,\delta \,L\,m\,a\,\gamma^3,
\end{align*}
must hold, where $\vec{F}$ is {\it the external force} on the sub-quantum molecule caused by the rest of the system and is defined by the quotient
 \begin{align*}
\vec{F}:=\frac{\delta\,\mathcal{W}}{\delta \,\vec{L}}.
 \end{align*}
$\delta\,\vec{L}$ is the infinitesimal displacement of the sub-quantum molecule caused by the rest of the system in the instantaneous coordinate system associated with the sub-quantum molecule at the instant just before the sub-quantum molecule suffers the interaction; $a$ is the value of the acceleration in the direction of the total exterior effort is done and the parameter $m$ is the inertial mass of the sub-quantum molecule $\mathcal{S}$; $\gamma$ is the relativistic factor. By the Assumption A.4, it holds that $\delta\,L=\,L_{min}$. Hence we have that
\begin{align*}
\delta\,\mathcal{W}= L_{min}\,\gamma^3\, m \,a.
 \end{align*}
  the infinitesimal work is given by the expression
   \begin{align*}
\delta\,\mathcal{W} =\,\frac{1}{2}\,\big( \gamma\delta {m}\,c^2+\,m\,\delta{\gamma}c^2)=L_{min}\, \gamma^3\,m \,a.
\end{align*}
Assume that there is no change in the matter content of $\mathcal{S}$. Then the relation
\begin{align*}
\delta m=0
\end{align*}
holds good. Since  the speed of any physical degree of freedom is bounded by the speed of light by the assumption of locality A.7, one has that
\begin{align*}
\delta{v}_{max}\leq\,v_{max}=c.
 \end{align*}
 Hence the maximal infinitesimal work calculated in a coordinate system  instantaneously at rest at the initial moment with the particle produced by the system on a point particle is such that
 \begin{align*}
 \mathcal{W}=\, m\,L_{min}\,a\leq\,\,m \,c^2.
  \end{align*}
 This relation implies an universal bound for the value of the proper acceleration $a$ for the sub-quantum system $\mathcal{S}$ given by
\begin{align*}
a\,\leq \,\frac{c^2}{\,L_{min}}.
\end{align*}
 Therefore, under assumptions A.3, A.4, A.6 and A.7  it is natural to require that the following additional assumption also holds,
\begin{itemize}
\item A.11. There is a maximal, universal proper acceleration for both sub-quantum atoms and di-atomic sub-quantum molecules. The value of the maximal proper acceleration is of order
    \begin{align}
    a_{max}\sim\,\frac{c^2}{\,L_{min}}.
    \label{maximalacceleration}
    \end{align}
\end{itemize}

Another derivation of a maximal acceleration does not need the assumption $\delta m=0$.
According to assumption A.6, there is a  maximal domain $\mathcal{U}$ that determines the effect of the dynamics on $\mathcal{S}$. As a consequence, we have that
  $ \delta {m}\leq \, C\,m$, with $C$ a constant of order $1$ that depends on the size of $\mathcal{U}$. Hence $C$ is a measure of the size of $\mathcal{U}$ respect to the size of the sub-quantum molecule. At this point, we make the assumption that $C$ is uniform in $\mathcal{M}$, independently that $\mathcal{M}$  being compact or non-compact. Then one obtains a similar expression than \eqref{maximalacceleration},
where the exact relation depends on the constant $C$.

In the continuous approximation, where the degrees of freedom follow a continuous dynamics instead of a discrete dynamics, the assumption A.11 cannot be derived heuristically as was done above. Hence in the continuous case, assumption A.11 is  an independent constraint imposed in the theory. That such kinematical theory exists at least in an {\it effective} sense and that it is indeed compatible with the action of the Lorentz group has been shown in \cite{Ricardo2014, Ricardo2020}, although in such theory the value of the maximal acceleration is not necessarily fixed by the expression \eqref{maximalacceleration}. The general justification for such structures comes in the case of continuous dynamics as a consequence of the violation of the clock hypothesis of relativity theories \cite{Einstein1922} in situations where back-reaction is significatively important \cite{Mashhoon1990,Ricardo2014, Ricardo2020}.

The geometric structures compatible with a maximal proper acceleration and a maximal speed were called {\it metrics of maximal acceleration} \cite{Ricardo2014,Ricardo2020}. They are not Lorentzian metrics or pseudo-Riemannian metrics, but metric structures defined in higher order jet bundles. In the case of spacetime geometry, the second jet bundle $J^2_0(\mathcal{M}_4)$ over the spacetime manifold $\mathcal{M}_4$; in the case of the dynamical systems that will be considered in this work, the second jet bundle $J^2_0(\mathcal{M})$ over the configuration space $\mathcal{M}$. However, the leading order term in a metric of maximal acceleration is a Lorentzian structure. In the present work, the leading order Lorentzian structure has been adopted as an approximation to the more accurate description of the spacetime structure. This leaves us with the need of a more precise treatment within the framework of spaces with maximal acceleration for a second stage of these investigations.

Let us remark that the concept of minimal coordinate distance is not diffeomorphism invariant. According to assumption A.4, the minimal coordinate distance happens in specific local coordinate systems of $\mathcal{M}$, which has associated a very specific class of local coordinates in the spacetime $M_4$. However, the concept of minimal distance is an useful notion to calculate the maximal proper acceleration, which is indeed an invariant concept. This construction implies that the geometry that we will need to use is based on the notion of  {\it maximal acceleration geometry} \cite{Ricardo2014, Ricardo2020}. Also, the dynamical systems that we will consider in the next chapter will be constrained to have uniform bounds for the proper acceleration and speeds.

There is a further argument in favour of the limitation of the proper acceleration in the context of the context that we are discussing. Let us consider the almost cyclic structure of the dynamics, as discussed in assumption A.5. Let us further assume that we consider parameters such that the dynamics is not only cyclic, but in some sense, it is also periodic. Then there is a {\it time scale} for the dynamics and there is also a maximal speed of propagation, which determines partially the causal structure of assumption A.7. Then there is a natural maximal acceleration. The interpretation as a proper acceleration is linked with the interpretation of the corresponding time parameters used.

The advantage of this argument respect to the previous one relies in the essential kinematical concepts used plus the assumption of periodicity of certain fundamental processes, in contrast with the classical dynamical concepts of work and force prevalent in the previous discussion.
\subsection{Which of the assumptions is rescindable?}
Although our set of assumptions should not be seen as a set of axioms, there are certain relations and hierarchies among them. Assumptions A.0 and A.1 appear as untouchable in our emergent approach to the quantum theory. We think that assumptions $A.2,\,A.3$ are also very natural and difficult to avoid in our scheme.  On the other hand, one could licitly rescind from assumption A.4 by adopting a geometry of maximal proper acceleration. Then one is forced to justify these geometries in the present context in a general way \cite{Ricardo2020}.

Regarding the assumptions on the dynamics, it is difficult to hide that they have been freely chosen. Determinism, locality and causality are generic properties for a dynamics, justified only by inherent simplicity respect to the opposite assumptions (randomness, non-locality and a-causality) and by the epistemological motivation to have a fundamental theory where the non-locality and acausal character of quantum mechanics phenomena is explained in terms of clearer notions and pictures. How these conditions can be implemented could different as it is proposed through our assumptions A.5, A.6 and A.7.

Assumption A.8, it is the core to our view of quantum description and current physical description as emergent. It is therefore, difficult to avoid in our approach.

Assumption A.9, is  natural one in the situations when system under consideration have a complex structure. On the other hand, it is consistent with the apparently unpredictable character of the phenomenology of quantum systems. Although in the future this assumption should be deduced from the details of the dynamics, it is useful to adopt it for the construction of the general frame of the dynamics.

\newpage

\section{\LARGE{Hamilton-Randers dynamical systems}}\label{chapter on classical dynamics Hamilton Randers}
\bigskip
\bigskip
In this {\it Chapter} we develop the structure of a theory of deterministic dynamical systems for the sub-quantum degrees of freedom. Such dynamical systems are constructed following as a guideline the assumptions discussed in {\it Chapter} \ref{chapter on Assumptions and General Theory} for a general, non-reversible dynamics. We will show that the dynamical systems considered here full-fill all the requirements for the fundamental dynamics discussed in  {\it Chapter} \ref{chapter on Assumptions and General Theory}.

We formulate the theory of fundamental dynamical systems in terms of differential geometry and differential equations.
The framework for the construction of the dynamical models that we adopt here is probably not the most general compatible with the ideas developed in {\it Chapter} \ref{chapter on Assumptions and General Theory}. Restrictions associated with the use of differential structures are introduced on the way. In particular, the theory of smooth real manifolds, with use of the real field $\mathbb{R}$ for the definition of the time parameters, is adopted. The adoption of differentiable models is justified in terms of the difference in scales between sub-quantum and quantum scales. This allow us to explode the power that differential models brings, in particular, on the existence and uniqueness of solutions and related geometric robust methods.

A specific characteristic of Hamilton-Randers dynamical systems is that they have a {\it two dimensional time dynamics} or {\it 2-time dynamics}, that we have denoted by $U_t$ and by $U_\tau$. This is similar to the situation with fast/slow dynamical systems in classical dynamics, but for the models that we will consider, each of the dynamics and time parameters are formally independent from each other: there is no a  bijective map between the values of the $t$-time and the values of $\tau$-time parameter as usually happens in fast/slow classical dynamical systems and the $U_\tau$ dynamics can only strictly be applied to quantum or macroscopic systems and does not apply to the more fundamental sub-quantum systems.
 The two-dimensional character of time and evolution, expressed by the need of using two time  parameters $(t,\tau)$, is of fundamental importance for Hamilton-Randers theory, since the interpretation of the quantum phenomena proposed relies on this $2$-time dynamics.
\subsection{Geometric framework}

In order to full fill {\it Assumptions 1,2,3} discussed in {\it Chapter} \ref{chapter on Assumptions and General Theory} we propose that the configuration manifolds $\mathcal{M}$ of the fundamental dynamical systems are tangent spaces $TM$ such that the manifold $M$ is diffeomorphic to a cartesian product manifold,
\begin{align}
M\cong \,\prod^N_{k=1}\,\times\,  M^k_4.
\label{manifoldM}
\end{align}
We assume that each of the manifolds $\{M^k_4,\,k=1,...,N\}$ is diffeomorphic to a given manifold $M_4$, that we call the {\it model manifold}.

 The implementation of {\it Assumption 3} does not require that the manifolds $\{M^k_4,\,k=1,...,N\}$, describing the configuration space of each degree of freedom, need to be diffeomorphic to the {\it model four manifold} $M_4$. The fundamental sub-quantum degrees of freedom are not observables from a quantum or classical scale. Therefore, the model manifold $M_4$ is generally different than the spacetime manifold $\mathcal{M}_4$. This is consistent with Einsteinis fundamental idea that spacetime and matter content are related, even in some circumstances, they determine each other. However, the matter content of a sub-quantum degree of freedom is different than the matter content of a macroscopic system, justifying the distinction between $M_4$ and $\mathcal{M}_4$.

For the dynamical systems that we shall consider, the configuration manifold $\mathcal{M}$  is the tangent space of the smooth manifold $M$ and is of the form
\begin{align}
\mathcal{M}\cong\,TM\cong \,\prod^N_{k=1}\,\times TM^k_4.
\label{structure of the configuration manifold}
\end{align}
However, since the dynamics will be described by a Hamiltonian formalism, the relevant geometric description of the dynamics is through the co-tangent bundle $\pi:T^*TM\to TM$.

The dimension of the configuration manifold $\mathcal{M}$ is
\begin{align*}
dim(\mathcal{M})=\,dim({TM})=\, 2 \,dim (M)=\,8\,N.
 \end{align*}
For the dynamical systems that we are interested, we assume that the dimension $dim(TM)=8N$ is large for all practical purposes compared with $dim(TM_4)=8$. Furthermore, choosing the configuration space $\mathcal{M}$ as the tangent space $TM$ instead than the base manifold $M$ allows to implement geometrically second order differential equations for the coordinates of $M$ as differential equations defining vector fields on $TM$.

 The canonical projections are the surjective maps
 \begin{align*}
 \pi_k:TM^k_4\to M^k_4,
 \end{align*}
 where the fiber $\pi^{-1}_k(x)$ over $x_k\in \,M^k_4$ is $\pi^{-1}_k(x_k)\,\subset TM^k_4$ is the tangent space of $M^k_4$ at $x$.

  We also introduce the co-tangent spaces $T^*TM^k_4$ and the projections
  \begin{align}\label{canonical projections}
  \begin{cases}
  & proj_k :T^*TM^k_4 \to TM^k_4,\\
    & proj :T^*TM_4 \to TM_4.
  \end{cases}
  \end{align}

It is assumed that each manifold  $M^k_4$ is diffeomorphic to the model manifold $M_4$.
This is part of the formal implementation of {\it Assumption} A.3 of {\it Chapter} \ref{chapter on Assumptions and General Theory},
 stating that each of  the $N$ fundamental degrees of freedom are indistinguishable and identical to each other. If they are identical and indistinguishable, then the mathematical description must be given in terms of equivalent mathematical structures. Since we are assuming models within the category of differentiable manifolds $M^k_4$, then the manifolds $TM_k$ must be diffeomorphic to each other. These diffeomorphisms are maps of the form
 \begin{align}
 \varphi_k: M^k_4 \to M_4,\quad k=1,...,N.
 \label{diffeomorphism conditions}
 \end{align}
 We denote by
\begin{align*}
\varphi_{k*}:T_{x_k}M^k_4\to T_{\varphi_k(x_k)} M_4
\end{align*}
 the differential map of $\varphi_k$ at $x_k$ and by
 \begin{align*}
 \varphi^*_k: T^*_{\varphi(x_k)} M^k_4\to T^*_{(x_k)}M_4
 \end{align*}
  the pull-back of $1$-forms at $\varphi_k(x)\in\,M_4$.

The notion of {\it sub-quantum molecule} was introduced in {\it chapter} \ref{chapter on Assumptions and General Theory}.
Each of the sub-quantum molecules is labeled by a natural number $k\in\, \{1,...,N\}$. At the current stage of the theory, it is taken as an assumption that the dimension of the model manifold $M_4$ is four. This construction can be justified because the standard description of physical phenomena in the framework of four dimensional manifold models is satisfactory and at some point, there must exist a connection between the spacetime arena of sub-quantum processes and quantum and macroscopic processes. Therefore, it is natural to assume that each element of the collection of manifolds $\{M^k_4,\,k=1,...,N\}$ associated with the sub-quantum molecules is also a four-manifold and that the differential geometry of four-manifolds is our basis for the formulation of the dynamical models, as a first step towards a more general formalism. In this way, the tangent manifold $TM^k_4$ is the configuration manifold for the $k$-th sub-quantum molecule.

Each point in the tangent space $TM^k_4 $ is described by four spacetime coordinates $({\xi}^1,{\xi}^2,{\xi}^3,{\xi}^4)$  of the point ${\xi}(k)\in\,M^k_4$ and four independent velocities coordinates
$(\dot{{\xi}}^1,\dot{{\xi}}^2,\dot{{\xi}}^3,\dot{{\xi}}^4)\in \,T_{\xi} M^k_4$.

This viewpoint can be weakened by demanding the existence of a bijection between detections and spacetime events,
in the form of sheaf theory and related structures. However, we adopt in the present differential manifold theory version of the theory because of the additional advantages that differential manifold theory offers. Considering continuous instead than discrete models for the configuration space allows the theory of ordinary differential equations as models, providing existence and uniqueness of the solutions. However, continuous models should be considered only as an approximation to discrete models with very large number of degrees of freedom.

Even in the framework of smooth dynamical systems theory as we are considering, it is possible to consider more general
configuration spaces for the description of the dynamics of additional degrees of freedom o associated with the sub-quantum molecules. However, we restrict our attention to {\it spacetime configuration manifolds}. This attitude is based on the following argument, that can be found, for instance in Bohm's work \cite{Bohm, Bohm1980}. Physical properties such as spin and other quantum numbers are associated with the quantum description of elementary particles and quantum systems. However, it is remarkable that in quantum mechanics, measurements of observables are ultimately reduced to local coordinate positions and time measurements. Hence we adopt the point of view that all possible measurements of observables can be reduced to the detection and analysis of spacetime events and that the dynamical description of a physical system can be developed ultimately in terms of a spacetime description. Moreover, we will illustrate this point of view by showing that quantum particles with intrinsic spin are described by Hamilton-Randers dynamical systems.

\subsubsection{Diffeomorphism invariance as a consistence requirement}
Given the model manifold $M_4$ and the collection of four-manifolds $\{M^k_4\}^N_{k=1}$, there is a collection of diffeomorphisms
\begin{align}
\Upsilon:=\{\varphi_k:M^k_4\to M_4,\, k=\,1,...,N\}.
\label{definitionofthefamilyofdiffeormosphism}
\end{align}
This collection realizes the {\it Assumption} A.3.

Similarly as it happens in classical statistical mechanics, a macroscopic observer is not able to identify the detailed evolution of the sub-quantum molecules and because the selection of the diffeomorphisms $\varphi_k$ is not fixed by the theory, one can choose any other arbitrary family of diffeomorphisms,
\begin{align*}
\tilde{\Upsilon}:=\{\tilde{\varphi}_k:\tilde{M}^k_4\to M_4,\, k=\,1,...,N\}.
\end{align*}
$\tilde{\varphi}_k$ is related with $\varphi_k$ by a diffeomorphism, defining a family of global diffeomorphisms of $M_4$,
\begin{align}
Tran_k:=\{\varphi_{ktran}\equiv \,\tilde{\varphi}_k\circ {\varphi_k}^{-1}:M_4\to M_4,\,\,k=1,...,N\}.
\label{definicion de Tran}
\end{align}
Since all these diffeomorphims between four manifolds are arbitrary, it is clear that one should identify
\begin{align}
Tran_k \cong \Diff (M_4),\quad  k=1,...,N,
\end{align}
where $\Diff(M_4)$ is the group of global diffeomorphisms of $M_4$. Therefore, we have that
\begin{proposicion}
The dynamical models describing the sub-quantum dynamics are $\Diff (M_4)$-invariant.
\end{proposicion}
Thus diffeomorphism invariance of the model manifold $M_4$ appears in the theory as a consistent condition of the mathematical structure of the fundamental dynamical systems.

On the other hand, since each of the sub-quantum molecules are described by different objects (points in different manifolds $M^k_4$), the theory must be invariant under the action of a {\it gauge diffeomorphism group},
\begin{align*}
Tran(M):=\,\prod^N_{k=1}\times\,Tran_k =\,\prod^N_{k=1}\,\times \,\Diff (M^k_4) =\,\prod^N_{k=1}\,\times\,\Diff (M_4)
\end{align*}
Therefore, we have that
\begin{align}
Tran(M)=\prod^N_{k=1} \times\,\Diff(M_4).
\label{relation diff and trans}
\end{align}
Thus the diffeomorphism group $\Diff(M_4)$ emerges as a sub-group of $Tran (M)$.

\subsubsection{Measure structures}
  One way to achieve the invariance under the exchange of the sub-quantum molecules that physical observables is to define them as {\it averaged objects} in such a way that the {\it averaging operations} are compatible with diffeomorphism invariance of $T^*TM$, the complete phase space. Such invariance on the particular choice of the diffeomorphisms $\{\varphi_k\}$ can be achieved if the probability measure $\mu_P$ used in the definition of the averages  of functions defined on $TM$ is a product measure. Therefore, we assume that the measure is of the form
\begin{align}
\mu_P=\,\prod^N_{k=1}\,\times \mu_P(k),
\label{general form of the measure}
\end{align}
where each $\mu_P(k),\,k=1,...,N$ is a $\Diff(M^k_4)$-invariant probability measure defined on $TM^k_4$ and since $\Diff(M^k_4)\cong \Diff(M_4)$, the measure is invariant under the transformations induced on $T^*TM$ by the action of the group $\Diff(M_4)$ on $M_4$.

A particular way to construct a measure like this is discussed in the following paragraphs, where the measure is constructed associated with tensorial structures defined in $T^*TM$.
\subsubsection{Metric structures}
The four-manifold $M_4$ is endowed with a Lorentzian metric $\eta_4$ of signature $(1,-1,-1,-1)$. Moreover, for each $k\in \, \{1,...,N\}$ there is a Lorentzian metric $\eta_4(k)$ on $M^k_4$ and  we assume that each of the Lorentzian structures $(M^k_4, \eta_4(k))$
is conformally equivalent to \footnote{Note that if $(M^k_4, \eta_4(k))$ and $(M_4,\eta_4)$ are conformally equivalent, then there is a diffeomorphism $\psi_k: M^k_4\to M_4$ such that $\psi^*\,\eta_4=\,\lambda_k\,\eta_4 (k)$ with $\lambda_k:M^k_4\to \mathbb{R}^+.$ It is not necessary that $\psi_k = \varphi_k$, but nothing it is lost by choosing the family $\{\varphi_k\}$ in this way.} to the Lorentzian model $(M_4,\eta_4)$. The Levi-Civita connection of $\eta_4(k)$ determines a horizontal distribution in a canonical way (see for instance \cite{KolarMichorSlovak}, {\it chapter} 2). Given such standard distribution, there is defined a pseudo-Riemannian metric $\eta^*_S(k)$ on $TM^k_4$ (the Sasaki-type metric), which is the Sasaki-type lift of the metric $\eta^k_4$ on $M^k_4$ to $TM^k_4$ by the distribution associated with the Levi-Civita connection. Then there is defined a pseudo-Riemannian metric $\eta^*_S$ on $TM$, which is given by the relation
 \begin{align}
 \eta^*_S=\,\frac{1}{N}\,\sum^N_{k=1}\oplus\,\eta^*_{S}(k),\quad k=1,...,N.
 \end{align}
 This metric defines the causal structure on the configuration manifold $\mathcal{M}=TM$.

  The dual metric of $\eta^*_S(k)$ is the dual pseudo-Riemannian metric
  \begin{align}
  \eta_S(k)=\,(\eta^*_S(k))^* .
  \end{align}
The dual Sasaki-type metrics $\{\eta_S(k)\}^N_{k=1}$ allows to define the dual pseudo-Riemannian metric\footnote{Note that the use of $*$-notation for dual metrics and norms here is partially the converse respect to the usual notation in Riemannian geometry. For instance, $\eta^*_S$ is a metric on $TM_4$, while $\eta_S$ is a metric in $T^*M_4$.}
\begin{align}
\eta=\,\frac{1}{N}\,\sum^N_{k=1}\oplus\, \eta_S(k).
\label{structureeta}
 \end{align}
which acts on the fibers of the co-tangent bundle $\pi:T^*TM \to TM$.
 \bigskip
 \\
{\bf Associated measures}.
 The construction of an invariant measure $\mu_P$ in $T^*TM$ and other associated measures is presented in the following lines.  The Lorentzian metric ${\eta}_4$ allows to define a $\Diff(M_4)$-invariant volume form $dvol_{{\eta}_4}$ on $M_4$ in a canonical way as the volume form
\begin{align*}
dvol_{{\eta}_4}=\,\sqrt{-\det \eta_4}\,dx^1\,\wedge dx^2\,\wedge dx^3 \,\wedge dx^4.
\end{align*}
This is the usual Lorentzian volume form of the spacetime $(M_4,\eta_4)$. The pull-back by $\varphi_k$ of this form is denoted by $dvol_{{\eta}^k_4}$. There is also a $\Diff(M_4)$-invariant vertical form $dvol_k(y_k)$ on each fiber $\pi^{-1}_k(x(k))$ of $TM^k_4$,
\begin{align}
d^4z_k=\,\sqrt{-\det \eta_4}\,\delta y^1_k \wedge\, \delta y^2_k \wedge\, \delta y^3_k\wedge\, \delta y^4_k,
\label{fibermeasure}
\end{align}
where $\delta y^\mu_k=\,dy^\mu_k -N^\mu_{k\rho}\,dx^\rho_k$, a covariant construction that makes use of the non-linear connection $N\mu_\rho$ as in Finsler geometry \cite{BaoChernShen}. Considering the pull
 Then we define the measure $\tilde{\mu}_P$ in $TM$ as
\begin{align}
\tilde{\mu}_P=\,\prod^N_{k=1}\,d^4z_k\wedge\,dvol_{\eta^k_4},
\label{medidamuinvariante0}
\end{align}
which is the product measure (exterior product) of $N$ number of $8$-forms $\sqrt{-\det \eta^k_4}\,d^4z_k\wedge\,dvol_{\eta^k_4}$.
The measure in $T^*TM$ is then
\begin{align}
\mu=\,\prod^N_{k=1}\,proj^*_k(\tilde{\mu})=\,\prod^N_{k=1}\,\left(d^4z_k\wedge \,dvol_{\eta_4}\right)^*,
\label{medidainvariante1}
\end{align}
where
\begin{align*}
\left(d^4z_k\wedge \,dvol_{\eta_4}\right)^*=\,proj^*_k(\tilde{\mu})
\end{align*}
 is the pull-back volume form  on $T^*TM_4$ of the $8$-form
 \begin{align*}
\tilde{\mu}=\,\sqrt{-\det \eta_4}\,dx^1\,\wedge dx^2\,\wedge dx^3 \,\wedge dx^4\wedge\left(\sqrt{-\det \eta_4}\,\delta y^1_k \wedge\, \delta y^2_k \wedge\, \delta y^3_k\wedge\, \delta y^4_k\right).
\end{align*}
The following result follows,
\begin{proposicion}
The measure \eqref{medidainvariante1} is invariant under diffeomorphisms of $T^*TM^k_4$ induced by diffeomorphisms of $M_4$.
\end{proposicion}

The measure \eqref{medidainvariante1} can be pulled-back to submanifolds of $TM$. In particular, it can be concentrated along the world line curves of the sub-quantum degrees of freedom. In this case, a product of delta functions with support on each $k$-esim world lines on $M^k_4$ is inserted in the measure. In this case, the probability measure is a product measure of the form
\begin{align}
\tilde{\mu}_P=\,\prod^N_{k=1}\,\delta(x_k-\xi_k)\,d^4z_k\wedge \,dvol_{\eta_4},
\label{medidainvariante2}
\end{align}
where $\xi_k$ is the coordinates of the $k$-essim sub-quantum  molecule.
 \subsection{Deterministic dynamics for the fundamental degrees of freedom}
 Let us consider  a local coordinate system $\{(x^\mu_k,y^\mu_k)\}^{4,N}_{\mu=1,k=1}$ on the tangent space $TM$. The dynamical systems that we shall consider for the {\it sub-quantum molecule degrees of freedom} are systems of ordinary first order differential equations of the form
 \begin{align}
 \begin{cases}
& \frac{dx^\mu_k}{dt}=\,\gamma^\mu_{kx}(x,y,t),\\
 &  \frac{dy^\mu_k}{dt}=\,\gamma^\mu_{ky}(x,y,t),\quad\quad \mu=1,2,3,4;\,\,k=1,...,N.
 \end{cases}
 \label{dynamicalsystem}
 \end{align}
 The functions $\gamma^\mu_{kx}:TM\times \,\mathbb{R}\to \mathbb{R}$ and $\gamma^\mu_{ky}:TM\times \,\mathbb{R}\to \mathbb{R}$ are assumed regular enough such that, in order to determine locally the solutions, it is necessary and sufficient to know the initial conditions $\{x^\mu_k(0),y^\mu_k(0)\}^{N, 4}_{k=1,\mu=1}$. For each value of $k=1,...,N$, these coordinates represent the collective system of one {\it sub-quantum atom} evolving on the positive direction of $t$-time and a {\it sub-quantum atom} evolving in the negative direction of $t$-time.
 Under the required assumptions on the smoothness of the functions $\gamma^\mu_{kx}$ and $\gamma^\mu_{ky}$ the system of equations \eqref{dynamicalsystem} constitute determine a local flow on $TM$, as consequence of the theorems of existence and uniqueness of ordinary differential equations \cite{Chicone}. In particular, it is required that at least $\gamma^\mu_{ky}$ are continuous and Lipshitz functions and that $\gamma^\mu_{kx}$ are $\mathcal{C}^1$.

 It is not necessary to identify $x_k$ with $x_{k'}$, even if the sub-quantum molecules are identical, according to {\it  Assumption 3}.

 The system of ordinary differential equations \eqref{dynamicalsystem} is a $8N$-dimensional coupled system of implicit ordinary differential equations whose solutions $\xi:I\to TM$ determine the fundamental dynamics of the {\it sub-quantum molecules}. However,
 the equations of motion \eqref{dynamicalsystem} are equivalent to a system of $4N$ second order differential equations of semi-spray type \cite{MironHrimiucShimadaSabau:2002}.  This is because we impose the {\it on-shell conditions}
 \begin{align}
 y^\mu_k =\,\frac{dx^\mu_k}{dt}\quad\quad \mu=1,2,3,4;\,\,k=1,...,N.
 \label{onshell condition on the dynamical system}
 \end{align}
 for each sub-quantum molecule degree of freedom, leading to the system of second order differential equations
 \begin{align}
 \frac{d^2{x}^\mu_k}{dt^2}=\,\gamma^\mu_{ky}(x,\dot{x},t),\quad\quad \mu=1,2,3,4;\,\,k=1,...,N
 \label{dynamicalsystem2}
 \end{align}
 with initial conditions $\{x^\mu_k(0),\dot{x}^\mu_k(0)\}^{N, 4}_{k=1,\mu=1}$ and with the consistence constrain of the form
 \begin{align}
 \gamma^\mu_{ky}(x,y,t)\,=\,\frac{d \gamma^\mu_{kx}}{dt},\quad\quad \mu=1,2,3,4;\,\,k=1,...,N
 \label{onshell condition on the dynamical system 2}
 \end{align}
Gradually, we will restrict this general class of systems, according to the requirements of the fundamental dynamics.
\subsection{General covariance of the fundamental dynamics}
 The equations \eqref{dynamicalsystem} subjected to the constraints \eqref{onshell condition on the dynamical system} and \eqref{onshell condition on the dynamical system 2} are written in local coordinates. In order to be consistent with local coordinates transformations of $M$ induced by changes in the local coordinates of $M_4$, the functions
 \begin{align*}
 \gamma^\mu_{kx},\,\gamma^\mu_{ky}:TM^k_4\to \mathbb{R}
 \end{align*}
  must transform in an specific way. In particular, we assume that the constraints \eqref{onshell condition on the dynamical system} and \eqref{onshell condition on the dynamical system 2} and the dynamical equations \eqref{dynamicalsystem} are formally the same in all the induced local coordinate systems on $TM$,
  \begin{align*}
\begin{cases}
 & \frac{d\tilde{x}^\mu_k}{dt}=\,\gamma^\mu_{kx}(\tilde{x},\dot{\tilde{x}},t),\quad  \frac{d\dot{\tilde{x}}^\mu_k}{dt}=\,\gamma^\mu_{ky}(\tilde{x},\dot{\tilde{x}},t),\\
 & {\tilde{y}}^\mu_k =\,\frac{d\tilde{x}^\mu_k}{dt}\quad\quad \mu=1,2,3,4;\,\,k=1,...,N.
\end{cases}
 \end{align*}
 Therefore, we have a consistency relation:
 \begin{align*}
 \tilde{\gamma}^\mu_{ky} =\,\frac{d \tilde{\gamma}^\mu_{kx}}{d t},\quad\quad \mu=1,2,3,4;\,\,k=1,...,N.
 \end{align*}
It follows from the above that the on-shell condition \eqref{onshell condition on the dynamical system 2} is general covariant.

We consider local coordinate transformations of the form
\begin{align}
(x^\mu_k,\dot{x}^\mu_k)\mapsto \left(\tilde{x}^\mu_k(x_k),\frac{\partial \tilde{x}^\mu_k}{\partial x^\rho_k}\,x^\rho_k\right),
\label{local coordinate transformation}
\end{align}
where sums over the indexes $k$ and $\rho$ are assumed. In order to that the dynamical system \eqref{dynamicalsystem} is covariant with respect to the transformations \eqref{local coordinate transformation}, it is enough that the functions $ \gamma^\mu_{x k},\,\gamma^\mu_{y k}$ transform by the rules
 \begin{align}
\gamma^\mu_{x k}\mapsto \,\tilde{\gamma}^\rho_{kx} =\,\frac{\partial \tilde{x}^\mu_{k}}{\partial x^\rho_{k}} \gamma^\rho_{kx}
\label{local change of coordinates of the dynamical system functions 1}
\end{align}
and for the vertical part,
\begin{align}
\nonumber {\gamma}^\mu_{ky}\mapsto \, \tilde{\gamma}^\mu_{ky} &=\,\frac{d\tilde{\gamma}^\mu_{kx}}{d t}=\,\frac{d}{d t}\,\left(\frac{\partial \tilde{x}^\mu_k}{\partial x^\rho_k}\,\gamma^\rho_{kx}\right)\\
\nonumber & =\,\frac{\partial \tilde{x}^\mu_k}{\partial x^\rho_k}\,\frac{d \gamma^\rho_{kx}}{d t}+\,\frac{\partial^2 \,\tilde{x}^\mu_k}{\partial x^\rho_k\,\partial x^\sigma_k}\,\frac{d x^\rho_k}{d t}\,\gamma^\rho_{kx}\\
& =\,\frac{\partial \tilde{x}^\mu_k}{\partial x^\rho_k}\,\gamma^\rho_{ky}+\,\frac{\partial^2 \,\tilde{x}^\mu_k}{\partial x^\rho_k\,\partial x^\sigma_k}\,\gamma^\sigma_{kx}\,\gamma^\rho_{kx}.
\label{local change of coordinates of the dynamical system functions 2}
\end{align}
If we assume these transformation rules \eqref{local change of coordinates of the dynamical system functions 1}-\eqref{local change of coordinates of the dynamical system functions 2},
the above considerations prove the following
\begin{proposicion}
The dynamical system of equations \eqref{dynamicalsystem} subjected to the constrains \eqref{onshell condition on the dynamical system} and is general covariant under the transformations \eqref{local change of coordinates of the dynamical system functions 1} and \eqref{local change of coordinates of the dynamical system functions 2}.
\end{proposicion}
\subsubsection{$t$-time re-parameterization invariance of the fundamental dynamics}
We pass now to analyze the conditions under which the dynamical system -\eqref{dynamicalsystem}-\eqref{dynamicalsystem2} together with the on-shell condition \eqref{onshell condition on the dynamical system} is $t$-time re-parametrization invariant. This invariance is required for a consistent theory, since the $t$-time parameters can be arbitrary. The equations of evolution are
\begin{align*}
\begin{cases}
&\frac{d x^\mu_k}{dt}=\,\gamma^{\mu}_{xk}(x,y,t)\\
& \frac{d y^\mu_k}{dt}=\,\gamma^{\mu}_{yk}(x,y,t)\\
& y^\mu_{k}=\,\frac{d x^\mu_k}{dt},\quad \mu=1,2,3,4,\,k=1,2,3,...,N .
\end{cases}
\end{align*}
This is equivalent to the system
\begin{align*}
\begin{cases}
&\frac{d x^\mu_k}{dt}=\,\gamma^{\mu}_{xk}(x,\frac{dx^\mu}{dt},t)\\
& \frac{d y^\mu_k}{dt}=\,\gamma^{\mu}_{yk}(x,\frac{d y}{dt},t),\quad \mu=1,2,3,4,\,k=1,2,3,...,N .
\end{cases}
\end{align*}
Under a monotone non-decreasing, $\mathcal{C}^1$ diffeomorphic change of the parameter $t\hat{t}$, the equations change to
 \begin{align*}
\begin{cases}
&\frac{d x^\mu_k}{d\hat{t}}\,\frac{d\hat{t}}{dt}=\,\hat{\gamma}^{\mu}_{xk}(x,y,t(\hat{t}))\\
& \frac{d y^\mu_k}{d\hat{t}}\,\frac{d\hat{t}}{dt}=\,\hat{\gamma}^{\mu}_{yk}(x,y,t(\hat{t}))\\
& y^\mu_{k}=\,\frac{d x^\mu_k}{d\hat{t}}\,\frac{d\hat{t}}{dt},\quad \mu=1,2,3,4,\,k=1,2,3,...,N .
\end{cases}
\end{align*}
Similarly, this system of ordinary differential equations is equivalent to
 \begin{align*}
\begin{cases}
&\frac{d x^\mu_k}{d\hat{t}}\,\frac{d\hat{t}}{dt}=\,\hat{\gamma}^{\mu}_{xk}(x,\frac{d x^\mu_k}{d\hat{t}}\,\frac{d\hat{t}}{dt},t(\hat{t}))\\
& \frac{d y^\mu_k}{d\hat{t}}\,\frac{d\hat{t}}{dt}=\,\hat{\gamma}^{\mu}_{yk}(x,\frac{d x^\mu_k}{d\hat{t}}\,\frac{d\hat{t}}{dt},t(\hat{t})),\quad \mu=1,2,3,4,\,k=1,2,3,...,N .
\end{cases}
\end{align*}
If the system of differential equaions needs to be invariant under $t$-reparametrizations, then it should be of the form
\begin{align*}
\begin{cases}
&\frac{d x^\mu_k}{d\hat{t}}=\,\gamma^{\mu}_{xk}(x,\frac{dx^\mu}{d\hat{t}},\hat{t})\\
& \frac{d y^\mu_k}{d\hat{t}}=\,\gamma^{\mu}_{yk}(x,\frac{d y}{d\hat{t}},\hat{t}),\quad \mu=1,2,3,4,\,k=1,2,3,...,N .
\end{cases}
\end{align*}
This is only possible if the following conditions hold good:
\begin{align}
\begin{cases}
& 1. \gamma^{\mu}_{xk}=\,\hat{\gamma}^{\mu}_{xk},\quad \gamma^{\mu}_{yk}=\,\hat{\gamma}^{\mu}_{yk},\\
& 2. \gamma^{\mu}_{xk}(x,\lambda y,t)=\,\lambda \, \gamma^{\mu}_{xk}(x,y,t),\quad \gamma^{\mu}_{yk}(x,\lambda y,t)=\,\lambda \, \gamma^{\mu}_{yk}(x,y,t)\,\quad \forall \lambda >0,\\
& 3. \gamma^{\mu}_{xk}(x,y,t)=\gamma^{\mu}_{xk}(x,y),\quad \gamma^{\mu}_{yk}(x,y,t)=\gamma^{\mu}_{yk}(x,y),
\end{cases}
\label{conditions of spray}
\end{align}
for $\mu=1,2,3,4,\,k=1,2,3,...,N .$

The differential equations \eqref{dynamicalsystem}-\eqref{dynamicalsystem2} together with the constrain \eqref{onshell condition on the dynamical system} such that the conditions \eqref{conditions of spray} holds good is a {\it spray} on $TTM$.

Because we impose the condition of $t$-time re-parameterization, from now on we will consider only sprays on $TTM$ as the models for Hamilton-Randers dynamical systems. We adopt the following notion of spray, which is an adaptation of the notion of spray \cite{MironHrimiucShimadaSabau:2002} our setting,
\begin{proposicion}
A spray $\gamma$ of $TTM$ is a vector field $\gamma\in \Gamma TTM$ such that the conditions
\begin{align*}
\begin{cases}
& 1. \gamma^{\mu}_{xk}=\,\hat{\gamma}^{\mu}_{xk},\quad \gamma^{\mu}_{yk}=\,\hat{\gamma}^{\mu}_{yk},\\
& 2. \gamma^{\mu}_{xk}(x,\lambda y,t)=\,\lambda \, \gamma^{\mu}_{xk}(x,y,t),\quad \gamma^{\mu}_{yk}(x,\lambda y,t)=\,\lambda \, \gamma^{\mu}_{yk}(x,y,t)\,\quad \forall \lambda >0,\\
& 3. \gamma^{\mu}_{xk}(x,y,t)=\gamma^{\mu}_{xk}(x,y),\quad \gamma^{\mu}_{yk}(x,y,t)=\gamma^{\mu}_{yk}(x,y),
\end{cases}
\end{align*}
holds good.
\label{proposicion sobre spray en TTM}
\end{proposicion}

\subsection{Notion of Hamilton-Randers space}
There are two mathematical formalisms that can be applied to describe dynamical systems of the type \eqref{dynamicalsystem}. The first of them is the geometric theory of Hamilton-Randers spaces (developed in this {\it chapter}). The second is the quantized theory of Hamilton-Randers spaces as an example of Koopman-von Neumann theory of dynamical systems (developed in {\it chapter} \ref{chapter on Koopman-von Neumann formulation}.
Based upon these formulations, we will develop in  {\it chapter} \ref{chapter on classical dynamics Hamilton Randers}, {\it chapter} \ref{chapter on Koopman-von Neumann formulation}, {\it chapter} \ref{chapter of the Hilbert space structure} and {\it chapter} \ref{chapter on concentration of measure} a theory from where quantum mechanics follows as an effective and emergent description of physical systems. That is, as an effective description of the dynamical and structural properties of the dynamical systems \eqref{dynamicalsystem} subjected to the constrains \eqref{onshell condition on the dynamical system} and \eqref{onshell condition on the dynamical system 2}.

The first mathematical formalism for the dynamical systems \eqref{dynamicalsystem} that we consider is a geometric formalism, based on the formal combination of two different notions extracted from differential geometry, the notion of {\it generalized Hamilton space} and the notion of {\it Randers space}. In this {\it section}, and integrate these two notions in the notion of {\it Hamilton-Randers space}, as a first step to build the notion of Hamilton-Randers dynamical system, is developed.

 The notion of {\it generalized Hamilton space} is almost literally the same than the standard notion as it appears in differential geometry \cite{MironHrimiucShimadaSabau:2002}. The notion of Randers space as it appears in standard treatments \cite{Randers,BaoChernShen} is considered here in a more formal way. The relevance of these geometric notions in the contest of deterministic dynamical systems is due to the formal properties of the corresponding kinematics, namely, the existence of a causal structure, the existence of a {\it maximal acceleration} and a formal time irreversible fundamental evolution \cite{Randers}, properties which are of fundamental relevance to prove  important consequences and properties of Hamilton-Randers dynamical systems.

\subsubsection{Notion of generalized Hamilton space}
 Let $\widetilde{M}$ be a smooth manifold, $\mathcal{D}\subset\,T^*\widetilde{M}$ a connected, fibered, open sub-manifold of the co-tangent manifold $T^*\widetilde{M}$, where each of the fibers $\pi^{-1}(u)$ is a subset of the corresponding fiber $T^*_u\widetilde{M}$. Let $\bar{\mathcal{D}}$ be the topological closure of $\mathcal{D}$ respect to the manifold topology of $T^*\widetilde{M}$. For each local coordinate system $(\widetilde{U}, u^i)$ of $\widetilde{M}$, there is  local natural induced coordinates $\left\{\left(T^*\widetilde{U},(u^i,\theta_j)\right),\,i,j=1,...,dim(\widetilde{M})\right\}$ on $T^*\widetilde{M}$. The choice of the coordinates $\theta_i$ are such that under local coordinates transformations on the manifold $\widetilde{M}$, the coordinates $\{\theta_j, \,j=1,...,\dim (\widetilde{M})\}$ transform tensorially as components of a local $1$-form $\theta\in \Gamma T^*_U\widetilde{M}$. Such coordinates are not {\it vertical coordinates}, that is, corresponding to a set of complete commuting vertical vectors. Later we shall introduce another set of holonomic and vertical, local coordinates ($p$-coordinates) on each fiber over $u\in \,\widetilde{M}$ on the relevant co-tangent spaces describing our dynamical systems that, although do not transform tensorially under local coordinate changes $\widetilde{{M}}$, they have a relevant physical interpretation.

 \begin{definicion}\label{generalized Hamilton space}
 A generalized Hamilton space is a triplet $(\widetilde{M},F,\mathcal{D})$  such that the function
 \begin{align*}
 F:\bar{\mathcal{D}}\to \mathbb{R}^+\cup\{0\}
 \end{align*}
has the following properties:
 \begin{itemize}
 \item It is positive homogeneous of degree one on each of the fiber coordinates on $\mathcal{D}\subset T^*M$,

 \item It is smooth on the open submanifold $\mathcal{D}\hookrightarrow T^*\widetilde{M}$,

 \item The partial Hessian of the function $F^2$,
\begin{align}
g^{ij}(u,\theta)=\,\frac{1}{2}\,\frac{\partial^2 F^2(u,\theta)}{\partial \theta_i\partial \theta_j}
\label{fundamentaltensor}
\end{align}
is a pointwise non-degenerate matrix
 \end{itemize}

The function $F$ is the {\it generalized Hamiltonian function}.

The {\it fundamental tensor} $g^{ij}$ is pointwise defined by the matrix \eqref{fundamentaltensor}.
 \end{definicion}
The non-degeneracy of the fundamental tensor is an analytical requirement for the construction of associated connections and also, for the existence of geodesics as local extremals of an action or energy functional \cite{Whitehead1932}.

As a consequence of the homogeneity property of $F:\bar{\mathcal{D}}\to \mathbb{R}^+$, the following properties hold good:
\begin{align*}
\begin{cases}
& F(u,\theta)=\,\sum^N_{i=1}\, \frac{\partial F}{\partial \theta_i} \,\theta_i ,\\
& F^2(u,\theta)=\,\sum^N_{i,j=1}\,g^{ij}(u,\theta)\,\theta_i\,\theta_j.
\end{cases}
\end{align*}

 The generalized Hamiltonian function $F$ must be globally defined on $\bar{\mathcal{D}}$, but not necessarily in the whole $T^*TM$. This possibility allows for the existence of kinematic singularities.

 $\bar{\mathcal{D}}$ has associated a causal structure as discussed in {\it chapter} \ref{chapter on Assumptions and General Theory}.
 It is direct that the topological closure of $\mathcal{D}$ is of the form
 \begin{align*}
 \bar{\mathcal{D}}=\{(u,\theta)\in\,T^*TM\,s.t.\,F(u,\theta)\geq \,0\}.
  \end{align*}
  and that the boundary is
  \begin{align*}
  \partial{\mathcal{D}}:=\{(u,\theta)\in\,\,T^*TM\,s.t. \,\,F(u,\theta)=0\}.
  \end{align*}
  It is also direct that $\bar{\mathcal{D}}$ and $\partial \mathcal{D}$ have fibered structures,
  \begin{align*}
  \bar{\mathcal{D}}=\,\bigsqcup_{u\in TM}\,\bar{\mathcal{D}}_u,\quad \partial{\mathcal{D}}=\,\bigsqcup_{u\in TM}\,\partial{\mathcal{D}}_u .
  \end{align*}
 The {\it causal arc-wise connection domains} are defined as
\begin{align*}
\mathcal{C}_{(u,\theta)}:=\{\xi\in\,\bar{\mathcal{D}} \,s.t.\, \exists \,\gamma:[0,1]\to \bar{\mathcal{D}},\, s.t. \gamma(0)=(u,\theta),\,\gamma(1)=\,\xi\}.
\end{align*}
Note that the curve $\gamma:[0,1]\to \bar{\mathcal{D}}$ is {\it causal}, since $F(proj(\gamma),\theta(\gamma))\geq 0$.

We can compare our notion of generalized Hamilton space with the corresponding notion developed in reference \cite{MironHrimiucShimadaSabau:2002}. In our definition it is required the homogeneity property respect to the $\theta$-coordinates and that the fundamental tensor is non-degenerate. With these requirements, the Hamilton equations associated with the function $F$ coincide with the auto-parallel curves of a non-linear connection associated with $F$, providing a natural geometric interpretation for the Hamiltonian dynamics.

The relation of the dynamical systems \eqref{dynamicalsystem}-\eqref{onshell condition on the dynamical system} with generalized Hamiltonian spaces strives on the identification of  the solutions of the system \eqref{dynamicalsystem} subjected to the on-shell conditions \eqref{onshell condition on the dynamical system}-\eqref{onshell condition on the dynamical system 2} with pairs of geodesics of a certain class of generalized Hamiltonian spaces. The interpretation is the following. For each sub-quantum molecule there is a pair of sub-quantum atoms. Each of the sub-quantum atoms evolves following a geodesic of a generalized Hamilton structure, but where one of the sub-quantum atoms evolves towards the positive direction of increasing the time parameter $t$, while the other evolves in the decreasing direction of the parameter $t$. The total Hamiltonian is the symmetrized version of the pair of Hamiltonian functions associated with sub-quantum atoms. Contrary with what happens in a time-reversible dynamics, in a non-reversible case, the symmetrized dynamics is not trivial, posing a fundamental non-reversible character of the dynamics.

\subsubsection{Notion of pseudo-Randers space} Non-reversibility can be introduced by perturbing a non-reversible dynamics. This is exactly the fundamental ingredient of the concept of Randers spaces as first discussed in Rander's work \cite{Randers}. We introduce a generalization of the notion of Randers space in the following paragraphs.

We start with the standard notion of Randers space. Let $\alpha^*(u,z)$ be the Riemannian norm of $z\in \,T_u\widetilde{M}$ determined by a Riemannian metric $\eta^*$, while $\beta^*(u,z)$ is the result of the action $1$-form $\beta^*\in \,\Gamma T^*\widetilde{M}$ on  $z$. Then
\begin{definicion}
In the category of Finsler spaces with Euclidean signature, a Randers structure defined on the manifold $\widetilde{M}$ is a Finsler structure with Finsler function is of the form
\begin{align*}
F^*:T\widetilde{M}\to \mathbb{R},\quad  (u,z)\mapsto \alpha^*(u,z)+\beta^*(u,z)
\end{align*}
and such that the condition
 \begin{align}
 \alpha^*(\beta^*,\beta^*)<1
 \label{randerscondition}
 \end{align}
is satisfied.
 \label{Randers space}
 \end{definicion}
 The condition \eqref{randerscondition} implies the non-degeneracy  and the positiveness of the associated fundamental tensor \eqref{fundamentaltensor}. The proofs of these properties are sketched for instance in \cite{BaoChernShen}.

 We now consider the analogous of a Randers structure in the category of generalized Hamiltonian spaces on a tangent space $TM$ whose fundamental tensors \eqref{fundamentaltensor} are non-degenerate and have indefinite signature. In this case the domain of definition of the Hamiltonian function $F$ should be restricted, since it is not possible to have a well defined Hamilton-Randers function on the whole cotangent space $T^*TM$. This is because $\eta$ is a pseudo-Riemannian metric and it can take negative values on certain regions of $T^*_uTM$, in which case the function $\alpha(u,\theta)$ is purely imaginary and cannot be the value of a reasonable Hamiltonian function. This argument motivates to consider the collection $\mathcal{D}_{Tu}$ of {\it time-like momenta} over $u\in TM$, which is defined by the set of co-vectors $\theta\in \,T^*_u TM$ such that
\begin{align}
\alpha(u,\theta)=\,\sum^{8N}_{i,j=1}\,\eta^{ij}(u)\,\theta_i \,\theta_j\,>0.
\label{definitionofthecone}
\end{align}

The domain of a Hamilton-Randers function is the topological closure of the open submanifold $\mathcal{D}_T$ of  time-like  momenta. This type of domain is indeed a cone: if $\theta\in \,\mathcal{D}_{Tu}$, then $\lambda\,\theta\in\mathcal{D}_{Tu}$ for $\lambda\in\,\mathbb{R}^+$. Also, $\mathcal{D}_{Tu}$  is the pre-image  of an open set $(0,+\infty)$ by the Randers type function $F(u,\theta)$, which is continuous function on the arguments. Therefore, $\mathcal{D}_{Tu}$ is an open sub-manifold of $T^*_uTM$.

The notion of pseudo-Randers space is formulated in terms of well defined geometric objects, namely, the vector field $\beta\in\,\Gamma TTM$ and the pseudo-Riemannian norm $\alpha$. Because of this reason, a metric of pseudo- Randers type is denoted by the pair $(\alpha,\beta)$.
\subsubsection{Notion of Hamilton-Randers space}
Let $\mathcal{D}\subset\,T^*TM$ be a connected, fibered, open sub-manifold of $T^*TM$, where the fibers $\pi^{-1}(u)$ are subsets of the corresponding fiber $T^*_uTM$.
Let $\beta \in \Gamma\,TTM$ be a vector field on $TM$ such that the dual condition to the Randers condition \eqref{randerscondition}, namely, the condition
\begin{align}
|\eta^*(\beta,\beta)|<\,1, \quad  \beta\in \,\Gamma TTM
\label{boundenesscondition}
\end{align}
holds good.
\begin{definicion}
A Hamilton-Randers space is a generalized Hamilton space whose  Hamiltonian function is of the form
\begin{align}
F:\mathcal{D}\to \mathbb{R}^+\cup \{0\},\quad
(u,\theta)\mapsto\, F (u,\theta)=\,\alpha(u,\theta)+\beta(u,\theta).
\label{corandersmetric}
\end{align}
with  $\alpha=\,\sqrt{\eta^{ij}(u)\theta_i\theta_j}$ real on $\mathcal{D}\,\subset T^*TM$ and where
\begin{align*}
\beta(u,\theta)=\,\sum^4_{\mu =1}\sum^{N}_{k=1}\,{\beta}^\mu_k(u) \theta_{\mu k},
\end{align*}
 such that the condition \eqref{boundenesscondition} holds good.
\label{DefinicionHR}
\end{definicion}
A Hamilton-Randers space is characterized by a triplet $(TM,F,\mathcal{D})$.
The space of Hamilton-Randers structures on a given tangent space $TM$ will be denoted by $\mathcal{F}_{HR}(TM)$. The function $F^2$ is the Hamiltonian function of a sub-quantum atom.

The motivation to introduce the notion of Hamilton-Randers space as above is based on the fact that the dynamical systems \eqref{dynamicalsystem} are directly related with a symmetrization of the Hamilton-Randers function \eqref{corandersmetric}. In fact, we have that
in Hamilton-Randers theory, the sub-quantum atoms follow auto-parallel equations of the non-linear connection, that coincide with Hamilton equations of $F$, but the sub-quantum molecules does not follow geodesics of the Hamiltonian $F$. We will introduce the Hamiltonian function for sub-quantum molecules as a result of a  {\it time symmetrization operation} acting on the Hamiltonian $F$. Indeed, the dynamics of the sub-quantum molecules heritages several  characteristics from the dynamics of sub-quantum atoms. In particular, we shall argue the existence of an uniform upper bound for the proper acceleration and the existence of an induced non-trivial causal structure for both sub-quantum atoms and molecules. Furthermore, it is also clear that in the geometric setting of Hamilton-Randers systems, the covariant manifestation of upper bound for the speed of propagation of ontological degrees of freedom is consistent with an upper bound for the proper acceleration, by the argument discussed at the end of chapter \ref{chapter on Assumptions and General Theory}. This requires the revision of the foundations of spacetime geometry \cite{Ricardo2014}.

\subsection{The geometric flow $\mathcal{U}_t$}
 Let us consider {\it Assumption A.5} of {\it chapter} \ref{chapter on Assumptions and General Theory} in the framework of Hamilton-Randers structures on $TM$, that we denote by $\mathcal{F}_{HR}(TM)$. The geometric structures already incorporated in a Hamilton-Randers model, when understood  as in definition \ref{DefinicionHR}, serve as motivation to postulate the existence of a geometric flow $\mathcal{U}_t$ in the category of Hamilton-Randers spaces
\begin{align*}
\mathcal{U}_t:\mathcal{F}_{HR}(TM)\to \mathcal{F}_{HR}(TM'),
\end{align*}
where $M$ and $M'$ does not need to coincide.
The flow $\mathcal{U}_t$ is assumed to be geometric, that is, given in terms of equations involving the geometric structures defined on tangent spaces.
 Such a flow induces another geometric flow on the topological closure $\bar{\mathcal{D}}_T$. Note that under the flow $U_t$, the manifold $\bar{\mathcal{D}}_T$ can also be transformed.

 In order to accommodate the characteristics of the hypothetical flow $\mathcal{U}_t$ to the characteristics of the dynamics of the sub-quantum degrees of freedom required by the assumptions of {\it chapter} \ref{chapter on Assumptions and General Theory}, and specifically to Hamilton-Randers dynamical spaces, it is assumed that $\mathcal{U}_t$ has the following  general properties:
\begin{enumerate}
\item The Hamilton-Randers structure $(M_t,\alpha(t),\beta(t))$ evolves continuously with the parameter $t$ under the flow $\mathcal{U}_t$.

\item Although $M$ and $M'$ can be different, we assume that $\dim (M) =\,\dim (M')$.

\item There are $t$-time parameters such that under the $\mathcal{U}_t$ evolution the left-limit conditions
\begin{align}
\lim_{t^-\to 2n T} \mathcal{U}_t(\mathcal{D})\subset \,\Sigma_{2n T},\,n\in \mathbb{Z}
\end{align}
holds good and such that the sub-manifolds $\{\Sigma_{2nT}\}$ are small in the sense that
\begin{align}
\mu_P(\mathcal{D})>>\,\mu_P(\Sigma_{2nT}).
\label{condition of contraction}
\end{align}
\end{enumerate}
The condition of {\it small manifold} \eqref{condition of contraction} is problematic when applied to non-compact manifolds, for instance in the case when $\mathcal{E}_u$ are hyperboloids. A natural solution of this issue is to use a measure applicable even for non-compact subsects, for instance, by introducing a weigh factor in the volume form.

The above properties suggest the introduction of the domain ${\bf D}_0\subset \,T^*TM$ containing the points corresponding to the evolution at the instants
$\{t=\,2 n\,T,\,n\in\mathbb{Z}\}.$
The open domain ${\bf D}_0$ will be called the {\it metastable equilibrium domain}.
When $2nT<\,t<\,2nT+\delta< (n+1)\,T$ with $\delta$ positive but sufficiently large, the contraction condition \eqref{condition of contraction} does not hold, otherwise, the $U_t$ flow will continuously contracting the allowable phase space of the system or finish the evolution in an physically un-motivated very small phase space manifold $\Sigma_{2n_0 T}$. After the system leaves the neighborhood of ${\bf D}_0$ for each $n\in \, \mathbb{Z}$ during the $\mathcal{U}_t$ flow, it is hypothesized that the $U_t$ evolution starts an expansion process in $T^*TM$ of the allowable space until the dynamical system reaches a non-expanding phase. After the systems has fill the space, then a contrative regime starts and the systems falls again the ${\bf D}_0$ domain.

In the view of these considerations the original {\it Assumption A.5} on the dynamics of the fundamental degrees of freedom can be substituted by the following one:
 \bigskip

 {\bf Fundamental assumption on the structure of the fundamental evolution}. The $\mathcal{U}_t$ flow of a Hamilton-Randers system is composed by fundamental cycles. Each cycle is composed by an {\it ergodic regime} followed by a {\it concentrating regime} followed by an {\it expanding regime}. The semi-period $T$ of each physical system depends on the particularities of the physical system. For stationary systems, there is a choice of the $t$-parameter such that the period $2T$ does not depend upon the cycle.
\bigskip

To specify the $\mathcal{U}_t$ dynamics with the required properties is an open problem in our approach. It is linked with the dynamics of the form \eqref{dynamicalsystem} compatible with the cyclic properties above. One possible flow is described by the following construction,
\begin{definicion}
The $U_t$ dynamics in the interval $[0,2T]\subset\,\mathbb{R}$ is a geometric flow of the form
\begin{align}
\begin{cases}
& \mathcal{U}_t: \mathcal{F}_{HR}(TM)  \to  \mathcal{F}_{HR}(TM),\\
& F\mapsto F_t =\sqrt{\kappa(u,\theta,t,\tau)\,\eta(\tau)+(1-\kappa(u,\theta,t,\tau))\,g(\tau)},\,
\label{evolutionofgeometricstructures}
\end{cases}
\end{align}
such that the function
$\kappa: \mathcal{F}_{HR}(TM) \times [0,2T]\to [0,1]$
 satisfies the boundary conditions
\begin{align}
\, \lim_{t\to\,0^+} \kappa(u,\theta,t,\tau)=0,\quad \lim_{t\to 2T^-} (1-\kappa(u,\theta,t,\tau))=0.
\label{limiteofk}
\end{align}
\label{definiciongeometricevolution}
\end{definicion}
We make the following further assumptions on the factor $\kappa$,
\begin{itemize}
\item We assume that the condition
\begin{align}
\kappa(u,\theta,t,\tau)=\,\kappa(t,\tau)
\label{condition on kappa}
\end{align}
holds good.

\item We assume that $\kappa$ is cyclic on the $t$-time parameter with period $2T$.

\end{itemize}

The $\mathcal{U}_t$-flow determines an homotopic transformation in the space $\mathcal{F}_{HR}(TM)$. Note that in this definition $t\in [0,2T]$, instead of  $t\in\,\mathbb{R}$, in concordance with the assumption that the  dynamics becomes cyclic or almost cyclic with semi-period $T$. However, this formulation should be considered only as an approximation to the exact dynamics, since the dynamics is almost-cyclic and not completely cyclic.

Due to the proposed structure of the $\mathcal{U}_t$ flow, there is a domain where $\mathcal{U}_t$ approximately corresponds to the identity operator. This domain is the metaestable domain ${\bf D}_0$.
\subsubsection{t-time inversion operation}
The parameter $t\in\,\mathbb{R}$ is interpreted as the time parameter for an {\it internal dynamics} of the system.
The {\it time inversion operation} $\mathcal{T}_t$ is defined in local natural coordinates on $T^*TM$ by the operator
\begin{align}
\begin{cases}
&\mathcal{T}_t:T^*TM \to T^*TM, \\
& (u,\theta)=(x,y,\theta_x,\theta_y)\mapsto (\mathcal{T}_t(u),\mathcal{T}^*_t(\theta))=(x,-y,\theta_x,-\theta_y).
 \end{cases}
 \label{timeinversionoperation}
\end{align}
This operation does not depend on the choice of the coordinate system: if the relation \eqref{timeinversionoperation} holds in a given natural coordinate system, it holds in any other natural coordinate system on $T^*TM$.
  $\mathcal{T}_t$ is an idempotent operator,
\begin{align}
\mathcal{T}\,^2_t=Id,\quad \forall \, t\in \,I\subset \,\mathbb{R},
\end{align}
where $Id$ is the identity operation on $T^*TM$, such that locally $(u,\theta)\mapsto (u,\theta)$.

It is also invariant under positive oriented re-parameterizations of the $t$-time parameter.

The induced action of $\mathcal{T}_t$ on an element $F\in\,\mathcal{F}_{HR}(TM)$ is given by the expression
\begin{align*}
\mathcal{T}_t(F)(u,\theta):=F(\mathcal{T}_t(u),\mathcal{T}_t(\theta)).
\end{align*}
Note that a Hamilton-Randers metric is non-reversible in the sense that
\begin{align*}
 F (u,\theta)\neq F (\mathcal{T}_t(u),\mathcal{T}_t(\theta)),
 \end{align*}
except for subsets of zero measure in $T^*_u TM$. The intrinsic irreversible character of the Randers geometry follows from this property. Furthermore, from the definition of the norm $\alpha$ as determined by a Sasaki type metric of the form $\eta=\sum^N_{k=1} \eta_4(k) \oplus \eta_4(k)$, we have that
 \begin{align}
 \mathcal{T}_t(\alpha)=\,\alpha.
 \end{align}
Effectively, the structure of the pseudo-Riemannian metric $\eta$ in Hamilton-Randers theory is of a sum of Sasaki types forms, $\eta = \,\sum^N_{k=1}\, \eta^k \oplus \eta^k$, where $\eta^k$ is a Lorentizian manifold on $M^k_4$ and depends only on $x$-coordinates.
On the other hand, the invariance of the dynamical system \eqref{dynamicalsystem} implies that $\beta_x$ in the form of the Randers function, it follows that
\begin{align}
 \mathcal{T}_t(\beta)=\,-\beta.
 \end{align}

We assume  that $\mathcal{T}_t$ commutes with the $\mathcal{U}_t$ dynamics in the sense that
\begin{align}
[\mathcal{U}_t,\mathcal{T}_t]=0, \quad \quad \textrm{when }  \,t=2 T n.
\label{relationUT}
\end{align}
If the relation \eqref{relationUT} holds, then we have that
   \begin{align}
   \lim_{t\to2 T}\,\mathcal{T}_t(\kappa(t,\tau))=\,\lim_{t\to 2T}\kappa(t,\tau)=1.
   \label{inversion of kappa}
   \end{align}
Assuming linearity in the action of the operator $\mathcal{T}_t$, we have that
   \begin{align*}
   \mathcal{T}_t(\eta) & =\,\mathcal{T}_t(\lim_{t\to 2Tn} \mathcal{U}_t(F))\\
    & =\,\mathcal{T}_t(\lim_{t\to 2 n \underline{}T} \kappa(t,\tau)\,\eta(u,\theta)+(1-\kappa(t,\tau))\,g(u,\theta))\\
    & =\,\lim_{t\to 2T} \mathcal{T}_t(\kappa(t,\tau)\eta(u,\theta)+(1-\kappa(t,\tau))\,g(u,\theta))\\
    & =\,\lim_{t\to 2T}\,(\mathcal{T}_t\kappa(t,\tau)\,\mathcal{T}_t(\eta)(u,\theta)+\mathcal{T}_t(1-\kappa(t,\tau))\,\mathcal{T}_t(g)(u,\theta))\\
      & =\,\lim_{t\to 2T} (\mathcal{T}_t\kappa(t,\tau)\,\eta(u,\theta)+\mathcal{T}_t(1-\kappa(t,\tau))\,g(u,\theta)).
   \end{align*}
   If the relation \eqref{relationUT} holds, the condition $\mathcal{T}_t \eta=\eta$ implies \eqref{inversion of kappa}.

\subsection{The Hamiltonian of the fundamental degrees of freedom}
Hamilton-Randers models and linear Hamiltonian functions in the momentum variables are related.
If $(TM, F, \mathcal{D})$ is a Hamilton-Randers space
that evolves towards the {\it averaged structure} $(TM,h)$ under the $U_t$ flow, for each value of $t$-parameter
there is an element $(TM,F_t)$ of $\mathcal{F}_{HR}(TM)$. In this setting, the implementation of {\it Assumption 3} and {\it Assumption 8} is achieved in the following way,

\bigskip
{\it Assumption 3+8}. On each individual sub-quantum molecule, one of the quantum atoms evolves with a Hamiltonian function $\frac{1}{2}F^2_t(u,\theta)$ and the second sub-quantum atom with a Hamiltonian $\frac{1}{2}F^2_t(\mathcal{T}_t(u),\mathcal{T}_t(\theta))$.
\bigskip

Applying the time inversion operation $\mathcal{T}_t$ to $F_t$ and taking into account that the function $\kappa(t,\tau)$ is invariant under $\mathcal{T}_t$, the corresponding Hamiltonian function of a Hamilton-Randers system at the instant $(t,\tau)$ must be of the form
\begin{align}
H_t(u,\theta)& =\,\frac{1}{4}\,\left(F^2_t (u,\theta)\,-F^2_t(\mathcal{T}_t(u),\mathcal{T}^*_t (\theta))\right).
\label{Hamiltonian of a Randers function}
\end{align}
The relative minus sign is because the Hamiltonian function is the generator function of time translations. Since the time inversion operation also changes the time lapse, the change in any quantity must be also affected for the relative change of sign in the $1$-form associated with time.

 Let us evaluate the difference
    \begin{align*}
 &\,F^2_t (u,\theta)\,-F^2_t(\mathcal{T}_t(u),\mathcal{T}^*_t (\theta)) \\
& =(1-\kappa(t,\tau))g^{ij}(u,\theta)\theta_i\theta_j+\,\kappa(t,\tau) \eta^{ij} \theta_i\theta_j
-(1-\kappa(t,\tau))g^{ij}(\mathcal{T}_t(u),\mathcal{T}^*_t(\theta))\theta_i\theta_j-\,\kappa(t,\tau)\eta^{ij} \theta_i\theta_j\\
& =\,\left[(1-\,\kappa(t,\tau))(\alpha+\,\beta)^2+\,\kappa(t,\tau)\,\eta^{ij}\theta_i \theta_j\right]
-\,\left[(1-\,\kappa(t,\tau))(\alpha-\,\beta)^2+\,\kappa(t,\tau)\,\eta^{ij}\theta_i \theta_j\right]\\
 &  =\left[(1-\,\kappa(t,\tau))(\alpha+\,\beta)^2\right]
-\,\left[(1-\,\kappa(t,\tau))(\alpha-\,\beta)^2\right]\\
& =\,4\,(1-\,\kappa(t,\tau))\alpha\,\beta.
\end{align*}

We fix the norm of the momenta by the relation
\begin{align}
\alpha (u,\theta)=\,\sqrt{\eta^{ij}\,\theta_i \theta_j}=\,1.
\label{momentum hyperboloid}
\end{align}
This condition is analogous to the condition of unit hyperboloid in general relativity.
Then we have
\begin{align*}
 \,\left(\frac{1}{4}F^2_t (u,\theta)\,-\frac{1}{4}F^2_t(\mathcal{T}_t(u),\mathcal{T}^*_t (\theta))\right) &=\,(1-\,\kappa(t,\tau))(\alpha+\beta - (\alpha -\beta))\\
& =\,4\,(1-\,\kappa(t,\tau))\,\alpha\,\beta\\
& =\,4\,(1-\,\kappa(t,\tau))\,\beta.
\end{align*}
Therefore, the Hamiltonian function associated with a HR-system  at $(t,\tau)\in\,\mathbb{R}\times \mathbb{R}$ is
\begin{align}
H_t(u,\theta)=\,(1-\kappa(t,\tau))\,\left(\,\sum^{N}_{k=1}\sum^4_{\mu=1}\,\beta^\mu_{x k}(u)\theta_{x k \mu}+\,\sum^{N}_{k=1}\sum^4_{\mu=1}\,\beta^\mu_{y k}(u)\theta_{y k \mu}\right),
\label{Classical Hamiltonian 0}
\end{align}
subjected to the constrain \eqref{momentum hyperboloid}.

We can appreciate the reason for the choice of the $t$-time operation by considering the transformation of the Hamilton equations for the local coordinate and the velocity coordinates of the sub-quantum molecule. If the $k$-essim molecule has coordiantes and velocities $(x^\mu_k,y^\mu_k)$,
the Hamilton equations are formulated in Poisson formalism as
\begin{align}
\begin{cases}
&\frac{d}{dt}x^\mu_k =\,\big\{ H_t, x^\mu_k\big\},\\
&\frac{d}{dt} y^\mu_k =\,\big\{ H_t, y^\mu_k\big\} .
\end{cases}
\label{Poison equations for x y}
\end{align}
The equivalent to the equation for the new $t$-time parameter associated with the $1$-form $-dt$, must be of the form
\begin{align*}
\begin{cases}
& -\frac{d}{dt}\mathcal{T}_t(x^\mu_k)=\,\big\{ \mathcal{T}_t(H_t), \mathcal{T}_t(x^\mu_k)\big\},\\
& -\frac{d}{dt} \mathcal{T}_t(y^\mu_k) =\,\big\{ \mathcal{T}_t(H_t), \mathcal{T}_t(y^\mu_k)\big\} .
\end{cases}
\end{align*}
According to the definition \eqref{timeinversionoperation} of the $t$-time inversion operation, the consistency of these equations with \eqref{Poison equations for x y} implies that
\begin{align}
\mathcal{T}_t (H_t)=\,-H_t.
\end{align}
This condition expresses the irreversible character of the dynamics associated with sub-quantum molecules. This irreversibility is the reason for the natural interpretation of several fundamental concepts in physics as emergent, among them, the association of irreversibility in the real world with the non-reversibility of the sub-quantum dynamics and the origin of the external time parameters as emergent.
\subsubsection{Canonically conjugate coordinates for the fibers of $T^*TM$} Until now we have use the coordinate systems of the type $\{(x,y,\theta_x,\theta_y)\}$ over $T^*TM$. In particular, it was stressed the tensorial character of the coordinates
 \begin{align*}
 \{\theta_{\mu kx},\theta_{\mu y k},\,=\mu=1,2,3,4,\,k=1,...,N\}.
  \end{align*}
  However, for the following developments, it is useful to consider, instead of these tensorial coordinates, a set of non-tensorial coordinates
 \begin{align*}
\{p_{\mu kx},p_{\mu y k},\,=\mu=1,2,3,4,\,k=1,...,N\}
\end{align*}
on each fiber over $u\in TM$. Contrary with the coordinates that we starting working, these new fiber coordinate systems are holonomic and canonical conjugate to the local vector fields $\frac{\partial}{\partial x^\mu}$. We will use these coordinates from now on, since they are useful in the proof of the consistency of the quantum conditions under changes of natural coordinate systems, as we shall discuss in the next {\it chapter}. The corresponding transformation rules for these coordinates induced by local coordinate transformations in $M$ will be specified later. Remarkably, these coordinates appear naturally in higher order geometry, for instance, in Finsler geometry constructions \cite{BaoChernShen,MironHrimiucShimadaSabau:2002}.

In order to keep a convenient notation for our purposes, the components $\tilde{\beta}^\mu_k$ of the vector field $\beta \in\,\Gamma TTM$ must also be renamed. It is convenient to make this in such a way that the Hamiltonian \eqref{Classical Hamiltonian 0}, is written in the form
\begin{align}
H_t(u,p)=\,(1-\kappa(t,\tau))\left(\,\sum^{N}_{k=1}\sum^4_{\mu=1}\,\beta^\mu_{x k}(u)p_{x k \mu}+\,\sum^{N}_{k=1}\sum^4_{\mu=1}\,\beta^\mu_{y k}(u)p_{y k \mu}\right).
\label{Classical Hamiltonian}
\end{align}
Note that in this new notation, the functions $\beta^\mu_k$ do not transform tensorially. The corresponding transformation laws will be discussed later in this {\it section}.

As  consequence of the limit conditions \eqref{limiteofk}, it holds the following
\begin{proposicion}
The Hamiltonian \eqref{Classical Hamiltonian} in the metastable domain ${\bf D_0}$ is identically zero,
\begin{align}
\lim_{t\to 2 n\,T}\,  H_t(u,p)=0.
\label{finalhamiltoniantevolution}
\end{align}
\label{proposicion of final Hamiltonian}
\end{proposicion}
This condition formally determines the metastable points $\{t=\,2 n T,\,n\in\,\mathbb{Z}\}$. Second, it is intrinsically related with the flow of geometric structures $(TM,F)\to (TM,h)$,
showing that the $\mathcal{U}_t$ flow is dissipative in the contractive domain. Third, it is compatible with the time re-parametrization invariance property of general covariant field theories. Note that this consequence of \eqref{finalhamiltoniantevolution} is not equivalent to a Hamiltonian constraint as it appears in general relativity, but it shows the compatibility of the $\mathcal{U}_t$-flow with the Hamiltonian  constraint, at least in weak form.
\subsection{Hamilton-Randers dynamical models for composite systems}\label{Hamilton Randers models for composite systems}
The configuration space $\mathcal{M}$ of a Hamilton-Randers system is a product manifold $M\simeq \prod^N_{N=1}\,M^k_4$. The dynamical structures are defined in the co-tangent space $T^*TM$, in a form such that the metric structure $\eta$ is a direct sum, but the $1$-form $\beta$ is defined on the disjoint union of co-tangent spaces. A direct generalization of the construction is provided in the following paragraphs.
\begin{definicion}
A {\it disjoint Hamilton-Randers system} $a\sqcup \,b$ of two Hamilton-Randers systems $a$ and $b$ is a Hamilton-Randers system whose set of degrees of freedom is the disjoint union of the sets of degrees of freedom of $a$ and $b$.
\end{definicion}
If the sub-quantum molecules determining the quantum system $a$ are described by a model $(M_a, (\alpha_a,\beta_a))$ and the sub-quantum degrees of $b$ are described by a model $(M_b,(\alpha_b,\beta_b))$ respectively, then there is a natural Randers structure constructed from the Hamilton-Randers structures associated with the systems $a$ and $b$. This is the {\it direct sum of structures},
   \begin{align}
   \left(M_a\times M_b,\, \left(\eta_a\oplus \eta_b,\,\beta_a\oplus \beta_b \right)\right).
   \label{sum of systems}
   \end{align}
The direct sum structure defines a particular class of disjoint  Hamilton-Randers systems of the form $a\sqcup \,b$,
   \begin{definicion}\label{disjoint union of non-interaction}
The disjoint union of non-interacting  Hamilton-Randers systems is the disjoint union as a sets endowed with the sum structure \eqref{sum of systems}.
\end{definicion}
When a disjoint union Hamilton-Randers system has the direct sum structure, the dynamical degrees describing $a$ evolve on the $U_t$ dynamics independently of the dynamical degrees of the system $b$: the systems $a$ and $b$ decouple from the sub-quantum point of view description. That is, the system is non-interacting from the point of view of sub-quantum degrees of freedom.

\subsection{The semi-period $T$ for a particular choice of the $t$-parameter}
In Hamilton-Randers theory there is no geometric structure defined on the tangent space $TM$ that can be used to define a natural $t$-time parameter for the $\mathcal{U}_t$ dynamics. Hence the models that we need to consider must be invariant under positive oriented $t$-time re-parameterizations. However, the freedom in the choice of the $t$-parameter required for consistent with the principle of general covariance does not contradict the existence of choices for the $t$-parameter particularly enlightening, similarly as it happens in other theories and models. In particular, it can be useful to identify the semi-period $T$ as a characteristic of the system to which it is associated, since fixed a $t$-time parameter, different systems can have different semi-periods.

There must be a minimal semi-period $T_{min}$, an universal scale that corresponds to the minimal period of the fundamental cycles for any Hamilton-Randers dynamical system. This minimal period must exists, since by assumption, the sub-quantum degrees of freedom are different from the quantum degrees of freedom and since, by assumption, quantum systems contain a finite number of them, then the limit case when $N=1$ imposes a limit to the complexity of the quantum system. However, let us remark that this shows the need for the existence of a $T_{min}$ if the period is directly linked with the number of sub-quantum degrees of freedom.

The value of $T_{min}$ depends on the choice of the arbitrary $t$-parameter but when the parameter has been chosen, then $T_{min}$ achieves the same universal value for all the Hamilton-Randers systems. We will show later that $T_{min}$ corresponds to the semi-period of the $\mathcal{U}_t$ of the {\it smallest} possible Hamilton-Randers system.

Given a $t$-parameter with the above properties, for each Hamilton-Randers system we postulate the existence of a class of $t$-time parameters of the form $[0,2\,T]\subset \,\mathbb{R}$ such that
\begin{align}
 \log_2\left(\frac{T}{T_{min}}\right)=\,\frac{T_{min}\,m\,c^2}{\hbar}
\label{ETrelation}
\end{align}
As a consequence of the expression \eqref{ETrelation} the minimal value of $m$, associated with $T_{min}$ is $m=0$. However, the above expression suggests an scale parameter $\widetilde{m}=\,{\hbar}/{T_{min}\,c^2}$.

It is through the expression \eqref{ETrelation} that we introduce in Hamilton-Randers theory the constant $\hbar$ with dimensions of action, that makes the quotient on the right hand side invariant under conformal changes of units in the $t$-time parameter.

There is no reason to adopt $T_{min}$ as the Planck time. Indeed, as we shall discuss later in relation with our interpretation of quantum entanglement from the point of view of Hamilton-Randers theory, for the models that we will investigate, the parameter $T_{min}$ could be larger than any time scale associated or close scale to the time of Planck.

The parameter $T$ is a measure of the complexity of the system. Intuitively, this is clear, because for simple systems one expect $T$ small than for large systems.
\subsubsection{The parameter {\it m} as a notion of mass}
It is useful to re-cast the relation \eqref{ETrelation} as a definition of the mass parameter $m$ in terms of the semi-period $T$ of the fundamental cycle.
\begin{definicion}
The mass parameter $m$ of a Hamilton-Randers dynamical system with fundamental semi-period $T$ is postulated to be given by the relation
\begin{align}
m=\,\frac{\hbar}{T_{min}\,c^2}\,\log_2 \left( \frac{T}{T_{min}}\right).
\label{definitionofinertialmass}
\end{align}
\end{definicion}

{\bf Fundamental properties of the $m$ parameter}:
\begin{enumerate}

\item \label{positivity of m}\label{non negativity of mass m} For any Hamilton-Randers system, the mass parameter $m$ is necessarily non-negative, with the minimum value for the parameter $m$ equal to $0$ when $T=\,T_{min}$; the first non-trivial value of $m$ corresponds to $T=2\,T_{min}$ and is $\widetilde{m}=\,\frac{\hbar}{T_{min}\,c^2}$.

\item \label{conservation of m} \label{preservation of mass for a free system} As long as the period $T$ is preserved for each fundamental cycle of the $\mathcal{U}_t$ dynamics, the mass parameter $m$ is preserved and remains the same for each cycle of the $\mathcal{U}_t$ dynamics.

\item \label{atribution of m} Since the semi-period $T$ is an attribute of the physical system under consideration, the mass parameter $m$ is also an attribute that increases with the {\it complexity of the system}.

\item \label{m is a measure of complexity of the system} Since the value of the semi-period is linked with the ergodic properties of the system, $T$ is a measure of the complexity of the system. Hence $m$ is a measure of the complexity of the system as well.

\end{enumerate}
The parameter $T$ has not been specified in terms of the ontological degrees of freedom of the system and hence and it has a high degree of arbitrariness yet. But due to the properties \ref{non negativity of mass m}-\ref{preservation of mass for a free system}, it is reasonably to identify the parameter $m$ given by the relation \eqref{definitionofinertialmass} as the mass parameter of the Hamilton-Randers system.

The above interpretation of the parameter $m$ is linked to a class of $t$-parameters with the properties above mentioned, namely, that the evolution is periodic on $t$-time and that there is a minimal period $T_{min}$ for some systems. Given a dynamical system, not all the possible $t$-time parameters satisfy such properties. Hence the parameter $m$ is not invariant under general $t$-time re-parameterizations. Indeed, given any positive $t$-parameter for the $\mathcal{U}_t$ dynamics, one could  define a parameter $m$ by the relation \eqref{definitionofinertialmass}, but in general, such a parameter $m$ will not be associated with the Hamilton-Randers system, neither will satisfy the properties  \ref{non negativity of mass m}-\ref{preservation of mass for a free system} above, specially, the constancy of mass property \eqref{preservation of mass for a free system}.

Two limiting cases are of special interest. If $T=T_{min}$ holds, then the mass parameter is zero, $m= 0$. This case describes Hamilton-Randers systems corresponding to quantum massless particles. Fixed a $t$-time parameter such that the relation \eqref{ETrelation} holds good, all massless quantum systems have the same semi-period $T=\,T_{min}$. The second obvious limit is when $T\to +\infty$,  which corresponds to the limit $m\to +\infty$. This situation corresponds to a Hamilton-Randers systems describing a system with infinite mass.

\subsubsection{Period and mass parameters for composite Hamilton-Randers systems}
Let us consider two arbitrary Hamilton-Randers systems $a,b$ with semi-periods $T_a$ and $T_b$.
If the systems $a$ and $b$ do not interact, then we further assume the existence of a semi-period $T_{a \,\sqcup\, b}$ for the system $a\,\sqcup\,b$ such that the multiplicative rule
 \begin{align}
T_{a \,\sqcup\, b}:=\,{T_a}\, {T_b}.
 \label{compatibilityconditionabwhennotmeasurable}
 \end{align}
 holds good.
 This is a natural generalization law for the periods of product of independent periodic functions. As before, note that this rule is not satisfied for a generic $t$-time parameter.

 The condition \eqref{compatibilityconditionabwhennotmeasurable} can be taken as the definition of non-interacting systems, where non-interacting means non-interacting at the quantum mechanical level.

\begin{ejemplo} According to Hamilton-Randers theory, a single electron system is described as a system composed by many sub-quantum degrees of freedom evolving coherently on the $\mathcal{U}_t$ dynamics, with a semi-period $T_e$, with $T_e > 1$ ($e$ stands for electron). The semi-period of a system composed by two identical, non-interacting electrons is then described in Hamilton-Randers systems by the product of semi-periods $T^2_e$. Similarly, if the system is of the form $e\sqcup e\sqcup e...\sqcup e$ $n$ times, then the semi-period of the system is $T^n_e$.
\end{ejemplo}
The above example illustrates several general feature of the theory:
\begin{itemize}

\item The semi-period associated with a composed system should increase exponentially with the components of the system.

\item From the relation \eqref{compatibilityconditionabwhennotmeasurable} between the mass $m$ and the semi-period $T$ given by the relation \eqref{ETrelation}, it follows that
the mass parameter of a composed non-interacting system $a\sqcup \,b$ is
\begin{align}
 m_{a \,\sqcup\, b}=\,m_a+\,m_b.
 \label{additivityofmass 0}
 \end{align}
 For this relation to hold, it is necessary to define the $t$-parameters such that $T_{min}=1$.

\item The relation \eqref{additivityofmass 0} only holds if the $t$-time parameters associated with the systems $a$, $b$ and $a\sqcup \,b$ are identical. Such a situation only happens if the co-moving parameters to $a$, $b$ and $a\sqcup \, b$ coincide, that is, if the systems are in relative rest to each other. In general, this is not the case. In that more general case, the relation \eqref{additivityofmass 0} must be substituted by the phenomenological relation among the {\it corresponding energies}
 \begin{align}
 \mathcal{E}_{a \,\sqcup\, b}=\,\mathcal{E}_a+\,\mathcal{E}_b.
 \label{additivityofmass}
 \end{align}
 This relation can be interpreted as the conservation of energy for close dynamics or as the additivity of energy for non-interacting systems.
\end{itemize}

\subsection{Models for the semi-periods of the internal dynamics}\label{models of periods in terms of mass} The character of the $t$-time parameters is arbitrary. However, the interesting features of the $t$-parameters for which the relations \eqref{ETrelation}, \eqref{compatibilityconditionabwhennotmeasurable} \eqref{additivityofmass} hold good, invites to consider the construction of specific models for such special $t$-parameters.

The fundamental characteristic of a Hamilton-Randers system $a$ is the number of sub-quantum degrees of freedom $N$. According to the above premises, let us consider the model where the semi-period $T$ of a Hamilton-Randers system $a$ is defined as a function of $N$. Also, we would like to pointed out that the dynamics $\mathcal{U}_t$ preserves $N$. Thus $N(t)$ is constant for any choice of the $t$-parameter. In a process where the system can be subdivided in two parts, then the number $N$ is also subdivided: if there is a process of the form $1\to 2\sqcup 3$, then we assume the condition
\begin{align}
 N_1=\,N_2+N_3\,-N_{min}.
 \label{preservation of the number of lines}
\end{align}
This relation can be interpreted by stating that the number of degrees of freedom of a composed system is the sum of the independent degrees of freedom; then $N_{min}$ appears as the number of degrees of freedom of a common border and is subtracted in order to do not count them twice.

In the following, we consider two models relating the semi-period $T$ as function of the number of degrees of freedom $N$,
\subsubsection{Model 1} Let us consider the following model for the semi-period,
\begin{align}
\frac{T}{T_{min}}:=\,\frac{T(N)}{T_{min}} =\,N^{\beta},
\label{model 1 for semiperiod}
\end{align}
where $\beta$ is a positive constant that does not depend upon the specific Hamilton-Randers  system, neither on  the number of sub-quantum degrees of freedom $N_a$. The model \eqref{model 1 for semiperiod} for the semi-period is consistent with the relations \eqref{ETrelation}, but it is not consistent with the conditions  \eqref{compatibilityconditionabwhennotmeasurable}-\eqref{preservation of the number of lines}.

\subsubsection{Model 2}  A model that it is consistent with the condition \eqref{compatibilityconditionabwhennotmeasurable} is the exponential model
\begin{align}
T=\,T_{min}\,2^{\left(\frac{\lambda\, T_{min}\,c^2}{\hbar}\left(N-N_{min}\right)\right)},
\label{model 2 for semiperiod}
\end{align}
where $\lambda$ is a constant relating the mass parameter $m$ and the number of degrees of freedom.
The period of a non-interacting composite system is of the form
\begin{align*}
T_{2\sqcup\,3} & =\,T_{min}\, 2^{\left(\frac{\lambda \,T_{min}\,c^2}{\hbar}\left(N_2+N_3-N_{min}-N_{min}\right)\right)} \\
& =\,T_{min}\,2^{\left(\frac{\lambda\, T_{min}\,c^2}{\hbar}\left(N_2-N_{min}\right)\right)}\cdot \,2^{\left(\frac{\lambda\, T_{min}\,c^2}{\hbar}\left(N_3-N_{min}\right)\right)}\\
&=\,\frac{1}{T_{min}}\,T_2\cdot \,T_3
\end{align*}
Thus the rule \eqref{compatibilityconditionabwhennotmeasurable} holds
if and only if $T_{min}=1$.

The model given by the expression \eqref{model 2 for semiperiod} $N_{min}$ corresponds to $T_{min}=1$.

By the model of semi-period \eqref{model 2 for semiperiod} the mass parameter is
\begin{align}
m = \, \lambda\left(N-\,N_{min}\right).
\label{relacion masa con N}
\end{align}
According to this model, the mass parameter $m$ is equivalent to the number $N$ of sub-quantum degrees of freedom minus the number $N_{min}$ of numbers of the {\it border}. Therefore, $m$ is a measure of the {\it quantity of sub-quantum degrees of freedom} and also we reach the following consequence,
 Hence, we reach the following result:
\begin{proposicion}
For any physical system, the inertial mass parameter $m$ is a non-negative, real number.
\end{proposicion}

Model \eqref{model 2 for semiperiod} has further consequences that make it convenient.
From the definition \eqref{definitionofinertialmass} of the mass parameter $m$ and the relation \eqref{relacion masa con N}, preservation of the number of lines $N$ in the sense of the relation \eqref{preservation of the number of lines} implies additivity of the mass parameter $m$,
\begin{proposicion}
For the model of the $t$-time parameter given by the expression \eqref{model 2 for semiperiod}, the preservation of the number of degrees of freedom $N$ in the sense of the relation \eqref{preservation of the number of lines} implies the additivity of the mass parameter \eqref{definitionofinertialmass} and of the energy.
\label{preservation of N implies preservation of mass}
\end{proposicion}

Note that, although $T_{min}=1$, $N_{min}\gg 1$, leading to an emergent interpretation for massless quantum systems.

Moreover, since for $N_{min}\mapsto m=0$, $N=\,2\, N_{min}$, if $\lambda=\frac{\hbar}{T_{min}\, c^2}$ and since $T_{min}=1$ then there is a three-fold quantization,
\begin{align}
\begin{cases}
& \,N, \\
& \,T,\\
& \,m
\end{cases}
\quad {\bf =} \quad
\begin{cases}
& \,n\,N_{min},\, \\
&\,2^{(n-1)\,N_{min}},\\
 &\,m=\frac{\hbar}{c^2}\,N_{min}(n-1),
\end{cases}
\label{quantization of mass}
\end{align}
with $n\in \,\mathbb{N}$.

From now on in Hamilton-Randers theory we adopt model \eqref{model 2 for semiperiod} for the semi-period $T$ with $T_{min}=1$ and $\lambda =\, \frac{h}{T_{min}\,c^2}$, if anything else it is not specified.

\subsection{The law of inertia as an emergent phenomenon: First considerations}\label{law of the inertia as emergent phenomenon 1}
Within this context of associating the mass parameter $m$ with the number of degrees of freedom of the system of the system, there is a notable consequence. According with the model \eqref{model 2 for semiperiod}, if $m$ is related with the number of sub-quantum degrees of freedom, then $T$ is related with the number of subsets of the set of sub-quantum degrees of freedom of the system. At each instant of the $t$-time, the system has associated one configuration, a partition of the set of fundamental degrees of freedom in different subsets. To each configuration there is associated a weight, given by the relative amount of $t$-time that the system has associated such configuration. The semi-period $T$ is a measure of the cardinality of the number of weights; a state is a selection of the possible weights.

According to such a picture, a coherent change in the state is a change in the collection of weights. We observe that the tendency of a system to change  from an structure of weights to another structure of weights when evolving from one cycle to another cycle by random changes due to a non-interacting environment in the macroscopic sense, is more unlike to happen for a complex system than for a simpler system. That is, the complexity of the system, in this case measured by the parameter $m$, opposes to dynamical consistent changes on the system in a way that the more the complexity, the more the opposition.

This interpretation of the {\it resistance to change the state} is based upon the assumption of the existence of stable configurations from cycle to cycle, the assumption that systems are composed by large number of sub-quantum degrees of freedom $N$ and the environment is not non-interacting in the quantum sense, although at the level of sub-quantum degrees of freedom of the quantum system, there are interactions. This argument constitutes an emergent interpretation of the law of inertia, since due to the large number of degrees of freedom $N$, the change from an stable configuration to another stable configuration, decreases  exponentially with $N$. It is a probabilistic interpretation, although based upon a deterministic framework.

Two characteristics properties of this interpretation are:
\begin{itemize}
\item The opposition to change in the dynamical state is larger the larger the semi-period $T$ is, or equivalently, the larger the mass $m$ is.

\item It is the degree of complexity $T$, the exponentially dependence on $N$ of the form $2^{-N}$, which makes the spontaneous change of state exponentially suppressed by a large exponential law if there is no macroscopic change due to an interaction with the environment.

\end{itemize}
Therefore, we have the following result:
\begin{teorema}
If a Hamilton-Randers system is evolving in a non-interacting environment, then there is no significant dynamical change in the quantum state of the system.
\label{ley de la inercia}
\end{teorema}
This is the emergent interpretation of the inertia law of mechanics.

\subsection{Emergent nature of the quantum uncertainty relation}
The relation \eqref{ETrelation} is not equivalent to the quantum mechanical energy-time uncertainty relation, since the $t$-parameter is not an external time parameter, as it is the one appearing in the quantum relation. This is because the relation between $m$ and $T$ is given by $\log T$ instead than being linear with $T^{-1}$, as it should be expected for a quantum relation. This is indeed a significant difference with the quantum mechanical energy-time uncertainty relation. Furthermore, according with the relation \eqref{ETrelation}, $m$ increases monotonically with $T$.

  However, if we consider the variation of the parameter $m$ due to a variation of the period $2\,T$ in the relation \eqref{ETrelation}, we have that
 \begin{align*}
 \Delta \left(m\,c^2\right)=\,(\Delta m)\,c^2 =\,\frac{\hbar}{T_{min}}\,\frac{\Delta{T}}{T},
 \end{align*}
 since the speed of light in vacuum is universal and does not depend on the specific system with $T_{min}$. The variation in the mass can be conceived as a variation due to a continuous interaction of the system with the environment; $\Delta T/T_{min}$ is the number of fundamental dynamical cycles that contribute to the stability of the quantum system. It is a large integer number. The expression  $m\,c^2$ is the energy of a system measured by an observer instantaneously co-moving with the system \cite{Ricardo2014} when the system has zero proper acceleration.  If there is an associated local instantaneous co-moving coordinate system associated with the Hamilton-Randers system as a whole, then we can apply the relativistic expression for the energy as given in \cite{Gallego-Torrome2019}. The value of the semi-period $T$ is a characteristic of the quantum system associated in such a local coordinate system.
 From the definition of $T_{min}$ we have that $\Delta{T}/T_{min}\geq 1$. The {\it spread of energy} for a massive system is  defined as
 \begin{align*}
 \Delta E :=\,\Delta(m\,c^2)=\,\Delta(m)\,c^2.
 \end{align*}
 We obtain an {\it emergent uncertainty relation},
 \begin{align}
 \Delta E\,2\,T\geq 2\,\hbar
 \label{Energytimeuncertainty}
 \end{align}
 Therefore, the {\it uncertainty} in the energy at rest $E$ associated  with the system is proportional to the inverse of the semi-period $T$, when the semi-period of the fundamental dynamics is described using a very particular kind of time parameter, as specified above.

 Note that the derivation of the relation \eqref{Energytimeuncertainty} has been for massive $m\neq 0$ systems.
\subsubsection{Interpretations of the emergent energy-time uncertainty relations}
The macroscopic interpretation that we propose is based on the assumption that there is an identification of the $t$-time parameter describing the $\mathcal{U}_t$ evolution of the sub-quantum degrees of freedom and where the model given by the relations \eqref{ETrelation} and \eqref{model 2 for semiperiod} holds good with the macroscopic coordinate time of a co-moving system with the quantum system.
This interpretation is consistent with $T$ being a parameter associated with the $\mathcal{U}_t$ evolution and simultaneously, a parameter associated with the system as a whole. It also fixes the operational meaning of the $t$-parameters such that the relation \eqref{definitionofinertialmass} holds.

Note that the identification of a particular $t$-time parameter with the proper time parameter of the system with mass $m$  do not contradict the essential $2$-dimensional character of the time, that is, the essential independence of $U_t$ and $U_\tau$ dynamics.

One natural interpretation of the energy-time relation \eqref{Energytimeuncertainty} is to think the quantity $\Delta E$ as the minimal exchange of energy between the Hamilton-Randers system and the environment in such a way that the system is stable at least during one whole cycle of semi-period $T$.
This energy exchange is measured in an instantaneous inertial coordinate frame in co-motion with the system just before the system changes to another different state or it decays to another different class of quantum system. Assuming that the exchange time is much longer than $2\,T$ and since the proper time is in terms of number of complete cycles of the system, the quantum mechanical energy-time uncertainty relation is obtained in the co-moving frame with respect to the proper time of the system measured in terms  the number of cycles and and where the $t$-time parameter specified by the equations \eqref{ETrelation} and \eqref{model 2 for semiperiod} by
   \begin{align*}
    \Delta E\,\tau\geq\,2\,\hbar .
   \end{align*}
the factor $2$ in this expression can be eliminated by a convenient redefinition of the model $\eqref{ETrelation}$ and the related expressions for the model of $t$-time parameter. Thus our final form of uncertainty relation is
   \begin{align}
    \Delta E\,\tau\geq\,\hbar .
    \label{Energytimeuncertainty 2}
   \end{align}

 If $2\, T$ is instead interpreted as a measure of the proper mean life time of the quantum system associated with the Hamilton-Randers system, then the relation \eqref{Energytimeuncertainty} is understood as the energy-time uncertainty relation in quantum mechanics. In this case, $\Delta E$ is the width in energy of the quantum system, which is an intrinsic parameter of the system, emergent from the underlying sub-quantum dynamics. Then one finds also the relation \eqref{Energytimeuncertainty 2}.
\subsubsection{Notion of external proper time parameter}
 By assumption, the $\mathcal{U}_t$ flow is {\it almost cyclic}. By this we mean that if the Hamilton-Randers system is isolated, then the  total $\mathcal{U}_t$ evolution is composed by a series of {\it fundamental cycles}, where the $n$-cycle takes place during the  intervals of internal $t$-time of the form
\begin{align*}
t\in [(2n+1)T,(2n+3)T],\quad n\in\mathbb{Z}.
\end{align*}
According to our {\it Assumption A.5+} of {\it Chapter} \ref{chapter on Assumptions and General Theory}, each of these fundamental cycles is composed itself by an {\it ergodic type regime}, followed by a concentrating regime, followed by an expanding regime. After these three phases regime cycle happens, a new ergodic regime start, defining the next fundamental cycle. The period of each fundamental cycle, if the corresponding quantum particle content of the system  does not change, is constant and equal to $2\,T$.

The relation \eqref{ETrelation} implies that the period $2\,T$ of a fundamental system depends on the physical system as a whole. Furthermore, following  our original {\it Assumption A.5}, the ergodic-concentrating-expansion regimes are universal, that is, such sequence of regimens in each fundamental cycle happen for every quantum system. By hypothesis, this holds for each degree of freedom of the Standard Model of particle physics but also for each atomic and molecular systems described according quantum mechanics.

In order to define a physical clock is needed either a quantum system which shows a high regular dynamics or a parameter associated with a very regular quantum process in the sense that a given process is stable under physical small fluctuations and whose characteristics are known with high precision.
Following this general scheme, a proper $\tau$-time parameter associated with a physical clock is constructed in the following way. Each fundamental period $2\,T$ of a fundamental cycle is identified with a fundamental or {\it minimal lapse} of $\tau$-time, $\delta \tau$, being the minimal lapse equal to $2\, T$. If we measure the time duration in physical processes in terms of a standard associated with the fundamental cycles descriptions associated with the process, then each minimal unit of measurable $\tau$-time is associated with a large number of fundamental periods, a multiple.

Different species of elementary particles or quantum systems have different  fundamental semi-periods $T$, which correspond to different {\it quantum clocks}.

It is natural a close map, isomorphism, between quantum clocks and the classical theory of congruences as it appears in classical number theory \cite{Dirichlet, Gauss Disquisitiones Arithmeticae, HardyWright}.

In the above context, it is then natural to consider the following notion of external time parameter:
\begin{definicion}
An external proper time parameter is an integer multiple function of the number of fundamental cycles of the $U_t$ flow associated with a given Hamilton-Randers system.
\label{definicion of tau proper time parameter}
\end{definicion}
Note that the {\it Definition} \ref{definicion of tau proper time parameter} of proper time parameter assumes the repeatability and stability of the fundamental cycles of the $\mathcal{U}_t$ dynamics and the stability of quantum processes associated with semi-periods $T$. Since physical clocks  are based upon the existence of stable, periodic processes, that can be reduced to the analysis of periodic quantum processes, they must also be  periodic in the number of fundamental cycles, according to the interpretation that any quantum system is a Hamilton-Randers system.
\subsubsection{Remarks on the notion of emergence of the macroscopic proper time}
Consequences of the above theory are the following.
First, it is inherent to Hamilton-Randers theory that time is necessarily described by two independent times and that one of its dimensions $\tau$-time has an emergent nature. This remark is valid for quantum clocks based upon counting fundamental cycles or in quantum processes from the point of view of Hamilton-Randers theory and hence for any clock defined in terms of them. Therefore, the statement is universal.

Second, there is a minimal theoretical lapse of proper time  identified with a period $2T$.

In the case of massless systems, there is no finite proper time, although there is a finite minimal period $2 \,T_{min}$.

One further consequence of the theory of the proper time parameter as counting fundamental cycles is that, for the dynamical cases when classical trajectories are well defined, it must exist a maximal macroscopic acceleration and maximal macroscopic speed. The maximal acceleration is because the relation $A_{max}=\, \frac{c}{\delta \tau_{min}}$ and $\delta \tau_{min}\geq 2 T$. In this case, regularity of the world line must be considered as a for all purposes good approximation.
\subsubsection{Relativistic implications of the energy-time uncertainty relation}
 In order to derive the relativistic version of the relation \eqref{Energytimeuncertainty}, let us to re-consider the meaning of {\it inertial coordinate system co-moving with a given Hamilton-Randers system}. The natural interpretation of such concept is that there is a macroscopic inertial coordinate system respect to which the exchange of energy-momentum vector of the quantum system associated with the Hamilton-Randers system and the environment is of the form
 \begin{align*}
 \Delta \,\mathcal{P}^\mu  \equiv \,(\Delta E, \vec{0}).
 \end{align*}
 Thus the relation \eqref{Energytimeuncertainty} can be re-written as
 \begin{align*}
  \Delta \,\mathcal{P}^\mu \, \Delta {X}_\mu =\,-\Delta E \,\delta \tau,
 \end{align*}
since the inertial system is co-moving with the coordinate system and where $\delta \tau$ is in this case the characteristic proper time of exchange of the energy of the system with the environment. In an arbitrary inertial coordinate system the relation is expressed as
\begin{align*}
2\,\hbar\leq \,\Delta E \,\delta \tau =\,\Delta E'\,\delta \tau'-\,\Delta \vec{P}'\,\cdot \delta \vec{x}'\leq \,\Delta E'\,\delta \tau',
\end{align*}
establishing the energy-time uncertainty relation for an arbitrary coordinate system,
\begin{align}
\hbar\leq \,\Delta E'\,\delta \tau' .
\label{energy time uncertainty 4}
\end{align}
By means of the relation $\vec{\mathcal{P}}=\,\frac{\mathcal{E}}{c^2}\,\vec{v}$,
\begin{align*}
\hbar\leq \Delta E'\,\delta \tau'=\, \Delta \vec{P}'\,\cdot \delta \vec{x}' \frac{c^2}{|\vec{v}|^2} \leq \,\Delta \vec{P}'\,\cdot \delta \vec{x}' .
\end{align*}
This last relation can be interpreted as a form of uncertainty principle for position and momentum observables,
\begin{align}
\hbar\leq \,\Delta \vec{P}'\,\cdot \delta \vec{x}' .
\label{momentum position uncertainty}
\end{align}
The uncertainty relations \eqref{energy time uncertainty 4} and \eqref{momentum position uncertainty} for massive particles are quantum mechanical uncertainty relations that, here derived from Hamilton-Randers theory, have an emergent nature.

Consistency with energy-momentum conservation implies the need to generalized the uncertainty relations \eqref{energy time uncertainty 4}-\eqref{momentum position uncertainty} to massless systems. This is direct if one consider an elementary vertex in quantum field theory, with one of the lines corresponding to a gauge boson like a gluon or photon.

\subsection{General properties of the bare $U_{\tau}$ flow}
Let us consider the trajectories of the sub-quantum molecules under the $\mathcal{U}_t$ flow.
The {\it word lines predecessors} are defined by the maps
\begin{align*}
{\xi_k}:\mathbb{R}\to M^k_4,\, t\mapsto \xi_{k}(t)\in \,M^k_4,\quad k=1,...N
\end{align*}
solutions of the $\mathcal{U}_t$ flow.  Equivalently, by the definition of the $\tau$-time parameter, we can consider instead
\begin{align*}
\xi_{tk}(\tau):=\,\xi_k(t+\,2 T,\tau),\quad t\in\,[0,2T],\quad \tau\in\,\mathbb{Z},\quad k=1,...,N.
\end{align*}
By the embedings  $\varphi_k:M^k_4\to M_4$, each manifold  $M^k_4$ is diffeomorphic to the manifold of observable  events or spacetime manifold $M_4$. Then for each fixed $t\in [0,2T]$ one can consider the embeddings of the predecessor world lines $\{\xi_{tk},\,k=1,...,N\}$ of the sub-quantum molecules $\{1,...,N\}$ in the model spacetime manifold $M_4$ given by $\varphi_k({\xi_t}_k)={\hat{\xi}}_{tk}\hookrightarrow M_4$,
 \begin{align}
\hat{\xi}_{tk}:\mathbb{R}\to M_4,\quad \tau\mapsto (\varphi_k\circ {\xi}_{tk})(\tau).
\label{worldlinesofsubquantumdegreesoffreedom}
  \end{align}
 Since each value of a $\tau$-time parameter is associated with a particular fundamental cycle in a series of consecutive cycles, changing $\tau$ but fixing $t$ is equivalent to consider the position of the sub-quantum molecule at different cycles at a fixed internal time  $t\,(\mod 2T)$. This succession of locations defines a world line of each $k$ sub-quantum particle in the spacetime manifold $M_4$.

 We can also see that this construction defines a string parameterized by $t\in\,[0,2T]$.

The speed vector and accelerations of the sub-quantum molecules respect to the $U_{\tau}$ evolution are bounded,
\begin{align}
\begin{cases}
& \eta_4(\beta_{{kx}},\beta_{{kx}})\leq c^2,\\
& \eta_4(\beta_{{ky}},\beta_{{ky}})\leq \,A^2_{max},\quad k=1,...,N.
\end{cases}
\label{boundaccelerationandspeed}
\end{align}
The conditions \eqref{boundaccelerationandspeed} are the fundamental constraints associated with a geometry of maximal proper acceleration \cite{Ricardo2012, Ricardo2014}.
\subsection{Kinematical properties of the $U_\tau$ dynamics from the $U_t$ dynamics}
For each Hamilton-Randers system, the non-degeneracy of the fundamental tensor $g$ of the underlying Hamilton-Randers space is ensured when
 the vector field  ${\beta}$ is constrained by the condition \eqref{boundenesscondition}. Such condition implies that
 \begin{enumerate}
 \item The velocity vector of the sub-quantum atoms is normalized by the condition $\eta^k_4(\dot{\xi}_{tk},\dot{\xi}_{tk})\leq v^2_{\max}=c^2$. This implies that  the world lines of sub-quantum molecules on the time $\tau$ are  sub-luminal or luminal.

 \item If the {\it on-shell} conditions $\{\dot{x}^k=y^k,\,k=1,...,4N\}$
 hold good, then there is a maximal bound for the proper acceleration $\eta^k_4(\ddot{\xi}_{tk},\ddot{\xi}_{tk})\leq A^2_{\max}.$
 \end{enumerate}
are invariant under diffeomorphisms of the four-dimensional manifold. These two kinematical properties are heritage from the corresponding $U_t$ dynamics.

As we showed, the classical Hamiltonian function \eqref{Classical Hamiltonian} of a Hamilton-Randers system  can also be defined by the {\it partially  averaged Hamiltonian function} $H_t(u,p)$.
Hence the Hamiltonian is associated with the evolution operator associated with the average description of a sub-quantum molecule, described by two sub-quantum atoms, one evolving on one direction of $t$-time and the other on the opposite direction of the $t$-time,
\begin{align*}
\begin{cases}
& U_t\,\equiv I-\,\,F^2_t(u,p)dt,\\
& U_{-t}\,\equiv I-\,F^2_{-t}(u,p)dt
\end{cases}
\quad \Rightarrow \quad\frac{1}{2}(U_t+\,U_{-t})=\,1-\,\,H_t(u,p)d t,
\end{align*}
where the second line applies to the evolution of some initial conditions in an appropriate way.
 It is then natural to read the Hamiltonian function \eqref{Classical Hamiltonian} as a {\it time orientation average} of the Hamilton-Randers function $F^2_t$ associated with a particular form of classical Hamiltonian.

The Hamilton equations for $(u,p)$ under the evolution generated by the Hamiltonian ${H}_t(u,p)$ are
 \begin{align}
 \nonumber\frac{d u^i}{dt}=\,{u'}^i & =\,\frac{\partial H(u,p)}{\partial p_i} =\,2(1-\kappa(t,\tau))\,\beta^i(u)-\,\frac{\partial\kappa(t,\tau)}{\partial p_i}\,\left(\sum^{8N}_{k=1}\,\beta^k(u)p_k\right)\\
 & =\,2(1-\kappa(t,\tau))\,\beta^i(u),
 \label{Hamilton Equation 1}
 \end{align}
 \begin{align}
  \nonumber\frac{d p^i}{dt}=\,{p'}^i & =\,-\frac{\partial H(u,p)}{\partial u^i} =\,-2(1-\kappa(t,\tau))\,\sum^{8N}_{k=1}\,\frac{\partial \beta^k(u)}{\partial u^i}p_k +\frac{\partial\kappa(t,\tau)}{\partial u^i}\,\left(\sum^{8N}_{k=1}\,\beta^k(u)p_k\right)\\
 & = \,-2(1-\kappa(t,\tau))\,\sum^{8N}_{k=1}\,\frac{\partial \beta^k(u)}{\partial u^i}p_k.
 \label{Hamilton Equation 2}
 \end{align}
 Note that in the last line we have applied that $\kappa$ does not depend upon $u\in TM$.

 \subsection{Re-definition of the $t$-time parameter and $\mathcal{U}_t$ flow}
 The $\mathcal{U}_t$ flow has been parameterized by the conformal factor $\kappa(t,\tau)$. However, in order to obtain dynamical equations of motion \eqref{dynamicalsystem} from a Hamiltonian theory it is necessary to conveniently normalize the $t$-time parameter, since the Hamilton equations \eqref{Hamilton Equation 1}-\eqref{Hamilton Equation 2} of the Hamiltonian function  \eqref{RandersHamiltonian} are not in general the equations \eqref{dynamicalsystem}. What we need is to re-scale the $t$-time parameter such that the extra-factor $1-\kappa$ dissapear from the equations. If the {\it old}  time parameter is denoted by $\tilde{t}$ and the new external time parameter by $t$, then we impose the condition
 In order to resolve such incompatibility, we re-define
 \begin{align}
 d\tilde{t}\mapsto t=\,d\tilde{t}\,(1-\kappa(\tilde{t},\tau)).
 \label{redefinitionoftau}
 \end{align}
  As it stands, the relation \eqref{redefinitionoftau}  is well defined, since $\kappa(\tilde{t},\tau)$ does not depend on $u\in TM$.
The integration of the condition \eqref{redefinitionoftau} provides the necessary change in the parameter,
 \begin{align}
 t-t_0 =\, \int^{\tilde{t}}_{\tilde{t_0}}\,d \varsigma (1-\kappa (\varsigma,\tau)).
 \label{redefinitionoft2}
 \end{align}
 Note that, besides the requirements of compactness of the domain of definition and positiveness of the time parameter $\tilde{t}$, which is also translate to $t$ by equation \eqref{redefinitionoft2}, the parameter $t$ is arbitrarily defined.

 \subsection{The deterministic, local dynamics of the sub-quantum degrees of freedom}
 The effective Hamiltonian that describes the dynamics of a full fundamental cycle, including the metaestable limit \eqref{finalhamiltoniantevolution}, under the time re-parametrization \eqref{redefinitionoftau} is given by
\begin{align}
  H_t(u,p) =\begin{cases}
    \sum^{N}_{k=1}\sum^4_{\mu=1} \beta^\mu_{kx}(u)\,p_{x\mu k}+\,\sum^{N}_{k=1}\sum^4_{\mu=1}\, \beta^\mu_{ky}(u)\,p_{y\mu k}, & t\neq 2n\,T,\,n\in\, \mathbb{Z},\\
    0, &  t= 2n\,T,\,n\in\, \mathbb{Z}.
  \end{cases}
  \label{RandersHamiltonian}
\end{align}
Since this Hamiltonian is equivalent to the previous Hamiltonian \eqref{Classical Hamiltonian}, the corresponding Hamiltonian function is still denoted  by $H_t(u,p)$. Note that the $\beta$ functions depend upon the coordinates $u$. Since for each Hamilton-Randers system, they are functions of $t$-time, the Hamiltonian function will also depend on the $t$-time parameter implicitly.

With respect to this external $t$-time parameter, the Hamilton equations of \eqref{RandersHamiltonian} are (in a slightly different notation) of the form
\begin{align}
\begin{cases}
&\dot{u}^\mu_k :=\frac{d u^\mu_k}{d t}=\,\frac{\partial H_t(u,p)}{\partial p_{k\mu}}=\,\beta^\mu_k(u),\\
 &\dot{p}_{k\mu}:=\frac{d p_{k\mu}}{d t}=\,-\frac{\partial H_t(u,p)}{\partial u^\mu_k}=\,-\sum^{N}_{i=1}\sum^4_{\rho=1}\,\frac{\partial \beta^\rho_i(u)}{\partial u^\mu_k}p_{i \rho},\quad i,k=1,...,N,\, \mu,\rho =1,2,3,4,
\end{cases}
\label{HamiltonEquations}
\end{align}
where $u=(x,y)$.

The explicit relation between the equations \eqref{dynamicalsystem}-\eqref{dynamicalsystem2} and the Hamilton equations of motion. This relation is given by the relations
\begin{align}
\beta^\mu_{x k}\equiv \gamma^\mu_{x k},\quad \beta^\mu_{y k} =\, \gamma^\mu_{y k},\quad k=1,...,N,\, \mu,=1,2,3,4.
\end{align}
These function describe sub-quantum molecules. These are the fundamental degrees of freedom of our models.

At the classical level of the formulation of the dynamics of Hamilton-Randers systems, the basis of the physical interpretation given by the theory developed above can be expressed as
\begin{teorema}
For each dynamical system given by the equations \eqref{dynamicalsystem}-\eqref{dynamicalsystem2} where $\gamma \in \Gamma TTM$, there exists a Hamilton-Randers system whose Hamiltonian function is \eqref{RandersHamiltonian} and whose Hamilton-Randers equations are \eqref{HamiltonEquations}.
\label{Teoremaregular HamiltonRandersHamilton}
\end{teorema}

Although deterministic and local, the dynamical systems that we shall consider within the category of Hamilton-Randers systems are not classical Newtonian of geodesic dynamics, as in classical mechanics. Indeed, the relation  of the
\begin{align*}
\dot{x}^\mu_{k}=\,{1}/{m}\,\eta^{\mu\rho}p_{xk\rho},
 \end{align*}
will not hold in general, since the Hamiltonian function \eqref{RandersHamiltonian} that we are considering is linear on the canonical momentum variables.

\subsection{Geometric character of the Hamiltonian dynamics}
Let us consider the problem of the invariance of the Hamiltonian transformation \eqref{RandersHamiltonian}. If we consider a local change of coordinates
\begin{align*}
\vartheta:\mathbb{R}^4\to \mathbb{R}^4,\,x^\mu \mapsto \tilde{x}^\mu(x^\nu),
 \end{align*}
 there is an induced change in local coordinates on $TM$ given schematically by
\begin{align}
\begin{cases}
& x^\mu_k \mapsto \tilde{x}^\mu(x^\nu)_k,\\
& Y^\mu\mapsto  \tilde{y}^\mu_k=\,\frac{\partial \tilde{x}^\mu_k}{\partial x^\rho_k}\,y^\rho_k.
\end{cases}
\label{local change of coordinates 1}
\end{align}
Then the invariance of the Hamilton equations of motion for the $u^\mu$ coordinates implies that
\begin{align*}
\frac{d \tilde{x}^\mu_k}{d\tau} =\,\frac{\partial \tilde{x}^\mu_k}{\partial x^\rho_k}\frac{dx^\rho_k}{d\tau}=\,\frac{\partial \tilde{x}^\mu_k}{\partial x^\rho_k} \beta^\rho_k.
\end{align*}
Thus, in a similar way as for the formulae \eqref{local change of coordinates of the dynamical system functions 1}-\eqref{local change of coordinates of the dynamical system functions 2}, we have that for the $\beta$ functions in a Hamilton-Randers system, we have
\begin{align}
\tilde{\beta}^\rho_k =\,\frac{\partial \tilde{x}^\mu_k}{\partial x^\rho_k} \beta^\rho_k.
\label{local change of coordinates 2}
\end{align}
Also, invariance of first group of Hamilton equations \eqref{HamiltonEquations} implies the consistency condition
\begin{align}
\nonumber\tilde{\beta}^\mu_{ky} &=\,\frac{d\tilde{y}^\mu_k}{d\tau}=\,\frac{d}{d\tau}\,\left(\frac{\partial \tilde{x}^\mu_k}{\partial x^\rho_k}\,y^\rho\right)\\
\nonumber & =\,\frac{\partial \tilde{x}^\mu_k}{\partial x^\rho_k}\,\frac{d y^\rho_k}{d\tau}+\,\frac{\partial^2 \,\tilde{x}^\mu_k}{\partial x^\rho_k\,\partial x^\sigma_k}\,\frac{d x^\mu_k}{d\tau}\,y^\rho_k\\
& =\,\frac{\partial \tilde{x}^\mu_k}{\partial x^\rho_k}\,\beta^\rho_{ky}+\,\frac{\partial^2 \,\tilde{x}^\mu_k}{\partial x^\rho_k\,\partial x^\sigma_k}\,\beta^\sigma_{kx}(x)\,y^\rho_k.
\label{local change of coordinates 3}
\end{align}
In order to obtain the transformation rules for the momentum coordinates $(p_{kx\mu},p_{ky\mu})$ it is required that the Hamiltonian \eqref{RandersHamiltonian} remains invariant,
\begin{align*}
H_t(\tilde{x},\tilde{y},\tilde{p}_x,\tilde{p}_y) & =\,\sum_k \,\tilde{\beta}^\mu_{kx}\,\tilde{p}_{kx\mu}+\,\sum_k \,\tilde{\beta}^\mu_{ky}\,\tilde{p}_{ky\mu}\\
& =\,\,\sum_k \,\beta^\mu_{kx}\,p_{kx\mu}+\,\sum_k \,\beta^\mu_{ky}\,p_{ky\mu}\\
& =\,H_t({x},y,p_x,p_y).
\end{align*}
Using the transformation rules \eqref{local change of coordinates 1}, \eqref{local change of coordinates 2} and \eqref{local change of coordinates 3} and re-arranging conveniently, the transformation rules for the momentum coordinates are
\begin{align}
& p_{kx\rho}=\,\frac{\partial \tilde{x}^\mu_k}{\partial x^\rho_k}\,\tilde{p}_{kx\mu}+\,\frac{\partial^2 \,\tilde{x}^\mu_k}{\partial x^\rho_k\,\partial x^\sigma_k}\,\tilde{p}_{ky\rho}(x)\,y^\rho_k,\\
& p_{ky\rho}=\,\frac{\partial \tilde{x}^\mu_k}{\partial x^\rho_k}\,\tilde{p}_{ky\mu}.
\label{local change of coordinates 4}
\end{align}

The set of transformations \eqref{local change of coordinates 1}, \eqref{local change of coordinates 2}, \eqref{local change of coordinates 3} and \eqref{local change of coordinates 4} leave invariant the Hamilton equations \eqref{HamiltonEquations}. They are similar in structure to the transformations that one founds in higher order geometry \cite{MironHrimiucShimadaSabau:2002}.

\subsection{Emergence of the spacetime metric and spacetime diffeomorphism invariance}
One notes that in the domain ${\bf D}_0$ as postulated above, implies that $\beta|_{{\bf D}_0}=0$.
One observes that the dynamical information of the system is encoded in the vector field $\beta\in \,\Gamma TM$. We also note that the Hamiltonian is zero in a hypersurface, corresponding to the domain ${\bf D}_0$, of $TM$. Formally, the null condition for the Hamiltonian in the domain ${\bf D}_0$ is of the form
\begin{align}
\sum^{N}_{k=1}\sum^4_{\mu=1} \beta^\mu_{kx}(u)\,p_{kx\mu }\big|_{{\bf D}_0}+\,\sum^{N}_{k=1}\sum^4_{\mu=1}\, \beta^\mu_{ky}(u)\,p_{ky\mu }\big|_{{\bf D}_0} =0.
\label{condition null Hamiltonian function}
\end{align}
In this sum, the first term corresponds to the contraction of timelike vector fields with time-like $1$-forms, since the momentum $1$-forms $\,p_{kx}$ are such that $\eta (p_{kx},p_{kx})=0$. The second term contains the contraction of spacelike vectors $\beta_{ky}$ with spacelike momentum $1$-forms $p_{ky}$, making the balance \eqref{condition null Hamiltonian function} possible. Also note that the condition \eqref{condition null Hamiltonian function} is compatible with the requirement of invariance under the time inversion $\mathcal{T}$ operation.

The mechanism to identify $(M_4,\eta_4)$ with $(\mathcal{M}_4,\eta'_4)$ is intuitive, if we identify the domain ${\bf D}_0$ as the classical domain where macroscopic observables are well defined. This is the domain where one can identify macroscopic events. Hence is a domain dominated by macroscopic properties and systems. The perturbation due to a sub-quantum system is then negligible. In fact,
one notes that in the domain ${\bf D}_0$ as postulated above, implies that $\beta|_{{\bf D}_0}=0$. This implies that in that regime, $F\equiv\,\alpha$. The domain ${\bf D}_0$ is identified with the domain of classicality, then implies that the four manifold $M_4$ must be identified with the spacetime manifold $\mathcal{M}_4$, at least up to diffeomorphism, and that the metric $\eta_4$ must be identified with the spacetime metric $\eta'_4$, at least up to conformal factor.

With respect to the properties inherited by the spacetime structure from the model spacetime structure, it is relevant to mention that both will be relativistic spacetimes with the same maximal speed of propagation and maximal proper acceleration $A_{\mathrm{max}}$.

Related with the emergence of spacetime structure in the ${\bf D}_0$ domain, that we should denote by $(\mathcal{M}_4, \lambda \eta_4)$, there is also the identification of the group of diffeomorphism $\Diff (M_4)$ with the group of spacetime diffeomorphisms $\Diff (\mathcal{M}_4)$, providing a derivation from first principles of one of the fundamental assumptions of general relativity.

\subsection{Macroscopic observers and metric structures}\label{Section on macroscopic observers}
 Since $M_4$ is endowed with a  Lorentzian metric $\eta_4$ with signature $(1,-1,-1,-1)$, there is a natural definition of {\it ideal macroscopic observer},
  \begin{definicion} Given the Lorentzian spacetime $(M_4,\eta_4)$
  an ideal macroscopic observer is a smooth, timelike vector field $W\in\,\Gamma TM_4$.
  \label{definiciondemacroscopicidealobserver}
  \end{definicion}
  This is the standard definition of observer in general relativity, based upon the notions of coincidence and congruence. The observers that we shall consider are ideal macroscopic observers defined as above.

  Given an observer $W$  there is  associated with the Lorentzian metric $\eta_4$ a Riemannian metric on $M_4$ given by the expression
\begin{align}
\bar{\eta}_4(u,v)=\,\eta_4(u,v)-\,2\,\frac{\eta_4(u,W(x))\eta_4(v,W(x))}{\eta_4(W(x),W(x))},\quad u,v\in \,T_x\,M_4,\,\,x\in M_4.
\label{RiemannianfromLorentz}
\end{align}
$d_W:M_4\times M_4\to \mathbb{R}$ is the distance function  associated with the  norm of the Riemannian metric $\bar{\eta}_4$.

We postulate that the relation between the metric $\eta_4$ on $M_4$ and the metric $\varepsilon_4$ on the spacetime manifold $\mathcal{M}_4$ is a conformal relation, $\varepsilon =\,\lambda (N) \,|eta_4$. This relation is suggested because the content of $M_4$ is determined, modulo diffeomorphism, by an individual sub-quantum molecule, while the matter content associated with $\mathcal{M}_4$ is associated with $N$ sub-quantum molecules. Furthermore, $\mathcal{M}_4$ is associated with the domain ${\bf D}_0$, that we will see in {chapter 6} and {chapter 7}, corresponds to the classical spacetime where events can happen. Since the source will appear as a pointwise particle $N$ times larger than for a sub-quantum particle, we suggest the dependence in $N$ of the conformal factor $\lambda$.
We can apply similar notions to the spacetime manifold $(\mathcal{M}_4,\lambda \eta_4)$. Given an ideal macroscopic observer $W$, for a fixed  point $x\in M_4$ and the world line $\hat{\xi}_{tk}$, the {\it distance function} between $x$ and $\hat{\xi}_{tk}$ is given by the expression
\begin{align*}
d_W(x,\hat{\xi}_{tk}):=\,\inf \big\{\lambda\,d_W(x,\tilde{x}),\,\tilde{x}\in\,\hat{\xi}_{tk})\big \}.
\end{align*}
The distance $d_W(x,\hat{\xi}_{tk})$ depends on the observer $W$. It also depends on the specification of the family of diffeomorphisms $\{\varphi_k,\,k=1,...,M\}$. Such condition is consistent with spacetime diffeomorphims invariance.

$\inf\{d_W(x,\hat{\xi}_{tk}),\,k=1,...,N\}$ depends on $t\in \mathbb{R}$. Since $t$-parameters are not observables, the corresponding distance is not observable. However, there is a definition of distance which is independent of $t$, namely,
 \begin{align}
 d_W(x,\hat{\xi}_k)=\,\inf\{d_W(x,\hat{\xi}_{tk}),\,t\in \,[2nT,2(n+1)T]\},
  \label{distancexxi}
 \end{align}
 with $n\in\,\mathbb{Z}$.
 The corresponding distance function $d_W(x,\hat{\xi}_k)$ is an observable depending on the observer $W$.

The distance function $d_W$ will be used in the construction of the quantum description of the system. Specifically, we shall define the density of probability at the point $x\in M_4$ associated with the quantum system in terms of the associated geometric distribution of the world lines associated with the sub-quantum degrees of freedom.
\subsection{Notion of physical observable in emergent quantum mechanics}
The freedom in the possible choices of the diffeomorphisms $\{\varphi_k,\,k=1,...,N\}$ is analogous to the the freedom in the choice of a gauge that appears in gauge theories. The associated ambiguities in the formulation of the theory are resolved automatically if we assume that all physical observable quantities must be insensitive to most of the details of the sub-quantum description, in particular, to the details invariant under the action of $\Diff(\mathcal{M}_4)$.

The ontological elements of the theory are described by the following notion:
\begin{definicion}
A property is ontological or real if it can be represented in $M_4$ by means of some of the diffeomorphisms $\varphi_k:M^k_4\to M_4$ or its induced maps.
\end{definicion}
For instance, the intersection of the curves associated with two different sub-quantum degrees of freedom $\varphi_l (\gamma_k(\mathbb{K}))\cap \varphi_l (\gamma_l(\mathbb{K}))$ is real.

Now we define a notion of observable:
\begin{definicion}
An observable property is a real property that does not depend upon each individual sub-quantum degree of freedom $k$.
\label{macroscopic observable}
\end{definicion}
An observable can be defined, for instance, in the concentration domain ${\bf D}_0$, where one can assume that for all practical purposes all the sub-quantum degrees of freedom are very close to each, in a way that one can modelled such a situation as that they meet at a given point of the spacetime manifold $T\mathcal{M}_4$. Observables, view in this way, constitute an effective description of the system.

The spacetime manifold $\mathcal{M}_4$, view in this way, constitutes an effective and emergent description.

\newpage

\section{\LARGE{Hilbert space formulation of Hamilton-Randers dynamical systems}}\label{chapter on Koopman-von Neumann formulation}
\bigskip
\bigskip
 The quantum mechanical description of physical systems arises naturally from the Hilbert space formulation of Hamilton-Randers dynamical systems. This formulation will be developed in this {\it chapter}. The application of Hilbert space theory to deterministic dynamical systems began with the work of Koopman and von Neumann on the study of ergodicity \cite{Koopman1931,von Neumann}. Koopman-von Neumann theory, combined with ideas related to notions of information  and complexity, was first applied in the context of emergent approaches to quantum mechanics by G. 't Hooft \cite{Hooft}. In our theory, however, a direct application of  Koopman-von Neumann theory  to Hamilton-Randers systems is developed. At this point, we would like to remark the different mathematical and ontological content of 't Hoot's theory and Hamilton-Randers theory.

 We develop our theory using local coordinate techniques. Local coordinate formulation of theories always leads to issues concerning the general covariance of the theory considered. We shall show that the Koopman-von Neumann formulation of Hamilton-Randers models are consistent with general covariant.
\subsection{On the pre-Hilbert space formulation of Hamilton-Randers dynamical systems}
  Each Hamilton-Randers dynamical system has associated a Hilbert space and quantum mechanical formulation of the dynamical evolution. Given the tangent manifold $TM$ corresponding to the configuration manifold of a Hamilton-Randers system, we can consider a local open set $TU\subset \,TM$ and a system of local coordinates $x^\mu_k,y^\mu_k:TU\to \, \mathbb{R},\quad \mu=1,2,3,4,\,k=1,...,N.$
We assume the existence of a vector space $\bar{\mathcal{H}}_{Fun}$ and a collection of linear operators
    \begin{align}
    \{\hat{u}^\mu_k,\,k=1,...,N\}^4_{\mu=1}\equiv\,\left\{\hat{x}^\mu_k,\hat{y}^\mu_k:\bar{\mathcal{H}}_{Fun}\to \bar{\mathcal{H}}_{Fun},\,k=1,...,N\right\}^4_{\mu=1}
    \end{align}
such that the values of the local position coordinates $\{x^\mu_k,\,k=1,...,N,\,\mu=1,2,3,4\}$
 and the local speed  coordinates $\{y^\mu_k,\,k=1,...,N,\,\mu=1,2,3,4\}$ of the sub-quantum molecules are the eigenvalues of the operators $\hat{u}^\mu_k$. We assume the existence of a {\it generator set}
 \begin{align*}
 \{|x^\mu_l,y^\mu_l\rangle\}^{N,4}_{k=1,\mu=1}\subset \bar{\mathcal{H}}_{Fun}
  \end{align*}
  of common eigenvectors of the operators $\{\hat{x}^\mu_k ,\,\hat{y}^\mu_k,\,\mu=1,2,3,4;\,\,k=1,...,N \}$ characterized by the relations
    \begin{align}
    \begin{cases}
 & \hat{x}^\mu_k |x^\mu_l,y^\nu_l\rangle=\,\sum_{l}\,\delta_{kl}\,x^\mu_l\,|x^\mu_l,y^\nu_l\rangle,\\
 & \hat{y}^\nu_k |x^\mu_l,y^\nu_l\rangle=\,\sum_{l}\delta_{kl}\,y^\nu_l\,|x^\mu_l,y^\nu_l\rangle,
 \end{cases}
    \label{positionvelocityeigenbasis}
    \end{align}
where $\delta_{kl}$ is the Kronecker delta function. A generic element of $\bar{\mathcal{H}}_{Fun}$ is of the form
\begin{align*}
\Psi =\sum^{N}_{k=1} \,\alpha_k\,|x^\mu_l,y^\nu_l\rangle,\quad \alpha_k\in\,\mathbb{K},
\end{align*}
if the numeric field used in the formulation of the dynamics is $\mathbb{K}$.

The elements of the generation set \eqref{positionvelocityeigenbasis} represent the configuration state of the sub-quantum degrees of freedom describing a Hamilton-Randers system. These elements of $\bar{\mathcal{H}}_{Fun}$ will be called {\it ontological states}, since they determine the dynamical evolution of the ontological degrees of freedom of Hamilton-Randers spaces. In the case that the dynamics and the degrees of freedom are discrete, the dimension of $\bar{\mathcal{H}}_{Fun}$ is finite.

By assumption, the position and velocity operators associated with the sub-quantum degrees of freedom commute among them when they are applied to each element of $\bar{\mathcal{H}}_{Fun}$ and when they are applied at equal values of the time $t\in \,\mathbb{R}$:
\begin{align}
\left[\hat{x}^\mu_k,\hat{x}^\nu_l\right]|_{t}=\,0,\quad \left[\hat{y}^\mu_k,\hat{y}^\nu_l\right]|_{t}
=\,0,\quad \left[\hat{x}^\mu_k,\hat{y}^\nu_l\right]|_{t}=0, \quad\,\mu,\nu=1,2,3,4;\, k,l=1,...,N.
\label{commutativityoftheontologicallabels}
\end{align}
The specification of the commutation relations of the algebra is done at each fixed value of the pair $(t,\tau)\in\,[0,2\,T]\times\mathbb{Z}$, that in the continuous approximation is $(t,\tau)\in\,[0,2\,T]\times\mathbb{Z}$. Specifying only the $\tau$-time value is not enough for the formulation of the commutation relation.

This construction implies an equivalence between the spectra of the operators $\{\hat{x}^\mu_k ,\,\hat{y}^\mu_k,\,\mu=1,2,3,4;\,\,k=1,...,N \}$ and the points in a local coordinate chart $TU\subset TM$. In particular, the local coordinate system over $TM$ determines the {\it local spectrum of eigenvalues}. In local coordinates, the relation is of the form
\begin{align}
Spec\left(\{\hat{x}^\mu_k ,\,\hat{y}^\mu_k\}^4_{\mu=1}\right)\to TM^k_4,\quad
\{|x^\mu_k,y^\nu_k\rangle\}^4_{\mu=1}\mapsto (x^\mu_k,y^\mu_k),\, \mu=1,2,3,4;\,\,k=1,...,N.
\end{align}
Different local spectrum must be consistent with each other, in consonance to how local coordinate systems must also be consistent between them.

We now introduce the procedure for the quantization of observables that depend upon the coordinate and position variables and  their time derivative functions. Since the commutativity conditions \eqref{commutativityoftheontologicallabels} and the fundamental states are eigenstates of the coordinate and velocity operators $\hat{u}^\mu_k=\,(\hat{x}^\mu_k,\hat{y}^\mu_k$), it is natural to define the quantization of a function and the coordinate derivative operator of a function by the corresponding operations on the eigenvalues. Thus to a function $\Xi : TM\to \mathbb{C}$, there is associated a linear operator in the vector space $\bar{\mathcal{H}}_{Fun}$ determined by the relation
\begin{align}
\widehat{\Xi}(x^\mu_k,y^\mu_k)\,|x^\mu,y^\mu\rangle :=\,\Xi({x}^\mu_k,{y}^\mu_k)\,|x^\mu,y^\mu\rangle .
\label{quantization of a function}
\end{align}
Similarly, the partial derivative operator is determined by the relation
\begin{align}
\widehat{\frac{\partial} {\partial \,x^\mu} \Xi}({x}^\mu_k,{y}^\mu_k)\,|x^\mu,y^\mu\rangle :=\,
\frac{\partial \Xi({x}^\mu_k,{y}^\mu_k)}{\partial x^\mu}\,|x^\mu,y^\mu\rangle
\label{quantization of the derivative of a function}
\end{align}
and similarly for other partial derivative operators,
\begin{align}
\widehat{\frac{d\,\Xi}{dt}}(x^\mu_k,y^\mu_k)\,|x^\mu,y^\mu\rangle :=\,\Xi({x}^\mu_k,{y}^\mu_k)\,|x^\mu,y^\mu\rangle,
\label{quantization of a time derivative of a function}
\end{align}
\begin{align}
\widehat{\frac{d}{dt}\frac{\partial\, \Xi}{\partial x^\mu}}({x}^\mu_k,{y}^\mu_k)\,|x^\mu,y^\mu\rangle :=\,
\frac{d}{d t}\,\left(\frac{\partial \Xi({x}^\mu_k,{y}^\mu_k)}{\partial x^\mu}\right)\,|x^\mu,y^\mu\rangle
\label{quantization of the time derivative of a function 2}
\end{align}

The commutation relations \eqref{commutativityoftheontologicallabels} are only a piece  of the full canonical quantization conditions of the system. The completion of the canonical quantization conditions is made by assuming the following additional {\it quantum conditions}:
\begin{itemize}

\item There is a set of symmetric linear operators $\left\{\hat{p}_{x k\mu},\hat{p}_{y k\mu},\,k=1,...,N,\,\mu=1,2,3,4\right\}$
  that generate local diffeomorphisms on $TM$ along the integral curves of the local vector fields
    \begin{align}
    \left\{\frac{\partial}{\partial x^\mu_k},\,\frac{\partial}{\partial y^\nu_k}\,\in \Gamma\,TTU\quad \,\mu,\nu=1,2,3,4;\,k=1,...,N\right\}.
    \label{interpretation of quantum operators}
    \end{align}

\item The phase space $T^*TM$ is commutative,
\begin{align}
\begin{cases}
& [\hat{x}^\mu_k,\hat{x}^\nu_l]=0,\,\left[y^\mu_k,y^\nu_l\right]=0,\,\left[x^\mu_k,y^\nu_l\right]=0,\\
& [\hat{p}_{x k\mu},\hat{p}_{y l \nu}]=0,\quad \mu,\nu=1,2,3,4;\, k=1,...,N.
 \end{cases}
\label{commutativequantumconditions}
\end{align}

\item The following commutation relations at each fixed two-times $(t,\tau)\in\,\mathbb{R}\times \mathbb{R}$ hold good,
\begin{align}
\begin{cases}
& [\hat{x}^\mu_k,\hat{p}_{x l\nu}]= \,\imath\,\hbar\,\delta^\mu_\nu\,\delta_{kl}\\
& [\hat{y}^\mu_k,\hat{p}_{y l\nu}]=\, \imath\,\hbar\,\delta^\mu_\nu\,\delta_{kl},\quad \mu,\nu=1,2,3,4;\, k=1,...,N.
\end{cases}
\label{quantumconditions}
\end{align}
\end{itemize}
The collection of operators on $\bar{\mathcal{H}}_{Fun}$
\begin{align}
\{\hat{u}^\mu_k,\,\hat{p}_{k\mu},\,\mu=1,2,3,4,\,k=1,...,N\}
\label{fundamental operators}
 \end{align}
 satisfying  the {\it fundamental algebraic relations} \eqref{quantumconditions}-\eqref{commutativequantumconditions} are called {\it fundamental canonical operators}.

 \begin{comentario}
 Since the dynamics in a Hamilton-Randers theory is discrete, it is natural that the spacetime should be modeled as a non-commutative space, for instance, as a quantum tangent spacetime model of Snyder's type \cite{Snyder}. According to this view, the continuous version of the models considered in this work is only an approximation and the commutativity conditions for the coordinates and velocities should be seen as an approximation to a more precise modelling.
 \end{comentario}
\subsubsection{Scalar product in $\mathcal{H}_{Fun}$}
 The vector space $\bar{\mathcal{H}}_{Fun}$ is a pre-Hilbert space with an inner product defined in the following way. Let us consider two sub-quantum degrees of freedom $k,l$ and the diffeomorphisms to the model manifold, $\varphi_k:M^k_4\to M_4$, $\varphi_l:M^l_4\to M_4$. First it is defined the inner scalar product of  ontological states,
 \begin{align}
\langle x_l,y_l|\tilde{x}_k,\tilde{y}_k\rangle := \,\delta_{lk}\,\delta(\varphi_l(x)-\varphi_{l}(\tilde{x}))\,\delta(\varphi_{l*}(y)-\varphi_{k*}(\tilde{y})).
\label{definiciondeproductoscalarenHPlanck}
\end{align}
This operation has the following properties,
\begin{itemize}
\item The symmetric property for this product operation,
 \begin{align*}
 \langle x_l,y_l|\tilde{x}_k,\tilde{y}_k\rangle = \,\langle \tilde{x}_k,\tilde{y}_k|  x_l,y_l \rangle
 \end{align*}
 holds good.
\item It is diffeomorphism invariant in the sense that the delta products are integrated on the correponding tangent space $TM^k_4$ geometrically. If the diffeomorphisms change to  $\tilde{\varphi}_k:M^k_4\to M_4$, $\tilde{\varphi}_l:M^l_4\to M_4$, then
\begin{align*}
\langle x_l,y_l|\tilde{x}_k,\tilde{y}_k\rangle := \,\delta_{lk}\,\delta(\tilde{\varphi}_l(x)-\tilde{\varphi}_{l}(\tilde{x}))\,\delta(\tilde{\varphi}_{l*}(y)-\tilde{\varphi}_{k*}(\tilde{y})).
\end{align*}

\end{itemize}
 The scalar product \eqref{definiciondeproductoscalarenHPlanck} is also the basis for the definition of a norm function in the vector space $\bar{\mathcal{H}}_{Fun}$.

 A particular example of the relation \eqref{definiciondeproductoscalarenHPlanck} is when $x_k$ and $x_l$ correspond to the same $x\in\,M_4$ via the inverse of the diffeomorphisms $\varphi_k$ and $\varphi_l$ respectively. In this case the relation can be re-cast formally as
 \begin{align}
\int_{M_4}\,dvol_{\eta_4}\langle x^\mu_l,y^\nu_l|x^\rho_k,y^\lambda_k\rangle = \,\delta^k_l\,\delta(y^\nu-y^\lambda),
\label{productoscalarenHPlanck}
\end{align}
which is  similar to the orthogonality property of position quantum states.

\subsection{The Hilbert space of ontological states}
The extension of this product rule to arbitrary linear combinations of fundamental vectors is achieved by assuming the adequate bilinear property of the product operation respect to the corresponding number field $\mathbb{K}$ in the vector space structure of $\bar{\mathcal{H}}_{Fun}$. We did not fix the number field $\mathbb{K}$ yet.
Let us  consider combinations of fundamental states of the form
\begin{align}
|\Psi\rangle=& \sum^N_{k=1} \,\alpha_{k} (\varphi^{-1}_{1}(x),z_{1},\varphi^{-1}_{2}(x),z_{2},...,\varphi^{-1}_{N}(x),z_{N})|\varphi^{-1}_{k}(x),z_{k}\rangle,\quad x\in\,M_4,
\label{general pre-quantum state}
\end{align}
Let us  remind here that sub-quantum atoms are described by curves $\gamma_k :\mathbb{K}\to M^k_4$, $\gamma_l:\mathbb{K}\to M^l_4$.
Locality implies that for most of the cases, except when intersecting on the image $\varphi_l (\gamma_k(\mathbb{K}))\cap \varphi_l (\gamma_l(\mathbb{K}))$, the $k$-degree of freedom and the $l$-degree of freedom are independent from each other. Therefore, except in a set of measure zero, the state \eqref{general pre-quantum state} is of the form
\begin{align}
|\Psi\rangle=& \sum^N_{k=1} \,\alpha_{k} (\varphi^{-1}_{k}(x),z_{k})|\varphi^{-1}_{k}(x),z_{k}\rangle,\quad x\in\,M_4,
\label{pre-quantum state}
\end{align}
where each coefficient defines a function on each fiber, $\alpha_{k} (\varphi^{-1}_{k}(x),\cdot):T_{(\varphi^{-1}_{k}(x),z_{k})}M^k_4\to \mathbb{R}$.

The state $|\Psi\rangle$ depends on the point $x\in M_4$ of the model manifold $M_4$ and on the lift from $M_4$ to the configuration space $M\cong \prod^N_{k=1}\,\times M^k_4$. Such lift is not natural, since it depends on the choice of the family of diffeomorphisms $\{\varphi_k: M^k_4 \to M_4,\,k=1,...,N\}$ and on the point $(z_1,...,z_N)$ of the fiber at the point $(\varphi^{-1}_1,...,\varphi^{-1}_N)\in \,M \cong \prod^N_{K=1} \,M^k_4$.

In the case of the dynamical evolution corresponding to a free, individual quantum  system, one expects that each fundamental cycles are identical to each other. Therefore, the vector $|\Psi\rangle$ must have $2T$-modular invariance. A natural way to implement this condition is by assuming that the coefficients  $\alpha_k (\varphi^{-1}_k(x),z_k)$ are complex numbers and that they are unimodular equivalent after a period $2T$ of  $U_t$ evolution. This argument partially motivates that the  fundamental states $|\varphi^{-1}_k(x),z_k\rangle$ that we shall consider will be complex combinations, $\alpha_k (\varphi^{-1}_k(x),z_k)\in \,\mathbb{C}$.
\begin{comentario}
One main difference between our theory and Dolce's theory of periodic boundary conditions on time coordinate is the times to respect to which the periodicity applies. While for our theory it applies to a $t$-time, in Dolce's theory it applies to the usual notion of macroscopic time \cite{Dolce}.
\end{comentario}

Other possibility to implement the $2T$-modulo invariance is that the time parameters corresponds to subset of quotient fields $\mathbb{Z}/p\mathbb{Z}$, for a large prime $p$. The large limit of this fields has been investigated in \cite{Ricardo quotient rings}, where using metric geometry methods, it was possible to demonstrate the existence of limits for $p\to \,+\infty$.

The norm of a vector $|\Psi\rangle$ is obtained by  assuming a bilinear extension  of the relation \eqref{definiciondeproductoscalarenHPlanck}. Therefore, the norm in $\mathcal{H}_{Fun}$ is  defined by the expression
\begin{align}
\|,\|_{Fun}:\mathcal{H}_{Fun}\to\mathbb{R},\quad \Psi\mapsto \|\Psi\|^2_{Fun}=\,\langle \Psi |\Psi\rangle :=\,\sum^N_{k=1}\,\int_{TM^k_4} \,d^8vol_{\eta^k_4}\,|\alpha_k|^2.
\label{norminHFun}
\end{align}
 The linear closure of the vectors $\Psi$ in $\bar{\mathcal{H}}_{Fun}$ with finite norms  is denoted by $\mathcal{H}_{Fun}$ and with the scalar product associated with this  norm. $\mathcal{H}_{Fun}$ is endowed with a product scale by assuming the parallelogram rule. In this way the vector space $H_{Fun}$ is endowed with a pre-Hilbert space structure. We further assume that for all practical purposes $\mathcal{H}_{Fun}$ is complete with such scalar product and hence, a Hilbert space.

\subsubsection{Mathematical naturalness of the Heisenberg representation}
The invariance in the product \eqref{definiciondeproductoscalarenHPlanck} with respect to the external $\tau$-time parameter is assumed, in accordance with the congruence $2T$-module  invariance property.
In order to achieve this property, we adopt the {\it Heisenberg picture of the dynamics}, where the ontological states $|x^\mu_k,y^\nu_k\rangle$ do not change with time, but the operators $\{\hat{x}^\mu_k,\hat{y}^k_k\}^{4,N}_{\mu,k=1,1}$ change with the $\tau$-time evolution.
\subsubsection{Resolution of the identity}
The constructions discussed above imply the existence of a {\it resolution of the identity of the unity},
\begin{align}
I=\,\sum^N_{k=1}\,\int_{T_{x_k}M^k_4}\,d^4z_k\,\int_{M^k_4}\,dvol_{\eta^k_4}(x_k)\,|x^\mu_k,z^\mu_k\rangle\langle x^\mu_k,y^\mu_k|.
\label{decomposition of unity in coordiante basis}
\end{align}
Note that there is no sum on the index $\mu$. This relation will be applied later.
\subsubsection{Momentum basis}
The momentum operators
\begin{align*}
\{\hat{p}_{xk\mu},\hat{p}_{yk\mu},\,\mu=1,...,4,\,k=1,...,N\},
\end{align*}
being the generators of diffeomorphisms in the real manifolds $M^k_4$,
have real eigenvalues. Since they commute to each other, there is a generator set of common eigenvectors of momentum operators,
\begin{align*}
\{|p_{xk\mu},p_{yk\mu}\rangle ,\,\mu=1,...,4,\,k=1,...,N \}\subset \bar{\mathcal{H}}_{Fun}
\end{align*}
such that
\begin{align*}
\begin{cases}
& \hat{p}_{xl\nu}\,|p_{xk\mu},p_{yk\mu}\rangle=\,\delta_{lk}\,p_{xl\nu}|p_{xk},p_{yk}\rangle,\\
& \hat{p}_{yl\nu}\,|p_{xk\mu},p_{yk\mu}\rangle=\,\delta_{lk}\,p_{yl\nu}\,|p_{xk\mu},p_{yk\mu}\rangle.
\end{cases}
\end{align*}

The vectors $|p_{xk\mu},p_{yk\mu}\rangle$  are not ontological states, since they do not represent the configuration of the sub-quantum degrees of freedom of a Hamilton-Randers systems. The collection $\{|p_{x_k\mu},p_{yk\mu}\rangle,\,\mu=1,...,4,\,k=1,...,N\}$ is the {\it momentum basis}.

 Since the vectors $|p_{xk\mu},p_{zk\mu}\rangle$ are elements of the Hilbert space $\mathcal{H}_{Fun}$, then one can consider the products
  \begin{align*}
  \langle x,y|p_x,p_y\rangle,\quad  (p_x,p_y)\in T^*_{(x,y)}TM.
  \end{align*}
The canonical quantization conditions and the meaning of the operators are analogous to the analogous appearing in standard quantum mechanics. Therefore it is reasonable that the above product has the form
  \begin{align}
  \langle x,y|p_x,p_y\rangle=\,\frac{1}{(2\pi\hbar)^{4N/2}}\,e^{\,\frac{\imath}{\hbar}\left(\sum^N_{k=1}\,p_{x k \mu}\,x^\mu_{k}+\,p_{y k \mu} \,y^\mu_k\right)}\in\,\mathbb{C}.
  \label{matrix element momentum space}
  \end{align}
  This expression, as the analogous one in quantum mechanics, is not invariant under local transformations \eqref{local change of coordinates 1}-\eqref{local change of coordinates 4}.

\subsection{Heisenberg dynamics of Hamilton-Randers dynamical systems}
The Hamiltonian operator of a Hamilton-Randers dynamical system is obtained by applying the quantization procedure described in the previous sections to the classical Hamiltonian \eqref{RandersHamiltonian}, using the canonical quantization rules \eqref{quantumconditions} and \eqref{commutativequantumconditions}. Note that at this level it is not required for the Hamiltonian $H_t$ to be an  Hermitian operator. This is because we are adopting the Heisenberg picture of dynamics only in order to reproduce the deterministic differential equations \eqref{HamiltonEquations} and not as a description of the dynamics of quantum observables.

With this aim, we consider the quantization of Hamiltonian operator \eqref{RandersHamiltonian} in the form
\begin{align}
\nonumber &\widehat{H}_t:\mathcal{H}_{Fun}\to \mathcal{H}_{Fun}\\
  &\widehat{H}_t =\begin{cases}
    \sum^{N}_{k=1}\sum^4_{\mu=1} \hat{\beta}^\mu_{kx}({u})\,\hat{p}_{x\mu k}+\,\sum^{N}_{k=1}\sum^4_{\mu=1}\, \hat{\beta}^\mu_{ky}({u})\,p_{y\mu k}, & t\neq 2 n\,T,\,n\in\, \mathbb{Z},\\
    0, &  t= 2 n\,T,\,n\in\, \mathbb{Z}.
  \end{cases}
  \label{QuantumRandersHamiltonian}
\end{align}
 with $t\in\,[0,2\pi],\,\tau\in\,\mathbb{Z}$.

The Koopman-von Neumann approach to classical dynamical systems explodes an application of Heisenberg's picture of dynamics to classical, deterministic systems described by certain type of ordinary differential equations, with the advantage that one can use the tools of Hilbert theory and operator theory on Hilbert spaces to study properties of the dynamics from a spectral point of view \cite{Koopman1931,von Neumann,ReedSimonI}. In the case of Hamilton-Randers deterministic systems, one can apply Koopman-von Neumann theory  to the $t$-time quantum evolution. Given an hermitian operator $\widehat{O}$, the Heisenberg equation are determined by the Hamiltonian  operator \eqref{QuantumRandersHamiltonian},
\begin{align}
\imath\, \hbar\,\frac{d \widehat{O}}{d t}=\,\left[\widehat{H}_t,\widehat{O}\right].
\label{Heisenberg equation}
\end{align}
 Koopman-von Neumann theory leads to the following fundamental result for Hamilton-Randers theory,
\begin{teorema}
 The Heisenberg equations for $\hat{x}^\mu_k$ and $\hat{y}^\mu_k$ reproduce the first set of Hamiltonian equations \eqref{HamiltonEquations} of the internal evolution.
 \label{HeisenbergHamilton}
\end{teorema}
\begin{proof}
 The canonical quantum theory determined by the commutations relations \eqref{quantumconditions} and \eqref{commutativequantumconditions}) andthe Hamiltonian \eqref{QuantumRandersHamiltonian} in the evaluation of the Heisenberg equations, imply that outside of the metastable ${\bf D}_0$ the Heisenberg equations of evolution for the $t$-time evolution are
 \begin{align*}
 \imath\, \hbar\,\frac{d \hat{u}^i}{d t}=\,\left[\widehat{H}_t,\hat{u}^i\right]=\,\imath\, \hbar\,\left[ \sum^N_{k=1}\,{\beta}^{k}(\hat{u})\hat{p}_{k},\hat{u}^i\right]=\,\imath\, \hbar\,{\beta}^{i}( \hat{u}),\quad i=1,...,N.
 \end{align*}
 That is, the following relations must hold,
 \begin{align}
  \imath\, \hbar\,\frac{d \hat{u}^i}{d t}=\,\imath\, \hbar\,{\beta}^{i}(\hat{u}).
 \label{Heisengerg equation for the fundamental degrees of freedom}
 \end{align}
   The application of this operator equation to each of the fundamental states $|u^\mu_l\rangle=\,|x^\mu_l,y^\mu_l\rangle$,
   one obtains the first set of Hamiltonian equations \eqref{HamiltonEquations} for the $t$-time parameter.
 \end{proof}

 Similarly,  the momentum operators $\hat{p}_{k\mu}$ follow the Heisenberg equations
  \begin{align*}
 \imath\hbar\,\frac{d \hat{p}_{i}}{d t}=\,\left[\widehat{H}_t,\hat{p}_i\right]  =\,\left[ \,\sum^N_{k=1}\,{\beta}^{k}(\hat{u})\hat{p}_{k},\hat{p}_i\right].
 \end{align*}
 By the algebraic rules and the interpretation of the quantum operators given above, the explicit form of this differential equation is
 \begin{align}
 \imath\hbar\,\frac{d \hat{p}_{i}}{d t}=\,-\sum^N_{k=1}\,\widehat{\frac{\partial {\beta}^{k}( u)}{\partial u^i}}\hat{p}_{k}.
 \label{Heisengerg equation for the momenta of fundamental degrees of freedom}
 \end{align}

There is a significant difference between the operational equation for the coordinate and speed operators of the sub-quantum degrees of freedom given by equations \eqref{Heisengerg equation for the fundamental degrees of freedom} and the equation of motion of the momentum operator \eqref{Heisengerg equation for the momenta of fundamental degrees of freedom}, since the system \eqref{Heisengerg equation for the  fundamental degrees of freedom} is an autonomous dynamical system, while the system \eqref{Heisengerg equation for the momenta of fundamental degrees of freedom} is not autonomous, fully depending on the solutions for the equations \eqref{Heisengerg equation for the momenta of fundamental degrees of freedom}. This  property justifies that the generator set of the Hilbert space $\bar{\mathcal{H}}_{Fun}$ is composed by the eigenvectors of the operators $\hat{u}^\mu$ as in \eqref{positionvelocityeigenbasis}. This set of states contains all the independent and sufficient information to describe the dynamics of the system.
For this reason, they are named {\it ontological states}.

 Given the Hamiltonian function \eqref{RandersHamiltonian}, the Hamiltonian operator $\widehat{H}_t$ is not Hermitian and it is not uniquely defined.  One can construct equivalent Hamiltonian operators of the form
 \begin{align}
  \widehat{H}_{\alpha t}(\hat{u},\hat{p}):=\,\frac{1}{2}\,\sum^{8N}_{k=1}\,\Big((1+\alpha){\beta}^{k}(t, \hat{u})\,\hat{p}_{k}-\,\alpha\,\hat{p}_{k}\,{\beta}^{k}(t, \hat{u})\Big)
\label{QuantumRandersHamiltonian alpha}
 \end{align}
 with $\alpha\in\mathbb{R}$. $\widehat{H}_{\alpha t}(\hat{u},\hat{p})$ determines the same equations than $\widehat{H}_{t}(\hat{u},\hat{p})$,
 all of them reproducing the Hamiltonian equations \eqref{HamiltonEquations}. However, if we want to relate this Hamiltonian with a quantum Hamiltonian, it is useful to adopt $  \widehat{H}_{\alpha t}(\hat{u},\hat{p})$ as the Hermitian Hamiltonian. Moreover, we require that the quantum conditions \eqref{commutativequantumconditions} are preserved. A quantization prescription ensuring the canonical commutation conditions is the Born-Jordan quantization prescription \cite{BornJordan1925, BornHeisenbergJordan1925, de Gosson}.

The constraints
\begin{align}
\hat{y}^\mu_k=\,\frac{d \hat{x}^\mu_k}{d t},\quad  \,\mu=1,2,3,4;\quad k=\,1,...,N
\label{constraintsonthevelocity}
\end{align}
are imposed, in parallel to the constraints \eqref{onshell condition on the dynamical system} in the geometric formulation of Hamilton-Randers dynamiccal systems. These constraints define in the Koopman-von Neumann formalism the {\it on-shell evolution} of the sub-quantum molecules as curves on $TM$.

The Heisenberg picture can be adopted as a natural way to describe the dynamics of the ontological degrees of freedom. In the Heisenberg picture, the ket space of sub-quantum states generated by the collection  $\{|x^\mu_l,y^\mu_l\rangle\}^{N,4}_{k=1,\mu=1}$ do not change on $t$-time, but the canonical operators associated with $u$-coordinates and conjugate $p$-momenta they do. Therefore, a generic state $|\Psi\rangle\in \, \bar{\mathcal{H}}_{Fun}$ does not change with $t$-time in this picture, while operators that are functional dependent of the canonical operators $\{\hat{u}^\mu_k,\hat{p}^\mu_k\}^{N,4}_{k=1,\mu=1}$ do evolve with $t$-time.
\subsubsection{Emergence of $\tau$-time diffeomorphism invariance}
Finally, let us note that the meta-stability condition \eqref{finalhamiltoniantevolution} is translated in the quantum formulation of Hamilton-Randers systems  as the following condition,
\begin{align}
\lim_{t\to nT}\,\widehat{H}_t(\hat{u},\hat{p})|\Psi\rangle=0,\quad n\in\,\mathbb{Z}.
\label{Hamiltonianconstrain2}
\end{align}
This is the pre-quantum formulation of $\tau$-time diffeomorphism invariance. Note that it only holds on the meta-estable equilibrium domain ${\bf D}_0$; outside of such a domain, the total Hamiltonian of the system is not constrained by \eqref{Hamiltonianconstrain2} when acting on fundamental states $|\Psi\rangle\in \bar{\mathcal{H}}_{Fun}$.

\subsection{Sub-quantum mechanical operators and quantum mechanical operators}
The exact relation between the emergent operators $\{\widehat{X}^a,\dot{\widehat{X}^a},\,a=2,3,4\}$ and the entire family of operators $\{(\hat{x}^a_k,\hat{y}^a_k,\,\hat{p}^a_{kx_k},\,\hat{p}^a_{ky_k})\}^{N,4}_{k=1,a=2}$ is still missing in our version of the theory.  We should not expect that such correspondence is given in a simple form, among other things because there is no conservation on the number of degrees of freedom. In particular, we have that
\begin{align*}
\sharp \left\{(\hat{x}^a_k,\hat{y}^a_k,\,\hat{p}^a_{kx_k},\,\hat{p}^a_{ky_k}),\,k=1,...,N,\,a=2,3,4\right\}\gg \sharp \left\{\widehat{X}^a,\dot{\widehat{X}^a},\,a=2,3,4\right\}.
\end{align*}
One possibility is to consider the {\it average operators of the fundamental degrees of freedom}. In the case of additive coordinates, this model will be used in \ref{Chapter on emergence of gravity}.

In Hamilton-Randers theory, the fundamental degrees of freedom are not a re-name of the standard quantum degrees of freedom. Indeed, there are many more fundamental degrees of freedom than quantum degrees of freedom.

\subsubsection{On the operator interpretation of the external time coordinate}
The covariant formalism adopted, treating on equal foot each coordinate on the $4$-dimensional manifolds $M^k_4$ is consistent with the principle of relativity, in particular, with the absence of geometric structures as {\it absolute time structures}. While this point of view is natural at the classical level as it was considered in {\it chapter 3}, at the quantum level it requires further attention. In particular, the requirement that $\{\hat{x}^0_k\}^N_{k=1}$ and $\widehat{X}^0$ are operators is in sharp contrast with non-relativistic quantum mechanics, where time is a parameter and does not correspond to a physical observable. This ambiguity is resolved naturally starting by considering the situation in quantum field theory, where the spacetime coordinates $\{X^\mu\}$ are parameters or labels for the quantum fields. This interpretation is also  applicable to the variables $\{x^\mu_k\}^N_{k=1}$, since they are used in Hamilton-Randers models as labels for the configuration of the sub-quantum degrees of freedom and do not define observable quantities.-for marking observables, but. Therefore, it is reasonable that the same happens generically for the spacetime coordinates $\{X^\mu\}$.

The application of Koopman- von Neumann theory to Hamilton-Randers systems does not imply that the coordinates $\{x^\mu_k\}^N_{k=1}$ or $X^\mu$ become observables in our theory, but it is only a formal description of the system.  The coordinates of the sub-quantum molecules are considered as the spectrum of a collection of linear operators and the speed coordinates are bounded.
\subsection{Geometric consistence of the Koopman-von Neumann formulation}
The fundamental algebra determined by the collection of operators \eqref{fundamental operators} satisfying the algebra \eqref{commutativequantumconditions} and \eqref{quantumconditions} was defined in the framework of a local representation of the Hamilton-Randers dynamics. We started with local coordinates $(u^\mu_k,p^\mu_k)$ on $T^*TM$ and promoted them to canonical operators such that the operator equations \eqref{commutativequantumconditions} and \eqref{quantumconditions} hold good. It is still open the question whether such quantization procedure and the Heisenberg dynamics generated by the quantization of the Hamiltonian \eqref{QuantumRandersHamiltonian alpha} are compatible with the transformations induced by changes of coordinates in $T^*TM$ naturally induced by changes of coordinates in $TM$ and whether the dynamics is compatible with such coordinate changes. We address this question in this {\it section}.

Let $\vartheta:\mathbb{R}^{4}\to \mathbb{R}^4$ be a local coordinate change in $M_4$ and consider the induced coordinate changes in $TM$, that we can express as
\begin{align*}
(x^\mu_k,y^\mu_k)\mapsto \left(\tilde{x}^\mu_k,\frac{\partial \tilde{x}^\mu_k}{\partial x^\nu_k}\,y^\nu_k\right)=\,(\tilde{x}^\mu_k, \tilde{y}^\mu_k).
\end{align*}
We want to investigate the consistency  of the relations \eqref{commutativequantumconditions} and \eqref{quantumconditions} with the corresponding canonical relations relations in the {\it tilde} coordinate system,
\begin{align}
\begin{cases}
& \left[\hat{\tilde{x}}^\mu_k,\hat{\tilde{x}}^\nu_l\right]=\,\left[\hat{\tilde{y}}^\mu_k,\hat{\tilde{y}}^\nu_l\right]=\,
\left[\hat{\tilde{x}}^\mu_k,\hat{\tilde{y}}^\nu_l\right]=0, \quad\,\mu,\nu=1,2,3,4;\, k,l=1,...,N,\\
&[\hat{\tilde{x}}^\mu_k,\hat{\tilde{p}} _{\nu y_l}]= [\hat{\tilde{y}}^\mu_k,\hat{\tilde{p}} _{\nu x_l}]=\,0, \quad\,\mu,\nu=1,2,3,4;\, k,l=1,...,N.
\label{commutativequantumconditionstilde}
\end{cases}
\end{align}
and
\begin{align}
[\hat{\tilde{x}}^\mu_k,\hat{\tilde{p}} _{\nu x_l}]= \,\imath\,\hbar\,\delta^\mu_\nu\,\delta_{kl},\,\,\quad [\hat{\tilde{y}}^\mu_k,\hat{\tilde{p}} _{\nu y_l}]=\, \imath\,\hbar\,\delta^\mu_\nu\,\delta_{kl},\quad \mu,\nu=1,2,3,4,\, k=1,...,N.
\label{quantumconditionstilde}
\end{align}
where the operators are taken at fixed $t\in\,\mathbb{R}$. The new momentum operators $\hat{\tilde{p}}^\mu_k$ are defined in analogous manner than the original momentum operators $\hat{p}^\mu_k$, namely, as generators of local diffeomorphism transformations on $TM$ along the respective directions.  Note that we have simplified the notation of the operators $\hat{\tilde{p}}$, in order to avoid excessive cluttering.

 We shall consider the consistency of the canonical quantization for the {\it hat}-operators, subjected to the quantum version of the local coordinate transformation,
\begin{align}
(\hat{x}^\mu_k,\hat{y}^\mu_k)\mapsto \left(\hat{\tilde{x}}^\mu_k,\widehat{\frac{\partial {x}^\mu_k}{\partial x^\nu_k}}\,\hat{y}^\nu_k\right)=\,(\hat{\tilde{x}}^\mu_k, \hat{\tilde{y}}^\mu_k).
\label{quantum local transformations}
\end{align}
The quantization of  the function $ x^\mu_k\mapsto {\tilde{x}}^\mu_k$ and its derivatives are taken according to the {\it Definitions} \eqref{quantization of a function} and \eqref{quantization of the derivative of a function}.

For the argument that follows below we shall need to consider integral operations on $TM $ and a natural volume element on $TM$. We can use the measure $\tilde{\mu}_P$, defined by the relation \eqref{medidamuinvariante0}. Since $\{|x^\mu_k,y^\mu_k\rangle,\,\mu=1,2,3,4,\,k=1,...,N\}$ is a generator system, we have the relation
\begin{align*}
\hat{p}^\mu_j|x^\nu_a,y^\nu_a\rangle=\,\int_{TM}\,\tilde{\mu}_P(x',y')\,
\lambda_l(x'^\nu_i,y'^\nu_i;x^\nu_a,y^\nu_a)\,|x'^\nu_i,y'^\nu_i\rangle .
\end{align*}
Similarly, in the new coordinate system,
\begin{align*}
\hat{\tilde{p}}^\mu_j|x^\nu_a,y^\nu_a\rangle=\,\sum_a\sum_\rho \,\int_{TM}\,\tilde{\mu}_P(x',y')\,\tilde{\lambda}_l(x'^\nu_i,y'^\nu_i;x^\nu_a,y^\nu_a)\,|x'^\nu_i,y'^\nu_i\rangle .
\end{align*}
Then the relation \eqref{quantumconditions} can be re-written as
\begin{align*}
&\int_{TM}\,\tilde{\mu}_P(x',y')\,\lambda_l(x'^\rho_i,y'^\rho_i;x^\nu_a,y^\nu_a)\,x'^\mu_k\,|x'^\nu_i,y'^\nu_i\rangle =\,-\imath \hbar \delta_{lk}|x^\nu_a,y^\nu_a\rangle\\
&+\,\int_{TM}\,\tilde{\mu}_P(x',y')\,\lambda_l(x'^\nu_i,y'^\nu_i;x^\nu_a,y^\nu_a)\,x^\mu_k\,|x'^\nu_i,y'^\nu_i\rangle .
\end{align*}
Similarly, the relation \eqref{quantumconditionstilde} can be re-written as
\begin{align*}
&\int_{TM}\,\tilde{\mu}_P(x',y')\,\tilde{\lambda}_l(x'^\rho_i,y'^\rho_i;x^\nu_a,y^\nu_a)\,\tilde{x}'^\mu_k\,|x'^\nu_i,y'^\nu_i\rangle =\,-\imath \hbar \delta_{lk}|x^\nu_a,y^\nu_a\rangle\\
&+\,\int_{TM}\,\tilde{\mu}_P(x',y')\,\tilde{\lambda}_l(x'^\nu_i,y'^\nu_i;x^\nu_a,y^\nu_a)\,\tilde{x}^\mu_k\,|x'^\nu_i,y'^\nu_i\rangle .
\end{align*}
A consistent criterion extracted from these conditions can be expressed as a condition on the coefficients $\lambda$,
\begin{align}
\nonumber &\lambda_l(x'^\rho_i,y'^\rho_i;x^\nu_a,y^\nu_a)\,x'^\mu_k\,-\lambda_l(x'^\nu_i,y'^\nu_i;x^\nu_a,y^\nu_a)\,x^\mu_k\\
&= \,\tilde{\lambda}_l(x'^\rho_i,y'^\rho_i;x^\nu_a,y^\nu_a)\,\tilde{x}'^\mu_k -\tilde{\lambda}_l(x'^\nu_i,y'^\nu_i;x^\nu_a,y^\nu_a)\,\tilde{x}^\mu_k .
\label{consistent condition for quantization}
\end{align}

We now investigate the consistency of the quantum $U_t$ dynamics determined by the Hamiltonian
\eqref{RandersHamiltonian} and the quantum relations. We need to consider the quantized version of the transformations rules  for the phase-space coordinates $(x^\mu_k,y^\mu_k,p_{xk\mu},p_{yk\mu})$. Therefore, apart from the quantum version of the local coordinate transformations \eqref{quantum local transformations}, we need to consider the corresponding quantized version of the transformation \eqref{local change of coordinates 4},
\begin{align}
& \hat{p}_{kx\rho}=\,\widehat{\frac{\partial \tilde{x}^\mu_k}{\partial x^\rho_k}}\,\hat{\tilde{p}}_{kx\mu}+\,\widehat{\frac{\partial^2 \,\tilde{x}^\mu_k}{\partial x^\rho_k\,\partial x^\sigma_k}}\,\hat{y}^\rho_k\,\hat{\tilde{p}}_{ky\mu}(x),\\
& \hat{p}_{ky\rho}=\,\widehat{\frac{\partial \tilde{x}^\mu_k}{\partial x^\rho_k}}\,\hat{\tilde{p}}_{ky\mu}.
\label{quantum local transformations 2}
\end{align}
Let us remark that the order of the products in the quantization in \eqref{quantum local transformations 2} is irrelevant for the evaluation of the canonical commutations \eqref{quantum local transformations 2}. This is because in the first entry of the commutators $[,]$ will entry combinations of positions $\hat{x}^\mu_k(\tilde{x})$ and velocity operators $\hat{\tilde{y}}^\mu$.

We can now show that the quantum relations are compatible with the relations \eqref{quantum local transformations}-\eqref{quantum local transformations 2}. Let us start considering the first set of equations in \eqref{commutativequantumconditionstilde}. That they are compatible with the quantum change of operators  \eqref{quantum local transformations} and \eqref{quantum local transformations 2} follows from the commutative sub-algebra generated by such operators and because of the structure of the transformations \eqref{quantum local transformations}.

Let us consider the second set of quantum relations in \eqref{commutativequantumconditions}. Then we have that for the first relation
\begin{align*}
0=\,[\hat{x}^\mu_{k},\hat{p}_{yl\rho}]=\,\left[x^\mu_k(\hat{\tilde{x}}),\widehat{\frac{\partial \tilde{x}^{\mu}_l}{\partial x^\rho_l}}\,\hat{\tilde{p}}_{y l \mu}\right]=\,\widehat{\frac{\partial \tilde{x}^{\mu}_l}{\partial x^\rho_l}}\,\left[x^\mu_k(\hat{\tilde{x}}),\hat{\tilde{p}}_{y l \mu}\right].
\end{align*}
Since this happen for any local diffeomorphism $\tilde{x}\mapsto x$, then it must hold the quantum condition
\begin{align*}
\left[\hat{\tilde{x}},\hat{\tilde{p}}_{y l \mu}\right]=0.
\end{align*}

This result is consistent with the geometric interpretation of the momentum operators.
If the momentum operator $\hat{\tilde{p}}_{yl\rho}$ is the generator of the transformation
\begin{align*}
|x^\mu,...,y^\mu_{l},...\rangle \mapsto
|x^\mu,...,y^\mu_{l}+\,d s \,\delta^\mu_\rho \delta_{kl},...\rangle,
\end{align*}
then a fundamental state such that
\begin{align*}
\hat{y}^\mu_k
|x^\mu,...,y^1_1,...,y^\mu_{l}+\,d s \,\delta^\mu_\rho \delta_{kl},...,y^4_{4N}\rangle & =\,
(y^\mu_{\rho}+\,d s \,\delta^\mu_\rho \delta_{kl})\cdot\\
& \cdot |x^\mu,...,y^1_1,...,y^1_{l}+\,d s \,\delta^\mu_\rho \delta_{kl},...,y^4_{4N}\rangle .
\end{align*}
Therefore, we have that
\begin{align*}
\left(\left[x^\mu_k(\hat{\tilde{x}}),\hat{\tilde{p}}_{y l \mu}\right]\right)\,
|x^\mu,...,y^\mu_{l},...\rangle & =\,x^\mu_k(\hat{\tilde{x}})\,\hat{\tilde{p}}_{y l \mu}|x^\mu,...,y^\mu_{l},...\rangle \\
& -\,x^\mu_k(\hat{\tilde{x}})\hat{\tilde{p}}_{y l \mu}|x^\mu,...,y^\mu_{l},...\,x^\mu_k(\hat{\tilde{x}})\rangle\\
& =\,x^\mu_k(\hat{\tilde{x}})\,
|x^\mu,...,y^1_1,...,y^1_{l}+\,d s \,\delta^\mu_\rho \delta_{kl},...,y^4_{4N}\rangle \\
& -\,\hat{\tilde{p}}_{y l \mu}\,x^\mu_k(\tilde{x})\,|x^\mu,...,y^\mu_{l},...\rangle \\
& =\,x^\mu_k(\hat{\tilde{x}})\,
|x^\mu,...,y^1_1,...,y^1_{l}+\,d s \,\delta^\mu_\rho \delta_{kl},...,y^4_{4N}\rangle \\
& -\,
|x^\mu,...,y^1_1,...,y^1_{l}+\,d s \,\delta^\mu_\rho \delta_{kl},...,y^4_{4N}\rangle =0.
\end{align*}

For the quantum relations \eqref{quantumconditionstilde}, we can argue in a similar way,
\begin{align*}
\imath\,\hbar\,\delta^{\mu}_\nu\,\delta_{kl}\,& =\,\left[\hat{x}^\mu_k,\hat{p}_{xl\nu}\right]  =\,\left[\hat{x}^\mu_k(\tilde{x}),\widehat{\frac{\partial \tilde{x}^\lambda_l}{\partial x^\nu_l}}\,\hat{\tilde{p}}_{xl\lambda}+\widehat{\frac{\partial^2 \tilde{x}^\sigma}{\partial x^a x^\nu}}\,\widehat{\frac{\partial x^a}{\partial \tilde{x}^b}}\,\hat{\tilde{y}}^b\hat{\tilde{p}}_{yl\sigma}\right]\\
& =\,\left[\hat{x}^\mu_k(\tilde{x}),\widehat{\frac{\partial \tilde{x}^\lambda_l}{\partial x^\nu_l}}\,\hat{\tilde{p}}_{xl\lambda}\right]+\left[\hat{x}^\mu_k(\tilde{x}),\widehat{\frac{\partial^2 \tilde{x}^\sigma}{\partial x^a x^\nu}}\,\widehat{\frac{\partial x^a}{\partial \tilde{x}^b}}\,\hat{\tilde{y}}^b\hat{\tilde{p}}_{yl\sigma}\right]\\
& =\,\widehat{\frac{\partial \tilde{x}^\lambda_l}{\partial {x}^\nu_l}}\,\left[\hat{x}^\mu_k(\tilde{x}),\hat{\tilde{p}}_{xl\lambda}\right]+\,\widehat{\frac{\partial^2 \tilde{x}^\sigma}{\partial x^a x^\nu}}\,\widehat{\frac{\partial x^a}{\partial \tilde{x}^b}}\,\hat{\tilde{y}}^b\,\left[\hat{x}^\mu_k(\tilde{x}),\hat{\tilde{p}}_{yl\sigma}\right].
\end{align*}
We have just seen that the second commutator  is zero. Therefore, we are led to
\begin{align*}
\imath\,\hbar\,\delta^{\mu}_\nu\,\delta_{kl}\,=\,\widehat{\frac{\partial \tilde{x}^\lambda_l}{\partial x^\nu_l}}\,\left[\hat{x}^\mu_k(\tilde{x}),\hat{\tilde{p}}_{xl\lambda}\right],
\end{align*}
which implies the canonical quantum condition
\begin{align*}
\imath\,\hbar\,\delta^{\mu}_\nu\,\delta_{kl}\,=\,\left[\hat{\tilde{x}}^\mu_k,\hat{\tilde{p}}_{xl\lambda}\right].
\end{align*}
For the second of the quantum conditions \eqref{quantumconditions}, the argument is similar:
\begin{align*}
\imath\,\hbar\,\delta^{\mu}_\nu\,\delta_{kl}\,& =\,[\hat{y}^\mu_k,\hat{p}_{yl\nu}]=\,\left[\widehat{\frac{\partial x^\mu_k}{\partial \tilde{x}^\rho_k}}\,\hat{\tilde{y}}^\rho_k,\widehat{\frac{\partial \tilde{x}^\sigma_l}{\partial x^\nu_l}}\,\hat{\tilde{p}}_{yl\sigma}\right]\\
& =\,\widehat{\frac{\partial \tilde{x}^\sigma_l}{\partial x^\nu_l}}\,\widehat{\frac{\partial x^\mu_k}{\partial \tilde{x}^\rho_k}}\,\left[\hat{\tilde{y}}^\rho_k,\,\hat{\tilde{p}}_{yl\sigma}\right].
\end{align*}
Since the last commutator will be zero, except if $k=l$ and $\sigma=\rho$, and since $\widehat{\frac{\partial \tilde{x}^\sigma_l}{\partial x^\nu_l}}\,\widehat{\frac{\partial x^\mu_k}{\partial \tilde{x}^\rho_k}}\,=\,\delta^\mu_\nu$, then the canonical commutation relation holds,
\begin{align*}
\imath\,\hbar\,\delta^{\mu}_\nu\,\delta_{kl}\,=\,\left[\hat{\tilde{y}}^\mu_k,\,\hat{\tilde{p}}_{yl\nu}\right].
\end{align*}

Similarly, one can argued for the last of the relations in \eqref{commutativequantumconditionstilde}
\begin{align*}
0 & =[\hat{y}^\mu_k,\hat{p}_{x l\nu}]=\,\left[\widehat{\frac{\partial x^\mu_k}{\partial \tilde{x}^\rho_k}}\,\hat{\tilde{y}}^\rho_k,\widehat{\frac{\partial \tilde{x}^\lambda_l}{\partial x^\nu_l}}\,\hat{\tilde{p}}_{xl\lambda}+\widehat{\frac{\partial^2 \tilde{x}^\sigma_l}{\partial x^a_l \partial x^\nu_l}}\,\hat{y}^a\hat{\tilde{p}}_{yl\sigma}\right]\\
& =\,\widehat{\frac{\partial \tilde{x}^\lambda_l}{\partial x^\nu_l}}\,\left(\left[\widehat{\frac{\partial x^\mu_k}{\partial \tilde{x}^\rho_k}},\,\hat{\tilde{p}}_{xl\lambda}\right]\,\hat{\tilde{y}}^\rho_k+\,\widehat{\frac{\partial x^\mu_k}{\partial \tilde{x}^\rho_k}}\,\left[\hat{\tilde{y}}^\rho_k,\hat{\tilde{p}}_{xl\lambda}\right]\right)\\
& +\,\widehat{\frac{\partial^2 \tilde{x}^\sigma}{\partial x^a \partial x^\nu}}\,\widehat{\frac{\partial x^a}{\partial \tilde{x}^b}}\,\hat{\tilde{y}}^b\,\left(\widehat{\frac{\partial x^\mu_k}{\partial \tilde{x}^\rho_k}}\,\,\left[\hat{\tilde{y}}^\rho_k,\hat{\tilde{p}}_{yl\sigma}\right]  +\, \left[\widehat{\frac{\partial x^\mu_k}{\partial \tilde{x}^\rho_k}},\,\hat{\tilde{p}}_{yl\sigma}\right]\hat{\tilde{y}}^\rho_k\,\right)\\
& =\,\imath \,\hbar\widehat{\frac{\partial \tilde{x}^\lambda_l}{\partial x^\nu_l}}\,\widehat{\frac{\partial^2 x^\mu_k}{\partial \tilde{x}^\rho_k\partial \tilde{x}^\lambda_k}}\,\hat{\tilde{y}}^\rho_k\,+\widehat{\frac{\partial \tilde{x}^\lambda_l}{\partial x^\nu_l}}\,\widehat{\frac{\partial x^\mu_k}{\partial \tilde{x}^\rho_k}}\,\left[\hat{\tilde{y}}^\rho_k,\hat{\tilde{p}}_{xl\lambda}\right]\\
& +\,\widehat{\frac{\partial^2 \tilde{x}^\sigma_l}{\partial x^a_l \partial x^\nu_l}}\,
\widehat{\frac{\partial x^a_k}{\partial \tilde{x}^b}_k}\,\hat{\tilde{y}}^b_k\,\left(\imath\,\hbar\,\delta^\rho_\sigma \delta_{kl}\widehat{\frac{\partial x^\mu_k}{\partial \tilde{x}^\rho_k}}+\,\left[\widehat{\frac{\partial x^\mu_k}{\partial \tilde{x}^\rho_k}},\,\hat{\tilde{p}}_{yl\sigma}\right]\hat{\tilde{y}}^\rho_k\,\right).
\end{align*}
Since the last commutator is zero, then we have
\begin{align*}
0 & =\,\imath\,\hbar\,\widehat{\frac{\partial \tilde{x}^\lambda_l}{\partial x^\nu_l}}\,\widehat{\frac{\partial^2 x^\mu_k}{\partial \tilde{x}^\rho_k\partial \tilde{x}^\lambda_k}}\,\hat{\tilde{y}}^\rho_k\,+\widehat{\frac{\partial \tilde{x}^\lambda_l}{\partial x^\nu_l}}\,\widehat{\frac{\partial x^\mu_k}{\partial \tilde{x}^\rho_k}}\,\left[\hat{\tilde{y}}^\rho_k,\hat{\tilde{p}}_{xl\lambda}\right]\\
& +\,\imath\,\hbar\,\widehat{\frac{\partial^2 \tilde{x}^\sigma_l}{\partial x^a_l \partial x^\nu_l}}\,
\widehat{\frac{\partial x^a_k}{\partial \tilde{x}^b}_k}\,\hat{\tilde{y}}^b_k\,\widehat{\frac{\partial x^\mu_k}{\partial \tilde{x}^\sigma_k}}
\end{align*}
After some algebraic manipulations, the first and third terms cancel, leading to the condition
\begin{align*}
0 =\widehat{\frac{\partial \tilde{x}^\lambda_l}{\partial x^\nu_l}}\,\widehat{\frac{\partial x^\mu_k}{\partial \tilde{x}^\rho_k}}\,\left[\hat{\tilde{y}}^\rho_k,\hat{\tilde{p}}_{xl\lambda}\right],
\end{align*}
It follows that
\begin{align*}
0=\,\left[\hat{\tilde{y}}^\rho_k,\hat{\tilde{p}}_{xl\lambda}\right].
\end{align*}

With this last developments we complete the proof of our statements on the compatibility of the canonical conditions, $U_t$ dynamics and  local coordinate transformations induced on $T^*TM$. Remarkably, some of the above arguments also apply to the case of the Dirac fermionic dynamical system.
\subsection{Koopman-von Neumann for composite systems}\label{Koopman-von Neumann of composite systems}
Let us now consider the Koopmann-von Neumann formulation of a composite system. We restrict first the treatment to non-interacting systems in the sense of Hamilton-Randers theory. This means that none of the sub-quantum degrees of freedom of the system $a$ interact with non of the sub-quantum degrees of freedom of the system $b$. If the system $a$ is described by the Hamilton-Randers system $(M_a,(\eta_a,\beta_a))$, the Koopman-von Neumann formulation associates a Hilbert space and a Hamilton-Randers Hamiltonian, $(\mathcal{H}_{Fun}[a], \widehat{H}[a])$. Similarly, for a second system described by $(M_b,(\eta_b,\beta_b))$ and in terms of Koopman-von Neuman by the Hilbert space and Hamiltonian $(\mathcal{H}_{Fun}[b], \widehat{H}[b])$. This mapping defines a functor, that we can appropriately call the {\it Koopman-von Neumann functor}, in the following way. Let us define the category of Hamilton-Randers dynamical systems, whose objects are configuration Hamilton-Randers dynamical models $(M,(\eta,\beta))$ and the morphisms are differential maps between configuration manifolds and the associated maps on tensors and forms. The category of fundamental Koopman-von Neumann models is defined by the fundamental Hilbert spaces associated with the Hamilton-Randers systems and that as morphisms, the homomorphisms between Hilbert spaces induced by the differential functions between the corresponding smooth manifolds. Then the functor $\Phi_{KvN}:Cat_{HR}\to \, Cat_{Fun}(H)$ is defined by $(M_a,(\eta_a,\beta_a))\mapsto \,(\mathcal{H}_{Fun}[a], \widehat{H})$ and the corresponding mappings of morphisms.

Let us consider in detail this construction. Recall from {\it section} \ref{Hamilton Randers models for composite systems} that $M_{a\sqcup b}=\,M_a\times M_b$, $\eta_{a\sqcup b}=\eta_a\oplus \eta_b$ and $\beta_{a\sqcup b}=\,\beta_{a}\oplus \beta_{b}$.
Therefore, we have the following,
\begin{proposicion}
Under the action of $\Phi_{KvN}$,
\begin{align*}
\begin{cases}
 & M_{a\sqcup b}\mapsto \mathcal{H}_{Fun}[a\sqcup b]=\,\mathcal{H}_{Fun}[a]\otimes \mathcal{H}_{Fun}[b],\\
 &(\eta_a\oplus \eta_b, \beta_a\oplus \beta_b)\mapsto \widehat{H}[a\sqcup b]=\,\widehat{H}[a]+\,\widehat{H}[b].
 \end{cases}
 \end{align*}
\end{proposicion}
The action of the functor on the morphisms is induced from the action on the ket basis vectors, such that the following diagram commutes:
\begin{align}
\xymatrix{M_a \ar[r]^{f_{ab}} \ar[d]^{\Phi_{KvN}} & M_b \ar[d]^{\Phi_{KvN}}\\
\mathcal{H}_{Fun}[a]\ar[r]^{\varphi_{ab}}  & \mathcal{H}_{Fun}[b] }
\label{commutative morphisms in HamiltonRanders theory}
\end{align}
where
\begin{align}
\varphi_{ab}(|x\, y\rangle) = \,|f_{ab}(x)\,(f_*)_{ab}(y)\rangle
\end{align}
determines the image of the morphism, $\Phi_{KvN}(f_{ab})=\varphi_{ab}$.

This leads to the following definition,
\begin{definicion}
Two systems "a" and "b" are non-interacting in the Hamilton-Randers sense if the state in the Koopman-von Neumann representation of an ampler system consistim of both of them is a product state of precursor states of the Hilbert spaces $\mathcal{H}_{Fun}[a]$ and $\mathcal{H}_{Fun}[b]$ and if the Hamiltonian of $a\sqcup b$ is the sum of Hamiltonian $\Phi_{KvN}((\eta_a\oplus \eta_b, \beta_a\oplus \beta_b))$.
\label{non-interacting Hamilton-Randers systems}
\end{definicion}

Therefore, if a state is written as a product state it does not guarantee the absence of interactions at the level of sub-quantum degrees of freedom.
\subsection{Set theoretical interpretation}
The formalism just described offers a remarkable picture of the relation between the geometric description and the Koopman-von Neumann description of Hamilton-Randers dynamical systems. Let us consider a composite quantum system formed by by two sub-systems, $a$ and $b$. In terms of sub-quantum degrees of freedom, the composite system is described by the disjoint union $a\sqcup b$. The geometric structure associated with the whole system, when the parts $a$ and $b$ are independent is such that
\begin{align*}
\begin{cases}
& a \mapsto (M_a,\eta_a,\beta_a),\\
& b \mapsto (M_b,\eta_b,\beta_b),\\
& a\sqcup b \mapsto (M_{a}\oplus M_{b}, \eta_a\oplus \eta_b,\beta_a\oplus \beta_b),
\end{cases}
\end{align*}
Clearly, we have that
\begin{align*}
\dim (M_a\sqcup M_b)=\, \dim (M_a)+\, \dim (M_b) .
\end{align*}

On the other hand, the Koopman-von Neumann formalism, when applied this composite system, leads to the following map,
\begin{align*}
\mathcal{H}_{Fun}[a\sqcup b] =\, \mathcal{H}_{Fun}[a]\otimes \mathcal{H}_{Fun}[b] .
\end{align*}
Thus we have that, for what is pointwise counting of dimensions,
\begin{align*}
\dim (\mathcal{H}_{Fun}[a\sqcup b] ) =\,\dim (\mathcal{H}_{Fun}[a])\cdot\, \dim (\mathcal{H}_{Fun}[b]) .
\end{align*}
We observe the existence of a rule relating the dimensions of the configuration manifold and the associated pointwise Hilbert space,
\begin{align*}
\dim (\mathcal{H}_{Fun}) = \,A^{\dim M }.
\end{align*}
 On the other hand, the dimension of $M$ is related with the number of degrees of freedom of the system, $\dim (M) =\, 4 N$. Then, by a suitable choice of $A$, one establishes the relation
 \begin{align}\label{relation geometry Koopman-von Neumann 1}
 \begin{cases}
 & \dim (M) =4 N,\\
 & \dim (\mathcal{H}_{Fun})=\,2^{N} .
\end{cases}
 \end{align}
This suggests the set-theoretical relation,
\begin{align}
\dim(\mathcal{H}_{Fun})=\,2^{M},
\label{relation geometry Koopman-von Neumann 2}
\end{align}
Fixing $x\in M_4$, the set of fundamental vectors is equivalent to the power (in the set theoretical sense) of $\dim(M)$. This interpretation is confirmed by the structure of the basis of $\mathcal{H}_{Fun}$, that as it can be notice from the definition of the states \eqref{pre-quantum state}, it has dimension $2^N$. Then we arrive to the following theorem,
\begin{teorema}
Fixed $x\in M_4$, the set $\mathcal{H}_{Fun}$ has pointwise on  the same cardinality than the power set of subquantum degrees of freedom.
\end{teorema}

\newpage

\section{\LARGE{Quantum mechanical formalism from Hamilton-Randers systems}} \label{chapter of the Hilbert space structure}
\bigskip
\bigskip
We have introduced in {Chapter} \ref{chapter on classical dynamics Hamilton Randers} the elements of Hamilton-Randers dynamical systems and in {Chapter} \ref{chapter on Koopman-von Neumann formulation} the corresponding Koopman-von Neumann formulation. In this {Chapter} the quantum mechanical Hilbert spaces are constructed from the Koopman-von Neumann formulation of Hamilton-Randers dynamical systems. This is achieved by a {\it process of fiber averaging}. This procedure is justified by the ergodic type properties of the fundamental $U_t$ dynamics. {\it Free quantum systems} are discussed from this point of view. This particular but relevant case leads to the conclusion that any quantum state has an emergent interpretation. This is part of {\it the emergence hypothesis}. We also discuss the meaning of the wave function. In particular, we discuss the emergence of the Born rule of quantum mechanics.

A theory of averaging operators is discussed. In particular, the emergent hypothesis is extended to make the assumption that any quantum operator is the average of an operator on the pre-Hilbert space associated with Koopman-von Neumann description of the dynamics of the fundamental states. From this point of view, the mapping from the models describing the fundamental interaction to the quantum description of the system is a functor.
Then the emergence nature of the quantum Heisenberg dynamics from Hamilton-Randers dynamics is discussed, as well as the relation between conservation quantities of the $U_t$ dynamics and the corresponding conserved quantities in the quantum dynamics.

\subsection{Quantum states as averages from predecessor states}
The purpose of this {\it section} is to show how quantum mechanical states can be expressed in terms of averages of the ontological states appearing in the Koopman-von Neumann description of the fundamental dynamics. We start considering a local coordinate system $\{X^\mu\}$ over the spacetime manifold $M_4$. The inverse of the diffeomorphisms  $\{\varphi_k:M^k_4\to M_4\}^k_{k=1}$ and the corresponding localizations provide local coordinates $\{(x^\mu_k,y^\mu_k),\mu=1,2,3,4;\,k=1,...,N\}$ on the configuration manifold $M$. A generic sub-quantum system is described by an element $|\Psi\rangle\in \,\bar{\mathcal{H}}_{Fun}$, with the notation discussed for the notion of predecessor vector \eqref{predecessor vector}, will be of the form
\begin{align}
|\Psi(u)\rangle=\,\frac{1}{\sqrt{N}}\,\sum^N_{k=1}\,e^{\imath\,\vartheta_{k\Psi}(\varphi^{-1}_{k}(x),z_k)} \,n_{k\Psi}(\varphi^{-1}_{k}(x),z_k)\,|\varphi^{-1}_{k}(x),z_k\rangle ,
\label{predecesorofquantumstate}
\end{align}
where the phases
\begin{align*}
\vartheta_{k\Psi} : TM^k_4\to \mathbb{R}, \quad n_{k\Psi}:TM^k_4\to \mathbb{R},\quad k=1,...,N
\end{align*}
are functions that depends on the dynamical system, but do not depend upon de family of diffeomorphisms $\{\varphi_k\}^N_{k=1}$, as we will discuss later.
$|\Psi(u)\rangle$ has finite $\|,\|_{Fun}$ norm as defined by the relation \eqref{norminHFun}. Such elements of $\mathcal{H}_{Fun}$ will be called {\it predecessor states}.

The fundamental assumption of our emergent approach to quantum mechanics is that every quantum state $|\psi\rangle$ is obtained from a predecessor state $|\Psi\rangle\in\,\mathcal{H}_{Fun}$ by averaging over the speed coordinates and by linear combinations of such averages. In Dirac notation, the average of a predecessor state is of the form
\begin{align}
|\psi(x)\rangle= \,\frac{1}{\sqrt{N}}\,\sum^N_{k=1}\,\int_{\varphi^{-1}_{k*} (T_xM_4)} \,d^4z_k\,e^{\imath\,\vartheta_{k\Psi}(\varphi^{-1}_{k}(x),z_k)} \,n_{k\Psi}(\varphi^{-1}_{k}(x),z_k)\,|\varphi^{-1}_{k}(x),z_k\rangle.
\label{ansatzforpsi}
\end{align}
The support of the measure $d^4z_k$ given by \eqref{fibermeasure} is restricted to the causal region of $TM$ determined by the conditions
 \begin{align*}
\eta_k(z_k,z_k)\leq 0,\quad k=1,...,N,
  \end{align*}
being $\eta_k$ the metric of the underlying pseudo-Randers space $(M^k_4,\eta_k, \beta_k)$. Therefore, if $\tilde{y}_k$ lays out of such causal domain, one formally has
\begin{align*}
d^4z_k (x,\tilde{y}_k)\equiv 0.
\end{align*}

The averaging operation is a linear operation: if $|\Psi(u)\rangle_1 \mapsto
  |\psi(x)\rangle_1$ and $|\Psi(u)\rangle_2\mapsto |\psi(x)\rangle_2$, then
  \begin{align*}
  |\Psi(u)\rangle_1+\,\lambda|\Psi(u)\rangle_2\mapsto
  |\psi(x)\rangle_1+\,\lambda|\psi(x)\rangle_2,\quad \forall \,\lambda\in\,\mathbb{C},\,u\in T_x M .
  \end{align*}

The fiber integration along the fibers $\varphi^{-1}_{k*} (T_xM_4)$ associated with the average operation is partially justified by the ergodic property of the $U_t$ flow in the ergodic regime. Given a particular point $x\in\, M_4$, it is assumed that a form of the ergodic theorem can be applied and as a result, the average along the $t$-time parameter in a period $[0,2T]$ is identified with the average along the speed coordinates of the sub-quantum degrees of freedom. Taking into account the un-distinguishability of the degrees of freedom, the average operation must take the form
 \begin{align*}
 \langle \cdot \rangle_t \equiv \, \frac{1}{\sqrt{N}}\,\sum^N_{k=1}\,\int_{\varphi^{-1}_{k*} (T_xM_4)} \,d^4z_k\, \cdot\,,
 \end{align*}
 where the equivalence relation is understood as an equality when operations are acting on objects living on $TM$. The position coordinates are not integrated: they are related with the macroscopic labels used by observers in the description of the quantum states and in the description of measurements and observation in the spacetime arena $M_4$. We are assuming that all possible measurements and observables by measurements of spacetime positions only is adapted from the ideas of D. Bohm \cite{Bohm} and J. Bell \cite{Bell}. In a similar way, we assume that any measurement in a experimental setting can ultimately be translated to measurements of positions of pointers and measurements of time lapses.

 Note that the underlying model manifold for the construction of $\mathcal{H}$ is the model manifold $M_4$, instead than the spacetime manifold $\mathcal{M}_4$.
\subsection{Remarks on the averaging operation in Hamilton-Randers theory}
One of the fundamental assumptions in Hamilton-Randers theory is that each fundamental cycle of the $U_t$ evolution  has three dynamical regimes: the ergodic regime, which is the longest one in terms of the $t$-time parameter, the short (in $t$-time with respect to the ergodic regime) concentrating regime and the short expanding regime. If the $U_t$ flow is exactly ergodic, it will be more natural to apply a strong form of the ergodic theorem (see for instance \cite{ArnoldAvez}) and assume that time average during the ergodic phase of the $U_t$ is equivalent to velocity integration (phase space integration or fiber average). However, the fundamental cycles are not exactly ergodic. For instance, they happen during a finite lapse of $t$-time and are disturbed by the environment. Ergodicity is then an idealization, useful to our constructions.

In our application of the ergodic theorem the time average is taken along the $t$-time evolution, in contrast with the usual time average along the $\tau$-time in dynamical systems theory and classical mechanics. The distinction among the two classes of time parameters is a key ingredient for Hamilton-Randers theory. It also illustrates the emergent character of any macroscopic $\tau$-time parameter.

In the usual applications the average is taken along the phase space or the configuration space of the dynamical system. However, in Hamilton-Randers models the average is taken along the velocity coordinates only. It is a {\it fiber averaging} in a tangent space, even if the evolution takes place in configuration space $TM$.

\subsection{Quantum states as equivalent classes}\label{Equivalence classes} The average state $\psi(x)$ as defined by the expression \eqref{ansatzforpsi} can be re-casted in the following way. Two predecessor states $\Psi_1,\Psi_2\in \,\mathcal{H}_{Fun}$ are say to be equivalent if they have the same average, given by the operation \eqref{ansatzforpsi}. This is an equivalence relation  $\sim_{\langle\rangle}$. A quantum state is an element of the coset space $\mathcal{H}_{Fun}/\sim_{\langle\rangle}$. An element of the coset space containing $\Psi$ is $\psi=[\Psi]:=\{\tilde{\Psi},\,|\,\langle \tilde{\Psi}\rangle_t =\,\psi\}$. The canonical projection is
\begin{align}
\mathcal{H}_{Fun}\to \mathcal{H}_{Fun}/\sim_{\langle\rangle} \,\quad \Psi\mapsto [\Psi]\equiv\psi.
\label{projectionplancktostandard}
\end{align}
The sum operation in the coset space $\mathcal{H}_{Fun}/\sim_{\langle\rangle}$ is defined as follows.
If $\psi_1=[\Psi_1],$ $\psi_2=[\Psi_2]$, then
\begin{align*}
[\Psi_1+\Psi_2]:=\,[\Psi_1]+[\Psi_2]=\,\psi_1 + \psi_2.
\end{align*}
The product by a scalar $\lambda\,\in \,\mathbb{K}$ is defined by
\begin{align*}
[\lambda \, \Psi]:=\,\lambda [\Psi]=\,\lambda  \,\psi.
\end{align*}
Since the operation of averaging that defines the equivalence class is linear, these operations are well defined and do not depend upon the representative of the equivalence classes.

$\psi$ is a section of a complex vector space $\mathcal{H}M_4$ over $M_4$, where the complex vector space is determined by assuming that its sections $\Gamma \mathcal{H}M_4$ coincide with the coset space $\mathcal{H}_{Fun}/\sim_{\langle\rangle}$. In particular, the projection \eqref{projectionplancktostandard} can be extended to lineal combinations of $\mathcal{H}_{Fun}$.
\begin{align*}
\mathcal{H}_{Fun}\to \Gamma\mathcal{H}M_4, \quad \sum_{a}\,\beta_a\Psi_a\mapsto \,\sum_a \beta_a\,[\Psi_a],\quad \Psi_a\in\,\mathcal{H}_{Fun},\,\beta_a\in\,\mathbb{C}.
\end{align*}
A generator system of pre-quantum states implies the existence of a generator system for $\mathcal{H}M_4$.

In a similar vein one can define finite tensor products of $\mathcal{H}_{Fun}/\sim_{\langle\rangle}$.

\subsection{States of highly oscillating relative phase}
Let us consider two states $A$ and $B$  associated with the same physical system. The possible differences in the physical observable attributes to  $A$ and $B$ are originated by different configurations during the $U_t$ evolution of the sub-quantum degrees of freedom.  Since $A$ and $B$ are associated with the same physical system, the associated sub-quantum degrees of freedom are denoted by the same set of sub-index $k=1,...,N$.
 Given a Hamilton-Randers system, we make the assumption that there is a collection of states corresponding to quantum states $\{A,B,...\}$ such that the associated collection of local phases
 \begin{align}
 \left\{ e^{\imath\,\vartheta_{kA}(\varphi^{-1}_{k}(x),z_k)}\, e^{-\imath\,\vartheta_{kB}(\varphi^{-1}_{k'}(x),z_k)},\quad x\in\,M_4,\,z_k\in\,T_{\varphi^{-1}_k(x)}M^k_4,\,k=1,...,N\right\}
  \label{highly oscillating phases}
  \end{align}
are highly oscillating compared to the characteristic frequency $\frac{1}{2T}$ of the corresponding fundamental cycles of the Hamilton-Randers system.

Formally, the highly oscillatory relative phase condition can be casted in the form
\begin{align}
   e^{\imath\,\vartheta_{kA}(\varphi^{-1}_{k}(x),z_k)}\, e^{-\imath\,\vartheta_{kB}(\varphi^{-1}_{k'}(x'),z'_{k'})}\equiv\, C_{A}\,\delta_{AB}\,\delta_{kk'} \delta(x-x')\,\delta(z_k-z'_{k'}),
  \label{highlyoscillatorycondition0}
  \end{align}
  etc..., where $C_{A}$ is a normalization factor depending on the state $A$ (by symmetry, one should have $C_A=C_B$, for all these states A, B,...) and where the equivalence means equality after the insertion of this expression in the fiber integration operation under consideration. The physical significance of this condition is that the sub-quantum degrees of freedom move fast compared with the slow trend associated with the cycles. Since the sub-quantum degrees of freedom are most of the time moving independently from each other, this notion is analogous to the notion of molecular chaos in statistical mechanics. Specifically, assuming the rapid motion assumption, the $\delta_{kk'}$ term in the relation \eqref{highly oscillating phases} arises from the independence between the degrees of freedom $k$ and $k'$, as well as the term $\delta(z_k-z'_{k'})$; the term $\delta_{AB}$ has its origin on the independence character of the degrees of freedom of $A$ with respect to the degrees of freedom of $B$; the term $\delta(x-x')$ is associated with the locality. Note also that the assumption of ergodic regime of the $U_t$ dynamics is fully compatible with the assumption of fast motion and highly oscillating relative phases.

Elements of $\mathcal{H}$ such that the relations \eqref{highly oscillating phases} hold are called {\it highly oscillating relative states}.
The volume measure $d^4z_k$ in $\varphi^{-1}_{k*}(T_xM_4)$ is normalized in such a way  that the relation
\begin{align}
C_A\,\int_{M_4}dvol_{{\eta}_4}\sum^N_{k=1}\int_{\varphi^{-1}_{k*} (T_xM_4)} \,d^4z_k\,|n_k(\varphi^{-1}_{k}(x),z_k)|^2=\,{N}
\label{normalizationcondition for the lines}
\end{align}
holds good.

The functions $\{\vartheta_k:TM^k_4\to \mathbb{R}\}^N_{k=1}$ and $\{n_k:TM^k_4\to M^k_4\}^N_{k=1}$ depend upon the family of diffeomorphisms $\{\varphi_k: M^k_4\to M_4\}^N_{k=1}$ and hence, on the family of manifolds $\{M^k_4\}^N_{k=1}$. However, the relation \eqref{highlyoscillatorycondition0} is independent of the specific choice of $\{\varphi_k\}^N_{k=1}$ . Indeed, after taking averages, the relation \eqref{highlyoscillatorycondition0}
leads to
  \begin{align}
   \langle e^{\imath\,\vartheta_{kA}(\varphi^{-1}_{k}(x),z_k)}\, e^{-\imath\,\vartheta_{kB}(\varphi^{-1}_{k'}(x'),z'_{k'})}\rangle=\, C_{A}\,\delta_{AB}\,\delta(x-x'),
  \label{highlyoscillatorycondition1}
  \end{align}
that is manifestly $M_4$-diffeomorphic invariant in a week form (if spacetime integrations on $M_4$ are performed). Furthermore, if we ask for the relations that leave invariant in a strong form the relation \eqref{highlyoscillatorycondition1} we have that at least there is a global $U(N)$ symmetry. This can be shown if we consider the vector
\begin{align*}
V_A=(e^{\imath\,\vartheta_{1A}(\varphi^{-1}_{1}(x),z_1)},...,e^{\imath\,\vartheta_{NA}(\varphi^{-1}_{N}(x),z_N)})\in\,\mathbb{C}^N.
\end{align*}
 Then the relative phases are casted in the form
\begin{align*}
V_A\cdot V^\dag_B =\,\sum^N_{k=1}\,e^{\imath\,\vartheta_{kA}(\varphi^{-1}_{k}(x),z_k)}\, e^{-\imath\,\vartheta_{kB}(\varphi^{-1}_{k'}(x'),z'_{k'})},
\end{align*}
which is manifestly invariant under $U(N)$ transformations of the form
\begin{align*}
V_A\mapsto \mathcal{U}\cdot V_A,\,V_B\mapsto \mathcal{U}\cdot V_B,\,\mathcal{U}\in\,U(N).
\end{align*}

\subsection{Definition of the quantum pre-Hilbert space $\mathcal{H}$}
In terms of the ontological states $\{|x^\mu_k,y^\mu_k\rangle\}^{N,4}_{k=1,\mu=1}$ and using the product rule \eqref{productoscalarenHPlanck} for the ontological states developing the pre-quantum Hilbert space, one finds that if the highly oscillating relative phase conditions \eqref{highly oscillating phases} hold, then
\begin{align*}
&\int_{M_4}dvol_{\eta_4}\langle \psi_A|\psi_B\rangle (x) =\, \frac{1}{N}\,\sum^N_{k=1}\int_{M_4}\,\int_{\varphi^{-1}_{k*} (T_xM_4)} dvol_{{\eta}_4}\wedge\,d^4z_k\,\sum^N_{k'=1}\,\int_{\varphi^{-1}_{k'*}(T_xM_4)}  d^4z'_{k'}\\
& e^{\imath\,\vartheta_{kA}(\varphi^{-1}_{k}(x),z_k)}\, e^{-\imath\,\vartheta_{kB}(\varphi^{-1}_{k}(x),z'_{k'})} \cdot \,n_{Ak}(\varphi^{-1}_{k}(x),z_k)\,n_{Bk}(\varphi^{-1}_{k'}(x),z'_{k'})\\
& =\, \frac{1}{N}\,\sum^N_{k,k'=1}\int_{M_4}dvol_{{\eta}_4}\,\int_{\varphi^{-1}_{k*} (T_xM_4)}d^4z_k\, \int_{\varphi^{-1}_{k'*}(T_xM_4)}  d^4z'_{k'} C_A\,\delta_{AB}\,\delta_{kk'}\cdot\\
& \cdot\,\delta(z_k-z'_{k'})\,\delta(x-x')\,n_{Ak}(\varphi^{-1}_{k}(x),z_k)\,n_{Bk}(\varphi^{-1}_{k'}(x),z'_{k'})\\
& =\,\delta_{AB}\frac{1}{N}\,\sum^N_{k=1}\int_{\varphi^{-1}_{k*} (T_xM_4)}d^4z_k \,C_A\,n_{Ak}(\varphi^{-1}_{k}(x),z_k)\,n_{Bk}(\varphi^{-1}_{k}(x),z_k).
\end{align*}
This relations imply the {\it orthogonality conditions}
\begin{align}
\int_{M_4}dvol_{\eta_4}\langle \psi_A|\psi_B\rangle =0, \quad \textrm{if}\, A\neq B
\label{orthonormalitycondition}
\end{align}
for highly oscillating  relative phases states.
The same procedure implies the normalization rule
\begin{align}
\int_{M_4}\,dvol_{{\eta}_4}\,\langle\psi|\psi\rangle=\,\int_{M_4}\,dvol_{{\eta}_4}\,\frac{1}{N}\,\sum^N_{k=1}
\int_{\varphi^{-1}_{k*} (T_xM_4)}\,d^4z_k\,C_\psi\,n^2_k(\varphi^{-1}_{k}(x),z_k)=\,1,
\label{normalizationofpsi}
\end{align}
consistent with our previous condition of normalization.
This condition fix the normalization factor
$C_\psi$.
Then we have proved the following result,
\begin{proposicion}
The collection of highly oscillatory relative phase elements of the form \eqref{ansatzforpsi} satisfying \eqref{highly oscillating phases} determines a set of orthogonal elements of $\Gamma\mathcal{H}M_4$.
\end{proposicion}
The extension of this product operation  to $L_2$-norm normalizable elements in $\Gamma\mathcal{H}M_4$ is achieved first by assuming bilinear property of the scalar product compatible with the orthonormal conditions \eqref{orthonormalitycondition} and second, by assuming that the set of orthonormal states of the form $\eqref{ansatzforpsi}$ is a generator system of the quantum Hilbert space, that we denote by $\mathcal{H}$. It is direct that $\mathcal{H}$ is a complex vector space, although not necessarily complete.
Note also that since the elements of $\mathcal{H}$ have finite norm, they are normalizable. Thus the condition \eqref{normalizationofpsi} is re-casted as
\begin{align*}
\|\psi\|_{L_2(M_4)}=1.
\end{align*}

These facts suggest that a natural candidate for the scalar product in $\mathcal{H}$ is the map
\begin{align}
(\cdot|\cdot):\mathcal{H}\times\mathcal{H}\to \mathbb{C},\quad (\psi,\chi)\mapsto (\psi|\chi):=\,\int_{M_4}\,dvol_{{\eta}_4}\,\langle\psi\,|\chi\rangle .
\label{scalar product operation}
\end{align}
For any $\psi,\chi\in\,\mathcal{H}$, the relation
\begin{align}
\nonumber \langle \psi|\chi\rangle =\,&\sum^N_{k=1}\int_{\varphi^{-1}_{k*} (T_xM_4)}d^4z_k
 e^{-\,\imath \vartheta_{\psi k}(\varphi^{-1}_{k}(x),z_k)}\cdot\\
&\cdot e^{\imath \vartheta_{\chi k}(\varphi^{-1}_{k}(x),z_k)}\,n_{\psi k}(\varphi^{-1}_{k}(x),z_k) \,n_{\chi k}(\varphi^{-1}_{k}(x),z_k)
\label{pre-scalar product}
\end{align}
holds good.
Therefore, the scalar product \eqref{scalar product operation} of two generic elements $\psi, \chi\in\,\mathcal{H}$ can be re-casted as
\begin{align}
\nonumber(\psi|\chi)=\,\int_{M_4}\,dvol_{{\eta}_4}\,\frac{1}{N}\,&\sum^N_{k=1}\int_{\varphi^{-1}_{k*} (T_xM_4)}d^4z_k
 e^{-\imath \vartheta_{\psi k}(\varphi^{-1}_{k}(x),z_k)}\cdot\\
&\cdot e^{\imath \vartheta_{\chi k}(\varphi^{-1}_{k}(x),z_k)}\,n_{\psi k}(\varphi^{-1}_{k}(x),z_k) \,n_{\chi k}(\varphi^{-1}_{k}(x),z_k).
\label{detail form of the scalar product}
\end{align}
We conclude that the vector space $\mathcal{H}$ with the $L_2$-norm is a pre-Hilbert space.
\begin{proposicion}
The linear space $\mathcal{H}$ with the $L_2$-norm is a complex pre-Hilbert space.
\label{pre-Hilbert space}
\end{proposicion}
The dual space of $\mathcal{H}$ is formed by linear functionals and will be denoted by $\mathcal{H}^*$.

Given a Hamilton-Randers system, the elements of the space $\mathcal{H}$ are associated with the states describing the system as a quantum dynamical system. Conversely, we make the assumption that any quantum system is described by a Hilbert space that can be constructed in the above form.
\begin{comentario}
In standard quantum mechanics there is a relevant class of wave functions which are the ones associated with {\it free quantum states}. Such wave functions are not $L_2$-finite when the spacetime manifold $M_4$ is not compact. Therefore, they are not elements of the Hilbert space. However, elements of the Hilbert space are linear combinations of them defining wave packets. Due to the relevance of such {\it free quantum states}, it is reasonable to consider free state wave functions as a {\it generator set} of $\mathcal{H}$. It will be shown that free quantum states have also an emergent origin.
\end{comentario}
\subsection{On the representation of position states in terms of sub-quantum degrees of freedom}
As direct application of the product rule \eqref{definiciondeproductoscalarenHPlanck} to the states \eqref{ansatzforpsi},  one obtains the relation
\begin{align}
\int_{\varphi^{-1}(M^k_4)}dvol_{\eta_4}(x')\,\langle \psi(x')|x'_k,z_k\rangle =\,\frac{1}{\sqrt{N}}\, e^{\imath\,\vartheta_{\Psi k}(\varphi^{-1}_{k}(x),z_k)}\,n_{\Psi k}(\varphi^{-1}_{k}(x),z_k).
\label{cambiodebase}
\end{align}
 The integral operation in the relation \eqref{cambiodebase} reflects the result of a domain of ergodicity for the $U_t$ evolution.

The expression \eqref{cambiodebase} bears certain resemblance with the quantum mechanical relation between the space representation and the momentum representation in quantum mechanics for systems with one degree of freedom in three spatial dimensions, a relation given by the {\it matrix elements} of the form
\begin{align}
\langle \vec{p}\,|\vec{x}\rangle=\,\frac{1}{(2\pi\,\hbar)^{3/2}}\,e^{-\frac{\imath\,{x}^\mu\,{p}_\mu}{\hbar}}.
\label{cambiodebaseenmecanicacuanticaemergente}
 \end{align}
 This representation is obtained in quantum mechanics from the coordinate representation of the translation operator by solving a first order linear differential equation (see \cite{Dirac1958}, pg. 94 and following). The relation  \eqref{cambiodebaseenmecanicacuanticaemergente} determines an unitary transformation relating different representations of the quantum Hilbert space.

The expression \eqref{cambiodebase} defines the phase $\vartheta(\varphi^{-1}_{k}(x),z_k)$ and the module $n_k(\varphi^{-1}_{k}(x),z_k)$ of the ontological state $| x_k,z_k\rangle \in\,\mathcal{H}_{Fun}$ in terms of the dual of the element of $|\psi\rangle\in \,\mathcal {H}$ in the expression \eqref{cambiodebaseenmecanicacuanticaemergente}. The role of $\langle \vec{p}\,|$ in  \eqref{cambiodebase} is played by  $\langle \psi|$ in \eqref{cambiodebaseenmecanicacuanticaemergente}, suggesting an emergent interpretation for the $1$-forms $\langle \vec{p}\,|\in\,\mathcal{H}^*$. Instead, the role of $|\vec{x}\rangle$ in \eqref{cambiodebase} should be played by the collection of vectors $\{|x_k,z_k\rangle\}^N_{k=1}$, instead than only by one vector $|x_k,z_k\rangle$. Noting that $x_k=\,\varphi^{-1}_k(x)$, the analogy between \eqref{cambiodebase} and \eqref{cambiodebaseenmecanicacuanticaemergente} can only be completed as equivalence between quantum states and emergent states after one integrates out the velocity coordinate fiber variables in a predecessor state $|\Psi(x)\rangle$ of the form
\begin{align}
|{x}\rangle :=\,\frac{1}{\sqrt{N}}\,\sum^N_{k=1}\,\int_{\varphi^{-1}_{k*}(T_{x}M_4)} \,d^4z_k\,\,e^{-\frac{\imath \sum _\mu z_{k\mu} (\varphi^{-1}(x))^\mu}{\hbar}}\,|\varphi^{-1}_{k}(x),z_k\rangle .
\label{four dimensional spacetime position state}
\end{align}
Let us remark the $4$-dimensional character of the state \eqref{four dimensional spacetime position state}.

\subsection{Emergent interpretation of the wave function and the Born rule}\label{Interpretation of wave function}
\begin{definicion}
 The number of sub-quantum degrees of freedom world lines at $x \in M_4$ of an arbitrary element $|\psi(x)\rangle\in\,\mathcal{H}$ is denoted by  $n^2(x)$ and is given by the expression
\begin{align}
n^2(x)=\,\sum^N_{k=1}\,\int_{\varphi^{-1}_{k*} (T_xM_4)} \,d^4z_k\,n^2_k(\varphi^{-1}_{k}(x),z_k).
\label{number of world lines}
\end{align}
\end{definicion}
As a consequence of the orthogonality relation \eqref{orthonormalitycondition}, for an arbitrary normalized combination $\sum_A\,\lambda_A \,|\psi(x)\rangle_A \in\,\mathcal{H}$ the density of world lines passing close to a given point $x\in\, M_4$ is not given by the expression
\begin{align*}
n^2(x)\neq\,\sum_A\,|\lambda_A|^2\,\sum^N_{k=1}\,\int_{\varphi^{-1}_{Ak*}(T_xM_4)} \,d^4z_{Ak}\,n^2_{Ak}(\varphi^{-1}_{Ak}(x),z_{Ak}),
\end{align*}
but by a more complicated expression that contains crossing terms, that leads to interference terms. However, it holds that
\begin{align}
\int_{M_4}\,dvol_{\eta_4}\,n^2(x)=\,\sum_{A}\,|\lambda_A|^2\,n^2_A=\,N.
\label{normalization of emergent states}
\end{align}

Given a macroscopic observer $W\in\,\Gamma TM_4$, there is defined the distance structure $d_W$ as it was introduced by the function \eqref{distancexxi}. Let us consider an arbitrary point $x\in\,M_4$. We say that the world line $\varphi_k({\xi_t}_k):I\to M_4$  passes close to $x\in \,M_4$ if its image in $M_4$ by the diffeomorphism $\varphi_k:M^k_4\to M_4$ is in the interior of an open set $\mathcal{U}(x,L_{min})$ whose points are at a distance less than $\L_{min}$ from the point $x\in M_4$, using the metric distance $d_W$ as defined in chapter 3. By application of a strong form of the ergodic theorem (we assume that such possibility is legitim), the density of lines $n^2(x)$ is the number of world lines $\varphi_k({\xi_t}_k):I\to M_4$ passing close to $x$ for a fixed internal time $t(l)$ for the  $l$-fundamental cycle $t(l)\in\,[(2l+1)T,(2l+3)T]$.

This interpretation of the density $n^2(x)$ is not $\Diff(M_4)$-invariant for generic models, a fact which is consistent with usual models in quantum mechanics.  However, the existence of measures that are $\Diff(M_4)$-invariant  allows to construct $\Diff(M_4)$-invariant models, where the density $n^2_k$ transform appropriately.

Fixed an observer $W\in\Gamma\,TM$ and a local coordinate system on $M_4$, $n^2(x)$ is a measure of {\it the relative presence} of the system at the point $x\in\,M_4$. Since these densities can be normalized, $n^2(x)$ can be interpreted as the probability to find a particle at $x$ if a spacetime position measurement is done. Note that this interpretation refers to a system composed of one particle, like it could be an electron or photon. However, for a preparation of a system of many individual particles in the same way, $n^2(x)$ can be associated with the statistic ensemble.

Since $n^2_k(\varphi^{-1}_k(x),z_k)$ is the density  of world lines passing close to $\varphi^{-1}_k(x)\in\,M^k_4$ with velocity speed vector $z_k$, then $n^2(x)$ can be read as the {\it total density} of world-lines of sub-quantum molecules close to the point $x\,\in M_4$.
 We have the following result,
\begin{proposicion}
  For any element $\psi\in \mathcal{H}$ of the form \eqref{ansatzforpsi} it holds that
\begin{align}
\frac{1}{N}\,n^2(x)=\,(\psi(x)|\psi(x)),
\label{Bornrule}
\end{align}
\label{teorema de la regla de Born}
\end{proposicion}
\begin{proof}
This is direct consequence of the relations \eqref{pre-scalar product} and \eqref{number of world lines}.
 \end{proof}

  The above result has the following interpretation.
  Since $n^2(x)$ is the number of world lines of sub-quantum degrees of freedom passing close to $x\in M_4$ in the above sense, $\frac{1}{N}\,n^2(x)$ measures the average density of such world lines at $x\in M_4$. Note that, during the ergodic regime, there is a multitude of points of $M_4$ such that are passed by sub-quantum degrees of freedom in the sense above, by means of the diffeomorphism $\varphi_k:M^k_4\to M_4$. Thus the contribution to the average value of the position is proportional to $n^2(x)/N$. View as average in $t$-time, if a collapse happened and the position was well defined, the probability to find the system in an infinitesimal neighborhood containing $x$ is
  \begin{align*}
  n^2(x)/N \,d^4x=:\,|\psi(x)|^2\,d^4 x,
   \end{align*}
   since it is the proportion of time that passes at that point, in the sense described above.  Therefore, $|\psi|^2(x)=\,(\psi(x)|\psi(x))$ measures the probability to find the system in an infinitesimal neighborhood of $x$, if a collapse happened evenly according to the relative time the system spends there  during the $U_t$ evolution.

  Let us make some remarks. First, this interpretation of the function $|\psi|^2(x)$ is not yet equivalent to the Born rule of quantum mechanics, since it has not yet been proved that there is a collapse, where the system is polarized. The existence of a mechanism for the collapse of the wave function will be considered in the next chapter. Second, the above interpretation for $|\psi|^2(x)$ applies to individual systems. However, it is only possible to check it when a multitude of identical systems are considered. Due to the fact that each quantum system is is a complex system from the point of view of Hamilton-Randers theory, then it is natural to think that the specific initial conditions for the sub-quantum degrees of freedom are different and that therefore, different values will be attained to measurements. The statistical profile of the observed events, however, is determined by $|\psi|^2(x)$, which is the same for every individual state that we prepare with the same wave function.

Let us now consider a general state of $\mathcal{H}$. For a general combination of elements from $\mathcal{H}$, the square of the modulus contains interference terms. In the case of the superposition of two states, one has the relation
\begin{align*}
|\lambda_A\psi_A +\,\lambda_B\psi_B|^2 & =\,|\lambda_A\,|\psi_A|^2+\,|\lambda_B\psi_B|^2 +2 \Re (\lambda_A\lambda^*_B\psi_A\,\psi^*_B)\\
& =\,|\lambda_A|^2\,|\psi_A|^2+\,|\lambda_B|^2\,\psi_B|^2 +2 \, \Re (\lambda_A\lambda^*_B\langle\psi_A|\psi^*_B\rangle).
\end{align*}
On the other hand, the number of lines function \eqref{number of world lines} is not of the form
\begin{align*}
n^2(x)_{\lambda_A\psi_A+\lambda_B\psi_B}=\,|\lambda_A|^2\,n^2_A(x)+\,|\lambda^2_B|n^2_B(x),
\end{align*}
but $n^2(x)_{\lambda_A\psi_A+\lambda_B\psi_B}$ also contains interference terms, that comes from the mixed phase products of $A$ and $B$ terms,
\begin{align*}
 n^2(x)_{\lambda_A\psi_A+\lambda_B\psi_B}=\,|\lambda_A|^2\,n^2_A(x)+\,|\lambda^2_B|n^2_B(x)+\,N\,2 \,\Re (\lambda_A\lambda^*_B\langle\psi_A|\psi^*_B\rangle).
 \end{align*}
 The interference terms are present on the density, but by the orthogonality condition \eqref{orthonormalitycondition}, they vanish when integrate on $M_4$. Also, the decomposition of a state $\psi$ does not translate in an homomorphism on the density. This is to be expected, since in terms of the fundamental degrees of freedom, one expects to have interactions between the degrees of freedom of $A$ and $B$. Specifically, due to the universal character of the degrees of freedom, one expects to have the same type of sub-quantum degrees of freedom for $A$ and for $B$. Thus there is a mixing and interaction among them.

 A typical example of quantum interference and the consequent modification of the associated classical probability theory will be discussed in {\it Chapter} \ref{chapter on properties of quantum mechanics}.

 \subsection{Properties of the functions $n^2_k(x)$ and $n^2(x)$}
The number of lines $n^2(x)$ is normalized by the condition
\begin{align*}
\int_{M_4}\,dvol_{{\eta}_4}\,n^2(x)= N
\end{align*}
and since in our models the number of degrees of freedom $N$ is conserved for an isolated quantum system, there is the additional constraint
\begin{align*}
\frac{d}{d\tau} \,\left(\int_{M_4}\,n^2(x)\,dvol_{{\eta}_4}\right)=0.
\end{align*}
The $\tau$-time derivative commutes with the integral operation,
\begin{align*}
\frac{d}{d\tau} \,\left(\int_{M_4}\,n^2(x)\,dvol_{{\eta}_4}\right) &=
\int_{M_4}\,\frac{d}{d\tau} \,\left(n^2(x)\,dvol_{{\eta}_4}\right)\\
& =\,\int_{M_4}\,\frac{d}{d\tau} \,\left(n^2(x)\right)\,dvol_{{\eta}_4}=0.
\end{align*}
Now let us assume the product structure for $M_4=\ \mathbb{R}\times\,M_3$ and that the invariant volume form is $dvol_{{\eta}_4}=\,d\tau\wedge\,d\mu_3$, where $d\mu_3(\tau,\vec{x})$ is an invariant volume form on each leave $\{\tau\}\times M_3$. Then the above constraint is expressed as
\begin{align*}
\int_{\mathbb{R}\times M_3}\,{d\tau}\wedge \,d\mu_3\,\frac{d}{d\tau}\left(\sum^N_{k=1}\,n^2_k\right)=0.
\end{align*}
Let us apply a {\it bump} diffeomorphism $\theta:M_3\to M'_3$ such that $\theta$ is the identity, except for a small region around a point $x_0\in M_3$, but such that it produces an arbitrary {\it bump} around $x_0$. Then by the invariance under diffeomorphism property of the integrals, at the point $x_0$ we have
\begin{align*}
\frac{d}{d\tau}\left(\sum^N_{k=1}\,n^2_k\right)=0.
\end{align*}
This expression can be re-written as a continuity equation, that in order to be consistent with Madelung continuity equation, must be of the form
\begin{align}
\frac{\partial n^2}{\partial \tau}\,+\vec{v}\,\vec{\nabla}\,n^2+\,n^2(x)\,\vec{\nabla}\vec{v}=0,
\label{v in terms of n}
\end{align}
where $n^2(x)=\,\sum^N_{k=1}n^2_k$.
In this expression the vector field $\vec{v}$ stands for the field of an hypothetical fluid that makes the above equation self-consistent. The fluid describe the evolution of the sub-quantum degrees of freedom. Hence it has a different interpretation than in quantum mechanics \cite{Bohmquantumtheory} and that in Bohm theory \cite{Bohm}.

The existence of the vector field $\vec{v}$ is assumed in Hamilton-Randers theory, in contrast with quantum mechanics, where it is fixed by demanding compatibility with Schr\"{o}dinger equation.
\subsection{On the emergent character of free quantum states}
As an example of the models that can be associated with Hamilton-Randers systems, let us consider a quantum state of a free quantum state $|P^\mu\rangle$. In terms of the basis of spacetime position eigenstates $|X^\mu\rangle$, the state is decomposed as
\begin{align}
|P^\mu\rangle=\,\int_{M_4}\,d^4 x\,\langle X^\mu|P^\mu\rangle\,|X^\mu\rangle,
\label{momentum state}
\end{align}
where the coefficient $\langle X^\mu|P^\mu\rangle$ is a scalar. If the states $|P^\mu\rangle$ and $|X^\mu\rangle$ are of the type described by Hamilton-Randers models, then by linearity
\begin{align*}
\langle X^\mu|P^\mu\rangle=\,\frac{1}{\sqrt{N}}\,\sum^N_{k=1}\,\int_{\varphi^{-1}_k(T_xM)}d^4z_k\,e^{\imath\,
\vartheta_{kP}(\varphi^{-1}_{k}(x),z_k)}\,n_{kP}(\varphi^{-1}_k(x),z_k)\,\langle X^\mu|\varphi^{-1}_k(x),z_k\rangle.
\end{align*}
The scalar $\langle X^\mu|\varphi^{-1}_k(x),z_k\rangle$ is understood in the sense of the scalar product in $\mathcal{H}_{Fun}$.
In order to compare with the free solution of the Schroedinger equation in quantum mechanics, represented by the expression
\begin{align}
\langle X^\mu|P^\mu\rangle=\,\frac{1}{(2\pi\hbar)^{4/2}}\,e^{-\frac{\imath X^\mu P_\mu}{\hbar}},
\label{freequantumstate}
\end{align}
we consider the vector ket $|x\rangle$ given by the expression \eqref{four dimensional spacetime position state}. Taking the product, we obtain
\begin{align*}
\langle X |P\rangle =\,\frac{1}{N}\sum^N_{k=1}\int_{\varphi^{-1}_{k*} (T_xM_4)}d^4z_k  e^{-\imath\left(\frac{\, z_{k\mu} \left(\varphi^{-1}_k(x)\right)^{\mu}}{\hbar}-\vartheta_{kP}(\varphi^{-1}_{k}(x),z_k )\right)}\,n_{kP}(\varphi^{-1}_k(x),z_k) .
\end{align*}
Comparing with \eqref{momentum state} we get the following conditions,
\begin{align}
\begin{cases}
& n_{kP}(\varphi^{-1}_k(x),z_k) =\, \frac{1}{(2\pi\hbar)^{4/2}} \, \delta_4 (z_{k\mu}-P_\mu),\\
& \vartheta_{kP}(\varphi^{-1}_{k}(x),z_k)) =0 ,
\end{cases}
\end{align}
which leads to
\begin{align*}
\langle X |P\rangle =\,\frac{1}{N}\,\sum^N_{k=1} \, \frac{1}{(2\pi\hbar)^{4/2}} \, e^{ -\imath \,\frac{\, P_\mu \left(\varphi^{-1}_k(x)\right)^{\mu}}{\hbar}} .
\end{align*}

Locally, there are on each manifold $M^k_4$ a special coordinate system such that $(\vartheta^{-1}_k(x))^\mu =X^\mu$. Then, in such local coordinate system of $M$, one has
\begin{align*}
\langle X |P\rangle =\,\frac{1}{N}\,\sum^N_{k=1} \, \frac{1}{(2\pi\hbar)^{4/2}} \, e^{ -\imath\,P_\mu \left(\varphi^{-1}_k(x)\right)^{\mu}} =\,\frac{1}{(2\pi\hbar)^{4/2}} \, e^{-\imath\,\frac{P_\mu\,X^{\mu}}{\hbar}}.
\end{align*}

In the above consideration, a four-dimensional formulation has been adopted, consistently with the formulation of Hamilton-Randers theory. In particular, it is consistent with the existence of a local maximal speed. However, note that the appearence of the delta function $\delta_4 (z_{k\mu}-P_\mu)$ makes the construction not general covariant, since there is the identification on dual forms pertaining to different spaces. Similarly, the construction builds on a particular class of local coordinates, hence breaking the general covariance.
\subsection{On the emergent character of quantum states in general}
It is worth good to remark the importance of the free quantum case. This is because, due to the linearity of the operations involved in the map $\mathcal{H}_{Fun}\to \mathcal{H}$, the emergent character can be extended to arbitrary linear combinations of free states with finite norm. Indeed, we have that
 \begin{proposicion}
 If a Hilbert space $\mathcal{H}$ admits a generator system composed by free states, then any element of $\mathcal{H}$ is recovered from an element of $\mathcal{H}_{Fun}$ by fiber averaging.
 \label{extension of emergene to wave packets}
 \end{proposicion}
According to this reasoning, every wave function of quantum mechanical system has an emergent character, in the sense that they correspond to elements $\psi\in\,\mathcal{H}$ with an ansatz of the form \eqref{ansatzforpsi} in terms of ontological states.

 Hamilton-Randers system does not apply directly to the quantum mechanical description, but refers to an underlying level of physical reality beneath the level usually described by quantum mechanics. The theory is not aimed to complete quantum mechanics in order to provide a justificiation of the values of the macroscopic observables, as it is the case of hidden variables \cite{Bell Introduction 1971, Redhead}. Indeed, the Hamilton-Randers description of an individual dynamical system does not determines the outcome of a given measurement on the system.

\subsection{An example of interacting Hamilton-Randers models} The following example is taken from \cite{Ricardo06} and describes an elementary method  to introduce non-trivial interactions in Hamilton-Randers theory.
Suppose a system composed by two identical elementary system, being their dynamics 
described by a deterministic Hamiltonian  of the form \eqref{RandersHamiltonian},  their Hamiltonian are determined by $(\alpha _1 ,\beta _1)$ and $(\alpha 
_2 ,\beta _2)$. The $1$-forms $\beta _i,\, i=1,2. $ have norm less than one by the 
corresponding Riemannian norms $\alpha _i,\, i=1,2$.
There are at least two ways to produce a bigger Randers space using just the above geometric 
data:
\begin{enumerate}
\item The first way is valid for complete general structures
\begin{displaymath}
\alpha =\alpha _1 \oplus \alpha _2 ;\,\, \beta =\beta _1 \oplus \beta _2 .
\end{displaymath}
This construction does not produce interaction terms in the total Hamiltonian. There is a 
priori not relation ${\alpha}_1$ ${\alpha}_2$.
\item The second form recovers the impossibility for a external observer to differentiate  
between identical particles:
\begin{displaymath}
\vec{p}=\vec{p}_1 \times \vec{0} +\, \,\vec{0}\times \vec{p}_2 ;\vec{\beta}=\vec{\beta}_1 
\times \vec{0} +\, \,\vec{0}\times \vec{\beta}_2 ,
\end{displaymath}
\begin{displaymath}
\alpha =\alpha _1 \oplus \alpha _2;\, \, \alpha _1 =\alpha _2.
\end{displaymath}
The quantum total Hamiltonian is given by:
\begin{displaymath}
\vec{\beta}(\vec{p})=\big(\vec{\beta} _1 (\vec{p}_1)+\vec{\beta} _1 
(\vec{p}_2)+\vec{\beta} _2 (\vec{p}_1)+\vec{\beta} _2 (\vec{p}_2)\big).
\end{displaymath}
The mixed terms produce the interaction. The condition $\alpha _1 =\alpha _2$ is to ensure that the above construction is a Randers space.
\end{enumerate}
The interaction of the quantum systems is related with the interchange of sub-quantum degrees among them. As we will discuss later, this general interaction implies a quantum mechanism for quantum entanglement.

\subsection{Emergence of the classical $\tau$-time diffeomorphism invariant constraint}
 The relation \eqref{Hamiltonianconstrain2} is a constraint on $\mathcal{H}_{Fun}$. Since $\mathcal{H}$ is a subset of $\mathcal{H}_{Fun}/\sim_{\langle \rangle}$, an analogous constraint can be applied to a subset  of elements in $\mathcal{H}$. Such constraint on the Hilbert space $\mathcal{H}$ is the quantum  version of the $\tau$-time diffeomorphism invariant condition,
\begin{proposicion}
For any physical state $|\psi\rangle\,\in \mathcal{H}$ the constraint
\begin{align}
\lim_{t\to (2n+1)T}\,  \widehat{H}_t(u,p)\,|\psi\rangle=0,\quad n\in\,\mathbb{Z}.
\label{averagehamiltonianevolution}
\end{align}
holds good.
\end{proposicion}
This constraint holds periodically in the $t$-time parameter, with a periodicity $2\,T$. Thus $\tau$-time re-parametrization invariance only holds in the metastable domain ${\bf D}_0$. This is in agreement with the diagonal $\Diff(M_4)$-invariance symmetry of the Hamilton-Randers systems formulated in chapter 3.

The constraint \eqref{averagehamiltonianevolution} bears a strong formal similarity with the Wheeler-DeWitt equation \cite{DeWitt}, but it is fundamentally different. First, while the Wheeler-DeWitt applies to the whole universe, the relation \eqref{averagehamiltonianevolution} applies to any isolated quantum system. Second, the constraint \eqref{averagehamiltonianevolution} only holds in the metastable  regime ${\bf D}_0$ and the expanding regime of each fundamental cycle and not during the whole $U_t$ flow evolution. Thus it is unlikely that \eqref{averagehamiltonianevolution} is a first order constraint as is the case of the Wheeler-Dewitt equation. Furthermore, we did not specified a particular model for the Hamiltonian $\widehat{H}_t$, that obviously contrast with the exact well-defined structure of the Wheeler-DeWitt equation. Therefore, the condition \eqref{averagehamiltonianevolution} expresses is the emergent origin of local time diffeomorphism invariance from Hamilton-Randers theory.

\subsection{On the emergent origin of Heisenberg equations for quantum systems}\label{Heisenberg dynamics}
\subsubsection{Average of operators}
The view of quantum states as equivalence classes can be generalized to include quantum operators as equivalence classes of operators acting on the Hilbert spaces describing the sub-quantum degrees of freedom. Given an operator automorphism
\begin{align*}
\widehat{\mathfrak{O}}:\mathcal{H}_{Fun}\to \mathcal{H}_{Fun} ,
\end{align*}
 not necessarily linear, acting on fundamental states, the induced operator on the coset space
\begin{align*}
\langle\widehat{\mathfrak{O}}\rangle:\mathcal{H}_{Fun}/\sim_{\langle,\rangle} \to \mathcal{H}_{Fun}/\sim_{\langle,\rangle}
\end{align*}
is such that the diagram
\begin{align}
\xymatrix{\mathcal{H}_{Fun} \ar[d]_{\langle,\rangle} \ar[r]^{\widehat{\mathfrak{O}}} &
\mathcal{H}_{Fun} \ar[d]^{\langle,\rangle}\\
\mathcal{H}_{Fun}/\sim_{\langle,\rangle} \ar[r]^{\langle\widehat{\mathfrak{O}}\,\rangle} & \mathcal{H}_{Fun}/\sim_{\langle,\rangle}}
\label{average of an operator raw space}
\end{align}
commutes. In particular, when the operator $\widehat{\mathfrak{O}}$ is linear, the commutativity of the diagram \eqref{average of an operator raw space} provides a characterization of the operator $\langle\widehat{\mathfrak{O}}\rangle$. Thus  $\langle\widehat{\mathfrak{O}}\rangle$ is determined by the relation
\begin{align}
\langle\widehat{\mathfrak{O}}\rangle|\psi\rangle_t := \,\langle \widehat{\mathfrak{O}}|\Psi\rangle_t ,
\label{average of an operator}
\end{align}
for each $|\Psi\rangle \in \,\mathcal{H}_{Fun}$ and with $|\psi\rangle =\,\langle \Psi\rangle_t.$
This construction defines a linear operator
\begin{align*}
\langle\widehat{\mathfrak{O}}\rangle \,:\mathcal{H} \to \mathcal{H}
\end{align*}
which is associated with a quantum operator acting on the pre-Hilbert space $\mathcal{H}$ of quantum states associated with the system.
The {\it average endomorphisms map} $\langle,\rangle :End (\mathcal{H}_{Fun})\to End (\mathcal{H})$ is such that \eqref{average of an operator} holds for each $|\Psi\rangle \in \,\mathcal{H}_{Fun}$ highly oscillating.
\begin{proposicion}
The average endomorphism map is an morphism for the corresponding composition laws.
\end{proposicion}
\begin{proof}
Let us consider to operators $\widehat{\mathfrak{O}}_1,\,\widehat{\mathfrak{O}}_2\in\,End(\mathcal{H}_{Fun})$.  The homomorphism property with respect to the sum of operators and with respect to the multiplication by scalars is direct, since the average operation $\langle ,\rangle$ is linear. To prove the homomorphism property with respect to the product of operators we consider the following expressions,
\begin{align*}
\langle \left(\widehat{\mathfrak{O}}_1\cdot \widehat{\mathfrak{O}}_2 \right)\rangle \langle\Psi\rangle =\,\langle \left(\widehat{\mathfrak{O}}_1\cdot \widehat{\mathfrak{O}}_2 \right) |\Psi\rangle =\,\langle \widehat{\mathfrak{O}}_1 \rangle \left(\langle \widehat{\mathfrak{O}}_2 \Psi\rangle\rangle \right) =\,\left(\langle \widehat{\mathfrak{O}}_1\rangle \cdot \langle \widehat{\mathfrak{O}}_2\rangle\right)(\langle \Psi\rangle),
\end{align*}
$\forall \,\,|\Psi\rangle\in\,\mathcal{H}_{Fun} $. From this relation follows the proposition.
\end{proof}
Direct consequences of the above proposition are the following
\begin{corolario}
The average operation commutes with the commutator and with anti-commutation operation,
\begin{align}
\begin{cases}
& \langle \left[\widehat{\mathfrak{O}}_1,\widehat{\mathfrak{O}}_2\right]\rangle_t =\,[\langle\widehat{\mathfrak{O}}_1\rangle_t,\langle \widehat{\mathfrak{O}}_2\rangle_t],\\
& \nonumber  \langle \{\widehat{\mathfrak{O}}_1,\widehat{\mathfrak{O}}_2\}\rangle_t =\,\{\langle\widehat{\mathfrak{O}}_1\rangle_t,\langle \widehat{\mathfrak{O}}_2\rangle_t\} .
\end{cases}
\end{align}
\end{corolario}
\begin{corolario}\label{preservation of Jacobi property for averages}
The average operation preserves the Jacobi identity,
\begin{align*}
[[\langle\widehat{\mathfrak{O}}_1\rangle,\langle\widehat{\mathfrak{O}}_2\rangle],\langle \widehat{\mathfrak{O}}_3\rangle]+\,[[\langle\widehat{\mathfrak{O}}_3\rangle,\langle\widehat{\mathfrak{O}}_1\rangle],\langle \widehat{\mathfrak{O}}_2\rangle]+\,[[\langle\widehat{\mathfrak{O}}_2\rangle,\langle\widehat{\mathfrak{O}}_3\rangle],\langle \widehat{\mathfrak{O}}_1\rangle]=\,0.
\end{align*}
\end{corolario}

The averaging operation can be extended to consider the average of homomorphisms between two different spaces $\widehat{\mathfrak{O}}:\mathcal{H}_{{Fun}_A}\to \mathcal{H}_{{Fun}_B}$ by the commutativity of the diagram
\begin{align}
\xymatrix{\mathcal{H}_{{Fun}_A} \ar[d]_{\langle,\rangle_A} \ar[r]^{\widehat{\mathfrak{O}}} &
\mathcal{H}_{{Fun}_B}\ar[d]^{\langle,\rangle_B}\\
\mathcal{H}_A \ar[r]^{\langle\widehat{\mathfrak{O}}\,\rangle} & \mathcal{H}_B}
\label{Diagrama commutativo sobre homomorphismos}
\end{align}
The average operation maps spaces of the form $\mathcal{H}_{Fun}$ to spaces of the form $\mathcal{H}$ and maps the corresponding homomorphisms and in particular, the corresponding endomorphisms. Since both classes of spaces are attached with a pre-Hilbert structure, the average operator is a functor between the categories of pre-Hilbert spaces of the form  $\mathfrak{C}(\mathcal{H}_{Fun} )$ where the objects are the pre-Hilbert spaces $\mathcal{H}_{Fun}$ and the morphisms are the homomorphism between them their endomorphisms to the category of Hilbert spaces $\mathfrak{C}(\mathcal{H} )$, where the objects are Hilbert spaces of the form $\mathcal{H} TM_4$ and the morphisms the corresponding endomorphisms The hypothesis of the emergent character of quantum mechanics can be stated as saying that the functor $\mathfrak{F}:\mathfrak{C}(\mathcal{H}_{Fun} ) \to \mathfrak{C}(\mathcal{H}) $ is surjective.

\subsubsection{Derivation of the quantum Heisenberg equation}
Assuming the surjectivity of the functor $\mathfrak{F}:\mathfrak{C}(\mathcal{H}_{Fun} ) \to \mathfrak{C}(\mathcal{H} )$, it is possible to derive the Heisenberg equations of quantum mechanics from the Koopman-Von Neumann formulation of the corresponding Hamilton-Randers systems. The form of the dynamics for the Koopman-von Neumann formulation can be written as
\begin{align*}
\imath \,\hbar \,d\widehat{\mathfrak{O}} |x_k,y_k\rangle =\, [\widehat{H}_t,\widehat{\mathfrak{O}}]\,|x_k,y_k\rangle dt,\quad k=1,...,N
\end{align*}
where $\widehat{H}_t$ is the quantum Hamiltonian \eqref{QuantumRandersHamiltonian}. Linearity of the operations involved implies the generalization to
\begin{align*}
\imath \,\hbar \,d\widehat{\mathfrak{O}} |\Psi \rangle =\, [\widehat{H}_t,\widehat{\mathfrak{O}}]\,|\Psi\rangle dt,
\end{align*}
for every predecessor state $\Psi\in \mathcal{H}_{Fun}$.
By taking averages in both sides it is obtained the relation
\begin{align}
\imath \,\hbar \,\langle d\widehat{\mathfrak{O}}|\Psi \rangle_t =\, \langle \int^{t_f}_{t_i} [\widehat{H}_t,\widehat{\mathfrak{O}}]\Psi\,dt\rangle_t ,
\label{average Heisenberg fundamental equation}
\end{align}
where the integral compress many fundamental cycles. The left hand side can be re-casted by means of a differential operator $\tilde{d}$, such that
\begin{align*}
 \tilde{d}\,\langle \widehat{\mathfrak{O}}|\Psi \rangle_t=\,\langle {d}\widehat{\mathfrak{O}} |\Psi\rangle_t +\,\langle \widehat{\mathfrak{O}}d (\Psi)\rangle_t .
\end{align*}
${d}$ denotes also the differential operation $End (\mathcal{H}_{Fun})$ due to the $U_t$ evolution. In the Heisenberg picture of the dynamics, the predecessor state $|\Psi\rangle$ do not change with the $t$-time, we have $d|\Psi\rangle =0 $, from where we have that $\tilde{d}\,\langle \widehat{\mathfrak{O}}|\Psi \rangle_t=\,\langle d\widehat{\mathfrak{O}} |\Psi\rangle_t$.

For the right hand side of relation \eqref{average Heisenberg fundamental equation}
\begin{align*}
\langle \int^{t_f}_{t_i} [\widehat{H}_t,\widehat{\mathfrak{O}}]|\Psi\,dt\rangle_t   = \,\int^{t_f}_{t_i}\langle  [\widehat{H}_t,\widehat{\mathfrak{O}}]|\Psi\rangle_t \,dt = \,\int^{t_f}_{t_i}\langle  [\widehat{H}_t,\widehat{\mathfrak{O}}]\rangle_t \,\langle \Psi\rangle_t dt
\end{align*}
  Thus we have
\begin{align*}
\imath\,\hbar\,\tilde{d}\,\left(\langle \widehat{\mathfrak{O}} \rangle_t \langle \Psi\rangle_t  \right)=\,\int^{t_f}_{t_i}\langle  [\widehat{H}_t,\widehat{\mathfrak{O}}]\rangle_t \, \langle \Psi\rangle_t\,dt.
\end{align*}
Due to the relation \eqref{average of an operator}, both sides of this expression are derivations. Furthermore, also because the average preserves the commutator of operators, we have that
\begin{align*}
\langle  [\widehat{H}_t,\widehat{\mathfrak{O}}]\rangle_t =\,[\langle\widehat{H}_t\,\rangle,\langle\widehat{\mathfrak{O}}\rangle_t].
\end{align*}
Therefore, we have
\begin{align*}
\imath\,\hbar\,\tilde{d}\,\langle \widehat{\mathfrak{O}} \rangle_t \,\langle \Psi\rangle_t =\, \int^{t_f}_{t_i} [\langle\widehat{H}_t\,\rangle_t,\langle\widehat{\mathfrak{O}}\rangle_t] \langle \Psi\rangle_t \,dt.
\end{align*}
In the case of minimum elapsed time to make the time average $t_f-t_i = 2T$ and the integral is taken along one single fundamental cycle. Thus the average $[\langle\widehat{H}_t\,\rangle_t,\langle\widehat{\mathfrak{O}}\rangle_t] = \langle  [\widehat{H}_t,\widehat{\mathfrak{O}}]\rangle_t $ is constant along that cycle. We take then that $d\tau =2T$ and hence we obtain the relations
\begin{align*}
\imath\,\hbar\,\tilde{d}\,\left(\langle \widehat{\mathfrak{O}} \rangle_t \langle \Psi\rangle_t \right)=\,[\langle\widehat{H}_t\,\rangle_t,\langle\widehat{\mathfrak{O}}\rangle_t] \langle \Psi\rangle_t \,d\tau.
\end{align*}
Since the average operation $\langle,\rangle:\mathcal{H}_{Fun}\to \mathcal{H}$ is surjective, this leads to the relations
\begin{align}
\imath\,\hbar\,\tilde{d}\,\langle \widehat{\mathfrak{O}} \rangle_t =\,  [\langle\widehat{H}_t\,\rangle_t,\langle\widehat{\mathfrak{O}}\rangle_t]\,d\tau.
\label{Heisenberg equation in differential form}
\end{align}
The relation \eqref{Heisenberg equation in differential form} is a relation equivalent to the Heisenberg equation of quantum mechanicsof emergent quantum mechanics. This derivation is concluded with assumption on the exhaustiveness of the functorial relation described above. Furthermore, the Hamiltonian of the quantum system appears to be the average of the fundamental Hamiltonian,
\begin{align}
\widehat{H}_q =\,\langle \widehat{H}_t \rangle_t,
\end{align}
where $\widehat{H}_t$ is the Hamiltonian \eqref{QuantumRandersHamiltonian} that corresponds to the Koopman-von Neumann formulation of the Hamilton-Randers dynamical system.
\subsection{On the conserved quantities} Fundamental for the analysis and understanding of a dynamical system is to identify the conserved quantities of a model. In quantum mechanics, the condition for a conserved quantity is the commutation of the associated operator with the Hamiltonian operator. Trivially $\widehat{H}_q$ is conserved by the Heisenberg quantum dynamics. In the case of the $U_t$ dynamics, the condition of conserved quantity in the $U_t$ dynamics is the commutation with $\widehat{H}_t$. Let us conserved quantity $\widehat{\mathfrak{C}}$ for the $U_t$ dynamics. Then one has $[\widehat{H}_t,\widehat{\mathfrak{C}}\,]=\,0$. It follows that $\widehat{\mathfrak{C}}$ is also conserved quantity for the quantum dynamics,
\begin{align*}
[\langle\widehat{H}_t\rangle_t,\langle\widehat{\mathfrak{C}}\,\rangle_t\,]=\,\langle [\widehat{H}_t,\widehat{\mathfrak{C}}\,] \rangle_t =\, 0.
\end{align*}
From the above reasoning the following result follows,
\begin{proposicion}
For each conserved quantity represented by a linear operator $\widehat{\mathfrak{C}}:\mathcal{H}_{Fun}\to \mathcal{H}_{Fun}$ of the fundamental $U_t$ dynamics there is associated a conserved quantity $\langle \widehat{\mathfrak{C}}\rangle:\mathcal{H}\to \mathcal{H}$ for the Heisenberg quantum mechanics, which is the average $\langle \widehat{\mathfrak{C}}\rangle$.
\label{Conservation quantities}
\end{proposicion}
Thus every conserved symmetry of the fundamental $U_t$ determines a conserved quantity of the quantum mechanical system. The logical converse is not true. This is due to the loss of detailed information originated by the averaging functor. Thus a generic conserved quantity of the quantum dynamics is not conserved by the $U_t$ dynamics. On the other hand, the hypothesis that Hamilton-Randers systems exhaust the category of quantum systems by means of the averaging functor when the last are seen as emergent objects, implies that given an operator $\widehat{O} \,:\mathcal{H} \to \mathcal{H}$, there is an operator $\widehat{\mathfrak{O}} \,:\mathcal{H}_{Fun} \to \mathcal{H}_{Fun}$ such that $ \widehat{O} =\, \langle\widehat{\mathfrak{O}}\rangle $. If this emergent hypothesis is adopted, then for an operator $\widehat{O}$ that commutes with the quantum Hamiltonian $\widehat{H}_q$, there is an operator $\widehat{\mathfrak{O}}$ such that $\widehat{O}=\,\langle \widehat{\mathfrak{O}}\rangle$. Therefore, we have
\begin{align*}
0=\,\left[\widehat{H}_q,\,\widehat{O}\right]=\,\left[\langle\widehat{H}_t\rangle,\,\langle\widehat{\mathfrak{O}}\rangle_t \right] =\,\langle\left[\widehat{H}_t,\,\widehat{\mathfrak{O}} \right]\rangle_t .
\end{align*}
Thus $\widehat{\mathfrak{O}} $, which is in inverse image of the average of $\widehat{O}$,  is weakly preserved in the sense that $\langle\left[\widehat{H}_t,\,\widehat{\mathfrak{O}} \right]\rangle_t $ =0, but not necessarily preserved  in the sense that $\left[\widehat{H}_t,\,\widehat{\mathfrak{O}} \right]=\,0$. That the preservation of quantities does not translate litereally from the quantum domain to the sub-quantum domain is natural, since the sub-quantum dynamics is of more complex than the averaged dynamics corresponding to the quantum dynamics.

\subsection{On the recurrence condition implicit in the averaging operation}
 We have associated a probabilistic interpretation to {\it the density of presence} of sub-quantum degrees of freedom at a given point of the spacetime. For the models considered the number of degrees of freedom is large but finite and of order $N$. How can we attach a probability of presence in this context? One way in which the probability interpretation can be implemented more effectively is when the degrees of freedom pass through the allowed domain of spacetime  many times, a notion consistent with the diffeomorphisms $\varphi_k: M^k_4\to M_4$. In the case of continuous evolution, this implies either the existence in abundance of closed timelike curves, as the most common world lines for sub-quantum degrees of freedom, or that the sub-quantum degrees of freedom need to surpass the assumed universal speed limit, the local speed of light. Since each $M^k_4$ is diffeomorphic to the spacetime model $M_4$, this will disregard many relevant four dimensional spacetime models, that are three of timelike curves. In a similar vein, within our framework, it is very unnatural to violate the causality conditions at the microscopic level.

However, if the $U_t$ dynamics is discrete, there the possibility to invert the speeds directions  of the sub-quantum degrees of freedom towards the past without surpassing the limiting speed and limiting acceleration in a form of zig-zag time symmetric motion. The speed  inversion is formally described by the map
\begin{align*}
(x,y,p_x,p_y)\mapsto (x,-y,-p_x,p_y).
\end{align*}
 Although this inversion operation is applied to the level of the sub-quantum atoms, it induces the {\it time inversion operation} \eqref{timeinversionoperation} acting on the quantum Hilbert space $\mathcal{H}$.

One direct implication of this idea is the determination of the value of the maximal acceleration. If the minimal distance that can change a sub-quantum degree of freedom at each step of the dynamics is of the order of  the Planck length $L_P$, then the acceleration due to a reverse interaction must limited by
\begin{align}
A_{max}=\,\frac{2\,c^2}{L_p}.
\label{maximalacceleration3}
\end{align}
The maximal acceleration given by this expression is four times larger than the given by \eqref{maximalacceleration}. The correct value should depend on the final model of sub-quantum dynamics dynamics, but this argument suggests that it must be on this scale of acceleration.

\newpage

 \section{\LARGE{Concentration of measure and natural spontaneous collapse}}\label{chapter on concentration of measure}
 \bigskip
 \bigskip
 In the previous chapter we have discussed how quantum dynamics arises from the Koopman-von Neumann formulation of the fundamental dynamics. In this chapter we describe a theory of observables and address the question of how observables appear to be well-defined when an observer measures them. Our description lies on the classical formulation of the theory described in chapter \ref{chapter on classical dynamics Hamilton Randers}. Observables will be typically described by functions $f:T^*TM\to \mathbb{K}$, although we will consider more often the field $\mathbb{R}$ be the field of the real numbers $\mathbb{R}$. The mechanism that we will describe is based upon the theory of {\it concentration of measure}. We will show that under certain assumption on the regularity of the functions, concentration of measure implies that the values of observables are well defined in the concentration domain, which is the domain directly accesible to a classical observer, even if one uses quantum devices measurements.
 \subsection{Regularity during the concentrating regime}
 It has been assumed that the fundamental cycles of the $U_t$ dynamics of Hamilton-Randes spaces are composed by a sequence of an {\it ergodic regimen} followed by a {\it concentration regime}, followed by an {\it expanding regime} and then followed of another analogous cycle and so on. We have seen in the previous chapter how the Koopman-von Neumann description of the Hamilton-Randers dynamics leads to a derivation of the existence of the quantum Hilbert space effectively describing the physical system. In such derivation, the ergodic properties of the $U_t$-flow played a fundamental role.

 In this chapter the general mathematical framework for the concentration regime of the fundamental dynamical cycles is discussed. The reduction of the allowable phase space that happens during such domain of the fundamental regime is mathematically described by means of $1$-Lipschitz functions, a regularity condition that extends to the pertinent operators.

The Lipschitz condition of regularity is defined as follows,
  \begin{definicion}
 The function $f:{\bf T}_1\to {\bf T}_2$ between two metric spaces $({\bf T}_1,d_1)$ and $({\bf T}_2,d_2)$ with $\lambda >0$ is $\lambda$-Lipschitz if
 \begin{align*}
 d_2(f(x_1),f(x_2))\leq \lambda\,d_1(x_1,x_2),\quad \forall \,x_1,x_2\in\,{\bf T}_1.
 \end{align*}
 \label{Lipschitz operator}
 \end{definicion}

Analogously, the operator  $O:\mathcal{F}({\bf T}_1)\to \mathcal{F}({\bf T}_2)$ is a {\it strongly} $\lambda$-{\it Lipschitz operator} if
\begin{align*}
d_2(O(f_1),O(f_2))\leq \lambda\,d_1(f_1,f_2),\quad \forall \,f_1,f_2\in\,\mathcal{F}({\bf T}_1).
\end{align*}

There also a notion of {\it weak Lipschitz condition} for endomorphisms of the metric space:  $O:{\bf T}_1\to {\bf T}_1$ is a weak Lipschitz operator if  for each Lipschitz funcion $f:{\bf T}_1\to {\bf T}_1$, $O\circ f$ is Lipschitz.

 There are several reasons for why the Lipschitz category of functions and operators is of relevance for the theory of dynamical systems. One of them is that the $\lambda$-Lipschitz condition implies the existence and uniqueness of solutions for ordinary differential equations in compact sets \cite{Elsgoltz}. This is a consistency condition. In the case of $\lambda=1$, the Lipschitz condition determines the natural class of maps that preserve invariant global properties of metric spaces based upon the notion of isometry (see for instance \cite{Gromov}). Moreover, composition of $1$-Lipschitz functions is $1$-Lipschitz, as well as the convex combinations of them, properties that converts such a category of particular relevance for the formulation of consistent dynamical models.

The $1$-Lipschitz condition accommodates the requirement of Hmailton-Randers theory that the difference on variations of coordinates are contained on uniform {\it cones}, since such condition is related with the mathematical assumptions in Hamilton-Randers theory on the finiteness of the speed and uniformly bounded proper acceleration of the sub-quantum atoms and sub-quantum molecules. Although these bounds are implemented at the level of the infinitesimal neighborhoods, they are consistent with $1$-Lipschitz condition, which is a much stronger condition. Indeed, the  $1$-Lipschitz condition is not required for the whole $U_t$-flow, but only in the contractive regime of each fundamental cycle, while the less strong condition of upper bounds for speed and proper acceleration remain valid during the whole evolution.

 Since each of the intervals $\{[(2n-1)T,(2n+1)T],\,n\in\mathbb{Z}\}$ is compact, for locally Lipschitz functions which are Lipschitz in the concentration regime there is a finite constant $K_n$ where the $K_n$-Lipschitz condition holds in the concentration regime of the $n$-enessim cycle labeled by $n\in\,\mathbb{Z}$. If we assume that $K=\sup\{K_n,\,n\in\,\mathbb{Z}\}<+\infty ,$  then by a suitable normalization one can consider $f\to f/K:=\,\tilde{f}$, the function $\tilde{f}$ is  $1$-Lipschitz. Such normalization can be understood as a change on the scale by which the observable described by the function $f$ is measured. Otherwise, one can assume the $1$-Lipschitz condition for all the concentrations regime. From now on, we adopt this assumption.

 The relevant mathematical theory for the description of physical properties of Hamilton-Randers systems in the concentration regime is {\it concentration of measure} \cite{Gromov,MilmanSchechtman2001,Talagrand}. The fundamental idea behind concentration of measure is that $1$-Lipschitz functions are {\it almost constant almost everywhere} for a rather general type of metric spaces of high dimension.
When concentration of measure is applied to Hamilton-Randers dynamical systems, it provides a mechanism to explain the fact that when any observable is measured for any physical quantum state, the measurement outcome has a well-defined value.

The key assumptions for the mechanism are the following:
  \begin{enumerate}
\item  During the concentration dynamical regime of each fundamental cycle, the operator $U_t$  is $1$-Lipschitz in the sense that when $U_t$ is applied to a $1$-Lipschitz function on certain domain of $T^*TM$ the result is a $1$-Lipschitz function. By application of the concentration of measure, such functions must be almost constant almost
    everywhere in the domain of concentration.

\item Only during the concentration regime a measurement of an observable is possible. That is, any physical measurement takes place during a concentration regime of a fundamental cycle.
     \end{enumerate}
In the first statement it is assumed that physical observables are associated with $1$-Lipschitz functions from $T^*TM$ to the real numbers $\mathbb{R}$ of the sub-quantum degrees coordinate and momentum variables, while the operator $U_t$ in  the concentration regime is described by a strong $1$-Lipschitz operator.

As a result of the concentration of measure and the induced mechanism in Hamilton-Randers systems, when a measurement described in classical terms as a function of the coordinate position of pointers of any property of an individual quantum system is performed, the outcome is always a well-defined value, without significant dispersion.

 \subsection{Concentration of measure}  The concentration of measure is a general property of regular enough functions defined in high dimensional topological spaces  ${\bf T}$ endowed with a metric function $d:{\bf T}\times {\bf T}\to\mathbb{R}$ and a Borel measure $\mu_P$ of finite measure, $\mu_P({\bf T})<\,+\infty$ or a $\sigma$-finite measure spaces (countable union of finite measure spaces)\footnote{There are other technical conditions on the topological and metric structure of {\bf T}, but we shall not entry on these details. For details, see  \cite{Gromov, MilmanSchechtman2001} and also the introduction from \cite{Talagrand}.}. For our applications,
 we also require that the topological space {\bf T} has associated a local dimension. This will be the case, since we shall have that ${\bf T}\cong T^*TM$,
 which are locally homeomorphic to $\mathbb{R}^{16\,N}$, ${\bf M}_4$ is the four dimensional spacetime and $N$ is the number of sub-quantum molecules
 defining the Hamilton-Randers system.

 The phenomenon of concentration of measure for the category of topological spaces with a well defined dimension can be stated in the following terms \cite{Talagrand}:
 \bigskip

{\bf Concentration of Measure Principle}. {\it In a  measure-metric space of large dimension, every real $1$-Lipschitz function of
many variables is almost constant almost everywhere.}
\bigskip

In the formalization of the concept of concentration of measure one makes use of the metric and measure structures of the space {\bf T} to provide a  precise meaning for the notions of {\it almost constant} and {\it almost everywhere}. The notions of measure structure $\mu_P$ and metric structure $d:{\bf T}\times {\bf T}\to \mathbb{R}$ are independent from each other. Indeed, the standard framework where  concentration is formulated is in the  category of mm-Gromov spaces \cite{Berger2002, Gromov}. The spaces that we shall consider are called $mm$-spaces and will be denoted by triplets of the form $({\bf T},\mu_P,d)$.

In a mm-space $({\bf T},\mu_P, d)$, $M_f$ is the {\it median} or {\it Levy's mean} of $f$, which is defined as the value attained by $f:{\bf T}\to \mathbb{R}$ such that
\begin{align*}
\mu_P(f>M_f)=1/2\,\textrm{ and }\, \mu_P(f<M_f)=1/2.
\end{align*}
 The {\it concentration function} $\alpha(\mu_P):\mathbb{R}\to \mathbb{R},\quad \rho\mapsto \alpha(\mu_P,\rho)$ is defined by the condition that $\alpha(\mu_P,\rho)$
 is the smallest real number such that\footnote{The concentration function $\alpha(\mu_P,\rho)$ is introduced here in a slightly different way than in \cite{Gromov,Talagrand}.}
\begin{align}
\mu_P(|f-M_f|>\rho)\leq 2\,\alpha(\mu_P,\rho),
\label{concentationformula}
\end{align}
for any $1$-Lipschitz function $f:{\bf T}\to \mathbb{R}$.
Therefore,  the {\it probability} that the function $f$  differs from the median $M_f$ in the sense of the
given measure $\mu_P$ by more than the given value $\rho\in \,\mathbb{R}$
is bounded by the concentration function $\alpha(\mu_P,\rho)$.

$\alpha(\mu_P,\rho)$ does not depend on the function $f$.

\begin{ejemplo}
The example of concentration that we consider here refers to  $1$-Lipschitz real functions on $\mathbb{R}^N$
 (compare with \cite{Talagrand}, pg. 8). In this case,  the concentration inequality is of the form
\begin{align}
\mu_P\left(\left|f-M_f\right|\,\frac{1}{\sigma_f}\,>\frac{\rho}{\rho_P}\right)\leq \, \frac{1}{2} \exp\left(-\frac{\rho^2}{2\rho^2_P}\right),
\label{concentration2}
\end{align}
where we have adapted the example from \cite{Talagrand} to a Gaussian measure $\mu_P$ with median $M_f$. In the application of this concentration to Hamilton-Randers models, $\rho_P$ is a measure of the minimal standard contribution to the distance $\rho$ per unit of degree of freedom of the Hamilton-Randers system. $\frac{\rho}{\rho_P}$ is independent of the function $f$, while $\sigma_f$ is associated with the most precise physical resolution of any measurement of the quantum  observable associated with the $1$-Lipschitz function $f:\mathbb{R}^N\to \mathbb{R}$.
\label{Example concentration 2}
\end{ejemplo}
For $1$-Lipschitz functions on a measure metric space  ${\bf T}$ of dimension $N$ there are analogous {\it exponential bounds}
as for examples the bounds in {\it Example} \ref{Example concentration 2}.
In general, the phenomenon of concentration  is a consequence of the Lipschitz regularity condition of the function $f:{\bf T}\to \mathbb{R}$ and the higher dimensionality of the space ${\bf T}$.
For $dim({\bf T})$ large, the concentration of measure implies that the values of $1$-Lipschitz functions are very picked around a certain constant value.

We can provide an heuristic interpretation of the concentration of $1$-Lipschitz functions.
Let $f:{\bf T}\to \mathbb{R}$ be a $1$-Lipschitz function on a normed topological space $({\bf T},\|,\|_{\bf T})$ locally homeomorphic to  $\mathbb{R}^N$. Then the $1$-Lipschitz condition is a form of {\it equipartition} of the  variation of $f$ originated by an arbitrary variation on the point on the topological space ${\bf T}$ where the function $f$ is evaluated. When the dimension of the space ${\bf T}$ is very large compared with $1$, the significance of the $1$-Lipschitz condition is that $f$ cannot admit large {\it standard variations} caused by the corresponding standard variations on the evaluation point of {\bf T}.
Otherwise, due to the assumption on the equipartition of the variation, a violation of the $1$-Lipschitz condition can happen, since the large dimension provides long contributions to the variation of $f$.

Note that for application in Hamilton-Randers models, in order to speak of large and small variations, one needs to introduce reference scales. This is our motivation to introduce the scales $\sigma_f$ and also the distance variation by degree of freedom $\rho_P$ in {\it Example} \ref{Example concentration 2}.

\subsection{Theory of natural spontaneous collapse as concentration of measure}
In the following paragraphs we describe a new theory of spontaneous collapse theory based upon the framework of Hamilton-Randers system and that makes explicit use of few more fundamental assumptions, namely a form of the equipartition principle, the identical nature of sub-quantum degrees of freedom and concentration of measure.

We apply the theory of concentration of measure to Hamilton-Randers systems by first modelling the domain  ${\bf D}_0\subset\, T^*TM$ locally
as homeomorphic to $\mathbb{R}^{16N}$ with $N\gg 1$. The measure $\mu_P$ on $T^*TM$ will be the pull-back of the standard product measure in $\mathbb{R}^{16 N}$.
We assume that there is concentration of measure determined by the  inequality \eqref{concentration2}.
Let $f:\mathbb{ R}^{16\,N}\to \mathbb{ R}$ be a real valued function locally $1$-Lipschitz describing a property of the Hamilton-Randers model and let us consider its flow under the $U_t$ dynamics in the domain ${\bf D}_0$ where $U_t$ is locally $1$-Lipschitz. In such dynamical regime, the induced map associated with the geometric flow \eqref{definiciongeometricevolution}
\begin{align*}
U_t:{\bf D}_0\times I_n\to T^*TM,
 \end{align*}
 is a $1$-Lipschitz evolution operator, where $ I=\cup_{n\in\mathbb{Z}}I_n,\,I_n\subset [2n-1,2n+1]\in\,\mathbb{R}$ determines the concentration regime of the $U_t$ dynamics. Then one considers the function $f_t:\mathbb{ R}^{16\,N}\times I_n\to \mathbb{R}$. $f_t$ is the evolution of $f$ by the $U_t$ dynamics during the regime defined by the $t$-time $t\in\,I_n\subset\, \mathbb{R}$.
Note that the co-domain is the full co-tangent space $T^*TM$ and
  not the concentration domain ${\bf D}_0$, since the system leaves ${\bf D}_0$ at the end of each interval $I_n$.

  In these conditions, the function ${f}_t$ must be almost constant almost everywhere along the image produced by the $U_t$ evolution on ${\bf D}_0\times \,I_n$. For measurements performed with a macroscopic or quantum device, one expects that the relation
 \begin{align}
|f_t-M_f (n)|\,\frac{1}{\sigma_f}\sim N, \quad 1\ll N,
\label{comparingcomplexities1}
 \end{align}
holds good, where $\sigma_f$ is associated with the standard change in the function $f_t$ by a small change in
 the location of one individual sub-quantum degree of freedom scale and $N$ is the number of sub-quantum molecules and where we assume that $\sigma_f$ does not depend on $n\in \,\mathbb{Z}$. The relation \eqref{comparingcomplexities1} expresses the equipartition to the contribution of observable functions due to each sub-quantum degree of freedom. It is consistent with the concentration inequality relation \eqref{concentration2} if also it holds that
\begin{align}
\frac{\rho}{\rho_P}\sim\, N,
\label{scales}
\end{align}
where $\rho$ is a distance function on the space $T^*T{\bf M}_4$. It is natural that the quotient of these two quantities $\frac{\rho}{\rho_P}$ is of order $N$, indicating an equipartition on the contribution to the distance $\rho$ among the number of degrees of freedom.

 If the quantum system corresponds to a Hamilton-Randers system of $N$ sub-quantum molecules, then the configuration space $T^*TM$ is locally homeomorphic to $\mathbb{R}^{16 N}.$  $\sigma_f$ is the minimal theoretical resolution for an observable associated with the function $f:T^*TM\to \mathbb{R}$. $\sigma_f$ is  associated with the variation of $f$ induced by a minimal variation of the configuration state of one sub-quantum particle. Therefore, for a quantum state associated with a Hamilton-Randers system described by $N$ sub-quantum atoms, the variation on $f_t$ that we should consider is of order $16 N\sigma_f$, which is the typical measurable minimal value for a quantum system under an additive rule. By the relation \eqref{comparingcomplexities1}, if the concentration of measure relation \eqref{concentration2} is applied to a $1$-Lipschitz function ${f}$ respect to all its arguments  in the $1$-Lipschitz dominated regime of $U_t$, then we have the relation
  \begin{align}
  \mu_P\left(|f_t-M_{f}(n)|>\,16 \,\sigma_f\,N\right)\leq \, \frac{1}{2} \exp\left(-\,\frac{N^2}{2} \right).
  \label{concentration3}
  \end{align}
  Note that although the quotient $|f_t-M_{f}(n)|/\sigma_f$ can be large, the value $8N$ can also be large, for $N$ large compared to $1$.

   The relation
  \begin{align*}
  |f_t-M_{f}(n)|\sim \,16 \,N\sigma_f
  \end{align*}
   provides the limit where the dynamics do not present  concentration of measure, in the sense that the concentration of measure will be non-effective. This threshold indicates the violation of the $1$-Lipschitz property of the function $f$ and with the additivity of $f$ respect to the values of  individual sub-quantum degrees of freedom.

The direct interpretation of the concentration of measure phenomenon described above is the following.
For functions associated with measurements of the properties of quantum systems and since $N\gg 1$ experiences concentration of measure around the median $M_{f}$ with probability almost equal to $1$. Thus if a measurement of an observable associated with $f$ is performed in the $\tau$-time associated with the cycle labeled by $n\in \,\mathbb{R}$, then the value $M_{f}(n)$ is to be found with probability very close to $1$ and in practice, equal to $1$.
\begin{comentario}
We have to mention the following points, consequences of the theory.
\begin{itemize}
\item It is relevant to mention that there is the possibility of non-Lipschitz evolution from cycle to cycle, since the evolution is only $1$-Lipschitz on each interval $I_n$. This explains the {\it quantum jumps}.

\item It is remarkable the intrinsic link in the theory above described between concentration and the emergence of intrinsic $\tau$-time.

\item All the possible properties associated with the above type of functions attain a well-defined value in the concentration regime. These functions, under standard operations of sum, convex products by scalars and multiplication is close, but do not form an algebra since the multiplication by scalars is not close.

\end{itemize}
These points of the theory will be developed below.

\end{comentario}

\subsection{Notions of classical and quantum interactions}
In Hamilton-Randers dynamics, the fundamental degrees of freedom interact during the $U_t$ dynamics. That there must be interactions between them is clearly manifest because the assumed structure of the cycles. Thus even in the case of an effective free quantum evolution, there are interactions among the sub-quantum degrees of freedom. Besides this notion of {\it interaction at the fundamental level}, there is a natural notion of non-interacting quantum system in terms of the dynamics of the sub-quantum degrees of freedom: we say that a quantum system is not interacting with the environment if the interaction between the sub-quantum degrees of freedom of the system and the environment can be disregarded without changing the quantum and dynamical properties of the quantum system. In particular, the measure $\mu_P$ is invariant under the $U_t$ evolution.

In most of the cases, an  {\it interacting system} is such that asymptotically for large number of cycles towards the past ($n\to -\infty$) and towards the future ($n\to +\infty$), the system is described by  non-interacting systems. This is the typical situation described in {\it scattering theory}. Therefore,  asymptotically large numbers of cycles towards the past, the system is of the form $A\sqcup B$ and similarly, for a large number of cycles towards the future, the form $A,\, B,\, A',\,B'\, C'...$.
 Thus, a scattering process is of the form
\begin{align*}
A\sqcup B\rightarrow A'\sqcup B'\sqcup C'\sqcup ...
\end{align*}
In this context, there is the following classification. Depending on the domain of the fundamental cycles where an interaction is dominant, a distinction between {classical interaction} and {quantum interactions} is drawn as follows,
\begin{definicion}
A {\it classical interaction on a system} is an interaction which is dominant only on the metastable domain ${\bf D}_0$ containing the metastable points $\{t=\,(2n+1)T,\,n\in \,\mathbb{Z}\}$ of each fundamental cycle; A {\it quantum interaction} is an interaction potentially dominant at least in the interior of the fundamental cycles $\left\{\cup_{n\in \mathbb{Z}}\, [(2n-1)T,(2n+1)T]\right\}\setminus {\bf D}_0$.
\label{quantumclassicalinteraction}
\end{definicion}
Since quantum interactions are not only restricted to the domain ${\bf D}_0$ of each fundamental cycle,
 after the formal projection $(t,\tau)\mapsto \tau$, under the  $U_\tau$ evolution  quantum interactions  are characterized by transition probability amplitudes associated with quantum transitions  between quantum states that have a non-local character from the point of view of spacetime. This is because the ergodic character of the $U_t$ dynamics in the interior of each cycle. Such ergodicity implies the  interactions between the sub-quantum degrees of freedom that correspond to different systems. Examples of quantum interactions are gauge theories, whose exact quantum mechanical description is given by holonomy variables. The Aharonov-Bohm effect is an example of non-local quantum effect which is usually interpreted in terms of non-local holonomy variables \cite{ChanTsou}.

In the above context of non-locality in the variables describing quantum gauge theories, is it there a reasonable notion of {\it local gauge theory}? In the case of evolution of the holonomy variables, however, the dynamics can be restricted by {\it locality conditions}. Let us consider the Hamiltonian formulation using loop variables of the gauge theory. Then the physical states of the field interaction must be such that the Hamiltonian density operator $\widehat{\mathcal{H}}(x)$ commute for spacetime points $x,x'\in \,M_4$ which are not causally connected.

In contrast, a classical interaction as defined above can be described directly in terms of local variables,  since a classical interaction is dominant only on the domain ${\bf D}_0$ where all the observable properties of the system defined by $1$-Lipschitz functions $f:T^*TM\to \mathbb{R}$ are well defined locally in spacetime ${\bf M}_4$. This constraint is consistent with the causal structure of the sub-quantum degrees of freedom.
 The relevant example of a classical interaction of this type is gravity, as we will discuss  in  {\it Chapter} \ref{Chapter on emergence of gravity}.

 \subsection{Emergence of the classical domain} We have identified the concentrating regime of the domain ${\bf D}_0$ with the {\it observable domain}, that is, the domain of the fundamental dynamics where the observable properties of the system, capable to be observed by the macroscopic observer $W\in\,\Gamma T{\bf M}_4$, are well defined. In the domain of such identification, we observe that:
\begin{enumerate}
\item The dynamics is dominated  by a strong $1$-Lipschitz operator $U_t$ and

\item Each observable associated with a $1$-Lipschitz function of the sub-quantum degrees of freedom have well defined values at each instant $\tau_n\in\,\mathbb{Z}$, associated with the $n$-cycle of the $U_t$ flow.
\end{enumerate}
The emergence of the classical domain is the identification of the classical framework of description of the physical systems, where the geometric arena is a $4$-dimensional manifold and the properties of the system are described by the algebra of real (or if convenient, with co-domain $\mathbb{K}$ different from the reals) functions, with the concentration domain. This notion is different than the notion of classical limit in quantum mechanics. Indeed, the notion of classical domain applies to both, small systems in the sense that a good description is only achieved by means of quantum mechanics, and to large systems, that are well-described by classical models.

 Under the assumptions that are discussed in the next chapter, the average position and speed coordinates associated with the system of $\{1,...,N\}$ of sub-quantum molecules of an arbitrary Hamilton-Randers model are $1$-Lipschitz functions. If we identify these functions with the observables {\it position} and {\it speed} of the system, these observables will have well defined values when they are measured.

 \subsection{Algebras associated with observables in Hamilton-Randers theory}
The discussion until now suggests the consideration of three diferent algebras of functions. First, it is the algebra of function
\begin{align*}
 \mathcal{F}_{\mathbb{C}}(T^*TM):=\{f:T^*TM\to \mathbb{C}\}.
  \end{align*}
  Alternative, one can consider real functions $\mathcal{F}_{\mathbb{R}}(T^*TM):=\{f:T^*TM\to \mathbb{R}\}$, although it can appear as too restrictive for algebraic considerations.

A sub-set, than not a sub-algebra of functions to consider are the $1$-Lipschitz functions or maps,
\begin{align*}
\mathcal{L}_{\mathbb{C}}(T^*TM):=\{f\in\, \mathcal{F}_{\mathbb{C}}(T^*TM)\,\,|\, \, f\,\textrm{ is locally $1$-Lipschitz on each ${\bf D}_0(n), n\in\mathbb{Z}$}\}.
\end{align*}
 It is clear that
\begin{align*}
\mathcal{L}(T^*TM)\subset \,\mathcal{F}_{\mathbb{C}}(T^*TM).
\end{align*}
The elements of $\mathcal{L}(T^*TM)$ have well-defined values on each of the  domains ${\bf D}_0(n)$, corresponding to the domains where the functions are Lipshcitz. But these functions are not the only one with well-defined values on each ${\bf D}_0(n)$. If we consider a functional $\mathcal{R}:\mathcal{L}_{\mathbb{C}}(T^*TM)\to  \mathcal{F}_{\mathbb{C}}(T^*TM)$, each $\mathcal{R}_{\mathbb{C}}(f)$ has well defined values on each local domain ${\bf D}_0(n)$.

Considering all the functions constructed from functionals in this way, we have
\begin{align*}
\mathcal{U}(T^*TM):=\cup_{\mathcal{L}_{\mathbb{C}}} \mathcal{R}(\mathcal{L}_{\mathbb{C}}(T^*TM)) .
\end{align*}

The second algebra of functions to be considered corresponds to the algebra of macroscopic functions, defined over the spacetime
\begin{align*}
 \mathcal{F}_{\mathbb{C}}(T^*T\mathcal{M}_4):=\{f:T^*T\mathcal{M}_4\to \mathbb{C}\}.
 \end{align*}
 This algebra corresponds to the attributes of macroscopic properties, that is, assigned pointwise according to the spacetime arena, to the physical system. They are potentially observable by a macroscopic observer.

The hypothesis on the emergent origin of quantum  mechanics is partially implemented as the condition
 \begin{align}
 \mathcal{U}(T^*TM)\, \subset\,  \mathcal{F}_{\mathbb{C}}(T^*T\mathcal{M}_4) .
 \label{hypothesis on the emergence of observables}
 \end{align}

 The third relevant algebra to consider is the algebra of quantum observables. Let us consider the operator
$\widehat{\mathfrak{O}}:\mathcal{H}_{Fun}\to \mathcal{H}_{Fun}$.
Its associated average operator
\begin{align*}
\langle\widehat{\mathfrak{O}}\rangle:\mathcal{H}_{Fun}/\sim_{\langle,\rangle} \to \mathcal{H}_{Fun}/\sim_{\langle,\rangle}
\end{align*}
 is defined as in {\it section} \ref{Heisenberg dynamics} by the commutativity of the diagram \eqref{average of an operator raw space}.
This construction defines a linear operator $\langle\widehat{\mathfrak{O}}\rangle \,:\mathcal{H} \to \mathcal{H}$,
which is associated with a quantum operator acting on the pre-Hilbert space $\mathcal{H}$ of quantum states associated with the system.

\subsection{Natural spontaneous quantum state reduction and Born rule}
The concentration of measure that takes place during the concentrating regime of the $U_t$ flow provides a natural mechanism for the reduction of the quantum state. However, such reduction processes not only happen when the quantum system is being {\it measured} by a measurement device. On the contrary, they are {\it spontaneous processes} that happen right after the ergodic regime in each fundamental cycle of the $U_t$-evolution. This is the reason of the name {\it natural reduction processes}, in contrast with induced reduction of the quantum state by an interaction with a quantum measurement system \cite{Diosi,GhirardiGrassiRimini, Penrose 1996, Penrose 2005}. The natural spontaneous reduction process does not necessarily change the quantum state of the system, neither it is necessarily associated with quantum measurement processes. In contrast, in a measurement process the measurement device can change the original quantum state, since there is an interaction between the system being measured and the apparatus measurement. This is described, for instance by von Neumann models of interaction between the quantum system with the measurement device.

The interaction responsible for the natural spontaneous collapse of the state is  classical, since it is dominant only in contractive domain containing the metastable equilibrium domain ${\bf D}_0$, when all the observable properties of the system are localized by the effect to the same interaction driving $U_t$. As a consequence of this phenomenon, the properties of the system appear as macroscopically well defined when such properties are measured by means of a macroscopic measurement device.

The existence of the natural spontaneous collapse implies a fully fledge form of the Born rule. In {\it section} \ref{Interpretation of wave function} we showed that the probability of finding in $t$-time the system described by the quantum state $|\psi(x)\rangle$ in an infinitesimal neighborhood of the point $x$ of the spacetime $\mathcal{M}_4$ is of the form $|\psi (x) |^2\,d^4x$. Since the measurement regime coincides with the classical regime, then the Born rule holds good:
\begin{proposicion}
For a system described quantically by the state $|\psi(x)\rangle\in \mathcal{H}$, the probability to find the particle in an infinitesimal neighborhood of $x\in \,{\bf M}_4$ is $|\psi (x) |^2\,d^4x $.
\end{proposicion}

Several remarks are in order. First, the standard formulation of the Born rule is presented in a three dimensional space (or the lower dimensional position configuration space) \cite{Dirac1958, Bohm}. In contrast, our formulation is developed in a four-dimensional spacetime arena. In order to recover the standard formulation it is enough to consider a fixed time measure, by considering the density $d^3\vec{x} := \,d^4 x\delta (x^0-x'^0)$. Thus the probability to find the system at the position $\vec{x}\in\,\mathbb{R}^3$ when the measurement is performed at the instance $\tau=\,x'^0$ is $|\psi (x) |^2\,d^3\vec{x}$. Second, the formulation above of the Born rules is given in terms of position eigenvectors. Taking into account the assumption that other quantum states can in many cases be re-written in terms of wave function basis by means of unitary transformations, the Born rule is induced on other quantum basis.

\subsection{Quantum fluctuation and prepared states}
A natural notion of {\it quantum fluctuation} arises in connection with the emergence character of observable quantities. This is because the ergodic property of the $U_t$ flow implies that the sequence in $\tau$-time of the values achieved by generic physical observables can be discontinuous. Given an observable $f:\mathcal{H}\to \mathcal{H}$, the values that we could assign after measuring on a time series is, according to Hamilton-Randers theory, given by an ordered time sequence $\{f(\tau_n),\, n\in \mathbb{Z}\}$. Then the uncertainty in the value of the observable can be measured by the differences
\begin{align*}
\delta f:=\,\min\{f(\tau_n)-f(\tau_{n+1}),\, n\in \mathbb{Z}\}.
\end{align*}
In general, for measurements on  generic quantum states, $\delta f \sim f$, indicating the existence of quantum fluctuation.

There is one exception to this rule, namely,  the case of {\it prepared states}. In a prepared state, a given physical observable is well defined during the whole $U_\tau$ evolution: repeated consecutive measurements of the observable will give the same value. How is this compatible with the emergent origin of quantum fluctuations? A characterization of prepared state is need in the language of Hamilton-Randers dynamics.
\begin{definicion}
A prepared is state is an element of $\mathcal{H}$ such that the mean $M_f (n)$ for a given observable is constant during the $U_t$ flow, that is, constant on $n\in\,\mathbb{Z}$.
 \end{definicion}
Therefore, $M_f(n)$ is constant and the same for all fundamental cycles. For example, to prepare an state of a given energy, the system is manipulated such that $M_E$ is fixed, where $E$ is an eigenstate of the matter Hamiltonian operator. If one selects another observable, for instance a spin component, the state will in general change. In the situation when the prepared state has constant means $M_E$ and $M_S$, the state is prepared with defined energy and spin. Thus the characteristics of prepared states are associated  to compatible conserved means $M_f (n)$ during the $U_t$ flow, that according to the discussion of the preceding chapter, coincide with conserved quantities of quantum mechanics.

\subsection{Universal limitations in the description of physical system}
If the initial conditions of the degrees of freedom of a Hamilton-Randers system are fixed, then the evolution of the median $M_f(n)$ is fixed as well. However, it is technically impossible for a macroscopic observer to determine the initial conditions of the sub-quantum degrees of freedom. The picture is similar to the description of an ideal gasin classical statistical mechanics.  Furthermore, although the Hamilton-Randers dynamics is deterministic, if we adopt Assumption A.9. of {\it Chapter} \ref{chapter on Assumptions and General Theory}, the system will necessarily be chaotic and hence unstable and sensitive to initial conditions. Indeed, for the dynamical systems that we are considering and if $\beta_k$ functions are not linear functions, enabling a complex dynamics.

 These circumstances impose a fundamental limitation in the knowledge and control that we can obtain on the dynamics of particular system. Such limitation is universal.

 The natural way to describe the long term dynamics for Hamilton-Randers systems is by using probabilistic methods, where the probability distribution function is an effective, emergent description of the fundamental $U_t$ flow of the system during the ergodic regime. In particular, the probability distribution function is associated with the density of world lines of sub-quantum molecules as discussed in  chapter \ref{chapter of the Hilbert space structure}.

\newpage

\section{\LARGE{Emergence of the gravitational interaction}}\label{Chapter on emergence of gravity}
\bigskip
\bigskip
   This {chapter} investigates further consequences of the application of concentration of measure phenomena to Hamilton-Randers models and other related issues. First, we provide a mechanism to bound from below the spectrum for a proposed {\it matter Hamiltonian}. This is necessary if we aim to relate the theory with standard quantum mechanics. The theory developed here provides a formal definition for the Hamiltonian for matter. Second, we use concentration of measure to show that gravity can be identified with a $1$-Lipschitz, classical interaction. By classical we mean that it is not quantum, neither it is possible to speak of superposition of classical spacetimes. This identification is based upon a formal analogy since the identified interaction is relativistic, it must have diffeomorphism invariance, inherited  from the underlying diffeomorphism invariance of Hamilton-Randers dynamical systems. A generalized form of the weak equivalence follows from first principles, completing the formal analogy with the properties of classical gravity. The regime where such interaction takes place corresponds to the regime where all the possible observables have a well defined value, that is, it corresponds to a classical dynamical regime. Hence gravity must be a classical interaction. Therefore, classical gravity appears as an essential ingredient in the mechanism to bound from below the Hamiltonian for matter and a dual ingredient of the natural spontaneous collapse mechanism. In addition, the $1$-Lipschitz interaction must be compatible with the existence of an universal maximal acceleration. This last condition implies a necessary deviation from current standard theories of gravitation as general relativity.

   \subsection{Hamiltonian decompoition in a $1$Lipschitz part and a non-Lipschitz parts}
   Let us consider the decomposition of the quantized Hamiltonian operator $\widehat{H}_t$ of a  Randers-Hamiltonian system \eqref{RandersHamiltonian} in a $1$-Lipschitz component $\widehat{H}_{Lipschitz,t}$ and a non-Lipschitz component $\widehat{H}_{matter,t}$,
   \begin{align}
   \widehat{H}_t(\hat{u},\hat{p})=\,\widehat{H}_{matter,t}(\hat{u},\hat{p})\,
   +\widehat{H}_{Lipschitz,t}(\hat{u},\hat{p}).
   \label{decompositionoftheHamiltonianHt}
   \end{align}
The {\it matter Hamiltonian} is defined in this expression by the piece of the Hamiltonian operator which is not $1$-Lipschitz. This is consistent  with the idea that matter (including gauge interactions) is {\it quantum matter}, in the sense that it is appropriately described by quantum models where the wave function is an effective description of the processes happening during the ergodic regime of the $U_t$ interaction and their interactions by quantum interactions of the form specified in {\it Definition} \ref{quantumclassicalinteraction}.

In general the decomposition \eqref{decompositionoftheHamiltonianHt} is not unique and it is not evident even its existence. However, if additional assumptions on the regularity of the Hamiltonian \eqref{RandersHamiltonian} are adopted, then it is possible to establish  for such type of decomposition.
Let us consider the classical dynamical version of Hamilton-Randers theory as described in chapter \ref{chapter on classical dynamics Hamilton Randers}. In particular, we have that
 \begin{lema}
    Let $H_t:T^*TM\to \mathbb{R}$ be a $\mathcal{C}^2$-smooth Randers Hamiltonian function \eqref{RandersHamiltonian}. Then there exists a compact domain $K'\subset\,T^*TM$ such that the restriction $H|_{K'}$ is $1$-Lipschitz continuous.
    \label{lemaon1lipschitz}
   \end{lema}
   The domain $K'$ is contained in the meta-stable domain ${\bf D}_0$. We will not make conceptual difference between $K'$ and ${\bf D}_0$.
   \begin{proof}
   By Taylor's expansion at the point $(\xi,\chi)\in T^*TM$ up to second order  one obtains the expressions
   \begin{align*}
   H_t(u,p) & =\,H_{t_0}+\,\sum^{8N}_{k=1}\,\frac{\partial H_t}{\partial u^{k}}|_{({\xi},\chi)}(u^k-{\xi}^k)+\,\sum^{8N}_{k=1}\,\frac{\partial H_t}{\partial p_{k}}|_{({\xi},\chi)}(p_k-\chi_k)\\
   & +\,\sum^{8N}_{k=1}\,R_{k}\,(u^k-{\xi}^k)^2+\,\sum^{8N}_{k=1}\,Q_{k}\,(p_k-\chi_k)^2,
   \end{align*}
   where the term
   \begin{align*}
   \sum^{8N}_{k=1}\,R_{k}\,(u^k-{\xi}^k)^2+\,\sum^{8N}_{k=1}\,Q_{k}\,(p_k-\chi_k)^2
   \end{align*}
   is the remaind term of the second order Taylor's expansion.
   The difference for the values of the Hamiltonian $H_t$ at two different points is given by the expressions
   \begin{align*}
   &|H_t(u(1),p(1))-\,H_t(u(2),p(2))|=\Big|\sum^{8N}_{k=1}\,\frac{\partial H_t}{\partial u^{k}}|_{({\xi},\chi)}(u^k(1)-{\xi}^k)+\,\sum^{8N}_{k=1}\,\beta^k({\chi})\,(p_k(1)-\chi_k)\\
   & +\sum^{8N}_{k=1}\,R_{k}(1)\,(u^k(1)-{\xi}^k)^2+\,\sum^{8N}_{k=1}\,Q^{k}(1)\,(p_k(1)-{\xi}_k)^2\\
   & -\sum^{8N}_{k=1}\,\frac{\partial H_t}{\partial u^{k}}|_{({\xi},\chi)}(u^k(2)-{\xi}^k)-\,\sum^{8N}_{k=1}\,\beta^k({\chi})\,(p_k(2)-\chi_k)\\
   & -\sum^{8N}_{k=1}\,R_{k}(2)\,(u^k(2)-{\xi}^k)^2-\,\sum^{8N}_{k=1}\,Q^{k}(2)\,(p_k(2)-\chi_k)^2\Big|\\
   & \leq \big|\sum^{8N}_{k=1}\,\frac{\partial H_t}{\partial u^{k}}|_{({\xi},\chi)}(u^k(1)-u^k(2))\big| +\,\big|\sum^{8N}_{k=1}\,\beta^k(\chi)(p_k(1)-p_k(2))\big|\\
   & +\big|\sum^{8N}_{k=1}\,R_{k}(1)\,(u^k(1)-{\xi}^k)^2-\,R_{k}(2)\,(u^k(2)-{\xi}^k)^2\big|\\
   & +\big|\sum^{8N}_{k=1}\,Q^{k}(1)\,(p_k(1)-\chi_k)^2-\,Q^{k}(2)\,(p_k(2)-\chi_k)^2\big|.
   \end{align*}
  Due to the continuity of the second derivatives of $H_t$, for each compact set $K\subset \,T^*TM$ containing the points $1$ and $2$, there are two constants  $C_R(K)>0$ and $C_Q(K)>0$ such that $|R_{k}(1)|,|R_{k}(2)|<C_R(K)$ and $|Q_{k}(1)|,|Q_{k}(2)|<C_Q(K)$, for each $k=1,...,8N$. Moreover, as a consequence of Taylor's theorem it holds that
    \begin{align*}
    \lim_{1\to 2}\,C_Q(K)=0,\quad\lim_{1\to 2}\,C_R(K)=0,
    \end{align*}
Since $K$ is compact the last two lines in the difference $|H(u(1),p(1))-\,H(u(2),p(2))|$ can be rewritten as
\begin{align*}
&\big|\sum^{8N}_{k=1}\,R_{k}(1)\,(u^k(1)-{\xi_t}^k)^2-\,R_{k}(2)\,(u^k(2)-{\xi_t}^k)^2\big|
\,\leq C_R(K)\big|\sum^{8N}_{k=1}(u^k(1)-u^k(2))^2\big|\\
   & \big|\sum^{8N}_{k=1}\,Q^{k}(1)\,(p_k(1)-\chi_k)^2-\,Q^{k}(2)\,(p_k(2)-\chi_k)^2\big|\,\leq
  C_Q(K) \big|\sum^{8N}_{k=1}\,(p_k(1)-p_k(2))^2\big|.
\end{align*}
The constants $C_Q(K)$ and $C_R(K)$ can be taken finite on $K$.
    Furthermore,  by further restricting the domain where the points $1$ and $2$ are to be included in a smaller compact set $\tilde{K}$, one can write the following relations,
    \begin{align}
   C_R(\tilde{K})|(u^k(1)-u^k(2))|\leq 1/2,\,\quad C_Q(\tilde{K})|(p_k(1)-p_k(2))|\leq 1/2.
\label{boundoftheconstant}
    \end{align}
Let us consider the  further restriction on the compact set ${K}'\subset \tilde{K}\,\subset\, T^*TM$ such that for each $({\xi},\chi)\in\,{K}'$
\begin{align}
\big|\frac{\partial H_t}{\partial u^{k}}|_{({\xi},\chi)}\big|\leq C_U,\,k=1,....,4N
\label{boundinacceleration}
\end{align}
holds good for some constant $C_U$. Also, on ${K}'$ it must hold that
\begin{align*}
& C_R(K)\big|\sum^{8N}_{k=1}(u^k(1)-u^k(2))^2\big|+\,  C_Q(K) \big|\sum^{8N}_{k=1}\,(p_k(1)-p_k(2))^2\big|\\
& \leq \,1/2\,\sum^{8N}_{k=1}\big|(u^k(1)-u^k(2))\big|+ 1/2\,\sum^{8N}_{k=1}\,\big|(p_k(1)-p_k(2))\big|.
\end{align*}
Moreover, the factors $|\beta^i|$ are bounded as a consequence of Randers condition \eqref{randerscondition}.
Then we have that
\begin{align*}
&|H(u(1),p(1))-\,H(u(2),p(2))|\big|_{{K}'}\leq  \,\tilde{C}_U\,\Big(\sum^{8N}_{k=1}\big|(u^k(1)-u^k(2))\big|\\
&+\,\sum^{8N}_{k=1}\,\,\big|\,(p_k(1)-p_k(2))\big|\Big) + 1/2\,\sum^{8N}_{k=1}\big|(u^k(1)-u^k(2))\big|+ 1/2\,\sum^{8N}_{k=1}\,\big|(p_k(1)-p_k(2))\big|
\end{align*}
with $\tilde{C}_U=\,\max \{C_U,1\}$.
This proves that $H|_{K'}$ is a Lipschitz function, with Lipschitz constant $M=\max\{\frac{1}{2},\tilde{C}_U\}$, which is necessarily finite. Now we can redefine the Hamiltonian dividing by $M$, which is a constant larger than $1$. This operation is equivalent to re-define the vector field $\beta\in \,\Gamma TTM$ in the domain ${\bf D}_0$. Such operation  does not change the equations of motion and the Randers condition \eqref{randerscondition}. Then we obtain a $1$-Lipschitz Hamiltonian on $K'$, restriction of the original Hamiltonian $H_t$.
    \end{proof}
    The proof of  {\it Lemma} \ref{lemaon1lipschitz} can be simplified further, making $Q_k=0$ from the beginning, but in this form is a kind of symmetric proof.
    \begin{comentario}
    Note that since it is assumed that the Hamiltonian $H_t$ is $\mathcal{C}^2$-smooth, the Randers condition \eqref{randerscondition} is not strictly necessary for the proof of  {\it Lemma} \ref{lemaon1lipschitz}. However, the Randers condition is useful to have well-defined causal structures associated with the underlying non-reversible Randers metric structure compatible with a macroscopic Lorentzian or causal structure.

    \end{comentario}
The compact domain $K'$ is not empty. In the metaestable domain ${\bf D}_0$, the Hamiltonian $H_t$ is equivalent to zero. Therefore, it is reasonable to think that in such domain $H_t$ is Lipschitz, thus providing an example where $K'$ can be contained (${\bf D}_0$ is not necessarily compact).

Extensions from $K'$ to the whole phase space $T^*TM$ can be constructed as follows. Consider the Sasaki metric on $T^*TM$ of the Hamilton-Randers structure \eqref{DefinicionHR}. For every observer $W$ one can associate by canonical methods a Finsler metric on $T^*TM$ and then an asymmetric distance function
   \begin{align*}
   \varrho_S:T^*TM\times T^*TM\to \mathbb{R}.
   \end{align*}
   Let us consider the projection on $K'$
   \begin{align}
   \pi_{K'}: T^*TM\to K',\,(u,p)\mapsto (\bar{u},\bar{p}),
   \end{align}
  where $(\bar{u},\bar{p})$ is defined by the condition that the distance from $(u,p)$ to $K'$ is achieved at $(\bar{u},\bar{p})$ in the boundary $\partial K'$. Then one defines the {\it radial decomposition} of $H_t$ by the expression
   \begin{align}
   H_t(u,p)=\,R\big(\varrho_S((u,p),(\bar{u},\bar{p}))\big)\,H_t(\bar{u},\bar{p})+\,\delta H_t(u,p).
   \label{decompositionofH}
   \end{align}
   The positive function $R\big(\varrho_S((u,p),(\bar{u},\bar{p}))\big)$ is such that the first piece of the Hamiltonian is $1$-Lipschitz. The second contribution is not $1$-Lipschitz. By assumption, $\delta H_t(u,p)$ is identified with the matter Hamiltonian $H_{matter}$,
    \begin{align}
    H_{matter, t}(u,p):=\,\delta H_t(u,p).
    \label{matterhamiltonian}
    \end{align}
With these redefinitions we obtain the following
   \begin{proposicion}
    Every Hamiltonian  \eqref{RandersHamiltonian} admits a normalization such that the decomposition \eqref{decompositionoftheHamiltonianHt} holds globally on $T^*TM$.
   \label{radialdecomposition}
   \end{proposicion}
   \begin{proof}
   One can perform the normalization
   \begin{align*}
   H_t(u,p)\to \,&\frac{1}{R\big(\varrho_S((u,p),(\bar{u},\bar{p}))\big)}\,H_t(u,p)\\
& = H_t(\bar{u}(u),\bar{p}(u))+\, \frac{1}{R\big(\varrho_S((u,p),(\bar{u}(u),\bar{p}(u)))\big)} \delta H_t(u,p).
 \end{align*}
 The first term is $1$-Lipschitz in $T^*TM$, since $H_t(\bar{u}(u),\bar{p}(u))$ is $1$-Lipschitz on $K'$, while the second term is not $1$-Lipschitz continuous.
   \end{proof}
   The uniqueness of this construction depends upon the uniqueness of the compact set $K'$, the uniqueness of the relation $(u,p)\mapsto (\bar{u},\bar{p})$. Thus in general the construction is not unique, except if a further criterium is added. For instance, one can consider the maximal set $K'$ obtained in that way.

  The above considerations concern the classical formulation of Hamilton-Randers systems. If we consider the quantized version of the condition  \eqref{decompositionoftheHamiltonianHt}, then the Hamiltonian constraint \eqref{averagehamiltonianevolution} is casted in the following way.
   From the properties of the $U_t$ flow it follows that
   \begin{align*}
   \lim_{t\to (2n+1)T} \big(\widehat{H}_{matter,t}\,+\widehat{H}_{Lipschitz,t}\big)|\Psi\rangle=0
   \end{align*}
   for each $|\Psi\rangle\in\,\Gamma\mathcal{H} TM$.
   However, each of the individual terms in this relation can be different from zero in the metastable domain ${\bf D}_0$,
   \begin{align*}
   \lim_{t\to (2n+1)T} \widehat{H}_{matter,t}|\Psi\rangle\neq \,0, \quad
    \lim_{t\to (2n+1)T}\widehat{H}_{Lipschitz,t}|\Psi\rangle\neq\,0.
   \end{align*}
   This implies that in order to have the metastable equilibrium point at the instant $t=(2n+1)T$, in addition with the matter Hamiltonian \eqref{matterhamiltonian}, an additional piece of dynamical variables whose described by the Hamiltonian $\widehat{H}_{Lipschitz,t}$ is needed. On the other hand,
if we assume that the matter Hamiltonian \eqref{matterhamiltonian} must be positive definite when acting on physical states, then the $1$-Lipschitz Hamiltonian should have negative eigenvalues only. Hence for Hamilton-Randers models the positiveness of the matter Hamiltonian is extended to all $t\in [0,(2n+1)T]$.
This implies the consistency of the positiveness of the energy level for the quantum Hamiltonian for matter \eqref{matterhamiltonian} in the whole process of the $U_t$-evolution.
We can reverse this argument. If $\widehat{H}_{Lipschitz,t}$ is negative definite, then $\widehat{H}_{matter,t}$ must be positive definite. This property is related with the analogous property of gravitational interaction.

\subsection{Emergence of the law of inertia and of the weak equivalence principle}
We organize this section in two parts.
\subsubsection{Preliminary considerations} Let us consider a physical system $\mathcal{S}$ that can be thought as composed by two sub-systems $A$ and $B$.
     We denote by  $X^\mu(\mathcal{S}_i),\,i\equiv \mathcal{S},A,B$ the {\it macroscopic observable} coordinates associated with the system $\mathcal{S}_i$, that is, the value of the coordinates that could be associated when local coordinates are assigned by a classical observer to each system $\mathcal{S},A,B$ by means of a measurements or by means of theoretical models.

Let $\xi_k:\mathbb{R}\to TM^k_4$ the world line of the $k$-essime sub-quantum molecule.
Then we adopt the following:
\\
{\bf Assumption I}.  The functions
\begin{align*}
X^\mu(\mathcal{S}_i):T^*TM\to \mathbb{R},\quad (u^{k_1},...,u^{k_N},p^{k_1},...,p^{k_N})\mapsto X^\mu(u^{k_1},...,u^{k_N},p^{k_1},...,p^{k_N})
\end{align*}
are smooth. This is an useful requirement to link the microscopic degrees of freedom with macroscopic degrees of freedom.

Under the additional requirement of the Randers type condition, that implies the existence of an universal bounded acceleration and speed for the sub-quantum molecules,  in the metastable  domain ${\bf D}_0\subset \,T^*TM$ containing the metastable points of the evolution of the system at the instants $\{t=(2n+1)T,\,n\in\,\mathbb{Z}\}$, the functions $X^\mu((2n+1)T)=X^\mu(\tau\equiv n):T^*TM\to \mathbb{R}$  are $1$-Lipschitz in $t$-time parameter. In order to show this, let us first remark that the relations
    \begin{align}
     \lim_{t\to (2n+1)T}\frac{\partial X^\mu(u,p,t)}{\partial t}=\, 0,\quad \mu=1,2,3,4
      \label{derivativeinthettoT1}
    \end{align}
hold good,  since in the metastable  domain ${\bf D}_0$ physical observables depending upon $(u,p)$ coordinates do not have $t$-time dependence. This condition can be re-written as
     \begin{align*}
    \lim_{t\to (2n+1)T} \frac{\partial X^\mu(u,p,t)}{\partial t}=\, \lim_{t\to (2n+1)T} \left(\sum^{8N}_{k=1}\,\frac{\partial X^\mu}{\partial u^\rho_k}\,\frac{d u^\rho_k}{d t}+\,\sum^{8N}_{k=1}\,\frac{\partial X^\mu}{\partial p_{\rho k}}\,\frac{d p_{\rho k}}{d t}\right)=0.
    \end{align*}
    \begin{proposicion}
   If the functions $\left\{u^\mu_k(t),\frac{d u^\rho_k}{d t},\,\mu=1,2,3,4;\, k=1,...,N\right\}$ are $\mathcal{C}^1$-functions, then
   the coordinate functions $\{X^\mu(\tau)\}^3_{\mu=0}$ are $\mathcal{C}^1$-functions with uniformly bounded derivatives in a restricted domain of ${\bf D}_0 \subset \,T^*TM$.
   \label{uniform bound for X}
   \end{proposicion}
   \begin{proof}
Let us also consider the differential expressions
    \begin{align*}
   \frac{d X^\mu(u,p,t)}{d t}=\,\sum^{8N}_{k=1}\,\frac{\partial X^\mu}{\partial u^\rho_k}\,\frac{d u^\rho_k}{d t}+\,\sum^{8N}_{k=1}\,\frac{\partial X^\mu}{\partial p_{\rho k}}\,\frac{d p_{\rho k}}{d t}.
    \end{align*}
 The derivatives $\{\frac{d u^\rho_k}{d t},\,\mu=1,2,3,4;\, k=1,...,N\}$ are uniformly bounded as a consequence of  the Randers condition \eqref{randerscondition}. Since the system of equations for the configuration coordinates $\{u^i\}^{8N}_{k=1}$ \eqref{HamiltonEquations} is autonomous for $u$, the derivatives
 \begin{align*}
 \left\{\frac{d p_{\rho k}}{d t},\,\rho=1,2,3,4;\, k=1,...,N\right\}
  \end{align*}
  are fully determined by
  \begin{align*}
  \left\{u^\mu_k(t),\frac{d u^\rho_k}{d t},\,\rho=1,2,3,4;\, k=1,...,N\right\}.
   \end{align*}
   Therefore, each $p_{\rho k}$ and its time derivative $\frac{d p_{\rho k}}{d t}$ are indeed uniformly bounded.
   \end{proof}

   Since the system of equations \eqref{HamiltonEquations} for the configuration coordinates $\{u^i\}^{8N}_{k=1}$  is autonomous for the variable $u$, the functions $\left\{u^\mu_k(t),\frac{d u^\rho_k}{dt}\right\}^{4,N}_{\mu,k=1,1}$ are also $1$-Lipschitz continuous.
    \begin{corolario}\label{Lipschitz on each cycle}
    The functions $\{X^\mu(t)\}^3_{\mu=0}$  are $1$-Lipschitz functions on a subdomain of ${\bf D}_0$ when restricted to each elementary cycle.
\end{corolario}
\begin{comentario}
From the above arguments do not follow the $1$-Lipschitz continuity with respect to the $\tau$ parameter, since the derivatives $ \frac{\partial X^\mu}{\partial u^\rho_k},\,\frac{\partial X^\mu}{\partial p_{\rho k}}$ can jump abruptly from cycle to cycle when evaluated in ${\bf D}_0$.
\end{comentario}

Let us consider two  subsystems $A$ and $B$ of the full system $\mathcal{S}$ under consideration. The sub-systems $A$, $B$ are embedded in $\mathcal{S}$ such that
$\mathcal{S}=\,A\sqcup B$, where $\sqcup$ is the {\it disjoint union operation} for systems described by collections of sub-quantum molecules. Let us consider a local coordinate system such that the identification for the description of the aggregate $\mathcal{S}$, but also for the aggregate of the systems  $A$ and $B$ by means of the following embeddings
   \begin{align}
   A\equiv (u_1(t),...,u_{N_A}(t),0,...,0) \quad \textrm{and}\quad B\equiv (0,...,0,v_1(t),...,v_{ N_B}(t)),
   \label{embedding A,BtoS}
   \end{align}
   with $N=\,N_A+N_B,N_A,N_B\gg 1$ holds good.
The whole system $\mathcal{S}$ can be represented in local coordinates as
   \begin{align*}
   \mathcal{S}\equiv (u_1(t),...,u_{ N_A}(t),v_1(t),...,v_{ N_B}(t)).
   \end{align*}
    By the action of the diffeomorphisms $\varphi_k:M^k_4\to {\bf M}_4$, one can consider the world lines of the sub-quantum molecules on ${\bf M}_4$ at each constant value of $t$ modulo $2T$. In the particular case of metaestable equilibrium points $\{t=(2n+1)T,\,n\in\,\mathbb{Z}\}$ we have a set of (discrete) world lines in ${\bf M}_4$. We are going to consider the system in the region where it appears in the most compact form, that is, where the distances between points in the image of ${\bf M}_4$ by the diffeomorphism $\varphi_k:M^k_4\to {\bf M}_4$ is the minimal possible. Since the notion of physical metric distance is not diffeomorphic invariant, the closest that we can consider is that for this regime all the points can be mapped in a common open set of ${\bf M}_4$. Then it is possible to define the {\it observable coordinates of the system} by the expression
    \begin{align}
    \tilde{X}^\mu_i(\tau(n))=\frac{1}{N}\,\lim_{t\to (2n+1)T}\sum^{N}_{k=1}\,\varphi^\mu_{k_i}(x_{k_i}(t)),\quad\,i=A,B,\mathcal{S},\,\mu=0,1,2,3,
    \label{definicionofXmu}
    \end{align}
    where here $\varphi^\mu_{k_i}$ are local coordinates on ${\bf M}_4$, defined after the action of the diffeomorphism $\varphi_{k_i}$. Note that the normalization factor $1/N$ is the same for all the systems $i\equiv A,B,\mathcal{S}$. This means that we are considering systems that eventually are sub-systems (proper or improper) of a larger sub-system $\mathcal{S}$. This formal constraint is however harmless for general purposes by the embedding \eqref{embedding A,BtoS}.

    We  identify $\tau(n)$ with the $\tau$-time parameter and consider it continuous, in relation with macroscopic or quantum time scales. Then by the embedding \eqref{embedding A,BtoS},
      \begin{align}
    \tilde{X}^\mu_i(\tau)=\frac{1}{N}\,\lim_{t\to (2n+1)T}\sum^{N}_{k=1}\,\varphi^\mu_{k}(x_{k_i}(t)),\quad\,i=A,B,\mathcal{S}.
    \label{definicionofXmu2}
    \end{align}
    These functions are $1$-Lipschitz on each cycle, by applying the arguments given above, specifically proposition \ref{uniform bound for X} and corollary \ref{Lipschitz on each cycle}.

In the domain ${\bf D}_0$, the macroscopic coordinate functions $\{X^\mu,\,\mu=1,2,3,4.\}$ and the world lines $\{\xi_k,k=1,...,N\}$ are $1$-Lipschitz on each cycle, because since $\beta=0$ and the derivatives associated are close to zero. Hence the composition $\tilde{X}^\mu_i =\,X^\mu\circ \chi (\mathcal{S}_i)$, where $\chi$ are the embeddings of system $\mathcal{S}_i$ on $T^*TM$,  are $1$-Lipschitz in each cycle intersected with ${\bf D}_0$.

The co-tangent space $T^*TM$ associated with a Hamilton-Randers dynamical systems is endowed with a geometric measure $\mu_P$ of the product form \eqref{general form of the measure}. Therefore, the function $X^\mu: T^*TM\to\mathbb{R}$ have associated the median $M^\mu$. Viewed as parameterized by $\tau$, the functions $X^\mu(\tau)$ implies a $\tau$-dependent median $M^\mu(\tau)$. The functions $M^\mu:\mathbb{R}\to \mathbb{R},\,\mu =0,1,2,3$ do not depend upon the particularities of the system, except for the initial conditions $M^\mu(0)$.

The pseudo-metric structure $\eta$ is a product of Sasaki type structures \eqref{structureeta}. This structure has a dual form $\eta^* =\frac{1}{N}\,\sum^K_{k=1}\,\oplus \eta^*_S(k)$ of direct sums of Sasaki metrics on $TM$, where each $\eta^*_S(k)$ is defined in $TM^k_4$. Given a macroscopic observer, that as it was discussed in {\it section} \ref{Section on macroscopic observers}, is a time-like vector field $W\in \,\Gamma T{\bf M}_4$, one can define copies of $W$ on each $M^k_4$ as follows. Let us consider the inverse diffeomorphism $\{\varphi^{-1}_k:{\bf M}_4\to M^k_4\}$. Then the pull-forwards of $W$ are defined pointwise, $W_k(u_k):=\,d(\varphi^{-1}_k|_x(W(x))$, where $x\in {\bf M}_4$ is the unique point such that $\varphi^{-1}_k(x))=y$. With the collection of vector fields $\{W_k\in\,\Gamma TM^k_4\}^k_{k=1}$ one can define from the Lorentzian metrics $\{{\eta}_4(k)\}$  the associated Riemannian metrics $\{\bar{\eta}_4(k)\}$ and the corresponding Sasaki type metrics $\{\bar{\eta}^*_S(k)\}$, that have Riemannian signature. Then one can define the Riemannian metric
\begin{align}
 \bar{\eta}^* =\frac{1}{N}\,\sum^K_{k=1}\,\oplus \bar{\eta}^*_S(k).
 \label{Riemannian metric in TM}
 \end{align}
 on $TM$, that has Riemannian signature.
\begin{proposicion}
The space $TM$ can be endowed with a mm-Gromov structure $(TM, \mu_P, \bar{\eta}^*)$, where $\mu_P$ is the induced measure on $TM$ from the homonym measure defined in $T^*TM$.
\label{mmGromov structure of TM}
\end{proposicion}
Note that mm-Gromov structure is not unique, since it depends on the vector field $W\in\,\Gamma T{\bf M}_4$.

Note also that  {\it Proposition} \ref{mmGromov structure of TM} the space $TM$ has been considered, instead than the space $T^*TM$. This is because the momentum variables are non-autonomous, allowing to reduce the kinematical description to $TM$. However, one can consider the Sasaki type metric $\bar{\eta}^*_S$ defined on $T^*TM$ constructed from the metric $\bar{\eta}^*$ (already Sasaki) on $TM$. Then we have
\begin{proposicion}
The space $T^*TM$ can be endowed with a mm-Gromov structure $(T^*TM, \mu_P, \bar{\eta}^*_S)$.
\label{mmGromov structure of T*TM}
\end{proposicion}
We will consider the class $(TM, \mu_P, \bar{\eta}^*)$ of mm-Spaces  in the following considerations.
\subsubsection{Emergence of the classical law of inertia and of the weak equivalence principle}
Now we apply the theory of concentration of measure of mm-Gromov states. We consider the mm-Gromov space $(TM, \mu_P, \bar{\eta}^*)$.
By the concentration property \eqref{concentration2} of the $U_t$ dynamics in the Lipschitz dynamical regime ${\bf D}_0$,
one has the relation
   \begin{align}
  \mu_P\left(\frac{1}{\sigma_{\tilde{X}^\mu}}\,|\tilde{X}^\mu(\mathcal{S}_i)-M^\mu|>\rho\right)_{t\to (2n+1)T}\sim C_1\exp \left(-\,C_2 \frac{\rho^2}{2\,\rho^2_p}\right),
\label{generalconcentrationofmeasure0}
\end{align}
$\mu=1,2,3,4,\,i=A,B,\mathcal{S}$ holds, where the metric used is the one associated with $\bar{\eta}^*_S$ and the measure is induced from $\mu_P$. Here $M^\mu$ is the median of the functions $\tilde{X}^\mu$.
The constants $C_1,C_2$ are of order $1$, where $C_2$ depends on the dimension of the spacetime ${\bf M}_4$. $\rho_p$ is  independent of the system $i=A,B,\mathcal{S}$ (see chapter \ref{chapter on concentration of measure}).
 The value of the constant $C_2$ cannot be fixed by the theory, but does not compensates the abrupt concentration caused by the difference between the sub-quantum scale and the corresponding quantum scale.

 One main point is that $M^\mu$ does not depend upon the system $i=A,B$, but it can depend upon the initial value of the coordinates $\tilde{X}^\mu$.

However, the relation \eqref{generalconcentrationofmeasure0} is not the concentration of measure required. It is the dynamical version of the relation \eqref{generalconcentrationofmeasure0} what is required. In Hamilton-Randers, each instant $\tau$ corresponds to the restriction to ${\bf D}_0$ to a given fundamental cycle. Then one defines a set of maps $\{\widetilde{X}^\mu|_{{\bf D}_0(n)}\}$, where ${\bf D}_0(n)$ is the restriction of the metastable domain to the cycle $n$ associated with the instant of time $\tau$.
The $\tau$-evolution of the coordinates $\tilde{X}^\mu(\mathcal{S}(\tau))$, $\tilde{X}^\mu(A(\tau))$ and $\tilde{X}^\mu(B(\tau))$ that have the same initial conditions  differ between each other after the dynamics at $\tau$-time such that
   \begin{align}
  \mu_P\left(\frac{1}{\sigma_{\tilde{X}^\mu}}\,|\tilde{X}^\mu(\mathcal{S}_i(\tau))-M^\mu(\mathcal{S}(\tau))|>\rho\right)_{t\to (2n+1)T}\sim C_1\exp \left(-\,C_2 \frac{\rho^2}{2\,\rho^2_p}\right),
\label{generalconcentrationofmeasure}
\end{align}
that are just the restrictions of the condition \eqref{generalconcentrationofmeasure0} to each of the functions $\widetilde{X}^\mu|_{{\bf D}_0(n)}:{\bf D}_0(n)\to \mathbb{R}$.

If $\rho/\rho_P$ is equal to the number of subquantum degrees of freedom $N$, then
  \begin{align*}
  \mu_P\left(\frac{1}{\sigma_{\tilde{X}^\mu}}\,|\tilde{X}^\mu(\mathcal{S}_i(\tau))-M^\mu(\mathcal{S}(\tau))|>\rho\right)_{t\to (2n+1)T}\sim C_1\exp \left(-\,C_2 N^2\right),
\end{align*}
Thus, under the same initial conditions, different systems have almost the same classical world-lines. However, such world lines do not need to be continuous on the $\tau$-time evolution description.

The explanation of the equivalence principle offered along these lines implies that theoretically the weak equivalence principle should be an almost exact law of Nature. It breaks down abruptly only when the system described is a quantum system composed by $N$ sub-quantum degrees of freedom or a system composed by few sub-quantum degrees of freedom.

The above argument also applies when the coordinate systems are globally defined. In this sense, the functions $\tilde{X}^\mu$ depends upon the initial conditions $(x^\mu (\tau=0),\dot{X}^\mu(\tau=0))$ and by the argument above, there is a global coordinate system where they are $(X^\mu =C^\mu, 0)$. In any other coordinate system related by a global boost, the velocity of the system will be constant. This completed our argument about the emergent origin of the inertial law in classical mechanics initiated in {\it section} \ref{law of the inertia as emergent phenomenon 1}.

  \subsection{On the emergent origin of the gravitational interaction}\label{euristicargumentforemergenceofgravity}
   Bringing together  the previous characteristics for the $1$-Lipschitz domain of the $U_t$ flow, we find  the following general features:
   \begin{itemize}
   \item Since the constraint \eqref{finalhamiltoniantevolution} holds good, the dynamical $U_t$ flow in the domain ${\bf D}_0$ is compatible with the Hamiltonian constraint of general relativity.

   \item A generalized weak equivalence principle for the observable coordinates $\widetilde{X}^\mu(S(\tau))$ holds good in the metastable domain ${\bf D}_0$.

   \item The dynamical $U_t$ flow in the domain ${\bf D}_0$ determines a classical interaction, since it is relevant only in the metastable domain ${\bf D}_0$.

   \item There is a local maximal speed for the sub-quantum molecules of a Hamilton-Randers dynamical systems and  invariance under a local relativity group invariance holds. This local relativity group is by construction the Lorentz group.

   \end{itemize}
   Furthermore, we have found the following two additional restrictions,
   \begin{itemize}

   \item The $U_t$ interaction in the $1$-Lipschitz domain must be compatible with the existence of a maximal and universal maximal proper acceleration.

   \end{itemize}
    In view of the formal similarity of these properties with the analogous properties of the current mathematical description of the gravitational interaction, the following conclusion follows:
\begin{center}
   {\it In the metastable domain the $1$-Lipschitz dynamics associated with} ${H}_{Lipshitz,t=(2n+1)T}$ {\it is the gravitational interaction.}
\end{center}

That gravity must be intrinsically involved in the collapse of the wave functions is an idea that appears in several modern approaches to the description of measurement problem \cite{Diosi, GhirardiGrassiRimini, Penrose 1996, Penrose 2005}. However, as we discuss explicitly before, there are fundamental differences between the models described here and spontaneous collapse models or collapse models induced by large mass measurement devices.

Thus according to Hamilton-Randers theory, gravity appears as classical, instead of semi-classical or quantum interaction. Classical means here not subjected to quantization and without superpositions of spacetimes. However, the interaction can be potentially fluctuating and free-fall worldlines discontinuous.

There are further essential differences with Einstein's general relativity, since our theory includes an universal  maximal proper acceleration. Universal in the sense that it affects all the interactions, gravitational and gauge interactions). It is very interesting the possibility that a generalization of Einstein gravity in the frameworks of metrics with a maximal acceleration compatible with the weak equivalence principle  could lead to a {\it classical resolution of curvature singularities} \cite{CaianielloGasperiniScarpetta, Ricardo2020}.

Regarding the incorporation of gravity in the standard description of quantum systems we shall discuss in {\it Section} \ref{chapter on properties of quantum mechanics} gravitational induced quantum interference and see how this procedure is compatible with our version of emergent gravity.

\subsection{Existence of a domain where Newtonian gravity is $1$-Lipschitz continuous}
   That there is a domain where the gravitational interaction is $1$-Lipschitz can be easily argued within the framework of newtonian gravity by showing that the Hamiltonian function is Lipschitz in the strong form. As a consequence, the corresponding infinitessimal evolution operator on functions defined over the spacetime ${\bf M}_4$ is also Lipschitz, and hence, the finite evolution operators.

    Let us consider the newtonian gravitational force between a massive point particle  with mass $m$ by a massive point particle with mass $M$  located at the origin of coordinates,
   \begin{align}
   F_2(\vec{r})=\,-G\,\frac{m\, M}{r^2}, \quad \vec{r}\in\,\mathbb{R}^3
   \label{Newtonlaw}
   \end{align}
  and $r=|\vec{r}|$ the distance to the origin in $\mathbb{R}^3$ of  the point $\vec{r}$. In order to compare different lengths or different mechanical forces, it is useful to consider dimensionless expressions, for which we need reference scales.
In doing this comparison we adopt as length scale the Planck length and for the force scale the Planck force and use homogenous quantities for length and force.  The Planck force provides a natural unit, respect to which  we can compare any other scale.
Let us consider the expression
\begin{align*}
\frac{|F_2(\vec{r}_2)-F_2(\vec{r}_1)|}{{F}_P}=\,\alpha \,\frac{|r_2-r_1|}{\l_P},
\end{align*}
where $F_P$ is the Planck force and $\l_P$ is the Planck length. After some algebraic manipulations, one finds an expression for the coefficient $\alpha$. In  the case of Newton law of universal gravitation \eqref{Newtonlaw}, $\alpha$ is given by the expression
\begin{align*}
\alpha=\,\l_P\,\frac{1}{c^4} \, G^2 \,m\,M\,\frac{1}{r^2_2\,r^2_1}\,|r_2+\,r_1|.
\end{align*}
In order to simplify the argument, let us  consider $m=M$. Furthermore, although the case $r_2=r_1$ is singular, in order to work in a fixed scale, we consider a relation $r_1= \lambda\, r_2$ with $\lambda\sim 1$ constant. Then one obtains a compact expression for $\alpha$,
\begin{align}
\alpha=\, \frac{1+\lambda}{\lambda^3}\,\frac{D}{D_p}\,\frac{E}{E_P},
\label{Lipschitzconstantforgravity}
\end{align}
where $D=\,{m_0}/{r^3}$ is a characteristic density of the system, $E=m_0 c^2$, $m_0=m_0 \,c^2$, $m_0=\max\{m,M\}$, $D_P$ is the Planck density and $E_p$ is the Planck energy. It follows from the expression \eqref{Lipschitzconstantforgravity} that for scales of the standard model, atomic physics systems or macroscopic systems,  $\alpha\ll 1$. Moreover, $\alpha$ is bounded by $1$. The bound is saturated at the Planck scale. This shows that at such scales, gravity is $1$-Lipschitz\footnote{A similar argument applies if instead of considering the case of the Planck scale as the fundamental scale, we consider another high energy scale as the fundamental scale.}. This is because the relative weakness of the gravitational interaction compared with the interactions of the Standard Model of particles.

A different form of the argument is the following. The Lipschitz condition applied to the Newtonian force is of the form
\begin{align*}
\frac{|\vec{F}_2(\vec{r}_1)-\,\vec{F}(\vec{r}_2)|}{F_P}=\,\frac{G m M}{F_P}\Big|\frac{\vec{r}_1}{r^3_1}-\frac{\vec{r}_2}{r^3_2}\Big|<\frac{|\vec{r}_1-\vec{r}_2|}{l_P}.
\end{align*}
Let us denote now $\bar{r}=\,\min\{r_1,r_2\}$. Then we have the sufficient condition for the Newtonian force to be Lipschitz
\begin{align*}
\frac{l_P}{F_P}\, \frac{G m M}{r^3} <1 .
\end{align*}
This condition can be re-cast as the following sufficient condition,
\begin{align}
\frac{D}{D_P}\,\frac{E}{E_P}<1 .
\end{align}

Although based on a Newtonian limit and under several approximations and assumptions, the conclusion that one can reach is the existence of a regime where gravity is a classical, $1$-Lipschitz interaction in the above sense and hence, in the relevant sense for Hamilton-Randers theory.

 As mathematical models, newtonian gravity or Einstein gravity can be extrapolated to domains where the evolution is not $1$-Lipschitz, specially in domains close to the curvature singularities. However, such extrapolations, by the arguments given in this {\it section}, should be considered un-physical: the domain of validity of physical gravitational models does not reach the domain of the singularities.

 However, the existence itself of curvature singularities for long distance fields implies in principle that such singularities have an effect far from their own spacetime location. Therefore, one could recognize the incompleteness of the model, even for scales where it is applicable. A possible way out of this dichotomy is to consider classical theories of gravity with a maximal acceleration \cite{Ricardo2020}. Consistent with the weak equivalence principle, such classical theories could be free of curvature singularities.

 \subsection{Modified Newtonian dynamics and the $1$-Lipschitz condition}
 Newtonian gravity is not the only case of gravitational models with domains compatibles with the $1$-Lipschitz condition and hence, compatible with the weak equivalence principle as an emergent law. Let us consider here the case of Modified Newtonian dynamics, where the force is determined by a logarithmic function of the radial distance,
 \begin{align}
 \vec{F}_1(\vec{r})=- k_0 M m \frac{\vec{r}}{r^2},
 \end{align}
 where $k_0$ is an universal constant.
 Then following an argument similar to the above, a sufficient condition for $\vec{F}_1$ being a Lipschitz force is that
 \begin{align}
 \frac{A_P}{A}\, \cdot \frac{k_0 m M}{E_P}<1,
 \end{align}
where $A=r^2$, $A_P$ is the Planck area.

Newtonian gravity and modified Newtonian gravity are not the only models of non-relativistic gravity compatible with the condition of Lipschitz force. If one consider a convex combination of them,
\begin{align}
\vec{F}_\lambda=\, \lambda\vec{F}_2(\vec{r})+\,(1-\lambda)\vec{F}_1 (\vec{r}),\,\quad\, \lambda\in \,[0,1] .
\end{align}
The condition $\lambda \in \, [0,1]$ is necessary to assure that the force is Lipschitz; otherwise, if $\lambda$ is not constrained in magnituded, the domain of consistence with the Lipschitz condition will be small.
\begin{align*}
\vec{F}(\vec{r})=\,- G m M\,\lambda\,\frac{r}{r^3}+\,\bar{k}\,m M \,(1-\lambda)\frac{\vec{r}}{r^2}.
\end{align*}
But note that from a purely formal point of view this expression for the force can be re-casted as
\begin{align}
\vec{F}(\vec{r})=\,- G_N m M\,\frac{r}{r^3}+\,\bar{k}_0\,m M \frac{\vec{r}}{r^2},
\label{Modified Newton force}
\end{align}
$G_N$ is then identified with the Newton constant, while the new constant $\bar{k}_0$ is independent of the masses $m$ and $M$.
A  sufficient condition for the force $\vec{F}$ be consistent with the Lipschitz condition is that
\begin{align*}
\lambda \frac{D}{D_P}\,\cdot \frac{E}{E_P}+\,(1-\lambda) \frac{A_P}{A}\, \cdot \frac{k_0 m M}{E_P}<1.
\end{align*}

This expression for the force implies that a large scales, the gravitational force is proportional to $1/r$, while at short distance (in astrophysical terms), the force is the usual Newtonian force. There is a radius where the forces are approximately equal,
\begin{align}
r_c=\,\frac{G_N}{\bar{k}_0}.
\end{align}
At this distance, a system with mass $m$ has an acceleration
\begin{align}
a_c =\, M \frac{\bar{k}^2_0}{G_N}.
\end{align}
Therefore, the critical radius is independent of $M$, but the critical acceleration, where it is expected to depart from Newtonian gravity, depends on the source mass $M$.

 An immediate consequence of the above reasoning is that the gravitational field is not composed of gravitons. In particular, a gravitational wave is not composed of gravitons in the way as a classical electromagnetic wave could be thought to be composed of photons. Instead, in the framework of Hamilton-Randers theory and its extension to gravity as developed in this {\it section},  gravitational waves are a classical and emergent effect, not reducible to quantum. This negative result can be turned a falsifiable prediction of our theory.

\subsection{Comparison with a quantum interaction: The case of electrodynamics}
If we  repeat this argument for the static electromagnetic field, formally an analogous result is obtained. However, if we take into account the relative intensity of the classical Coulomb field with the Newtonian field,  for instance for the electron, the corresponding $\alpha$  is a factor of order $10^{42}$ larger that for gravity. This suggests that at such scales the electromagnetic field cannot be $1$-Lipschitz.

Another way to see this issue is by recognizing the fundamental fact that the electromagnetic field is quantized, which is a very different regime than a $1$-Lipschitz regular dynamic. Instead, we find that it is the full quantum electrodynamics which is required to have consistent predictions with experience at the atomic and sub-atomic scales. Hence we should not extend the argument directly from the Newtonian gravitational field to the Coulomb field.

 We showed that one of the assumptions in our derivation of the classical weak equivalence principle was the absence of exchange of sub-quantum molecules with the ambient or the source of the field. This assumption seems not to hold in the case of the electrodynamics and in general, gauge interactions, as we have discussed previously.

\newpage
\section{\LARGE{Conceptual problems in quantum mechanics from the point of view of Hamilton-Randers Theory}}\label{chapter on properties of quantum mechanics}
\bigskip
\bigskip
 Since the advent of the quantum theory there have been conceptual difficulties on the foundations of the theory, its relation with general relativity, the inclusion of gravity on the quantum framework and the interpretation of the phenomenology of quantum description. Among the relevant conceptual issues, following Isham , are the {\it quaternity of problems} \cite{Isham 1995}. These problems and how they are addressed in Hamilton-Randers theory, are the following:
\begin{enumerate}
\item {\bf The meaning of probability in quantum mechanical systems}. We have considered this problem in Chapters \ref{chapter of the Hilbert space structure} and \ref{chapter on concentration of measure}. In Hamilton-Randers theory, the origin of the probability description in the quantum theory is based on the ergodic behaviour of the sub-quantum degrees of freedom during the $U_t$ evolution. It was proved the emergent character of the Born rule, arguing from first principles that $|\psi|^2(x)$ represents the probability to find the system in an infinitesimal neighborhood of the point $x\in \,{\bf M}_4$, if a position measurement is performed and that such interpretation applies to individual quantum systems or to collection of identical quantum systems.

\item {\bf The role of measurement}. In Chapter \ref{chapter on concentration of measure}, a {\it Theory of Measurement on Quantum Systems} was sketched, illustrating its relation with the notion of {\it natural instantaneous collapse of the quantum state}, a notion that was also introduced in that Chapter. As a consequence of the theory, it turns out that in Hamilton-Randers theory, observable magnitudes are well defined during the measurement in the sense that every measurement is performed in the metastable domain ${\bf D}_0$ of the $t$-time dynamics, which is the domain where classical properties can be associated with the system objectively.

\item {\bf The collapse of the quantum state}. In Hamilton-Randers theory there is no collapse of the wave function. Instead of this type of process, Hamilton-Randers theory introduces the notion of natural instantaneous collapse of the quantum state. The intrinsic difference with the instantaneous collapse or other notions of the {\it collapse of the wave function} investigated in the literature is that the new mechanism does not require the interaction of the quantum system with an external agent or device to happen. Natural instantaneous collapse is based upon the assumptions on that $U_t$ flow, in particular, on the assumptions of the concentrating regime on each fundamental cycle, where the concentration of measure happens, according to the theory developed in chapter \ref{chapter on concentration of measure}.

\item {\bf Quantum entanglement}. For pure states,  quantum entanglement is the property that for systems describing more than one quantum particle, the system cannot be described in the form of a tensor product of states. This mathematical fact leads to surprising consequences on the local properties of measurements in separate parts. Specifically, entanglement implies an apparent non-local correlation among measurements performed spacetime spacelike separate systems. In our view, standard quantum mechanics describes effectively such non-local behavior due to entanglement and related properties, but it lacks of a explanatory mechanism for it. We will discuss in this chapter a mechanism, based on mixing and interaction during the $U_t$ dynamics.
\end{enumerate}

 A sub-quantum explanation of entanglement and quantum non-locality is based on the ergodic property of the $U_t$ flow for sub-quantum dynamics. Sub-quantum degrees of freedom in spacelike separated spacetime points can be correlated by means of interactions during the fundamental $U_t$ dynamics and that since such description is absent in the standard quantum mechanical description, it appears as a {\it spooky action effect}.

We being this chapter discussing in the framework of Hamilton-Randers theory the mechanism explaining quantum interference and quantum entanglement. This is done by considering the classical two slit experiment as discussed by Feynman et al. \cite{Feynman, Feynman Hibbs}. We also describe a mechanism to explain quantum non-local correlations in the framework of Hamilton-Randers theory. Besides of a fundamental description of entanglement, the highly non-locality when considered form a spacetime viewpoint of the $U_t$ dynamics provides a mechanism for contextuality, avoiding in this way the constraints imposed by Bell's theory and by the Kochen-Specken theorem \cite{Kochen-Specken1967, Isham 1995, Redhead}.
\subsection{On the quantum Young experiment}
\subsubsection{Experimental setting}
In a schematic version of the quantum Young experiment, let us consider a two-dimensional spatial Euclidean space, where $x_1$-direction is the direction of propagation of the quantum particles and the $x_2$ direction is the vertical direction of orientation for the slits. The vertical detector screen are perpendicularly located at a fixed distance from the source, after the slits.  The particles are in a translational state, that corresponds to free particles and define the {\it beam of particles}. We assume that the intensity of the beam can be regulated to allow only for one quantum particle  on flight  each time that the experiment is performed. The states are pure quantum states, describing individual quantum particles.
At a fixed distance from the source there is a screen with two apertures, the slits $1$ and $2$. Otherwise this first screen is un-penetrable for the particles. The apertures are separated by a distance larger than the quantum wavelength associated with the particles. After the screen with the slits, at some distance, there is a second screen (parallel to the first one) with a pointwise detection system where the particles are being detected. This is where the photographic plate for the detection is located, for instance.

 The experiment is repeated a large number of times with different identical quantum particles, under the constraint that the macroscopic initial momenta at the source is the same for each of the particles. We assume that other conditions on the experiment, as the value of the external gravitational field, interactions with the ambient, spin dynamical degrees of freedom and other dynamical properties are, either the same for each of the particles or that the variance of these factors do not affect the outcomes of the experiment.

\subsubsection{Qualitative quantum description of the Young experiment}
The quantum mechanical description of the experiment can be summarized as follows. The quantum system is associated with the slit $1$ (resp. for the slit $2$) if, closing the slit $2$, the particle is detected on the detection screen after some time has passed since the particle was generated. The quantum state associated with the slit $1$ is described by means of a wave function $\psi_1$ such that, if we close the slit $2$, $|\psi_1|^2$ reproduces the statistical patron observed in the detection screen, after the experiment has been repeated a large number of times with identical quantum systems. The slit $2$ has associated an analogous wave function denoted by $\psi_2$. The quantum system is characterized by the initial momentum and by which slits are open. Assuming that the experimental setting is stationary, if the two slits are open, then the quantum mechanical state at the instant $\tau_0$, just after the system goes through the slits, is described by a vector $\psi\in\,\mathcal{H}\simeq L_2(\mathbb{R}^2)$ of the form
\begin{align}
\psi(\vec{x},\tau_0)=\, C\,\left(\psi_1(\vec{x},\tau_0)+\,\psi_2(\vec{x},\tau_0)\right),
\label{initialcondition}
 \end{align}
 with $C$ a normalization real constant such that $\|\psi\|_{L_2}=1$. The evolution after passing the slits is linear and determined by a Schr\"odinger's equation and prescribed boundary or initial conditions,
 \begin{align}
\psi(\vec{x},\tau)=\, C\,\left(U_{\tau}(\tau,\tau_0)\psi_1(\vec{x},\tau_0)+\,U_{\tau}(\tau,\tau_0)\psi_2(\vec{x},\tau_0)\right).
\label{solutiontwoslits}
 \end{align}
  Unitary evolution implies
 \begin{align*}
 |\psi(\vec{x},\tau_0)|^2 =\,|\psi(\vec{x},\tau)|^2 =\,C^2\left( |\psi_1(\vec{x},\tau)|^2+|\psi_2(\vec{x},\tau)|^2+ 2\,Re (\psi^*_1(\vec{x},\tau)\psi_2(\vec{x},\tau))\right) .
 \end{align*}
 The third term describes the characteristic interference patterns. The fundamental properties of these patters can be better understood if we use the polar coordinates for the description of the wave functions,
 \begin{align*}
 \psi_1 =\,|\psi_1|\,e^{arg(\psi_1)},\quad \psi_2 =\,|\psi_2|\,e^{arg(\psi_2)}.
 \end{align*}
If the two slits are placed symmetrically from $x_2=0$, then
symmetry considerations imply that $\psi_1\simeq \,\psi_2$ in the central region of the screen $x_2\approx 0$. Outside of the central axis, one expects that either $|\psi_1|\neq \,|\psi_2|$ or
$arg (\psi_1)\neq \,arg (\psi_2)$ or that both conditions  hold. In the case where there is a relative phase between $\psi_1$ and $\psi_2$, an interference pattern depending on the geometric arrangement of the experiment must appear and, as  we move out from the axis $x_2=0$, a relative phase between $\psi_1$ and $\psi_2$ will initially increase. This increase in the relative phase is translated in a complex exponential modulation of the amplitude of $\psi(x_1,x_2,\tau)$ as a function of the variable $x_2$. Moreover, one also expects that the condition $|\psi_1|\neq |\psi_2|$ holds along the $x_2$-axis, although by symmetric considerations, the distributions must be such that they are in a symmetrically related with respect to reflection around the central axis $x_2=0$.

From the point of view of the operational interpretations of quantum mechanics, there is no further direct interpretation of the two slit experiment: if $\psi_1$ and $\psi_2$ are constructed in the form above described and if the two slits are open, then there is no way to identify a more detailed description without disturbing the system to know by which slit each individual particle passed. Furthermore, the appearance of interference patterns, inexplicable from the classical point of view, indicates that it is not possible to state by which slit the particle passed\footnote{This is clearly not true in the framework Bohm's interpretation and the de 2roglie-2ohm's interpretation of quantum mechanics.}.

\subsubsection{Qualitative description of the quantum Young experiment from the point of view of Hamilton-Randers theory}
The physical description and interpretation of the two slit experiment from the point of view of Hamilton-Randers theory is the following. First, we consider the predecessor state \eqref{predecesorofquantumstate} for the prepared quantum state. The predecessor state must be of the form
\begin{align*}
\Psi(u)=\,\frac{1}{\sqrt{N}}\,\sum^N_{k=1}\,e^{\imath\,\vartheta_{\Psi k}(\varphi^{-1}_{k}(x),z_k)} \,n_{\Psi k}(\varphi^{-1}_{k}(x),z_k)\,|\varphi^{-1}_{k}(x),z_k\rangle ,
\end{align*}
where $x=(x_1,x_2,\tau)$ and $z_k$ stand for the coordinates at ${\bf M}_4$ and the components of the velocity variable associated with the $k$ sub-quantum molecule.
Let us consider first the case when slit 1 is open. Then the predecessor is re-casted in the form
 \begin{align*}
 \Psi_1(u)=\,\frac{1}{\sqrt{N_1}}\,\sum^{N_1}_{1k=1}\, e^{\imath\,\vartheta_{1k}(\varphi^{-1}_{1k}(x),z_{1k})} \,n_{1k}(\varphi^{-1}_{1k}(x),z_{1k})\,|\varphi^{-1}_{1k}(x),z_{1k}\rangle .
 \end{align*}
 Fixed the variable $x\in {\bf M}_4$, under the assumption of ergodicity, the average on time during the cyclic $U_t$ evolution is equivalent to an average on $\varphi^{-1}_{k}*(T_x{\bf M}_4)$. Thus the quantum state for a quantum system if only the slit $1$ is open is of the form
  \begin{align}
 \psi_1(x)=\,\frac{1}{\sqrt{N_1}}\,\sum^{N_1}_{1k=1}\,\int_{\varphi^{-1}_{1k}*(T_x{\bf M}_4)} \,d^4z_{1k}\,  e^{\imath\,\vartheta_{1k}(\varphi^{-1}_{1k}(x),z_{1k})} \,n_{1k}(\varphi^{-1}_{1k}(x),z_{1k})\,|\varphi^{-1}_{1k}(x),z_{1k}\rangle .
 \end{align}
2y a similar argument, if only the slit $2$ is open, then the pre-quantum state is of the form
  \begin{align}
 \Psi_2(u)=\,\frac{1}{\sqrt{N_2}}\,\sum^{N_2}_{2k=1}\,  e^{\imath\,\vartheta_{2k}(\varphi^{-1}_{2k}(x),z_{2k})} \,n_{2k}(\varphi^{-2}_{2k}(x),z_{2k})\,|\varphi^{-1}_{2k}(x),z_{2k}\rangle .
 \label{Estado 1}
 \end{align}
The associated quantum state is
  \begin{align}
 \psi_2(x)=\,\frac{1}{\sqrt{N_1}}\,\sum^{N_1}_{2k=1}\,\int_{\varphi^{-1}_{2k}*(T_x{\bf M}_4)} \,d^4z_{2k}\, e^{\imath\,\vartheta_{2k}(\varphi^{-1}_{2k}(x),z_{2k})} \,n_{2k}(\varphi^{-1}_{2k}(x),z_{2k})\,|\varphi^{-1}_{2k}(x),z_{2k}\rangle  .
 \label{Estado 2}
 \end{align}

 By reasons of symmetry, $N_1=\,N_2=\,N/2$.

 When both slits are open, during the $U_t$ evolution each of the sub-quantum degree of freedom passes through both slits, but since the slits are not the same and the degrees of freedom are moving independently, there is a partition  of the sub-quantum degrees of freedom, either as associated with slit $1$ or associated with slit $2$. A natural classification is that the $k$-ene sub-quantum molecule is in class $1$ if during the $U_t$ evolution it expends more time close to $1$ than to slit $2$.

 The characterization of the system should reflect this partition. Indeed, by applying a suitable form of ergodic theorem, the system is characterized by an average from the contribution arising from associated with $1$ and the corresponding situation associated with $2$. Hence the pre-quantum state is of the form
 \begin{align*}
 \Psi_{12} (u) =\, \frac{1}{\sqrt{2}}\,\left(\Psi_1 (u)+\,\Psi_2 (u)\right),
 \end{align*}
 where the coefficient $\frac{1}{\sqrt{2}}$ can be understood using symmetric considerations. Since $N_1=N_2=N/2$, we have that the associated quantum state is of the form
\begin{align*}
\psi_{12}(x) = &\, \frac{1}{\sqrt{N}}\,\sum^{N_1}_{1k=1}\int_{\varphi^{-1}_{1k}*(T_x{\bf M}_4)} \,d^4z_{1k} e^{\imath\,\vartheta_{1 k}(\varphi^{-1}_{1k}(x),z_{1k})} \,n_{1 k}(\varphi^{-1}_{1k}(x),z_{1k}) |\varphi^{-1}_{1k}(x),z_{1k}\rangle \\
& + \frac{1}{\sqrt{N}}\sum^{N_2}_{2k=1} \int_{\varphi^{-1}_{2k}*(T_x{\bf M}_4)} \,d^4z_{2k} e^{\imath\,\vartheta_{2 k}(\varphi^{-1}_{2k}(x),z_{2k})} \,n_{2 k}(\varphi^{-1}_{2k}(x),z_{2k}) |\varphi^{-1}_{2k}(x),z_{2k}\rangle .
\end{align*}

Note that separating the degrees of freedom in $1$ and $2$ is, in this case nothing more than a partition of the $N$ degrees of freedom describing the particle, which is an unique quantum system. Therefore, the highly oscillating condition \eqref{highlyoscillatorycondition0} cannot be applied to each terms of the partition as individual particles. This has the most remarkable consequences, because the probability of passing through the slit systems is not the sum of the probabilities through $1$ or $2$ separetely, but it appears an interference term. Indeed, the expression for the probability distribution is of the form
\begin{align*}
&|\psi_{12}(x)|^2 =\,\frac{1}{{N}}\,\sum^{N}_{k=1}\int_{\varphi^{-1}_{k}*(T_x{\bf M}_4)} \,d^4z_k \Big( n^2_{1k}(\varphi^{-1}_{1k}(x),z_{1k})+n^2_{2k}(\varphi^{-1}_{2k}(x),z_{2k})\Big)\\
& + \,\sum^{N_1}_{1k=1}\int_{\varphi^{-1}_{1k}*(T_x{\bf M}_4)} \,d^4z_{1k} \sum^{N_2}_{2k=1}\int_{\varphi^{-1}_{2k}*(T_x{\bf M}_4)} \,d^4z_{2k} 2 n_{1k}(\varphi^{-1}_{1k}(x),z_{1k})\, n_{2k}(\varphi^{-1}_{2k}(x),z_{2k})\cdot\\
& \cdot Re \Big(e^{\imath(\vartheta_{1 k}(\varphi^{-1}_{1k}(x),z_{1k})-\vartheta_{2 k}(\varphi^{-1}_{2k}(x),z_{2k}))}\,\langle\varphi^{-1}_{1k}(x),z_{1k}|\varphi^{-1}_{2k}(x),z_{2k}\rangle\Big) .
\end{align*}
Orthogonality relations implies
\begin{align*}
&|\psi_{12}(x)|^2 =\, 1+ \frac{1}{{N}}\,\sum^{N_1}_{1k=1}\int_{\varphi^{-1}_{1k}*(T_x{\bf M}_4)} \,d^4z_{1k}  2 n_{1k}(\varphi^{-1}_{1k}(x),z_{1k})\, n_{2k}(\varphi^{-1}_{1k}(x),z_{1k})\cdot\\
& \cdot Re \Big(e^{\imath(\vartheta_{1 k}(\varphi^{-1}_{1k}(x),z_{1k})-\vartheta_{2 k}(\varphi^{-1}_{1k}(x),z_{1k}))}\Big) .
\end{align*}

For the double slit experiment, the symmetries of the setting implies the existence of nodes, or regions where $|\psi(x)|^2_{12}=0$. If normalized conveniently, one has the relation
\begin{align*}
0=&\, 1+ \frac{1}{{N}}\,\sum^{N_1}_{1k=1}\int_{\varphi^{-1}_{1k}*(T_x{\bf M}_4)} \,d^4z_{1k}  2 n_{1k}(\varphi^{-1}_{1k}(x),z_{1k})\, n_{2k}(\varphi^{-1}_{1k}(x),z_{1k})\cdot\\
& \cdot Re \Big(e^{\imath(\vartheta_{1 k}(\varphi^{-1}_{1k}(x),z_{1k})-\vartheta_{2 k}(\varphi^{-1}_{1k}(x),z_{1k}))}\Big)
\end{align*}
Since $0\leq n_{1k}\leq 1$ and $0\leq n_{2k}\leq 1$, then one needs that, in order that second term compensates the constant term 1,
\begin{align}
\vartheta_{1 k}(\varphi^{-1}_{k}(x),z_k)-\vartheta_{2 k}(\varphi^{-1}_{k}(x),z_k))=\,a\pi,\,a\in \mathbb{Z}
\label{existence of nodes condition 1}
\end{align}
for each $k=1,...,N$ and for each $z_k$.
In addition, the condition
\begin{align}
n_{1k}(\varphi^{-1}_{k}(x),z_k)=\, n_{2k}(\varphi^{-1}_{k}(x),z_k),\,k=1,...,N.
\label{existence of nodes condition 2}
\end{align}

The picture for the interference phenomena in the quantum Young experiment that Hamilton-Randers theory is the following. When both slits are open, the quantum system passes through both slits in the sense that during the $U_t$ evolution each of the sub-quantum degrees of freedom passes through both slits $1$ and $2$. The overlook of the fundamental $U_t$ dynamics in the quantum mechanical description of the Young experiment produces the impression of quantum interference, as if the system passed  by both slits at the same value of the $\tau$-time. This ascription to $1$ and $2$ has as a consequences that the spontaneous natural collapse can happen either close to slit $1$ or close to slit $2$. Also note that in Hamilton-Randers theory, there is no superposition when using the detailed description of double dynamics. The linear combination of states are for mathematical tool, rooted in the Koopman-von Neumann formulation of classical dynamics and in the ergodic property of the underlying $U_t$ dynamics.
Since each sub-quantum molecule passes by each of the slits, one can assign a probability of finding the system at each value of the coordinate $x_2$ in the detection screen is given by the modulus square of the complex amplitude $\psi$, with interference pattern between the terms obtained from the state $\psi_1$ and $\psi_2$. The probabilities assigned are determined by the time that each of the sub-quantum molecule passes close to $1$ or close to $2$ and are given by the Born's rule, as discussed in Chapter \ref{chapter of the Hilbert space structure} and Chapter \ref{chapter on concentration of measure}.

On the other hand, the detection events have always a pointwise character, according to the theory developed in Chapter \ref{chapter on concentration of measure}. From the point of view of Hamilton-Randers theory, it must be defined in terms of sub-quantum degrees of freedom. The properties are well defined in the concentration domain, where there is no superpositions on the macroscopic position of the system: either the system is found at $1$ or at $2$.

If an experimental device is settle to control by which of the slits the particle passes, then the experimental setting changes. In the extreme case, when one of the slits is closed, the result from repeating the experiment many times, the statistical distribution of events at the detection screen corresponds to a pattern as if the particles passed by one of the slits only (the one open) or if it did not passed. In mathematical terms, by the insertion of an observer, the predecessor state changes:
\begin{align*}
 |\Psi_{12}\rangle \to |\Psi'\rangle_1\,\,\textrm{or}\,\, |\Psi_{12}\rangle\to |\Psi'\rangle_2
 \end{align*}
  in such a way that all the sub-quantum degrees are associated either to $1$ or to $2$. Moreover, the averages are such that $\langle \Psi'_1\rangle =\, \langle\psi\rangle_1$ and analogously, $\langle \Psi'_2 \rangle =\,\langle\psi\rangle_2$.
This implies that a measurement by which of the slits the system passes is described quantum mechanically by the projections
\begin{align*}
p_1:\mathcal{H}\to\mathcal{H},\quad \psi \mapsto \langle\psi\rangle_1,\quad\quad
p_2:\mathcal{H}\to\mathcal{H},\quad \psi\mapsto \langle\psi\rangle_2 .
\end{align*}
Such projection transformations are caused by the experimental setting, that changes the nature of the system.

When the measurement device is not close to one of the slits, then there are two consequences for the description of the phenomenology of the quantum Young experiment. First, since the interaction associated with the measurement happens when the system is in the concentration domain of ${\bf D_0}$, a detection has always a well-defined value. Second, even if the system concentrates, let say at $1$ in the fundamental $n$-fundamental cycle, it can happen that it concentrates at the slit $2$ during the next fundamental cycle. The process, although determined by the dynamics at the fundamental level, appears as undetermined from a quantum mechanical description. For a macroscopic observer, that describes the experiment using the quantum formalism, it is not possible to surpass this accuracy in the description. This is because the nature of $t$-time and $\tau$-time are different.

\subsection{Quantum interferometry in presence of an external gravitational field}
 We discuss in this section the interaction of a quantum system with an external  classical external gravitational field from the point of view of Hamilton-Randers theory. It will be shown that an external gravitational interaction is modelled in an analogous way as any other external potential. This could appear as a surprise if we think that in Hamilton-Randers theory gravity corresponds to the dominant interaction only in the $1$-Lipschitz domain of the $U_t$ flow and that, therefore, gravity emerges when the system is naturally spontaneously collapsed and localized.

Let us consider an ideal interferometer setting composed by two classical paths $S1D$ and $S2D$ such that the location of each path are located at different potentials of the Earth's surface gravitational field. The details can be found in  \cite{ColellaOverhauser1974, ColellaOverhauserWerner1975} or the exposition developed in \cite{Sakurai}, from where we adopt the notation. It can be shown that for a non-relativistic quantum particle that arrives at the screen and that interferences with itself by being allowed to follow each of the classical paths $S1D,\,S2D$, the phase difference along the paths is given by the expression
\begin{align}
\Phi_{S1D}\,-\Phi_{S2D}=\,\frac{m^2}{\hbar^2}\, g_0\,\frac{\lambda}{2\pi}\,\left(l_h\,l_v\,\sin \delta\right),
\label{interferencespattern}
\end{align}
where $l_h$ is the length of the horizonal tram, $l_v$ is the length of the vertical and $\delta$ is the angle of inclination of the interference plane with the horizontal, varying from $0$ to $\pi/2$.
In the derivation of this expression it is applied the path integral approach to quantum amplitudes, the local expression for the gravitational potential of the Earth and de Broglie relation between the momenta and wave vectors.

The particle has two possible classical paths where it can naturally spontaneously concentrate. If the system describing the particle concentrates at one point along the classical path $S1D$, then the local classical gravitational potential that the quantum particle feels is $\mathfrak{V}(x_{S1D})$ at the point $x_{S1D}$ over the path $S1D$ and similarly if it concentrates along the second path $S2D$. If no measurement is performed on the system, then the possible concentration region is $S1D\cup S2D$. Furthermore, the localization of the quantum particle can switch between the paths $S1D$ to $S2D$ and viceversa, since the evolution is not necessarily continuous in spacetime, but it has jumps between these paths that must be compatible with the conserved quantities of the fundamental $U_t$ dynamics. For the particular case under consideration, the jumps between paths must be compatible the with momentum and energy conservation laws.

Despite the jumps, since the gravitational potential $\mathfrak{V}$ is local in the spacetime manifold, it is implemented in the standard way in the quantum evolution operator formalism.

Following this interpretation, although the system is collapsed at any instant that interacts gravitationally with the gravitational potential $V$ (but this collapse is not caused by such an interaction), there is no smooth classical trajectory defined by the system, since the path where the system naturally collapses can fluctuate between the classical paths $12D$ and $1CD$. Only when the size of the system makes the fluctuations to be small compared with the notion of center of mass motion, then to speak of classical trajectory as an approximation becomes a good description of the evolution.

\subsection{Models of gravity with fluctuating sources}
 According to Hamilton-Randers theory viewpoint, quantum systems do not follow a classical path, but a fluctuating path among the possible classical paths. This view has consequences for the possible models of sources for gravitation. Let us consider an interferometric setting and let us assume that the quantum interferometric system acts as the source for a gravitational field, given as solution of the Einstein field equations,
\begin{align}
G_{\mu\nu}(x)=\, \frac{8\,\pi\,G_N}{c^4}\, T_{\mu\nu}(x),\quad x\in \,{\bf M}_4,
\end{align}
where $T_{\mu\nu}$ is the energy-momentum tensor of an individual quantum system experiencing the interference.
Besides the problem of the consistence use of distributional sources for the matter fields, the source of the gravitational field appears as a fluctuating source between two paths $S1D$ and $S2D$ compatible with the conservation quantities of the fundamental mechanics, that also lead to certain conserved quantum mechanical quantities, as discussed in Chapter \ref{chapter of the Hilbert space structure}. Since the quantum system follows a zig-zag trajectory between the classical paths $S1D$ and $S2D$, it is natural to describe the path by a random combination of paths,
\begin{align}
\gamma^\mu(\tau)=\,R_1(\tau)\,\gamma^{\mu}_{S1D}(\tau)+\,R_2(\tau)\,\gamma^{\mu}_{S2D}(\tau),
\label{random world line}
\end{align}
where $\gamma^{\mu}_{S1D}:\mathbb{R}\to \mathbb{R}$ and $\gamma^\mu_{S2D}:\mathbb{R}\to \mathbb{R}$ are coordinates along the spacetime paths $S1D$ and $S2D$ respectively. The functions $R_1,R_2:\mathbb{R}\to \{0,1\}$ are constrained by the condition $R_{1}+R_2 =1$. Since the fundamental $U_t$ dynamics is very complex, sensitive to environment and initial conditions, the evaluation of each individual trajectory from first principles $\gamma^\mu(\tau)$ is technically impossible. Therefore, to describe effectively spacetime trajectory associate to the quantum system, it is assumed that the functions $R_1$ and $_2$ are random for all practical purposes.

It is difficult to implement this matter source in Einstein equations, among other things, because Einstein equations do not admit $1$-dimensional distributional sources in the Schwartz distributions. A way to avoid this issue is a way to avoid the problem of considering low dimensional distributions matters. Two ways to do this is first, to consider linearised Einstein equations with fluctuating sources and second, to consider the Newtonian limit.

By considering the Newtonian case, we obtain a framework were to observe the characteristics of the models of fluctuating gravity in Hamilton-Randers theory. For this it is necessary to implement as a source the mass density with a fluctuating $1$-dimensional distrubution mass density.  A general covariant formulation along the lines of Cartan can be introduced, but it is enough to establish the geometric setting if we consider that a time coordinate $\tau$ is given and that the manifold is foliated in the form ${\bf M}_4\cong \mathbb{R}\times \mathbb{R}^3$.

Proceeding in this way, at each fixed value of the time parameter $\tau =\,\tau_0$, the Poisson equation can be formulated,
\begin{align*}
\nabla \mathfrak{V} =\,4\pi\,G_N \delta \left(x^\mu-R_1(\tau_0)\,\gamma^{\mu}_{S1D}(\tau)-\,R_2(\tau_0)\,\gamma^{\mu}_{S2D}(\tau_0)\right).
\end{align*}
But this equation is equivalent to a  Poisson equation for a random superposition of sources,
\begin{align}
\nabla \mathfrak{V} =\,4\pi\,G_N\,\left( R_1(\tau_0)\,\delta\left( x^\mu-\gamma^{\mu}_{S1D}(\tau_0)\right)+\,R_2(\tau)\delta\left( x^\mu-\gamma^{\mu}_{S2D}(\tau_0)\right)\right).
\label{newtonian fluctuating model of gravity}
\end{align}
Because the linearity of the $3$-dimensional Laplacian operator, the solution is direct in the form
\begin{align*}
\mathfrak{V}(\tau,\vec{x}) = \,R_1(\tau)\,\mathfrak{V}_{S1D}(\vec{x})+\,\left(1-R_1(\tau)\right)\,\mathfrak{V}_{S2D}(\vec{x}),
\end{align*}
 The solution $\mathfrak{V}$ mimics the random path and becomes a random gravitational potential. Clearly, the structure of this solution resorts on the sparation in space coordinates $\vec{x}$ and time coordinates $\tau$. It is a non-relativistic solution. On the other hand, it is not necessarily discontinuous, depending on the closeness on the paths $S1D$ and $S2D$
 For classical situations, one has that $S1D\approx S2D$ in a very good degree for any potential pair of paths where natural spontaneous collapse can happen, while for quantum particles, the fluctuating nature of the gravitational field becomes apparent. The gravitational field is described by means of the metric component, since the Newtonian potential is directly related with the $g_{00}$ component of the metric.

As discussed in {\it Chapter} \ref{chapter on concentration of measure}, quantum interactions are associated with the interactions between sub-quantum degrees of freedom during the ergodic regime of the $U_t$ dynamics, while classical interactions are associated with interactions active predominantly on the concentration domain. Therefore, the models discussed above of fluctuating gravity confirm the classical character of gravitation.

\subsubsection{Compatibility with general covariance} The gravitational fluctuating models as the one discussed above are compatible with general covariance. In order to clarify this issue, let us consider the analogies and differences with superposition of spacetimes. Under the assumption that all measurements are fundamentally position measurements of some short, the argument for the absence of super-position of spacetimes as discussed in \cite{Ricardo 2022} or in {\it Appendix} \ref{Heisenberg uncertainty principle against general covariance} applies. For gravity with fluctuating sources, however, the models are fully spacetime general covariant, since the notion of spacetime is the usual one. With a fluctuating distributional source, the metric is fluctuating as well, but it is not a fluctuation among different spacetime manifolds, as long as the fluctuations are such that the manifold where the metrics live is keep the same. Thus the problem of identifying points from different manifolds in a general covariant way, that was on the basis of Penrose construction of gravitationally induced wave function collapse model, is not present in our theory. Fluctuations and general covariance are therefore, subjected to the constraint of being formulated in the same spacetime manifold ${\bf M}_4$.

The physical significance of the fluctuation is associated with the ascription of the proper time functional. Since the potential $\mathfrak{V}$ is directly associated with the $g_{00}$ component of the metric, the metric itself becomes fluctuating and hence, the proper time becomes fluctuating. Thus the fluctuations are of the a ideal proper clock.

\subsection{On the nature of non-local spacetime quantum correlations in general}\label{nature of the quantum correlations}
 In the framework of Hamilton-Randers theory there is a natural mechanism to describe in terms of the sub-quantum dynamics the nature of the spacetime non-local quantum correlations.  Let us consider an entangled state with two parties $A$ and $B$. During the $U_t$ evolution, the sub-quantum degrees of freedom of $A$ and $B$ {\it fill} the part of the phase space $T^*TM$ compatible with the corresponding causality conditions of the $U_t$ dynamics, that is, the limitation due to the local maximality of the speed of light in vacuum. During such dynamical filling of the phase space, the degrees of freedom of $A$ and $B$ interact. Hence they can correlate their $t$-time  evolution. Once the projection $(t,\tau)\mapsto \tau$ is imposed in the mathematical description, the information on the details of the $U_t$ dynamics is unaccessible. This interacting mechanism combined with the formal projection is the origin of the quantum non-local correlations.

 There are two main consequences of this theory on the nature of quantum correlations. The first is that quantum correlations appear as instantaneous for any macroscopic observer, despite that the ergodic motion of the sub-quantum molecules is constrained by hypothesis to the sub-luminal speed condition $\vec{v}<\,c$. Let us remark that this notion of instantaneous non-locality is compatible with relativistic physics, in the sense that it is a notion independent of the reference system and do not implies spacetime propagation of energy or signaling, since signaling is a macroscopic effect. Also, note that the mechanism explains why quantum correlations are apparently instantaneous in the $1$-dimensional $\tau$-time description, but the mechanism is not instantaneous when the two dimensional description $(t,\tau)$ of the full Hamilton-Randers dynamics.

 The second consequence is that, due to the sub-luminal constraint in the sense that $|v_k|<\, c$ for each sub-quantum degrees of freedom. First, let us recall that the metric structure of each of the manifolds $M_k$ are isometric to each other, that is, there are isometries relating the metrics $\eta_{k_1}= I_{12}(\eta_{k_2})$, for each $k_1, k_2$. On the other hand, the spacetime structure $({\bf M}_4, \eta_4)$ is not isometric to $(M_k,\eta_k)$. However, given a Hamilton-Randers dynamical system, all the metric structures describing the sub-quantum degrees of freedom are isometric and hence, the relation between the distances functions in each of them and the distance in ${\bf M}_4$ are characterized by the system only, not by each degree of freedom. Thus the macroscopic distance range for the existence of instantaneous quantum correlations is bounded by the relation
\begin{align*}
d_{cor}\leq \,c\,T ,
 \end{align*}
 where $T$ is the semi-period of the quantum system. Note that this relation is the same for every Hamilton-Randers system, since the left side is constructed by means of geometric elements defined over ${\bf M}_4$ as discussed in {\it section} \ref{Section on macroscopic observers}, while the right hand side is constructed by means of isometric structures that only depend upon the system itself.

 By the relation \eqref{ETrelation}, the bound $d_{cor}$ on the range of the non-local interactions depends on the mass $m$ of the quantum system in question,
 \begin{align}
d_{cor}\leq \,c \,T_{min}\,\exp\left[\frac{T_{min}\,m \,c^2}{\hbar}\right].
\end{align}
\label{boundfordistanceofcorrelations1}
For massless quantum systems, the expression  \eqref{boundfordistanceofcorrelations1} reduces to
  \begin{align}
d_{cor}\leq \,c \,T_{min}.
\label{boundfordistanceofcorrelations2}
 \end{align}
It is interesting to mention the kind of dependence on the mass in the expression \eqref{boundfordistanceofcorrelations1}. It implies that quantum non-local correlations will vanish quicker for massless systems than for massive systems.

Therefore, Hamilton-Randers theory implies the existence of a maximal range for the distance where photon pairs can exhibit quantum correlations.

 Both expressions \eqref{boundfordistanceofcorrelations1} and \eqref{boundfordistanceofcorrelations2} are falsifiable. If a particular value of $T_{min}$ is determined experimentally, for instance, by measuring a maximal range of correlations for photon pairs, then also the relation \eqref{boundfordistanceofcorrelations1} could  be tested for systems of different inertial masses. Indeed, according to our theory, the following relation
 \begin{align*}
 \frac{\Delta d_{cor}}{d_{cor}} =\,\frac{T_{min} \,c^2}{\hbar}\,\Delta m
 \end{align*}
 must hold good for small difference on mass $\Delta m$. Furthermore, the pattern by which quantum correlations has a limited range and the maximal range is given in terms of the relation \eqref{boundfordistanceofcorrelations1} is {\it per se} a falsifiable prediction of a sector of Hamilton-Randers theory.

\subsection{Entangle states as emergent states}
We discussed in Chapter \ref {chapter of the Hilbert space structure}, in particular in {\it proposition} \ref{extension of emergene to wave packets} that any wave function, that is, for a element of the quantum Hilbert space $\mathcal{H}$,  can be represented as an emergent state. We illustrate this by showing that one particle free quantum states of any representation of the Lorentz group, are emergent states. We will consider now in detail the case of entangled states from the point of view of Hamilton-Randers theory to show that they are emergent states.

Let us consider a predecessor state describing two different systems $A$ and $B$, but of identical type of systems, for instance, two photons or two electrons. From our theory of inertial mass as developed in {\it section} \ref{models of periods in terms of mass}. For each of the systems, the predecessor states are of the form
\begin{align*}
&\Psi_A=\,\frac{1}{\sqrt{N_A}}\,\sum^{N_A}_{k=1}\Big(n_{1k}(\varphi^{-1}_k(x),z_k)\,\exp(\imath \,\vartheta_{1k}(\varphi^{-1}_k(x),z_k))\break\\
 &\nonumber +\,e^{\imath \theta_A}\,n_{2k}(\varphi^{-1}_{k}(x),z_k)\,\exp(\imath \,\vartheta_{2k}(\varphi^{-1}_k(x),z_k))\Big)|\varphi^{-1}_{k}(x),z_k\rangle,
\end{align*}
and
\begin{align*}
&\Psi_B=\,\frac{1}{\sqrt{N_B}}\,\sum^{N_B}_{l=1}\Big(n_{1l}(\varphi^{-1}_l(x),z_l)\,\exp(\imath \,\vartheta_{1l}(\varphi^{-1}_l(x),z_l))\break\\
 &\nonumber +\, e^{\imath \theta_B}\,n_{2l}(\varphi^{-1}_{l}(x),z_l)\,\exp(\imath \,\vartheta_{2l}(\varphi^{-1}_l(x),z_l))\Big)|\varphi^{-1}_{l}(x),z_l\rangle
\end{align*}
respectively, where $\theta_A,\theta_B$ are arbitrary, constant phases. In these expressions the indices $k$ and $l$ run over different sets of sub-quantum degrees of freedom associated with two quantum particles referred by $A$ and $B$. The labels $1$ and $2$ stand for two different spacetime locations $x_1, x_2\in {\bf M}_4$, where measurements take place. Thus physical properties for each of the particles $A$ and $B$ can be measured in the location $1$ or in the location $2$. The result of the localization is determined by the details of the $U_t$ dynamics and is determined by the processes of concentration of measure as discussed in Chapter \ref{chapter on concentration of measure}. Note that although sub-quantum molecules labeled by $1$ are different from the sub-quantum molecules labeled by $2$, during the ergodic regime of the $U_t$ evolution, there is an interaction  of the degrees of freedom. Although sub-quantum molecules labeled by $1$ expend most of their time near the location $1$, they also expend time at $2$ and in between $1$ and $2$, because of the ergodic properties of the $U_t$ dynamics.

 The pre-quantum state associated with $A\sqcup B$ is defined to be the product state
 \begin{align*}
\Psi_{A\otimes B}=\Psi_A\otimes \Psi_B,
\end{align*}
The product state is re-casted as
 \begin{align*}
&\Psi_{A\otimes B}(u)= \,\frac{1}{\sqrt{N_A}}\,\frac{1}{\sqrt{N_B}}\sum^{N_A}_{k=1}\sum^{N_B}_{l=1}
\Big(n_{1k}(\varphi^{-1}_k(x),z_k)\,\exp(\imath \,\vartheta_{1k}(\varphi^{-1}_k(x),z_k))\\
&+\,e^{\imath\theta_A}\,n_{2k}(\varphi^{-1}_{k}(x),z_k)\,\exp(\imath \,\vartheta_{2k}(\varphi^{-1}_k(x),z_k))\Big)\\
 &\Big(n_{1l}(\varphi^{-1}_l(x),z_l)\,\exp(\imath \,\vartheta_{1l}(\varphi^{-1}_l(x),z_l)) +\,e^{\imath\theta_B}\,n_{2l}(\varphi^{-1}_{l}(x),z_l)\,\exp(\imath \,\vartheta_{2l}(\varphi^{-1}_l(x),z_l))\Big)\\
 & \cdot |\varphi^{-1}_{k}(x),z_k\rangle\otimes|\varphi^{-1}_{l}(x),z_l\rangle
\end{align*}
The corresponding quantum state, obtained by averaging the predecessor state $\Psi_{A\otimes B}$ is
\begin{align*}
&\psi_{AB}(x)= \,\frac{1}{\sqrt{N_A}}\,\frac{1}{\sqrt{N_B}}\,\sum^{N_A}_{k=1}\sum^{N_B}_{l=1}  \int_{\varphi^{-1}_{k}*(T_x{\bf M}_4)}\,d^4z_k \, \int_{\varphi^{-1}_{l*}(T_x{\bf M}_4)}\,d^4z_l\\
&\Big(n_{1k}(\varphi^{-1}_k(x),z_k)\,\exp(\imath \,\vartheta_{1k}(\varphi^{-1}_k(x),z_k))+\,e^{\imath\theta_A}n_{2k}(\varphi^{-1}_{k}(x),z_k)\,\exp(\imath \,\vartheta_{2k}(\varphi^{-1}_k(x),z_k))\Big)\\
 &\Big(n_{1l}(\varphi^{-1}_l(x),z_l)\,\exp(\imath \,\vartheta_{1l}(\varphi^{-1}_l(x),z_l)) +\,e^{\imath\theta_B}\,n_{2l}(\varphi^{-1}_{l}(x),z_l)\,\exp(\imath \,\vartheta_{2l}(\varphi^{-1}_l(x),z_l))\Big)\\
 & \cdot |\varphi^{-1}_{k}(x),z_k\rangle\otimes|\varphi^{-1}_{l}(x),z_l\rangle
\end{align*}

 We are interested in discussing conditions under which the predecessor state $\Psi_{A\otimes B}$ has associated a quantum state $\psi_{AB}$ of the form
 \begin{align}
 \psi_{AB}(x)=\,\psi_{1A}(x)\otimes\,\psi_{2B}(x)+\,e^{\imath\theta}\,\psi_{2A}(x)\otimes\,\psi_{1B}(x),
 \label{entanglestate}
 \end{align}
 where $\theta$ is a given relative phase and the states are conveniently normalised such that $|\psi_{AB}|=1$.
 Thus although the predecessor state is a product, the quantum effective state is an entangled state. This condition reads in full generality as
  \begin{align}\label{entangledconditions}
&\nonumber\sum^{N_A}_{k=1}\sum^{N_B}_{l=1}  \int_{\varphi^{-1}_{k}*(T_x{\bf M}_4)}\,d^4z_k \, \int_{\varphi^{-1}_{l*}(T_x{\bf M}_4)}\,d^4z_l\Big(n_{1k}(\varphi^{-1}_k(x),z_k)\,\exp(\imath \,\vartheta_{1k}(\varphi^{-1}_k(x),z_k))n_{1l}(\varphi^{-1}_{l}(x),z_l)\\
&\nonumber\exp(\imath \,\vartheta_{1l}(\varphi^{-1}_l(x),z_l)) + \,e^{\imath (\theta_A+\theta_B)}\,n_{2k}(\varphi^{-1}_k(x),z_k)\,\exp(\imath \,\vartheta_{2k}(\varphi^{-1}_k(x),z_k))n_{2l}(\varphi^{-1}_{l}(x),z_l)\cdot\\
&\nonumber\,\exp(\imath \,\vartheta_{2l}(\varphi^{-1}_l(x),z_l))\Big)|\varphi^{-1}_{k}(x),z_k\rangle\otimes|\varphi^{-1}_{l}(x),z_l\rangle=0 .\\
 \end{align}
 These conditions imply that, although the predecessor state is a product, there are interactions between the subquantum degrees of freedom of $A$ and the subquantum degrees of freedom of $B$. Therefore, the systems $A$ and $B$ are not separated in the Hamilton-Randers sense. Indeed, when the system $A\sqcup B$ has a common origin, the non-separability at the sub-quantum level could appear as a reasonable hypothesis, while when they do not have a common origin, it is reasonable to assume separability. However, a main difference between Hamilton-Randers theory and quantum mechanics is the degeneration of the non-separability with time, as it has been discussed above in {\it section} \ref{nature of the quantum correlations}.

If the condition \eqref{entangledconditions} holds, then $\psi_{AB}$ is of the form
\begin{align}
\psi_{AB}(x)\,=&\,\frac{1}{\sqrt{N_A}}\,\frac{1}{\sqrt{N_B}}\,\Big(\sum^{N_A}_{k=1} \int_{\varphi^{-1}_{k}*(T_x{\bf M}_4)}\,d^4z_k \,n_{1k}(\varphi^{-1}_k(x),z_k)\,\exp(\imath \,\vartheta_{1k}(\varphi^{-1}_k(x),z_k))|\varphi^{-1}_{k}(x),z_k\rangle\Big)\break\\
& \nonumber \otimes e^{\imath\,\theta_B}\,\Big(\sum^{N_B}_{l=1} \int_{\varphi^{-1}_{l*}(T_x{\bf M}_4)}\,d^4z_l\, \cdot n_{2l}(\varphi^{-1}_{l}(x),z_l)\,\exp(\imath \,\vartheta_{2l}(\varphi^{-1}_l(x),z_l))|\varphi^{-1}_{l}(x),z_l\rangle\Big)\break\\
 &\nonumber +\Big(\sum^{N_A}_{k=1} \int_{\varphi^{-1}_{k}*(T_x{\bf M}_4)}\,d^4z_k \,n_{2k}(\varphi^{-1}_k(x),z_k)\,\exp(\imath \,\vartheta_{2k}(\varphi^{-1}_k(x),z_k))|\varphi^{-1}_{k}(x),z_k\rangle\Big)\break\\
& \nonumber\otimes e^{\imath\,\theta_A}\,\Big(\sum^{N_B}_{l=1} \int_{\varphi^{-1}_{l*}(T_x{\bf M}_4)}\,d^4z_l\, \cdot n_{1l}(\varphi^{-1}_{l}(x),z_l)\,\exp(\imath \,\vartheta_{1l}(\varphi^{-1}_l(x),z_l))|\varphi^{-1}_{l}(x),z_l\rangle\Big),\\
\label{entangledstate2}
\end{align}
which is of the form  \eqref{entanglestate} modulo a global phase, if one defines the states:
\begin{align*}
\begin{cases}
&  \psi'_{1A}(x)=\,\sum^{N_A}_{k=1} \int_{\varphi^{-1}_{k}*(T_x{\bf M}_4)}\,d^4z_k \,n_{1k}(\varphi^{-1}_k(x),z_k)\,\exp(\imath \,\vartheta_{1k}(\varphi^{-1}_k(x),z_k))|\varphi^{-1}_{k}(x),z_k\rangle,\break\\
&\nonumber \psi'_{2A}(x)=\,\sum^{N_A}_{k=1} \int_{\varphi^{-1}_{k}*(T_x{\bf M}_4)}\,d^4z_k \,n_{2k}(\varphi^{-1}_k(x),z_k)\,\exp(\imath \,\vartheta_{2k}(\varphi^{-1}_k(x),z_k))|\varphi^{-1}_{k}(x),z_k\rangle,\break\\
&\nonumber \psi'_{1B}(x)=\,\sum^{N_B}_{l=1} \int_{\varphi^{-1}_{l*}(T_x{\bf M}_4)}\,d^4z_l\, \cdot n_{1l}(\varphi^{-1}_{l}(x),z_l)\,\exp(\imath \,\vartheta_{1l}(\varphi^{-1}_l(x),z_l))|\varphi^{-1}_{l}(x),z_l\rangle,\break \\
&\nonumber \psi'_{2B}(x)=\,\sum^{N_B}_{l=1} \int_{\varphi^{-1}_{l*}(T_x{\bf M}_4)}\,d^4z_l\, \cdot n_{2l}(\varphi^{-1}_{l}(x),z_l)\,\exp(\imath \,\vartheta_{2l}(\varphi^{-1}_l(x),z_l))|\varphi^{-1}_{l}(x),z_l\rangle .
\end{cases}
\end{align*}

We have not considered the norms of the states yet. However, let us  consider
\begin{align*}
\begin{cases}
&\psi_{1A}(x)=\frac{1}{\sqrt{N_{1A}}}\psi'_{1A}(x),\\
&\,\psi_{2A}(x)=\frac{1}{\sqrt{N_{2A}}}\psi'_{2A}(x),\\
&\,\psi_{1B}(x)=\frac{1}{\sqrt{N_{1B}}}\psi'_{1B}(x),\\
&\,\psi_{2B}(x)=\frac{1}{\sqrt{N_{2B}}}\psi'_{2B}(x),
\end{cases}
\end{align*}
By the normalization \eqref{normalization of emergent states}, each of these states is normalized to $1$.
Then the entangled state \eqref{entangledstate2} is re-cast as
\begin{align}
\psi_{AB}(x)=\,\frac{\sqrt{N_{1A}N_{2B}}}{\sqrt{N_{A}N_{B}}}\,e^{\imath\theta_B}\psi_{1A}\otimes \psi_{2B}+\,\frac{\sqrt{N_{1B}N_{2A}}}{\sqrt{N_{A}N_{B}}}\,e^{\imath\theta_A}\psi_{1B}\otimes \psi_{2A} .
\end{align}
Let us note that by properly defining the densities $n_{1k},\,n_{2k},\,n_{1l},\,n_{2l}$, the fractions $\frac{\sqrt{N_{1A}N_{2B}}}{\sqrt{N_{A}N_{B}}}$ and $\frac{\sqrt{N_{1_B}N_{2A}}}{\sqrt{N_{A}N_{B}}}$ are in the domain $[0,1]$.

 Therefore, we have the following result:
 \begin{teorema}
 Any entangled state of the form
 \begin{align}
 \psi(x)=\,\alpha\psi_{1A}\otimes\psi_{2B}+\,\beta\psi_{1B}\otimes \psi_{2A}
 \label{general entangled state}
\end{align}
such that $|\alpha|^2+\,|\beta |^2 =\,1$ is an emergent state from a predecessor state.
 \end{teorema}
 The modulus $|\alpha_1|$ and $|\alpha_2|$ can be different; they  can be modelled by the different partitions $(n_{1l},n_{2l}),\,(n_{1k},n_{2k})$.

The emergent character of pure entangled states as described above can be generalized to other entangled states and to systems with several particles and observers. Any emergent state can be obtained by an analogous emergent mechanism.

 \subsection{Consequence from the Hamilton-Randers theory of quantum entanglement}
 The above discussion suggests the following notion of {\it macroscopic separable system},
\begin{definicion}
Two Hamilton-Randers systems $A\sqcup B$ are {\it macroscopically separated} if the associated quantum state $\psi =\langle \Psi\rangle_t$ is a product state.
\end{definicion}
The notion beneath this definition is different from the notion of non-interacting Hamilton-Randers systems as formalized in {\it Definition} \ref{disjoint union of non-interaction} or in {\it Definition} \ref{non-interacting Hamilton-Randers systems}. This is because the heritage mechanism of separation from microscopic to macroscopic description discussed in {\it section}  \ref{nature of the quantum correlations},
\begin{proposicion}
If two systems are separated in the Hamilton-Randers sense, then they are macroscopically separated.
\end{proposicion}

Let us consider a quantum system composed by two particles spatially separating in opposite directions, but originally with the same origin. The situation that we should keep in mind is the Einstein-Podolsky-Rosen experiment. If the $U_t$ evolution of sub-quantum degrees of freedom provides the mechanism responsible for quantum correlations and entanglement, then a processes of the form
\begin{align}
\psi_{AB}\to \rho_{AB}
\label{universal process of long t evolution}
\end{align}
from entangled to macroscopic mixed state of product states, must happen generically for large enough $\tau$-time evolution. The suppression of the quantum correlations is caused by the sub-luminal kinematical constraint on the $U_t$ dynamics of the sub-quantum molecules and the preservation of the period of the cycles. Thus, according to Hamilton-Randers theory, the quantum correlations decay with $\tau$-time, as it was discussed in {\it section} \ref{nature of the quantum correlations}.

On the other hand, the microscopic mechanism for the emergence of the macroscopic separation is microscopic separation. This is because for large evolution time, on a system spacetime separating, the conditions \eqref{entangledconditions} are not satisfied.

The mechanism has the following realization. Let us consider two points $1$ and $2$ of the spacetime point. If $A$ and $B$ could spontaneously collapse to $1$ and $2$, then $\psi_A =\beta_{1A} \psi_1+\beta_{2A}\psi_2$, $\psi_B =\beta_{1B} \psi_1+\beta_{2B}\psi_2$. On the other hand, the states $A$ and $B$ cannot naturally spontaneously collapse necessarily in different points $1$ and $2$ of the spacetime in a consistent way with conservation laws. Thus for a product state in the microscopic sense of the form $\Psi_{A\otimes B}$, the state is of the form translated to a mixed state of the form
\begin{align*}
\rho_{AB}=\,\frac{1}{2}\,|\psi_{1A}\rangle\otimes \,|\psi_{2B}\rangle\langle \psi_{2B}|\otimes \langle \psi_{1A}|+\,\frac{1}{2}\,|\psi_{1B}\rangle|\otimes \,|\psi_{2A}\rangle\langle \psi_{2A}|\otimes \langle \psi_{1B} |.
\end{align*}

Let us highlight two qualitative properties of the mechanism described above:
\begin{itemize}

\item Individually, the particles follow a defined path close to $1$ or close to $2$, but there are no oscillations between $1$ and $2$.

\item Note that the transition \eqref{universal process of long t evolution} happens suddenly, at most, in a $t$-time period of evolution.

\end{itemize}

 \subsection{Expectation values of operators when the quantum state is an entangled state}
Given two points $x_1$ and $x_2$ of the spacetime, let us consider the Hilbert spaces
\begin{align*}
\begin{cases}
& \mathcal{H}_1:=\,\langle |p_1\rangle, \,\imath \frac{\partial}{\partial x}|_{x=x_1}\,|p_1\rangle =\,p_1|p_1\rangle \,\rangle \\
& \mathcal{H}_2:=\, \langle |p_2\rangle, \,\imath \frac{\partial}{\partial x}|_{x=x_1}\,|p_2\rangle =\,p_2|p_2\rangle \,\rangle
\end{cases}
\end{align*}
Assuming that $\mathcal{H}_A\otimes \mathcal{H}_B \cong \,\mathcal{H}_1\otimes \mathcal{H}_2$, we can consider the expectation value of product of linear endomorphisms $O_1\otimes O_2:\mathcal{H}_1\otimes \mathcal{H}_2\to \mathcal{H}_1\otimes \mathcal{H}_2$ viewed as an endomorphism  $O_1\otimes O_2 :\mathcal{H}_A\otimes\mathcal{H}_B\to \,\mathcal{H}_A\otimes \mathcal{H}_B$.
On the other hand, the linear  operators $O_1$ and $O_2$ have associated linear operators $\bf{O}_1:\mathcal{H}_1\to\mathcal{H}_1$, ${\bf O}_2:\mathcal{H}_2\to\mathcal{H}_2$ determined by the relations
\begin{align*}
\begin{cases}
&\langle |O_1|\psi_A\rangle =\, \langle \Psi_A|{\bf O}_1|\Psi_A\rangle \\
&\langle \psi_B|O_2|\psi_B\rangle =\, \langle \Psi_B|{\bf O}_2|\Psi_B\rangle .
\end{cases}
\end{align*}
Applying the orthogonality conditions \eqref{entangledconditions} anf the highly oscillating conditions \eqref{highlyoscillatorycondition1}, the expectation value of the operator ${O}_1\otimes {O}_2$ is of the form
 \begin{align}
 \label{average value of product operators and entanglement}
 \langle \psi_{AB}|{ O}_1\otimes {O}_2 |\psi_{AB}\rangle & =\sum^{N_A}_{k=1}\sum^{N_B}_{l=1}  \int_{\varphi^{-1}_{k}*(T_x{\bf M}_4)}\,d^4z_k \, \int_{\varphi^{-1}_{l*}(T_x{\bf M}_4)}\,d^4z_l\,n^2_{1k}(\varphi^{-1}_k(x),z_k)\cdot\\
 \nonumber & \cdot n^2_{2l}(\varphi^{-1}_l(x),z_l)\langle\varphi^{-1}_{k}(x),z_k|{\bf O}_1|\varphi^{-1}_{k}(x),z_k\rangle\,\langle
 \varphi^{-1}_{l}(x),z_l|{\bf O}_2|\varphi^{-1}_{l}(x),z_l\rangle +\\
 \nonumber & e^{\imath\theta}\,\sum^{N_A}_{k=1}\sum^{N_B}_{l=1}  \int_{\varphi^{-1}_{k}*(T_x{\bf M}_4)}\,d^4z_k \, \int_{\varphi^{-1}_{l*}(T_x{\bf M}_4)}\,d^4z_l\,n^2_{2k}(\varphi^{-1}_k(x),z_k)\cdot\\
 \nonumber & \cdot n ^2_{1l}(\varphi^{-1}_l(x),z_l)\langle\varphi^{-1}_{k}(x),z_k|{\bf O}_2|\varphi^{-1}_{k}(x),z_k\rangle\,\langle
 \varphi^{-1}_{l}(x),z_l|{\bf O}_1|\varphi^{-1}_{l}(x),z_l\rangle .
 \end{align}

Let us consider the expectation value of the product of two spin operators $\vec{\sigma}\cdot\vec{u}_1$, $\vec{\sigma}\cdot\vec{u}_2$ associated with spin measurements of the particles $A$ and $B$. Let $\vec{\Sigma}\cdot \vec{u}_1:\mathcal{H}_{Fun}(A)\to \mathcal{H}_{Fun}(B)$ and $\vec{\Sigma}\cdot\vec{u}_2:\mathcal{H}_{Fun}(B)\to\mathcal{H}_{Fun}(B)$ be the microscopic operators associated with $\vec{\sigma}\cdot\vec{u}_1$, $\vec{\sigma}\cdot\vec{u}_2$ respectively. These sub-quantum linear operators are determined by the relations
\begin{align*}
\begin{cases}
&\langle \psi_A|\vec{\sigma}\cdot\vec{u}_1|\psi_A\rangle =\, \langle \Psi_A|\vec{\Sigma}\cdot\vec{u}_1|\Psi_A\rangle \\
&\langle \psi_B|\vec{\sigma}\cdot\vec{u}_2|\psi_B\rangle =\, \langle \Psi_B|\vec{\Sigma}\cdot\vec{u}_2|\Psi_B\rangle
\end{cases}
\end{align*}
We observe that the choice by a macroscopic observer of the different directions of measuring spin has its direct imprint on the form of the sub-quantum operators  $\vec{\Sigma}\cdot\vec{u}_1$ and $\vec{\Sigma}\cdot\vec{u}_1$.

 The expectation value for the product of the operators $\vec{\sigma}\cdot\vec{u}_1$ and $\vec{\sigma}\cdot\vec{u}_2$ is of the form
 \begin{align}
 \nonumber\langle \vec{\sigma}\cdot{\vec{u}_1}\otimes\vec{\sigma}\cdot\vec{u}_2&\rangle_{\psi_{AB}}
 =\,\int_{{\bf M}_4}\,d^4x\,\sum^{N_A}_{k=1}\sum^{N_B}_{l=1}\int_{\varphi^{-1}_{k}*(T_x{\bf M}_4)}\,d^4z_k \, \int_{\varphi^{-1}_{l*}(T_x{\bf M}_4)}\,d^4z_l\break\\
 &\nonumber \Big(n^2_{1k}(\varphi^{-1}_k(x),z_k)\,n ^2_{2l}(\varphi^{-1}_l(x),z_l)+e^{\imath\theta}\,n^2_{2k}(\varphi^{-1}_k(x),z_k)\,n ^2_{1l}(\varphi^{-1}_l(x),z_l)\Big)\cdot\break\\
 &\cdot \langle\varphi^{-1}_{k}(x),z_k|\vec{\Sigma}\cdot\vec{u}_1|\varphi^{-1}_{k}(x),z_k\rangle\,
 \langle\varphi^{-1}_{l}(x),z_l|\vec{\Sigma}\cdot\vec{u}_2|\varphi^{-1}_{l}(x),z_l\rangle.
 \label{expectationvalueofentangledstate}
 \end{align}
  There are constraints between the variables appearing in \eqref{expectationvalueofentangledstate}, since they must be subjected to the entanglement conditions \eqref{entangledconditions}. This implies a generic expression for the expectation values of the form
  \begin{align} \label{contextuality}
 \langle \vec{\sigma}\cdot{\vec{u}_1}\otimes\vec{\sigma}\cdot\vec{u}_2\rangle_{\psi_{AB}}=\,&\int_{{\bf M}_4}\,d^4x\,\sum^{N_A}_{k=1}
 \sum^{N_B}_{l=1}\int_{\varphi^{-1}_{k}*(T_x{\bf M}_4)}\,d^4z_k \, \int_{\varphi^{-1}_{l*}(T_x{\bf M}_4)}\,d^4z_l\break\\
&\nonumber f(n_{1k},n_{1l},n_{2l},\vartheta_{1k},\vartheta_{1l},\vartheta_{2k},\vartheta_{2l},\theta)\cdot\break\\
 & \nonumber \cdot \langle\varphi^{-1}_{k}(x),z_k|\vec{\Sigma}\cdot\vec{u}_1|\varphi^{-1}_{k}(x),z_k\rangle\,
 \langle\varphi^{-1}_{l}(x),z_l|\vec{\Sigma}\cdot\vec{u}_2|\varphi^{-1}_{l}(x),z_l\rangle,
 \end{align}
 where the variables $n^2_{2k}$ has been substituted by the rest of variables.

\subsection{On the quantum statistics of emergent states}
The formal structure \eqref{contextuality} of the expectation value of observables in Hamilton-Randers theory as given by relations of the type \eqref{expectationvalueofentangledstate} differs in several ways from the expression for the expectation valued used as a fundamental equation in the derivation of Bell's inequalities  in quantum mechanics \cite{Bell1964}. Let us recall that Bell's theory is based upon the assumption that the expectation value of the product $\vec{\sigma}\cdot{\vec{u}_1}\otimes\vec{\sigma}\cdot\vec{u}_2$ takes the form
\begin{align}
P(\vec{u}_1,\vec{u}_2)=\int \,d\lambda\,\rho(\lambda)\,A_1(\vec{u}_1,\lambda)\,B_2(\vec{u}_2,\lambda),
\label{bellprobabilitydistribution}
\end{align}
where $\lambda$ labels the hidden variables, $A$ and $B$ are the possible values of the respective observables and $\rho$ is the density distribution of the hidden variables. One can consider more general models as the ones discussed in \cite{Bell Introduction 1971}, but there is no change in the argument.

 Let us apply Bell's model \eqref{bellprobabilitydistribution} to a Hamilton-Randers model describing a quantum system composed of two entangled particles at $1$ and $2$. First, it is natural to assume that there are $N=N_A+N_B$ sub-quantum degrees of freedom for the whole system. Second, the result of the values of the observables is only determined by the $U_t$-evolution of the sub-quantum degrees of freedom accordingly with the ergodic theorem as discussed in Chapter \ref{chapter of the Hilbert space structure}. If Hamilton-Randers models is to be interpreted as a hidden variable model according to Bell's model, then the expression  \eqref{bellprobabilitydistribution} must apply. In this case, the expectation value must be of the form
 \begin{align}\label{bellHamiltonRanders}
 \langle \vec{\sigma}\cdot{\vec{u}_1}\,\otimes\vec{\sigma}\cdot\vec{u}_2\rangle_{\psi_{AB}}=\,\int_{{\bf M}_4}\,d^4x\,\sum^N_{k=1}\,
 \int_{\varphi^{-1}_{k}*(T_x{\bf M}_4)}\,d^4z_k\,n^2_k(\varphi^{-1}(x),z_k)&\,A(\vec{\sigma}\cdot{\vec{u}_1},(\varphi^{-1}(x),z_k))\cdot\\
 \break\nonumber
& \cdot B(\vec{\sigma}\cdot\vec{u}_2,(\varphi^{-1}(x),z_k)).
 \end{align}
But the expressions \eqref{contextuality} and \eqref{bellHamiltonRanders} are not compatible. We conclude that Hamilton-Randers models do not satisfy Bell's assumptions \eqref{bellprobabilitydistribution}.

The theory of quantum correlations developed in the preceding paragraphs overcomes the assumptions of  Bell's theorem. The essence of the mechanism is the interacting during the ergodic regime of the $U_t$ dynamics. In practical terms, Bell's theory does not apply once the defining property for the expectation value of product of  quantum mechanical operators, namely, the expression \eqref{bellprobabilitydistribution} in Bell's theory, does not hold for Hamilton-Randers models.
The surprising aspect of the expression \eqref{expectationvalueofentangledstate} is the appearance of fourth degree factors on $n_k$, when one should expect quadratic terms. This indicates that there is a different term to $\rho(\lambda)$ in the expression for the probability \eqref{bellHamiltonRanders}. Specifically, the term
\begin{align*}
 n^2_{1k}(\varphi^{-1}_k(x),z_k)\,n ^2_{2l}(\varphi^{-1}_l(x),z_l)+n^2_{2k}(\varphi^{-1}_k(x),z_k)\,n ^2_{1l}(\varphi^{-1}_l(x),z_l)
 \end{align*}
does not corresponds to a {\it classical density term}, which should be instead a quadratic term in $n$.

The theory developed implies the strong result,
\begin{teorema}
Hamilton-Randers theory  implies for the expectation values of the observables  $\langle \vec{\sigma}\cdot \vec{u}_1\otimes \vec{\sigma}\cdot \vec{u}_2\rangle_{\psi_{AB}}$ the quantum mechanical values.
\end{teorema}
\begin{proof}
This is because the state $\psi_{AB}$ is an emergent state from a Hamilton-Randers system.
\end{proof}

 \subsection{Emergence of contextuality from locality at the sub-quantum level}
In a nutshell, the mechanism to avoid the conclusions of Kochen-Specker theorem are the same than the one that permits to avoid the conclusions of Bell's theorem: a deep, large non-locality in spacetime due to a very complex dynamics at the sub-quantum level and the two dimensional character of time and the respective evolution laws.

 \newpage

 \section{\LARGE{Remarks on the notion of time in Hamilton-Randers theory}}\label{Chapter On the notion of time}
 \bigskip
 \bigskip
 We have discussed several aspects of the notion of time in the above {\it chapters}, centered around the notion of emergence and two-dimensional character of time. Due to the relevance of the notion of time for our theory, in the present {\it chapter} we shall consider several aspects and thoughts about these notions.
 \subsection{On the notion of two-dimensional time in Hamilton-Randers theory}
 We have discussed, starting in {\it Chapter} \ref{chapter on classical dynamics Hamilton Randers}, the two-dimensional character of time in Hamilton-Randers theory.
One can be inclined to think that the two-dimensional time theory as it appears in our theory is analogous to the one found in classical dynamical systems in connections with averaging methods in classical mechanics \cite{Arnold}. Indeed, it is notable the formal resemblance between the dynamical system with the $(U_t,U_\tau)$ evolution along the two time parameters $(t,\tau)$ and  fast/slow classical dynamical systems. Moreover, G.'t  Hooft and others have suggested interpretations of quantum mechanics in terms of fast/slow variables in the vein as they appear in classical mechanics \cite{Hooft2021}. In such classical dynamical systems there are two time scales for the scale of change of the dynamical degrees of freedom: there are {\it slow degrees} of motion and {\it fast degrees of motion}. The possibility  to identify in Hamilton-Randers models the fast degrees of motion with the sub-quantum molecules and the slow degrees of motions with the effective quantum densities and wave functions described by elements of a Hilbert space $\mathcal{H}$, which are determined by $t$-time averages and by means of an assumed extension of the ergodic theorem, by velocity averages) of the fast dynamical degrees of freedom appears as a fundamental insight of our theory, appears rather natural.

However, a closer inspection refrains us to make a complete identification between Hamilton-Randers systems and fast/slow classical dynamical systems. First, in a classical fast/slow dynamics, there is a one-to-one mapping between the domain of the $t$-time and the domain of the $\tau$-time. In contrast, in Hamilton-Randers theory  such bijection between $t$-time parameter and $\tau$-time parameter fails, since the values of $\tau$ correspond to a measure zero set of the domain $t$-time parameter. From their interpretaion as emergent parameters, the values of the $\tau$-time attach as a discrete set of values corresponding to an ordering of the different fundamental cycles of the system, namely we have adapted the identification
\begin{align*}
\tau \leftrightarrow 2\,n\,T,\quad n\in\,\mathbb{Z},
\end{align*}
where $T$ is the semi-period of the fundamental cycles of the $U_t$ dynamics.

Due to the impossibility to trace the details of the sub-quantum dynamics, the $\tau$-parameters need to be considered as independent from the $t$-parameter in the mathematical description of each Hamilton-Randers system. However, although the $\tau$-parameters have an origin in {\it counting cycles}, the way that the parameters appear in the quantum description makes them independent from its emergent origin.

These arguments reinforce the two-dimensional character of time in Hamilton-Randers theory, a genuine aspect of Hamilton-Randers theory, although of a rather peculiar {\it two-dimensionality}, since there is no association with a geometric spacetime representation. It is as if one of these time dimensions, the $t$-time, {\it lives inside} the standard notion of $\tau$-time parameter.  Strongly related with this conclusion is the emergent character of external $\tau$-time parameters in Hamilton-Randers theory.

\subsection{Models for the spacetime with a $2$-dimensional time}
The fundamental discreteness of the $U_t$ and the $U_\tau$ dynamics suggests that the configuration manifold or configuration space $M$ must also be discrete. Models for such discrete spacetimes could be locally described as an open domain of the lattice $\mathbb{Z}^5$. In this case, each space $M^k_4$ is locally homeomorphic to the lattice $\mathbb{Z}^4$ or to a subset of it. From this bare, discrete description, several further assumptions can be adopted.
One direct implication of {\it proposition} \ref{proposicion of final Hamiltonian} is that the  spacetime arena where the sub-quantum mechanical evolution takes place is a topological space $M_4\times \,\mathbb{Z}$, where each point is $(x,t)\in \,M_4\times \mathbb{Z}$. Second, if the continuous approximation is adopted, the {topological space} $\mathbb{R}\times \mathbb{Z}$  associated with the parameter time $(t,\tau)$  is replaced by the {\it time manifold} $\mathbb{R}\times \mathbb{R}$.  Therefore, the spacetime that emerges in this description is a smooth manifold (indeed, a topological manifold that we assume to be smooth) which is locally  a foliation,
   \begin{align*}
   M_5\simeq M_3\times \mathbb{R}^2.
   \end{align*}
Other models could be of the form $ M_5\simeq M_3\times \mathbb{C} $ or  $M_5\simeq M_3\times \mathbb{P}_2$, where $\mathbb{P}_2$ is here a paracomplex algebra of dimension $2$.

Generically, this relation only holds locally, that is, for small enough neighborhoods $U\subset \,M_5$, $U_3\subset M_3$.
Each of the slits
\begin{align*}
\{t=\,(2n+1)T,\,n\in\,\mathbb{Z}\}
 \end{align*}
 defines a four-dimensional submanifold $M_4\hookrightarrow M_5$. This is not tbe associated with spacetime manifold associated to observers in general.

There are in the literature theories where time is two dimensional.
Let us mention the {\it two-time physics} framework developed by I. Bars, where the new dimension of time has accompanied a compact space extra-dimension \cite{Bars2001}. In order to show the differences between Bars's notion and our notion of two-dimensional time, let us make specific here that in Bars's theory the metric structure is defined in an extended $6$-dimensional manifold and has signature $(-1,-1,1,1,1,1)$.  In contrast, in our case, the metric structure is defined in the $4$-manifold $M_4$ and has Lorentzian signature: there is no extension of the spacetime structure to incorporate the $t$-time direction in the form of an extended metric structure. In contrast, our notion of two-dimensional time refers more to a {\it time ($\tau$ type parameter) {\it inside} a time ($t$-time parameter)}, rather than a geometric two-dimensional time or a higher order geometric time.

Another geometric theory where the geometric structure is extended to contain an additional time is the theory developed by Ragun\'i \cite{Raguni}, a theory that bears certain similarity with ours in what regards to the aim of explaining the quantum non-local phenomenology as an effect partially due to the incomplete description of time in quantum mechanics, but which is fundamentally different.

\subsubsection{Continuous models of the $\tau$-time parameters}
Due to the form in which the $\tau$-time parameters are defined in Hamilton-Randers theory, they necessarily attain discrete values on discrete domains of definitions.  Both, the degrees of freedom and dynamics of the sub-quantum degrees of freedom in Hamilton-Randers theory are attached to discrete variables, as well as the values of the $\tau$-time parameters. Hence discreteness is inherent in our models in a fundamental and natural way. However, it is useful to treat the $\tau$-time parameter as continuous.  This is the aptitude advocated in the present work, where we describe external time parameters, corresponding to the $\tau$-time parameters, as real valued  $\tau\in\,I\subset\, \mathbb{R}$. The rationale beneath this approach is that the phenomena associated with quantum or classical clocks involve large numbers of fundamental cycles. Thus it is possible to define the notion of infinitesimal time $\delta \tau$ in the theory as an approximation or model. Infinitesimal time lapses correspond to sub-quantum transitions involving very short semi-periods compared with the periods associated with quantum systems. This approximation has certainly a limitation, and the discrete nature of time parameters should be observable if enough accuracy in the determination of time lapses is achievable.

\subsection{Consequences of the emergent character of the $\tau$-time}\label{emergenceoftime}
According to Hamilton-Randers theory, $\tau$-time parameters are emergent from the $U_t$ dynamics. There are several far reaching consequences of this thesis, that we shall consider in the following paragraphs.

The first one is the discreteness of the $\tau$-time, while the use of continuous parameters must be considered as an approximation.

The second consequence of the theory is the emergent nature of any external $\tau$-time parameter. Every $\tau$-time parameters appear as determined by the $U_t$ flow happening at the fundamental scale and in particular, as a consequence of the dynamics of the sub-quantum degrees of freedom.
This emergent character of the $\tau$-time is not in contradiction with the requirement of {\it time reparametrization invariance} and, when adopting continuous models of $\tau$-time parameters, it can be make consistent with the requirement of general covariance, since our interpretation of the $\tau$-time parameters applies to any parameter associated with a physical clock and does not introduce a preferred notion notion of $\tau$-time in our theory.

Direct consequences of the emergent character of the $\tau$-time in the framework of Hamilton-Randers theory is the irreversibility character of the dynamics from several points of view. This is imprinted in several different ways:
\begin{itemize}
\item Hamilton-Randers theory implies that the underlying dynamics beneath quantum mechanics must be deterministic but irreversible.
 This conclusion is supported by the emergent character of $\tau$-time as discussed in {\it Chapter} \ref{chapter on classical dynamics Hamilton Randers} and the physical impossibility to control the details of the $U_t$ dynamics by means of quantum or classical systems. By this we mean the impossibility of controlling the details of the initial conditions of the fundamental sub-quantum degrees of freedom. Since the $U_\tau$ evolution is applied to the mathematical description involving a larger scale than the $U_t$ dynamics, it is not possible to control the fundamental $U_t$ flow by means of systems described by the $U_\tau$ dynamics, due to the complexity of the dynamics of sub-quantum degrees of freedom and the impossibility to control the $U_t$ flow using quantum interactions. This implies the existence of an irreversible $\tau$-time evolution for any physical system, including the case of the whole universe.

\item The irreversibility is not limited to the $U_t$ dynamics for the above reasons. Indeed, the fundamental dynamics is also irreversible, as it can be established from the underlying Randers type structure of the dynamics.

\item Hamilton-Randers theory implies the impossibility of {\it travel back} in the $\tau$-time for any quantum or classical degree of freedom. This is because the way  the $U_t$ flow defines emergent $\tau$-time parameters.  Any consistent model with Hamilton-Randers theory should be consistent with this emergent character of time. This is indeed a form of {\it chronological protection} \cite{Hawking1992}, that when applied to macroscopic spacetimes, should impose the absence of closed timelike curves. The particular realization of the conjecture in Hamilton-Randers theory needs to be fleshed in terms of gravitational models with a maximal acceleration relativistic, generalizations of general relativity, but the essential argument is the intrinsic irreversibility of the fundamental dynamics.

\end{itemize}

These consequences of the general formalism of Hamilton-Randers theory are falsifiable. For instance, one could look for evidence of travel back on time for matter. Our theory insists that such a possibility is not possible, because statistical considerations. Thus any evidence from travel back on $\tau$-time evolution will falsified our theory, besies of being a revolutionary fact.

\subsection{Conclusion and discussion} The notion of two-dimensional makes its appearence in several instance in the field of theoretical physics. It was introduced and it is still under intense development in the contest of string theory and super-symmetry by I. Bars and collaborators \cite{Bars2001}. In a different contest, pseudo-Riemannian manifolds with more than one timelike eigenvalues were investigated in \cite{CraigWeinstein}, in particular the associated Cauchy problem for wave functions in geometries with multiple dimensions of time. Also a geometrical construction, but linked with the nature of quantum correlations is \cite{Raguni}

However, the notion developed by these authors is radically different from our notion of two-dimensional time. The main difference relies on the emergent character of the two-dimensional time in Hamilton-Randers theory, compared with the geometric character of time of these mentioned theories. The emergence is such that it is not imprinted in the spacetime structure. Said this, the nature of time is not geometric: t-time is just the parameter of the fundamental dynamics and it is not geometric, while $\tau$-time or macroscopic time parameter has their origin on the fundamental dynamics, while the appearance of a consistent macroscopic dynamics, makes necessary the introduction of the property of $\tau$-time reparametrization invariance, that is implemented by means of spacetime general covariance.

The cyclic nature of the fundamental interaction appeared first in the preliminary versions of our theory and also, in the interpretation of quantum mechanics advocated by D. Dolce in \cite{Dolce} and subsequent works. Despite this similarity, Dolce's theory and the theory advocated by the author are rather different. In Dolce's theory the notion of time is the usual one as it appears in current physical models, but where physical fields are subjected to cyclic boundary conditions with the periods cycle given by the associated inverse Compton frequency, while in Hamilton-Randers theory the fundamental almost-cyclic evolution is parameterized by a different time parameters than usually and the relation between the period and the mass of the system is exponential \cite{Ricardo2014,Ricardo2023}.

Finally, our theory of time has as a consequence the fundamental irreversibility of the dynamics. This is in the context of a deterministic, local fundamental theory, which makes the approach radically different from Prigogine's one, although the common use in both theories of Koopman-von Neumann theory \cite{Prigrogine 2019}.

\newpage

\section{\LARGE{Hamilton-Randers models and number theory}}\label{chapter on Hamilton-Randers theory and number theory}
\bigskip
\bigskip

\subsection{On the relation between fundamental periods and primes}
The discussion in chapter \ref{chapter on classical dynamics Hamilton Randers} suggested a simple  relation between the series of prime numbers $\{p_i,\,i=1,2,...\}$ and fundamental semi-periods of non-interacting Hamilton-Randers systems. One possibility for such a relation is
\begin{align}
i\mapsto \frac{T_i}{T_{min}}=\,p_i,\,i=1,2,...
\end{align}
where $\mathcal{P}:=\,\{p_1,p_2,p_3,...\}$ is the increasing sequence of prime numbers. For identical physical systems $a$ and $a'$ we can associate the same prime, $p_a=\,p_{a'}$, if $a$ and $a'$ are identical particles. For compose systems $a\cup a'$, the product rule is adopted. Another possibility is to associate meta-stable quantum systems and pseudo-particle systems with sets of close prime, like twin prime pairs, for example.

 Following the above lines, the collection of fundamental semi-periods is associated with the collection of prime numbers. Let us assume first that the minimal semi-period must be, in convenient units of the $t$-parameter,
 \begin{align}
 T_{min}\equiv p_1 =2.
 \end{align}
  This semi-period corresponds to the smallest possible Hamilton-Randers systems, that have inertial mass parameter zero, the relation \eqref{definitionofinertialmass}.
The associated mass spectrum of fundamental particles is related with the series of prime numbers by
\begin{align*}
p_i\mapsto \frac{\hbar}{2\,c^2}\,\log \left(p_i\right),\quad i=1,2,3,...
\end{align*}
The {\it spectrum} depends upon the value chosen for $T_{min}\equiv p_1=2$.
º
There is no evidence that such a spectrum is not realized in Nature. Also, there is the possibility that $T_{min}$ is reached for another value.
If one picks up a different, $T_{min}$ then one has a different spectrum. It could be possible that there is a sequence of prime numbers such that the spectrum than in the Standard Model of elementary particles. By this we mean that the masses corresponding to the ratios
\begin{align}
\frac{m_i}{m_{i+1}}=\,\frac{\log\left(\tilde{p}_i\right)}{\log\left(\tilde{p}_{i+1}\right)},\quad i=1,2,3,...
\label{spectrum for particles 2}
\end{align}
should be in relation with the corresponding ratios of the Standard Model masses for the first values of the primes sequence $\{\tilde{p}_i\}$, with $\tilde{p}_i$ running in the family $\{\tilde{p}_i\}$.  We make the conjecture that such sequence of primes exists.
Such sequence can be conceived that happens for large prime numbers.
 There can also be some elements in $P'$ have associated twin prime numbers.

The relation \eqref{spectrum for particles 2} can be re-written as
\begin{align*}
\frac{m_1}{m_{i}}=\,\frac{\log\left(\tilde{p}_1\right)}{\log\left(\tilde{p}_{i}\right)},\,i=1,2,3,...
\end{align*}
 $m_1$ is identified the lowest positive mass of the spectrum of the Standard Model.

The number of elementary particles with associated semi-period $T_a$ is bounded as a consequence of the  {\it number prime theorem} \cite{HardyWright}. Asymptotically, the number of primes less than a given number $x\in \,\mathbb{R}$ is
  \begin{align}
  \sum_{p<x} 1\sim\,Li(x),
  \label{numberofstates}
  \end{align}
  where the logarithmic integral is
  \begin{align}
  Li(x)=\,\int^x_2\,\frac{dt}{\log t}.
  \end{align}

\subsection{The Hamilton-Randers Hamiltonian and the Hilbert-P\'olya conjecture}
 There exists a formal resemblance between Hamilton-Randers systems and what in the literature are called $xp$-models. Such a models are of relevance as a possible realization of the Hilbert-P\'olya conjecture on prime numbers (see for instance \cite{Berry, BerryKeating, Connes, Sierra2014}). In such approaches, the Hamiltonian is $1$-dimensional and it is schematically of the form
   \begin{align}
   H_{xp}=xp+\,V(x,p),
   \label{xp Hamiltonian}
   \end{align}
   where $V$ is a potential. There are several constraints on the properties of the Hamiltonian $H$ that comes from the analogy of the spectral and large asymptotic properties of $H_{xp}$ with the properties of the Riemann dynamics \cite{BerryKeating, Connes}, analogies which are on the basis for the quantum mechanical and functional approach to the Riemann hypothesis. As a matter of fact, the average number of non-trivial zeroes of the $\zeta$ Riemann function and the expression for the partial fluctuations in the number of zeroes of $\zeta$ are formally the same than the counting number of eigenvalues of $H_{xp}$ and asymptotic of the fluctuation of the zeroes of $H_{xp}$.

We can  compare this Hamiltonian  $H_{xp}$ with the Hamiltonian of a Hamilton-Randers dynamical system,
  \begin{align*}
  H=\,\sum^N_{k=1}\,y^k\,p_{xk}+\,\beta_y(x,y)^k\,p_{yk},
  \end{align*}
  which determines the $U_t$ dynamics.
  Under the approximate hyperbolic motion,
  \begin{align*}
  y^k =\,\dot{x}^k\approx x^k,
   \end{align*}
the Hamilton-Randers Hamiltonian is
  \begin{align*}
  H\approx\,\sum^N_{k=1}\,x^k\,p_{xk}+\,\beta_y(x,x)^k\,p_{xk},
  \end{align*}
showing that the Hamilton-Randers Hamiltonian appears as a $N$-copy generalization of the $xp$-Hamiltonian \eqref{xp Hamiltonian}, but corrected  with additional interaction terms from the terms $\beta_y(x,x)^k\,p_{xk}$. These interacting terms may play a significant role in solving many of the current difficulties of the $xp$-models. For instance, it is known that the correct $xp$-Hamiltonian cannot be $1$-dimensional hyperbolic. Comparison between chaotic dynamics (that is, Hamiltonian dynamical systems with bounded, unstable trajectories) and fluctuaction formulae for the zeroes of the Riemann zeta function implies that the Hamiltonian of the Hilbert-P\'olya conjecture must be {\it chaotic}. These are properties that the classical Hamiltonian \eqref{xp Hamiltonian} must have. These dynamical properties could be achieved by interacting terms as the proposed here.

The above discussion suggests that the dynamics of Hamilton-Randers systems is related with the fundamental laws of arithmetics and in particular with the structure of the distribution of prime numbers.

\newpage

\section{\LARGE{Discussion and open problems in Hamilton-Randers theory}}\label{Chapter on Discussion}
\bigskip
\bigskip
In this {\it chapter} we discuss several aspects of Hamilton-Randers theory in relation with other theories of emergent quantum mechanics and emergent gravity. We briefly resume some of the, although qualitative, falsifiable predictions of our theory and indicate open problems and further developments.

\subsection{Comparing Hamilton-Randers theory with other deterministic theories of quantum mechanics}
The theory presented in this work has several remarkable similarities with other approaches to emergent quantum mechanics. Let us start by considering several aspects of the approach developed in Cellular Automaton Interpretation by G. 't Hooft and by others (see for instance \cite{Hooft, Hooft2016}). Among the common points between the Cellular Automaton Interpretation as developed in the work of Hooft and Hamilton-Randers theory, are  the systematic use of Hilbert spaces theory for the description of classical systems, the introduction of a dissipative dynamics  in order to recover the notion of quantum state as an equivalence class of sub-quantum states  and the fact that the wave function is $\psi$-epistemic. Another point in common is the fundamental discreteness of the evolution law, although we have taken in this work a pragmatic approach to it and considered continuous models.

However, there are important differences between our approach and  cellular automaton approaches to emergent quantum mechanics. The mathematical formalisms used in both theories are different and also the notion of deterministic system attached to a quantum particle is different.

Another significant difference between our approach and cellular automaton is on the notion of time. In Hamilton-Randers theory it is described by  two parameters $(t,\tau)$, instead than the usual one parameters, with values in the product of partial ordered fields $\mathbb{K}$ (usually assumed to be the real field $\mathbb{R}$) and the set of integer numbers $\mathbb{Z}$. In the continuous limit the domain of the time parameters $(t,\tau)$ are open subsets $I\times \,\tilde{I}$ of the cartesian product $\mathbb{R}\times \mathbb{R}$. The further identification of the space of the time parameters with the field of complex numbers $\mathbb{C}$ imposes further constraints in the theory that we have not consider in this work.
Also, note that $\mathbb{R}\times \,\mathbb{R}$ cannot be made an ordered field. Thus strictly speaking, there is no  notion of global causation for Hamilton-Randers systems. The causal spacetime structures emerge with the projection $(t,\tau)\mapsto \tau$.

The parameters $t$ and $\tau$ are qualitatively different and irreducible to each other, which implies the  two-dimensional character of time and that the associated two dynamical evolution $(U_t,U_\tau)$ are irreducible to each other. Note that in the cellular automaton of Hooft's approach, unitarity is recovered at the level of equivalence classes. In contrast, in our description, where there is a  dissipative dynamics in the form of an internal $U_t$ flow which is assumed of a rather intricate structure, there is no preservation of the volume phase element for the $U_t$ flow, since the dynamics is driven by a time dependent Hamiltonian, but the $U_\tau$ evolution associated with the quantum mechanical evolution is unitary.

\subsection{Comparing Hamilton-Randers theory with de Broglie-Bohm theory}
 There are certain similarities between some properties present in Hamilton-Randers theory and in de Broglie-Bohm theory \cite{Bohm}. In Hamilton-Randers theory, the value of the observables of the system are well-defined before any measurement is done and are independent of the possible decisions that a particular observer makes. Indeed, in Hamilton-Randers theory, the natural spontaneous collapse processes  provides a natural description for the measurement  and provides  a deterministic and  local picture of the dynamics in the extended configuration space $TM$ and under the evolution determined by the double dynamics $(U_t,U_\tau)$. The quantum system is in a localized state prior to any measurement performed by a standard macroscopic observer, but it does not necessarily follows a  smooth world line in the spacetime ${\bf M}_4$.  In such localized state, all the possible classical observables have well definite values.

 Interestingly, that classical observables have definite values does not necessarily implies the existence of a classical smooth trajectory, because the intrinsic discrete character of the $\tau$-time. As our discussion of the double slit experiment reveals, transitions between separate paths can happen, as long as there is compatibility with conservation laws. Therefore, there are no Bohmian trajectories in Hamilton-Randers theory, in the sense that the succession of concentration domains does not necessarily defines a smooth curve in spacetime, although Bohmian trajectories could be associated with statistical averages when experiments are performed with ensembles of identical quantum particles.

In Hamilton-Randers theory is natural to interpret the wave function as a {\it presence of matter} during the non-concentrating phase of the $U_t$ dynamics. However, the matter refers to the sub-quantum degrees of freedom and the only extrapolation to the Standard Models degrees of freedom comes from the use of a form of ergodicity. Therefore, in Hamilton-Randers theory  the wave function has an epistemic interpretation, in sharp contrast with de Broglie-Bohm theory (at least in some interpretations, that includes the original formulation due to D. Bohm \cite{Bohm}), where the wave function is a real ontological field.

A possible test of the ontological character of the wave function arises if  the wave function is associated with the source of the gravitational interaction. In the de Broglie-Bohm theory, since the wave function $\psi$ is a  physical field, it must carry energy and momentum and henceforth, must affect the surrounding gravitational field. In contrast, in Hamilton-Randers theory such modification of the gravitational field due to the wave function should not be expected.

\subsection{Comparing Hamilton-Randers theory with theories of emergent classical gravity} The general idea of emergent gravity is not new to our work. Probably, a relevant example for us is Verlinde's theory of entropic gravity, where it was also argued that gravity is a classical non-fundamental interaction \cite{Verlinde2011}. We aim to briefly clarify the differences between our description of emergent gravity and Verlinde's theory of gravity and eventually to show that, in fact, both are rather different theories.

Apart from being based in very different principles and assumptions, a qualitative difference is on the universality of the corresponding interaction. While the theory of gravity as entropic force requires that the system is described necessarily by many macroscopic or quantum degrees of freedom, in order to define the temperature and entropy of the system, it seems that this inevitably leads to abandon universality of gravity, since a simple system as an electron is not such a  complex system where macroscopic entropy and temperature can be defined. However, it is well known that gravity affects also quantum systems in a remarkably universal way. For instance, neutron interferometry is usually used to test how a classical gravitational potential interacts with a quantum particle \cite{ColellaOverhauserWerner1975}. In contrast, this conceptual problem does not appear in our version of emergent gravity, since it can be applied to a single electron (that eventually, it appears as a rather complex dynamical system). In this sense, gravitational interaction in Hamilton-Randers theory is universal.

Another important difference between both theories of gravity is on the use of the holographic principle, which is fundamental in Verlinde's approach but that it is apparently absent in Hamilton-Randers theory. Nevertheless, this point of view and the relation of thermodynamics and gravity is an interesting point to develop further from the point of view of Hamilton-Randers theory.

In resume, despite that in both theories gravity appears as an emergent classical interaction, after a rapid inspection of the above differences between Verlinde's theory and the proposal suggested by Hamilton-Randers theory, they appear as very different concepts and not only  different formalisms of the same underlying theory.

\subsection{Formal predictions of the Hamilton-Randers theory} Despite the general description of Hamilton-Randers theory that we have provided, we can make several predictions which are independent of model. Let us resume them here.

\begin{enumerate}
\item {\bf Existence of a maximal acceleration}. This is a general consequence of the assumptions in {\it section} 2. The search for experimental phenomenological signatures of maximal acceleration is currently an active research line, developed by several groups. For us, one of the most interesting possibilities to detect effects due to maximal acceleration is on the spectrum and properties of ultra-high cosmic rays. Maximal acceleration effects allow an increase in the number of ultra-high events compared with the predictions of a special relativistic theory \cite{Ricardo2014}.

\item
{\bf Exactness of the weak equivalence principle}. If our interpretation of emergent quantum mechanics is correct, there will not be an experimental observation of violation of WEP until the energy scales of the probe particles are close enough to the fundamental energy scale. This prediction can be cast as follows: the WEP will hold exactly up to a given scale (close to the fundamental scale) and after this energy scale is reached or it is close enough, the principle will be totally violated. The scale where this happens is associated with the fundamental scale where the deterministic sub-quantum degrees of freedom live.

This prediction contrast with standard phenomenological models associated with quantum gravity, where violation of the weak equivalence happens (DSR, rainbow metrics and also some Finsler spacetimes, to give some examples) and where the predictions are in the form of smooth violations of the principle, where deviations could occur even at relatively low scales.

\item
{\bf  Absence of the graviton}. According to Hamilton-Randers theory, there is no graviton (massless particle of spin $2$, that is associated with the formal quantization of the gravitational field). Hence if a graviton is discovered, our theory has to be withdraw. This can be either falsified by the study of primordial gravitational waves in cosmology or in high energy experiments.

\item
{\bf  Quantum correlations have a macroscopic distance range bounded and are macroscopically instantaneous}.  If experimental limits are found in future experiments on the apparent speed of the quantum correlations, our theory has to be deeply modified, if not refuted. Moreover, the expression for the\eqref{boundfordistanceofcorrelations1} and the expression \eqref{boundfordistanceofcorrelations2} must hold.

\item
{\bf Impossibility of time travel}. By the way that $\tau$-time is defined as emergent concept in Hamilton-Randers theory, it is not possible to time travel back to the past for a macroscopic or quantum matter system. The reason underlies in the emergent character of $\tau$-time from the point of view of the $U_t$ flow. Although rather qualitative, it is enough to confirm that in Hamilton-Randers theory the chronological protection conjecture holds \cite{Hawking1992}. This can be interpreted as a theoretical prediction of our theory.
\end{enumerate}
The above are general predictions of Hamilton-Randers theory. We did not discussed the implications on quantum computing that our emergent quantum theory should have. Further consequences are to expected once the theory is developed in more detail.

\subsection{Further developments}

\subsubsection{The relation between quantum mechanical observables and microscopic operators} Another open question in our approach to the foundations of quantum mechanics is how to construct quantum mechanical observables  in terms of {\it fundamental operators} acting on ontological states $\{|x^\mu_k,y^\mu_k\rangle\}^{N,4}_{k=1,\mu=1}$. The transition from the description of the dynamics provided by $\widehat{H}_{matter,t}(\hat{u},\hat{p})$ to the description by $\widehat{H}_{matter,t}(\widehat{X},\widehat{P})$ requires either to know the structure of the operators $\{\widehat{X},\widehat{P}\}$ in terms of the operators $\{\hat{u}, \hat{p}\}$ or a formal argument to identify the quantum operators  $\{\widehat{X},\widehat{P}\}$ from the sub-quantum operators.

A partial result related with  this problem has been given in {\it Chapter} \ref{chapter of the Hilbert space structure}, where we have explicitly constructed free quantum states as emergent states from pre-quantum states. However, our construction is rather heuristic and another more formal approach is need.

Related with this problem is the issue of which are the quantum operators that must be translated to the level of {\it fundamental operators}. For instance, if the underlying models are Lorentz invariant or Lorentz covariant, must spin operators  be translated to fundamental operators? Are there representations of the (possibly deformed) Lorentz group in terms of fundamental operators? Quantum spacetime models as Snyder spacetime could be involved in answering these questions, but it is only if we have on hand complete examples of Hamilton-Randers models that we can address conveniently these problems. Examples of such models are free particles described in {\it section 5}. But further analysis is required of how these models can be compatible with quantum operators. One possibility is to study invariants under the $U_t$-flow as possible values defining the compatible representations of the Lorentz group with Hamilton-Randers theory.

 Another important question to address is how quantum field theories arise from Hamilton-Randers models. The operators of a field theory are different than the operators of a first quantized theory. Thus it is possible that one needs to consider first deterministic field theories from the beginning in the formulation of Hamilton-Randers theories. Other possibility is to interpret  relativistic quantum field theories as an effective description of an otherwise discrete, first order quantized theory.
\subsubsection{Quantum non-locality and entanglement}
 On the basis of our interpretation of the fundamental non-local properties of quantum mechanical systems is the notion of two-dimensional time parameter $(t,\tau)\in \,\mathbb{R}\times \mathbb{R}$.  In order to describe the state of a system at a given instant, in Hamilton-Randers theory we need to specify the two time parameters $(t,\tau)\in \,\mathbb{R}\times\,\mathbb{R}$. If only the parameter $\tau$ is specified in a dynamical description of a physical system (as it is done in usual field theory and quantum mechanics) and if the system is not in an localized state for the $U_t$ dynamics, an effective non-locality in the phenomenology of quantum systems appears. This is the origin of the phenomenological non-local properties of quantum systems.
\subsubsection{Interpretation of entangled states and derivation of associated inequalities} A dynamical system interpretation of the quantum entanglement could be developed in our framework, with a theory of entangled states, based on the embedding  $\mathcal{H}\hookrightarrow \mathcal{H}_{Fun}$ in  the construction given by the expression  \eqref{ansatzforpsi} can be constructed using the methods of {\it sections 4} and {\it section 5}. In particular, we suspect that for product states, there is few interactions between the sub-quantum degrees of freedom associated with $a$ and $b$, while for entangled states, there is many interactions that imply the constraints \eqref{entangledconditions}.

Although we have envisaged a geometric way to understand quantum non-local phenomena, a complete mathematical treatment must be investigated: how exactly is related the projection $(t,\tau)\mapsto \tau$ and ergodicity  with Bell's inequalities \cite{Bell1964}?
\subsubsection{The problem of synchronization different emergent systems} If different quantum system have different periods associated, how is it possible that we perceive different quantum objects as co-existing at the same local time, for a given observer? One possibility is to invoke super-determinism and argue that there is an underlying conservation law that induces an {\it universal synchronization}. However, we found this possibility rather un-natural in the present framework, even if our theory is consistent with super-determinism. Other possibility is to invoke a higher order mechanism, currently present in quantum mechanics, that allows for a transition from a {\it emergent quantized time} as it is suggested in sub-section \ref{emergenceoftime}, to a continuously local time synchronized macroscopic reality in systems with many parts or components. A third mechanism exploits the relation between the primes and periods and look for collection of consecutive primes that mach the spectrum of elementary quantum particles.
\subsubsection{Absolute structures in Hamilton-Randers theory}
 It is a problematic point of our theory that the metric $\eta_4$ and the collection of metrics $\{\eta^k_4\}^N_{k=1}$ are background structures. We expect that an improved version of our theory should provide a natural dynamics for $U_t$ and $U_\tau$, hence for the Hamilton-Randers geometric structures. The existence of the metaestable domain  $D_0$ implies that the $U_t$ flow determines a thermodynamical limit. Hence it is possible to derive field equations for the metrics  $\{\eta^k_4\}^N_{k=1}$ valid in the metastable domain as equation of state for the system described in terms of physical observables, in the thermodynamical equilibrium domain.  A suggestion in this direction is that it could exists a fundamental relation between our notion of emergent gravity and several proposals that view general relativity as a thermodynamical limit of classical systems. Furthermore, our geometric structures are (locally) Lorentz invariant, although a more precise description of the geometry structure in presence of proper maximal acceleration is required \cite{Ricardo2014}. This could be important as a mechanism to avoid singularities within a completely classical theory of gravity.
 \subsubsection{The relation with weak measurements} It has been discussed that in Hamilton-Randers theory the physical observables are well defined prior a measurement is made by an observer. On the other hand, weak measurements have been intensively investigated and applied to {\it proof} the reality of the wave function. We think that weak measurements can be indeed reconsidered in the framework of Hamilton-Randers theory and at the same time keep the epistemic interpretation of the wave function.

\newpage

\appendix
\section{Basic notions of Category Theory}\label{Appendix on Categories}

\newpage

\section{Geometric framework for Finsler spacetimes}\label{Appendix on Finsler Geometry}
Following  J. Beem \cite{Beem1} but using the notation from \cite{GallegoPiccioneVitorio:2012},
we introduce the basic notation and fundamental notions of Finsler spacetimes theory. As we mentioned before, we formulate the theory for $n$-dimensional manifolds. Thus we start with the following
\begin{definicion}
A Finsler spacetime is a pair $(M,L)$ where
\begin{enumerate}
\item $M$ is an $n$-dimensional real, second countable, Hausdorff $C^{\infty}$-manifold.
\item $L:N\longrightarrow R$ is a real smooth function such that
\begin{enumerate}
\item $L(x,\cdot)$ is positive homogeneous of degree two in the variable $y$,
\begin{align}
L(x,ky)=\,k^2\,L(x,y),\quad \forall\, k\in ]0,\infty[,
\label{homogeneouscondition}
\end{align}
\item The {\it vertical Hessian}
\begin{align}
g_{ij}(x,y)=\,\frac{\partial^2\,L(x,y)}{\partial y^i\,\partial y^j}
\label{nondegeracy-signature}
\end{align}
is non-degenerate and with signature $n-1$ for all $(x,y)\in\, N$.
\end{enumerate}
\end{enumerate}
\label{Finslerspacetime}
\end{definicion}
This will be our definition of Finsler spacetime with signature $n-1$ (in short, Finsler spacetime).

Direct consequences of this definition and Euler's theorem for positive homogeneous functions are the following relations,
\begin{align}
\frac{\partial L(x,y)}{\partial y^k}\,y^k=\,2\,L(x,y),\quad \frac{\partial L(x,y)}{\partial y^i}=\,g_{ij}(x,y)y^j,
\quad L(x,y)=\frac{1}{2}\,g_{ij}(x,y)y^iy^j.
\end{align}
Furthermore, the function $L$ determines a function $F$ as $F:=L^{1/2}$. One can extend the formulae from positive definite metrics to  metrics with signature $n-1$ by substituting $F$ in terms of $L$, when $L\neq 0$.

Given a Finsler spacetime $(M,L)$,
a vector field $X\in\,\Gamma TM$ is timelike if $L(x, X(x))<0$ at all point $x\in\,M$ and a curve $\lambda:I\longrightarrow M$ is timelike
if the tangent vector field is timelike $L(\lambda(s),\dot{\lambda}(s))<0$. A vector field $X\in \,
 \Gamma TN$ is lightlike if $L(x,X(x))=0,\,\forall x\in\,M$ and a curve is lightlike if its tangent
  vector field is lightlike. Similar notions are for spacelike.
A curve is causal if either is timelike and has constant
  speed $g_{\dot{\lambda}}(\dot{\lambda},\dot{\lambda}):=L(\lambda,\dot{\lambda})=g_{ij}(\dot{\lambda},T)\dot{\lambda}^i\dot{\lambda}^j$ or if it is lightlike.

In particular
 Cartan tensor components are defined as
\begin{align}
C_{ijk}:=\,\frac{1}{2} \frac{\partial g_{ij}}{\partial y^{k}}, \quad i,j,k=1,...,n,
\label{Cartantensor coefficients}
\end{align}
differently to the way it is introduced in \cite{BaoChernShen}.
Therefore, because of the homogeneity of the tensor $g$, Euler's theorem implies
\begin{align}
C_{(x,y)}(y,\cdot,\cdot)=\,\,\frac{1}{2} y^k\,\frac{\partial g_{ij}}{\partial y^{k}}=0.
\label{homogeneityofcartan tensor}
\end{align}

The introduction of the Chern's connection for Finsler spacetimes can be done in similar terms as in the case of positive definite metrics. We introduce here the connection following the index-free formulation that can be found in \cite{GallegoPiccioneVitorio:2012}.
\begin{proposicion}
Let $h(X)$ and $v(X)$ be the horizontal and vertical lifts of $X\in\,\Gamma TM$ to $TN$, and $\pi^*g$ the pull back-metric.
For the Chern connection the following properties hold:
\begin{enumerate}
\item  The almost $g$-compatibility metric condition is equivalent to
\begin{align}
{\nabla}_{v({\hat{X}})}\pi^*g=2C(\hat{X},\cdot,\cdot),\quad
{\nabla}_{h({X})} \pi^*g=0,\quad \hat{X}\in\,\Gamma TN.
\label{covariantalmostmetriccondition}
\end{align}

\item The torsion-free
condition of the
Chern connection is equivalent to the 
following:
\begin{enumerate}
\item Null vertical covariant derivative of sections of ${\pi}^*{TM}$:
\begin{align}
{\nabla}_{{V({X})}} {\pi}^* Y=0,
\label{covariantensorfreecondition1}
\end{align}
for any vertical component $V(X)$ of $X$.
\item  Let us consider $X,Y\in {TM}$ and their horizontal
lifts $h(X)$ and $h(Y)$. Then
\begin{align}
\nabla_{h(X)} {\pi}^* Y-{\nabla}_{h(Y)}{\pi}^*X-{\pi}^* ([X,Y])=0.
\label{covariantensorfreecondition2}
\end{align}
\end{enumerate}
\end{enumerate}
\end{proposicion}
The connection coefficients $\Gamma^i_{jk}(x,y) $ of the Chern's connection are constructed  in terms of the fundamental tensor components and the Cartan tensor components. The detail of how they are constructed are not necessary for the results described in this note.
\newpage
\section{Koopman-von Neumann theory}
\newpage
\section{Concentration of measure}
\newpage
\section{Heisenberg uncertainty principle against general covariance}\label{Heisenberg uncertainty principle against general covariance}
\subsection{Short review of Penrose's argument} In this appendix we will follow the arguments and notation of the paper \cite{Ricardo 2022}.
Penrose's argument aims not to be a complete theory of objective reduction of the wave function, but to illustrate a general mechanism for the reduction that involves gravity in a fundamental way, using  the less possible technical context. Therefore, a Newtonian gravitation framework is chosen for the development of the argument. In this context, it is highlighted the existence of a geometric framework for Newton theory of gravitation which is manifestly general covariant, namely, Cartan formulation of Newtonian gravity. The adoption of Cartan general covariant formulation of Newtonian gravity is useful to directly full-fill the required compatibility of the with the principle of general covariance of the theory of gravitation used.
 The adoption of Newton-Cartan theory is also justified by the type of experiments suggested by the argument can be implemented in the Newtonian limit.

 A brief introduction of Cartan's geometric framework is in order (see for instance \cite{Cartan1923, Cartan1924} and \cite{MisnerThorneWheeler}, chapter 12). The spacetime is a four dimensional manifold ${\bf M}_4$.  There is defined on ${\bf M}_4$ an affine, torsion-free connection $\nabla$ that determines free-fall motion as prescribed in Newtonian gravity. There is a $1$-form $dt$ determining the absolute Newtonian time and there is a three dimensional Riemannian metric $g_3$ on each section $M_3$ transverse to $dt$. The connection preserves the $1$-form $dt$. The metric $g_3$ is compatible with the connection. An analogue of Einstein's equations define the dynamical equations for $\nabla$ in terms of the density of matter $\rho$. The Riemannian metric $g_3$ defined on each transverse space $M_3(t)$ is determined by further geometric conditions. For practical purposes, can be thought $g_3$ as given. The scalar field $t:{\bf M}_4\to \mathbb{R}$ is called the {\it absolute time} and it determines the $1$-form $dt$.

 There are two further assumptions that Penrose uses in his argument and that are worthily to mention:
 \begin{itemize}

 \item Slightly different spacetimes ${\bf M}_4(1)$, ${\bf M}_4(2)$, corresponding two different configurations of the system are such that the absolute time functions $t_i:{\bf M}_4(i)\to\,\mathbb{R}, \,i=1,2$ can be identified.

 \item The spacetime is stationary, that is, there exists a Killing vector of the metric $g_3$,
 \begin{align}
 \nabla_T\,g_3=0.
 \label{timelike killing vector}
 \end{align}
 \end{itemize}

The reason to identify the absolute coordinate $t_1\in \,{\bf M}_4(1)$ with the absolute coordinate $t_2\in\,{\bf M}_4(2)$ is technical, since it avoids to consider several fine details in the calculations.

 The reason to assume the existence of a Killing vector is that when there is a timelike Killing vector on ${\bf M}_4$, then there is a well-defined notion of stationary state in an stationary Newton-Cartan spacetime. Regarded as a derivation of the algebra of complex functions $\mathcal{F}({\bf M}_4)$, the fact that $T$ is globally defined implies that the eigenvalue equation
 \begin{align}
 T\,\psi =\,-\imath\,\hbar\,E_\psi \,\psi
 \label{stationary state}
 \end{align}
 is well defined on ${\bf M}_4$. This equation defines the stationary state $\psi\in\,\mathcal{F}({\bf M}_4)$ and the energy $E_\psi$ as eigenvalue of $T$. Let us remark that stationary stated defined by relation \eqref{stationary state}, instead than as eigenvectors of an Hermitian Hamiltonian.

After the introduction of the geometric setting, let us highlight the logical steps of Penrose's argument as it appears for instance in  \cite{Penrose 1996, Penrose 2014a} or in  \cite{Penrose 2005}, chapter 30. The thought experiment considers the situation of superposition of two quantum states of a lump of matter, when the gravitational field of the lump.

We have articulated Penrose argument as follows:
\\
{\bf 1.} First, the principle of general covariance is invoked and it is shown its incompatibility with the formulation of stationary Schr\"odinger equations for superpositions when the gravitational fields of the quantum systems are taking into consideration. In particular, Penrose emphasizes why one cannot identify two diffeomorphic spacetimes in a pointwise way. Therefore, the identification of timelike killing vectors pertaining to different spacetimes is not possible and one cannot formulate the stationarity condition \eqref{stationary state} for superpositions of lumps, since each lump determines its own spacetime. Each time derivative operation is determined by the corresponding timelike Killing field and each of those Killing fields lives over a different spacetime ${\bf M}_4(1)$ and ${\bf M}_4(2)$. Note that this obstruction is absent in usual quantum mechanical systems, where the gravitation fields of the quantum system is systematically disregarded as influencing the spacetime arena, namely, the Galilean or Minkowski spacetime or in general, a classical spacetime.
\\
{\bf 2.} If one insists in identifying Killing vectors of different spaces, it will be an {\it error}. Such an error is assumed to be a measure of the amount of violation of general covariance. Penrose suggested a particular measure $\Delta$  and evaluated it in the framework of Cartan formulation of Newtonian gravitational theory by identifying the corresponding $3$-vector accelerations associated with the corresponding notions of free falls.
\\
{\bf 3.} An interpretation of the error in terms of the difference between the gravitational self energies of the lumps configurations $\Delta E_G$ is developed, with the result that
$\Delta =\,\Delta E_G.$
\\
{\bf 4.} Heisenberg's enery/time uncertainty relation is applied in an analogous way as it is applied for unstable quantum systems to evaluate the lifetime for decay $\tau$ due to an instability. In the present case,  the energy uncertainty is the gravitational self-energy of the system $\Delta\,E_G$. Therefore, the lifetime than the violation of covariance persists before the system decays to a particular well-defined spacetime is of the form
\begin{align}
\tau\sim \,\frac{\hbar}{\Delta\,E_G}.
\label{life time}
\end{align}
For macroscopic systems $\hbar/\Delta\,E_G$ is a very short time. Therefore, the argument provides an universal  mechanism for objective reduction of the wave function for macroscopic systems.
The cause for the instability is reasonably linked to the violation of the general covariance that one introduces when  identifying the Killing vectors $T_1$ and $T_2$. However, Penrose's argument does not discuss a particular dynamics for the collapse.

It is worthily to mention that Penrose's argument relies on the  following construction. The estimated errors used in points 2. and 3. is given in Penrose's theory by the integral
 \begin{align}
 \Delta =\,\int_{M_3(1)}\,d^3 x\,\sqrt{\det g_3}\,g_3(\vec{f}_1(t,x)-\vec{f}_2(t,x),\vec{f}_1(t,x)-\vec{f}_2(t,x)),
 \label{error in approximation}
 \end{align}
 where $\vec{f}_1(t,x)$ and $\vec{f}_2(t,x)$ are the acceleration 3-vectors of the free-fall motions when the lumps are at position $1$ and $2$, but regarded as points of $M_3(1)$, submanifold of ${\bf M}_4(1)$. $M_3(1)$ depends on the value of the absolute time parameter $t=t_0$.
 Although the measure $\int_{M_3(1)}\,d^3 x\, \sqrt{\det \,g_3}$ is well defined and invariant under transformations leaving the $1$-form $dt$ and the space submanifold $M_3(1)$ invariant, the integral \eqref{error in approximation} is an ill-defined object. This is because $\vec{f}_2(t,x)$ is not defined over $M_3(1)$ but over $M_3(2)$ and hence, the difference $\vec{f}_1(t,x)-\vec{f}_2(t,x)$ is not a well-defined geometric object. In Penrose's argument, the integral operation \eqref{error in approximation} is understood as an indicator of an error due to an assumed violation of general covariance. There is the hope that a deeper theory of quantum gravity will justify such expression properly.
To play the role of a meaningful error estimate, the integral \eqref{error in approximation} must have an invariant meaning, independent of any diffeomorphism leaving the $1$-form $dt$  and each space manifold $M_3(1)$ invariants. This is indeed the case, module the issue of the problematic nature of  $g_3(\vec{f}_1(t,x)-\vec{f}_2(t,x),\vec{f}_1(t,x)-\vec{f}_2(t,x))$.

Assuming this procedure, the integral \eqref{error in approximation} can be related with the differences between Newtonian gravitational self-energies of the two lumps configurations $\Delta =\,\Delta E_G$, as Penrose showed \cite{Penrose 1996},
\begin{align}
\nonumber \Delta =\,\Delta E_G =\,-4\,\pi\,G\,\int_{M_3(1)}\,\int_{M_3(1)}\,d^3x\,d^3 y \,&\sqrt{\det \,g_3(x)}\sqrt{\det\,g_3(y)}\cdot\\
& \cdot\frac{\left(\rho_1(x)-\rho_2(x)\right)\,\left(\rho_1(y)-\rho_2(y)\right)}{|x-y|}.
\end{align}

\subsection{On the application of Penrose argument to microscopic systems}
In the following paragraphs we discuss a paradoxical consequence of Penrose's gravitationally induced reduction mechanism of the quantum state argument.
 We remark that strictly speaking,  Penrose's argument is also applicable to {\it small scale systems} in the following form. First,  there is no scale in Penrose argument limiting the applicability of the error measure \eqref{error in approximation}  for the violation of the general covariance. Therefore, the same considerations as Penrose highlights for macroscopic systems, must apply to small systems. For small scale systems, quantum coherence seems to be a experimentally corroborated fact and it is one of the fundamental concepts of quantum mechanical description of phenomena. Thus by Penrose's argument, the existence of microscopic quantum coherence seems inevitably to be interpreted as an observable violation of general covariance. We think this is a dangerous consequence of Penrose's argument that can lead to a contradiction between the aims of the theory and the consequences of the theory.

 Let us first agree that the measure of the violation of  general covariance is measured by the identification $\Delta =\,\Delta E_G$, which is as we have discussed, a function of the absolute time coordinate function $t:{\bf M}_4(1)\to \,\mathbb{R}$. For typical quantum systems, $\Delta E_G (t)$ is very small, since gravity is weak for them. However, due to the eventual large time that coherence can happen for a microscopic system,  $\Delta E_G(t)$  is not the best measure of violation of general covariance. For microscopic systems, coherence in energy due to superpositions of spacetimes could persists for a long interval of time $t$. In such situations, it is  $\tau\,\Delta E_G$ what appears to be a better measure for the violation of the general covariance, where here $\tau$ is the span of absolute time such that the superposition of spacetimes survives.

The measure of the violation of general covariance proposed by Penrose is an integral of the modulus of the difference between two vector fields in space given by \eqref{error in approximation}, which does not measures the possible accumulation effect on time of the violation of general covariance. Thinking on this way, $\Delta E_G$ appears to be a {\it density of error}, while the {\it error} should be obtained by integrating $\Delta E_G(t)$ along absolute time $t$. This reasoning suggests to consider the following measure of the error in a violation of general covariance by a local identifications of spacetimes due to superposition of lumps,
\begin{align}
\widetilde{\Delta} =\,\int_{{\bf M}_4(1)}\,dt\wedge \,d^3x\,\sqrt{\det g_3}\,g_3(\vec{f}_1(x,t)-\vec{f}_2(x,t),\vec{f}_1(x,t)-\vec{f}_2(x,t)),
\label{covariant error in the approximation}
\end{align}
Note that with the geometric structures available in the Cartan-Newton space considered by Penrose, there is no other invariant volume form in ${\bf M}_4(1)$ that one can construct than $dt\wedge\,d^3 x\,\sqrt{\det\,g_3}$.  Furthermore, the volume form $dt\wedge \,d^3x\,\sqrt{\det g_3}$ is invariant under the most general diffeomorphism living the $1$-form $dt$ invariant, a property that it is not shared by the form $d^3x \,\sqrt{\det g_3}$. In these senses, the measure $\widetilde{\Delta}$ is unique, supporting \eqref{covariant error in the approximation} instead than \eqref{error in approximation} as measure.

Furthermore, we have that if  $g_3(\vec{f}_1(x,t)-\vec{f}_2(x,t),\vec{f}_1(x,t)-\vec{f}_2(x,t))\neq \,0$ but constant only in an interval $[0,\tau]$ and  otherwise it is zero, then
\begin{align*}
\widetilde{\Delta} & =\,\int^\tau_0 dt \,\int_{M_3(1)}\,d^3 x\,\sqrt{\det g_3}\,g_3(\vec{f}_1(t,x)-\vec{f}_2(t,x),\vec{f}_1(t,x)-\vec{f}_2(t,x)) =\,\tau\,\Delta\\
& =\,\tau\,\Delta E_G
\end{align*}
holds good, demonstrating the interpretation of $\widetilde{\Delta}=\,\tau\,\Delta E_G$ as the value of a four dimensional integral on ${\bf M}_4(1)$.

Mimicking Penrose's argument, in order to determine the lifetime $\tau$ one could apply Heisenberg energy/time uncertainty relation. In doing this it is implicitly assumed that there is coherence in energy. The error in the above identification of the vector $\vec{f}_2$ as a vector in $M_3(1)$ is such that
\begin{align}
\widetilde{\Delta}=\,\tau\,\Delta E_G \sim \,\hbar,
 \end{align}
 for any quantum system, large or small. Hence the amount of violation of the general covariance principle due to quantum coherence does not depend upon the size of the system, since it is always of order $\hbar$. This is despite the lifetime $\tau$ can be large or small, depending on the size of the system.

The direct consequence of this form of the argument is that assuming Heisenberg energy/time uncertainty relation and superposition of spacetimes in the conditions described above,  a fundamental violation of general covariance for microscopic systems of the same amount than for macroscopic systems.
This conclusion, that follows from taking as measure of the error the expression \eqref{covariant error in the approximation}, implies a deviation from Penrose argument. However, we have seen that even in the form of small violations of general covariance as given by the integral $\eqref{error in approximation}$, there is an observable violation of general covariance, since coherence is to be understood as a form of violation of general covariance.

The form of the paradox that we have reached by further pursuing Penrose's argument is that, although the argument is presented as an aim to preserve in an approximated way general covariance as much as it could be possible in settings where superpositions of gravitational spacetimes are considered, it leads  to a mechanism that violates general covariance. Adopting Penrose's measure $\Delta$ and his explanation of the reduction of the wave function for macroscopic objects, then the same interpretation yields to infer that the experimental observation of microscopic quantum fluctuations must be interpreted as observable violation of general covariance. Even if the violation is small in the sense of a small $\Delta$, it will be observable. Furthermore, if one instead adopts the measure $\widetilde{\Delta}$, then there is no objective meaning for attributing a very small violation of the principle of general covariance for microscopic systems and large violation for large or macroscopic systems, because the violation of general covariance measured using $\widetilde{\Delta}$ is universal and the same for all systems obeying Heisenberg energy/time uncertainty relation.
\section{Elements of Number Theory}
\newpage

\bigskip
Vale!


\newpage
\begin{thebibliography}{99}


\bibitem{Adler2003} S. L. Adler, {\it Why Decoherence has not Solved the Measurement Problem: A Response to P. W. Anderson}, Stud.Hist.Philos.Mod.Phys. {\bf 34}:135-142 (2003).

\bibitem{Adler}  S. L. Adler, {\it Quantum Theory as an Emergent Phenomenon: The Statistical Mechanics of Matrix Models as the Precursor of Quantum Field Theory}, Cambridge University Press (2004).


\bibitem{Arnold} V. Arnold, {\it Mathematical Methods of Classical Mechanics}, Springer (1989).

\bibitem{ArnoldAvez} V. I. Arnold and A. Avez, {\it Ergodic problems of classical mechanics}, Princeton University (1968).

\bibitem{BaoChernShen} D. Bao, S.S. Chern and Z. Shen, {\it An Introduction to Riemann-Finsler
Geometry}, Graduate Texts in Mathematics 200, Springer-Verlag.


 \bibitem{Bars2001} I. Bars, {\it Survey of Two-Time Physics}, Class.Quant.Grav. {\bf 18} 3113-3130 (2001).

\bibitem{Bell1964} J. S. Bell, {\it On the Einstein-Podolski-Rosen paradox}, Physica {\bf 1}, 195 (1964).

\bibitem{Bell Introduction 1971} J. S. Bell, {\it Introduction to the hidden-variable question}, Foundations of Quantum Mechanics. Proceedings of the International School of Physics Enrico Fermi, New York, Academic (1971), pp 171-81.


\bibitem{Bell} J. S. Bell, {\it Speakable and Unspeakable in Quantum Mechanics}, Cambridge University Press (1987).

\bibitem{Beem1} J. K. Beem, {\it Indefinite Finsler Spaces and Timelike Spaces},
Canad. J. Math. {\bf 22}, 1035 (1970).


\bibitem{Berger2002} M. Berger {A panoramic view of Riemannian Geometry}, Springer-Verlag (2002).

\bibitem{Berry} M. V. Berry, {\it The Bakerian Lecture, 1987: Quantum Chaology}, Proc. R. Soc. Lond. A {\bf 413}, 183-198 (1987).

\bibitem{BerryKeating} M. V. Berry and J. P. Keating, {\it $H=xp$ and the Riemann zeros}, in {\it Supersymmetry and Trace Formulae: Chaos and Disorder},
ed. J. P. Keating, D. E. Khemelnitskii, L. V. Lerner, Kuwler 1999; M. V. Berry and J. P. Keating, {\it The Riemann zeros and eigenvalue asymptotics}, SIAM REVIEW {\bf 41} (2), 236 (1999).

\bibitem{Blasone2}M. Blasone, P. Jizba and F. Scardigli, {\it Can Quantum Mechanics be an Emergent
Phenomenon?}, J. Phys. Conf. Ser. {\bf 174} (2009) 012034, arXiv:0901.3907[quant-ph].

\bibitem{Bohmquantumtheory} D. Bohm, {\it Quantum Theory}, Prentice Hall, Inc (1951).

\bibitem{Bohm} D. Bohm, {\it A suggested Interpretation of the Quantum Theory in Terms of Hidden variables. I}, Phys. Rev. {\bf 85}, 166-179 (1952); {\it A suggested Interpretation of the Quantum Theory in Terms of Hidden variables. II}, Phys. Rev. {\bf 85}, 180-193 (1952).

\bibitem{Bohm1980} D. Bohm, {\it Wholeness and the Implicate Order} , London: Routledge,  (1980).


\bibitem{BornJordan1925} M. Born and C. Jordan, {\it Zur Quantenmechanik}, Z. Physik {\bf 34}, 858-888 (1925).

\bibitem{BornHeisenbergJordan1925} M. Born, W. Heisenberg and C. Jordan, {\it Zur Quantenmechanik II}, Z. Physik {\bf 35},  557-615 (1925).

\bibitem{CaianielloGasperiniScarpetta} E. R. Caianiello, M. Gasperini and G. Scarpetta, {\it Inflaction and singularity
    prevention in a model for extended-object-dominated cosmology}, Clas. Quan. Grav. {\bf 8}, 659 (1991).

\bibitem{Cartan1923} E. Cartan, {\it Sur les vari\'et\'es \`a connexion affine et la th\'eorie de la relativité g\'en\'ralis\'ee (Premire partie)}, Annales Scientifiques de l'\'Ecole Normale Sup\'erieure, {\bf 40}, 325 (1923).


\bibitem{Cartan1924} E. Cartan, {\it  Sur les vari\'et\'es \`a connexion affine et la th\'eorie de la relativité g\'en\'ralis\'ee (Premire partie)(Suite)}, Annales Scientifiques de l'\'Ecole Normale Sup\'erieure, {\bf 41}, 1 (1924).


\bibitem{ChanTsou} H. M. Chan and S. T. Tsou, {\it Some elementary gauge theory concepts}, World Scientific (1993).


\bibitem{Chern} S. S. Chern, {\it Finsler geometry is just Riemannian geometry without the quadratic restriction}, Notices of the American Mathematical Society {\bf 43}, 959 (1996).

\bibitem{Chicone} C. Chicone, {\it Ordinary Differential Equations with Applications}, Springer (2006).

 \bibitem{ColellaOverhauser1974} R. Colella and A. W. Overhauser, {\it Experimental test of gravitational induced quantum interference}, Phys. Rev. Lett. {\bf 33}, 1237 (1974).

\bibitem{ColellaOverhauserWerner1975} R. Colella, A. W. Overhauser and S. A. Werner, {\it Observation of gravitationally induced interference}, Phys. Rev. Lett. {\bf 34}, 1472 (1975).

    \bibitem{Connes} A. Connes, {\it Trace formula in noncommutative geometry and the zeros of the Riemann zeta function}, Selecta Mathematica (New Series) {\bf 5}, 29 (1999).

\bibitem{CraigWeinstein} W. Craig and S. Weinstein, {\it On determinism and well-possedness in multiple time dimensions}, Proc. Royal. Soc. London, A: Mathematical, Physical and Engineering Science {\bf 465}, 3023–3046 (2009).


\bibitem{DeSinghVarma2019} S. De, T. P. Singh, A. Varma, {\it Quantum gravity as an emergent phenomenon},
         Int.J.Mod.Phys.D {\bf 28} 14, 1944003, (2019).

 \bibitem{DeWitt} B. S. DeWitt, {\it Quantum Theory of Gravity. I. The Canonical Theory}. Phys. Rev. {\bf 160} (5): 1113–1148 (1967).




\bibitem{Diosi 1987}  L. Di\'osi, {\it A universal master equation for the gravitational violation of quantum mechanics} Physics Letters A {\bf 120} (8), 377 (1987).

\bibitem{Diosi} L. Di\'osi, {\it Models for universal reduction of macroscopic quantum
fluctuations}, Phys. Rev. A {\bf 40}, 1165–1174 (1989).

\bibitem{Dirac1958} P. A. M. Dirac, {\it The Principles of Quantum Mechanics}, Fourth Edition, Oxford University Press (1958).

\bibitem{Dirichlet} P. G. L. Dirichlet, {\it Lectures in Number Theory, with supplements by R. Dedeckind}, History od Mathematics Sources {\it 16}, American Mathematical Society (1999).

\bibitem{Dolce} D. Dolce, {\it Elementary Spacetime cycles}, Europhys. Lett., {\bf 102}, 31002 (2013); {\it Internal times” and how to second-quantize fields by means of periodic boundary conditions}, Annals of Physics Vol. {\bf 457}, October 2023, 169398 .

\bibitem{Einstein 1905a} A. Einstein, {\it Zur Elektrodynamik bewegter k\"orper}, Annalen de r Physik {\bf 17}, (1905).

\bibitem{Einstein 1916} A. Einstein, {\it Die grundlage der allgemeinen Relativit\"astheorie}, Annalen der Physic vol. {\bf 49}, 769 (1916).

\bibitem{Einstein 1918} A. Einstein, {\it Principielles zur allgemeinen Relativit\"atstheorie}, Annalen der Physik, vol. {\bf 55}, 242 (1918).

\bibitem{Einstein et al.} H. A. Lorentz, A. Einstein, H. Minkowski and H. Weyl. {\it The principle of the relativity}, Dover publications (1923).

\bibitem{Einstein1922} A. Einstein, {\it The meaning of relativity}, Princenton University Press (1923).

\bibitem{EinsteinPodolskiRosen} A. Einstein, B. Podolsky and N. Rosen, {\it Can quantum-mechanical desscription of physical reality be considered complete?}, Phys. Rev. {\bf 47}, 777 (1935).

\bibitem{Ehlers Pirani Schild} J. Ehlers, F.A.E. Pirani and  A. Schild, {\it The geometry of free fall and light propagation}, In L. O' Raifeartaigh (ed.) General Relativity: Papers in Honour of J. L. Synge, pp.63-84. Claredon Press, Oxford (1972); J. Ehlers, F.A.E. Pirani and  A. Schild,  {\it Republication of: The geometry of free fall and light propagation}   Gen. Relativ. Gravit. {\bf 44}, 1587 (2012).

\bibitem{Elsgoltz} L. Elsgoltz, {\it Ecuaciones diferenciales y c\'alculo variacional}, Editorial URSS (1994).

\bibitem{Earman1994} J. Earman, {\it Bangs, Crunches, Whimpers and Schrieks: Singularities and Acausalities in Relativistic Spacetimes}, Oxford University Press (1995).

\bibitem{Elze2003} H. T. Elze, {\it Quantum mechanics emerging from timeless classical dynamics},  Trends in General Relativity and Quantum Cosmology, ed. C.V. Benton (Nova Science Publ., Hauppauge, NY, 79-101 (2006), quant-ph/0306096 [quant-ph].

\bibitem{Elze} H.T. Elze, {\it The Attractor and the Quantum States}, Int. J. Qu. Info. {\bf 7}, 83 (2009).

\bibitem{Elze2009} H. T. Elze, {\it Symmetry aspects in emergent quantum mechanics},
         J. Phys. Conf. Ser. {\bf 171}, 012034 (2009).

\bibitem{FernandezIsidroPerea2013} P. Fern\'andez de C\'ordoba, J. M. Isidro and Milton H. Perea {\it Emergent Quantum Mechanics as thermal  essemble}, Int. J. Geom. Meth. Mod. Phys. {\bf 11}, 1450068 (2014).

\bibitem{FernandezIsidroVazquez2016} P. Fernandez de Cordoba, J.M. Isidro, J. Vazquez Molina, {\it The holographic quantum},
Published in Found.Phys. {\bf 46} no.7, 787-803 (2016).

\bibitem{Feynman} R. P. Feynman, R. B. Leighton and M. Sands, {\it The Feynman lectures on Physics}, Ad. Wesley (1964).

\bibitem{Feynman Hibbs} R. P. Feynman and A. R. Hibbs, {\it Quantum mechanics and path integrals}, MacGraw-Hill Companies, New York (1965).

\bibitem{Ricardo2005}{\it On the Maximal Universal Acceleration in Deterministic Finslerian Models}, arXiv:gr-qc/0501094 .

\bibitem{Ricardo05b} R. Gallego Torrom\'e,  {\it Quantum systems as results of geometric evolutions}, arXiv:math-ph/0506038 .

 \bibitem{Ricardo06}R. Gallego Torrom\'e, {\it A Finslerian version of 't Hooft Deterministic Quantum Models},
J. Math. Phys. {\bf  47}, 072101 (2006).

\bibitem{Ricardo Randers-Lorentz} R. Gallego Torrom\'e, {\it On the notion of Randers spacetime}, arXiv:0906.1940v5 [math-ph].

\bibitem{Ricardo2010} R. Gallego Torrom\'e, {\it Averaged dynamics of ultra-relativisitc charged particles beams}, PhD Thesis, Lancaster University (2010), arXiv:1206.4184 [math-ph].

\bibitem{GallegoPiccioneVitorio:2012} R. Gallego Torrom\'e, P. Piccione and H. Vit\'orio,
{\it On Fermat's principle for causal curves in time oriented {Finsler} spacetimes},
{J. Math. Phys.} {\bf 53}, 123511 (2012).

\bibitem{Ricardo2012} R. Gallego Torrom\'e, {\it Geometry of generalized higher
order fields and applications to classical linear electrodynamics}, arXiv:1207.3791.

\bibitem{Ricardo2014} R. Gallego Torrom\'e, {\it An effective theory of metrics with maximal acceleration},  Class. Quantum Grav. {\bf 32} 245007 (2015).

\bibitem{Ricardo2015b} R. Gallego Torrom\'e, {\it Emergence of classical gravity and the objective reduction of the quantum state in deterministic models of quantum mechanics}, Journal of Physics: Conference Series {\bf 626} (1), 012073 (2015).

 \bibitem{Ricardo2016} R. Gallego Torrom\'e, {\it Classical gravity from certain models of emergent quantum mechanics},
Journal of Physics: Conference Series {\bf 701} (1), 012033 (2016).

\bibitem{Ricardo 2017} R. Gallego Torrom\'e, {A second order differential equation for a point
Charged particle}, International Journal of Geometric Methods in Modern Physics,
Vol. {\bf 14}, No. 04, 1750049 (2017).

\bibitem{Ricardo 2017b} R. Gallego Torrom\'e, {\it On singular generalized Berwald spacetimes and the
equivalence principle}, Vol. {\bf 14}, No. 06, 1750091 (2017).

\bibitem{Ricardo2017b} R. Gallego Torrm\'e, {\it Emergent Quantum Mechanics and the Origin of Quantum Non-local Correlations},
International Journal of Theoretical Physics, {\bf 56}, 3323(2017).

\bibitem{Gallego-Torrome2019} R. Gallego Torrom\'e, {\it Some consequences of theories with maximal acceleration in laser-plasma acceleration}, Modern Physics Letters A Vol. {\bf 34} (2019) 1950118.

\bibitem{Ricardo2020} R. Gallego Torrom\'e, {\it Maximal acceleration geometries and spacetime curvature bounds}, International Journal of Geometric Methods in Modern Physics, Vol. {\bf 17}, 2050060 (2020).



\bibitem{Ricardo2019a} R. Gallego Torrom\'e, {\it On the origin of the weak equivalence principle in a theory of emergent quantum mechanics}, Int. J. Geom. Methods Mod. Phys., Vol. {\bf 17}, No. 10, 2050157 (2020).

\bibitem{Ricardo general theory of dynamics} R. Gallego Torrom\'e, {\it General theory of non-reversible local dynamics}, International Journal of Geometric Methods in Modern Physics  Vol. {\bf 18}, No. 07, 2150111 (2021).

\bibitem{Ricardo quotient rings} R. Gallego Torrom\'e, {\it Quotient rings and fields of integers from the metric geometry point of view},  arXiv:2010.06817v2 [math.RA].

\bibitem{Ricardo 2022} R. Gallego Torrom\'e, {\it On Penrose’s theory of objective gravitationally induced wave function reduction}, International Journal of Geometric Methods in Modern Physics, Vol. {\bf 19}, Issue No. 02, 2250020 (2022).

 \bibitem{Ricardo2023} R. Gallego Torrom\'e, J.M. Isidro, Pedro Fern\'andez de C\'ordoba, {\it On the emergent origin of the inertial mass}, Foundations of Physics,  {\bf 53}, Article number: 52 (2023).

\bibitem{Ricardo 2024} R. Gallego Torrom\'e, {On the emergence of gravity in a framework of emergent quantum mechanics},
International Journal of Geometric Methods in Modern Physics, Vol. {\bf 21}, No. 14, 2450245 (2024).

\bibitem{Gauss Disquisitiones Arithmeticae} C. Friedrich Gauss, {\it Disquisitiones Arithmeticae}, Yale University (1966).

\bibitem{GhirardiRiminiWeber} G.C. Ghirardi,  A. Rimini and T. Weber, {\it Unified dynamics for microscopic and macroscopic systems}, Physical Review, D {\bf 34}, 470 (1986).

\bibitem{GhirardiGrassiRimini} G. Ghirardi, R. Grassi and A. Rimini, {\it Continuous-spontaneous-reduction model involving gravity}, Phys. Rev. A {\bf 42} 1057 (1990).
\bibitem{Godement} R. Godement, {\it Topologie Algebrique et Theorie des Faisceaux}, Hermann Paris (1958).

\bibitem{de Gosson} M. A. de Gosson, {\it Born-Jordan quantization and the equivalence of matrix and wave mechanics},
	eprint arXiv:1405.2519.

\bibitem{Groessing2013} G. Gr\"{o}essing, {\it Emergence of Quantum Mechanics from a Sub-Quantum Statistical Mechanics}, Int. J. Mod. Phys. B, {\bf 28}, 1450179 (2014).

\bibitem{Gromov} M. Gromov, {\it Riemannian structures for Riemannian and non-Riemannian spaces}, Birkh$\ddot{a}$user (1999).

\bibitem{Hatfield} B. Hatfield, {\it Quantum Field Theory of Point Particles and Strings}, Frontiers in Physics, CRC Press (1992).

\bibitem{Hawking Ellis 1973} S. W. Hawking and G. F. R. Ellis, {\it The Large Scale Structure of the Universe}, Cambridge University Press (1973).

\bibitem{Hawking1992} S. W. Hawking, {\it Chronology protection conjecture}, Phys. Rev. D {\bf 46}, 603 (1992).

\bibitem{HardyWright} G. H. Hardy and E. M. Wright, {\it An introduction to the theory of number }, Sixth Edition, Oxford University Press (2008).

\bibitem{Hooft} G. 't Hooft, {\it Determinism and Dissipation in Quantum Gravity}, Class. Quantum Grav. {\bf 16}, 3263 (1999).

\bibitem{Hooft 2001} Gerard't Hooft, {\it How does God play
dies?(Pre-) Determinism at the Planck Scale}, {
hep-th/0104219}.

\bibitem{Hooft2006} G. 't Hooft, {\it Emergent Quantum Mechanics and Emergent Symmetries}, 13th International Symposium on Particles, Strings, and Cosmology-PASCOS 2007. AIP Conference Proceedings {\bf 957}, pp. 154-163 (2007).

\bibitem{Hooft2012wavefunctionschroedingercollapse} G. 't Hooft, {\it How a wave function can collapse without violating Schr\"odinger's equation, and how to understand the Born's rule}, arXiv:1112.1811v3.

\bibitem{Hooft2016} G. 't Hooft, {\it The Cellular Automaton Interpretation of Quantum Mechanics}, Fundamental Theories in Physics Vol. {\bf 185}, Springer Verlag (2016).

\bibitem{Hooft2021} G. 't Hooft, {\it Fast Vacuum Fluctuations and the Emergence of Quantum Mechanics}, Foundations of Physics {\bf 51}(3), 63 (2021).

\bibitem{Hirzebruch1978} F. Hirzebruch, {\it Topological methods in Algebraic Geometry}, Classical in Mathematics, Springer Verlag (1995).

\bibitem{Isham 1995} C. Isham, {\it Lectures on quantum theory: mathematical and theoretical foundations}, Imperial College Press (1995).

\bibitem{Javaloyes et al.} M. A. Javaloyes, L. Lichtenfelz and P. Piccione {\it Almost isometries of non-reversible metrics with applications to stationary spacetimes}, Journal of Geometry and Physics {\bf 89}, 38 (2015).

\bibitem{JacobsonI} T. Jacobson, {\it Thermodynamics of Spacetime: The Einstein Equation of State}, Phys. Rev. Lett., {\bf 75}, 1260 (1995).

\bibitem{Kobakhidze2011} A. Kobakhidze, {\it Gravity is not an entropic force}
 Phys. Rev. D {\bf 83} Rapid Communication, 0215202 (2011).

 \bibitem{Kochen-Specken1967} S. Kochen and E. Specken, {\it The problem of hidden variables in quantum mechanics}, J. of Mathematics and Mechanics {\bf 17}, 59-87 (1967).

\bibitem{KolarMichorSlovak} I. Kolar, P. W. Michor and J. Slovak, {\it Natural operators in differential geometry}, Springer-Verlag (1993).

\bibitem{Kondepudi Prigogine 2015} D. Kondepudi adnd I. Prigogine, {\it Modern Thermodynamics,. From Heat Engines to Dissipative Structures}, Second edition, Wiley (2015).

\bibitem{Koopman1931} B. O. Koopman, {\it Hamiltonian Systems and Transformations in Hilbert Space} Proceedings of the National Academy of Sciences {\bf 17} (5), 315 (1931).

\bibitem{Koopman-von Neumann1932} B. O. Koopman and J.v. Neumann, {\it Dynamical systems of continuous spectra}, Proceedings of the National Academy of Science {\bf 18}, 255 (1932).

\bibitem{Mac Lane} S. Mac Lane, {\it Cateories for the Working Mathematician}, Graduate Texts in Mathematics 5, Springer (1971).

\bibitem{MaldacenaSusskind} J. Maldacena and L. Susskind, {\it Cool horizons for entangled black holes}, Progress of Physics, Vol. {\bf 61}, 781 (2013).

\bibitem{Mashhoon1990}  B. Mashhoon, {\it Limitations of spacetime measurements}, Phys. Lett. A {\bf 143}, 176-182 (1990); B. Mashhoon,
{\it The Hypothesis of Locality in relativistic physics}, Phys. Lett. A {\bf 145}, 147 (1990).

\bibitem{MilmanSchechtman2001} V. D. Milman and G. Schechtman, {\it Asymptotic theory of Finite Dimensional normed spaces}, Lecture notes in Mathematics 1200, Springer (2001).

\bibitem{MironHrimiucShimadaSabau:2002} {R. Miron, H. Hrimiuc, H. Shimada and S.V. Sabau}, {\it The geometry of {Hamilton} and {Lagrange} spaces}, {Kluwer Academic Publishers} (2002).
\bibitem{MisnerThorneWheeler} C. W. Misner, K. S. Thorne and J. A. Wheeler, {\it Gravitation}, Princeton University Press (2017).

\bibitem{Minkowski 1908} H. Minkowski, {Space and Time, with notes by A. Sommerfeld}, in {\it The Principle of Relativity}, Dover Edition (1952).

\bibitem{Padmanabhan2012a} T. Padmanabhan, {\it Gravity as an emergent phenomenon: Conceptual aspects}
         AIP Conf.Proc. {\bf 1458}  1, 238-252 (2012).

\bibitem{Padmanabhan2012b} T. Padmanabhan, {\it Emergent perspective of Gravity and Dark Energy} T. Padmanabhan
         Res.Astron.Astrophys. {\bf 12}, 891-916 (2012).

\bibitem{Padmanabhan2015} T. Padmanabhan, {\it Emergent Gravity Paradigm: Recent Progress}, T. Padmanabhan
         Mod.Phys.Lett.A {\bf 30}, 1540007 (2015).

\bibitem{Pauli 1958} W. Pauli, {\it Theory of Relativity}, Pergamon Press (1958).

\bibitem{Penrose 1996} R. Penrose, {\it On gravity's role in quantum state reduction}, Gen. Rel. and Gravit. {\bf 8}, No. 5, 581 (1996).

\bibitem{Penrose 2005} R. Penrose, {\it The Road to Reality}, Vintage, London (2005).

\bibitem{Penrose 2014a} R. Penrose, {\it On the Gravitization of Quantum Mechanics 1: Quantum State Reduction}, Foundations of Physics {\bf 44} (5), 557 (2014).

\bibitem{Prigrogine 2019} I. Prigogine, {\it Las leyes del caos}, Cr\'itica (2019).

\bibitem{Rademacher} H. B. Rademacher, {\it Nonreversible Finsler Metrics of positive flag curvature}, in Riemann-Finsler Geometry, MSRI Publications Vol. {\bf 80}, 261-303 (2004).

\bibitem{Raguni} G. Raguni, {\it Two-Time Relativistic Bohmian Model of Quantum Mechanics}, https://www.researchgate.net/publication/383840820 .

\bibitem{Randers} G. Randers, {\it On an Asymmetrical Metric in the Four-Space of General Relativity}, Phys. Rev. {\bf 59}, 195-199 (1941).

\bibitem{Redhead} M. Redhead, {\it Incompleteness, Nonlocality and Realism: A prolegomenon to the Philosophy of Quantum Mechanics},  Claredon Press, Oxford (1989).

\bibitem{ReedSimonI} M. Reed M and B. Simon, {\it Methods of Modern Mathematical Physics I: Functional Analysis, Revised and Enlarged Edition}, New York: Academic Press (1980).

\bibitem{Robinson} A. Robinson, {\it Non-standard Analysis}, Revised Edition, Princeton University Press (1996).

\bibitem{Sakurai} J.J. Sakurai, {\it Modern Quantum Mechanics} Addison-Wesley (1994).

\bibitem{SierraTownsend} G. Sierra and P. K. Townsend, {\it Landau levels and Riemann zeroes}, Phys. Rev. Lett. {\bf 101}, 110201 (2008).

\bibitem{Sierra2014} G. Sierra, {\it The Riemann zeros as energy levels of a Dirac fermion in a potential built from the prime numbers in Rindler spacetime},  J. Phys. A: Math. Theor. {\bf 47} 325204 (2014).

\bibitem{SharmaSingh} A. Sharma, T. P. Singh, {\it How the quantum emerges from gravity},
         Int.J.Mod.Phys.D {\bf 23},  12, 1442007 (2014).

\bibitem{Singh} T. P. Singh, {\it From quantum foundations, to spontaneous quantum gravity: an overview of the new theory}, Z. Naturforsch. A {\bf 75}, 833 (2020).

\bibitem{Smolin2012} L. Smolin, {\it A real ensemble interpretation of quantum mechanics}, Found. of Phys. {\bf 42}, Issue 10, pp 1239–1261 (2012).

\bibitem{Snyder}  H. S. Snyder, {\it Quantized Space-Time}, Phys. Rev. {\bf 71} (1), 38-41 (1947).

\bibitem{Sternberg} S. Sternberg, {\it Dynamical Systems}, Dover Publications (2010).

\bibitem{Stone} M. H. Stone, {\it On One-Parameter Unitary Groups in Hilbert Space}  Ann. of Math.
Second Series, Vol. {\bf 33} pp. 643-648 (1932).

\bibitem{Susskind 2014} L. Susskind, {\it ER=EPR, GHZ, and the Consistency of Quantum Measurements}, arXiv:1412.8483v1 [hep-th].

\bibitem{Talagrand} M. Talagrand, {\it A new look at independence}, Ann. Probab. {\bf 24}, Number 1, 1-34 (1996).

\bibitem{Tolman} R. C. Tolman, {\it The prinicples of Statistical Mechanics}, Dover (1979).

\bibitem{Verlinde2011} E. Verlinde, {\it On the origin of gravity and the Laws of Newton}, JHEP {\bf 1104}, 029 (2011).

\bibitem{von Neumann} J. von Neumann,  {\it Zur Operatorenmethode In Der Klassischen Mechanik}, Annals of Mathematics {\bf 33} (3): 587–642 (1932);
J. von Neumann,  {\it Zusatze Zur Arbeit "Zur Operatorenmethode.... }, Annals of Mathematics {\bf 33} (4): 789–791 (1932).

\bibitem{Wald1984} R. M. Wald, {\it General Relativity}, The University of Chicago Press (1984).

\bibitem{Warner} F. Warner, {\it Foundations of Differentiable Manifolds and Lie Groups}, Scott, Foresman and Company (1971).

\bibitem{Weinberg1995} S. Weinberg, {\it The Quantum Theory of Fields. Volume I: Foundations}, Cambridge University Press (1995).

\bibitem{Wilson1931} W. A. Wilson, {\it On Quasi-Metric Spaces}, American Journal of Mathematics  {\bf 53}, No. 3 , pp. 675 (Jul., 1931).

\bibitem{Whitehead1932} J. H. C. Whitehead, {\it Convex regions in the geometry of paths}, Quarterly Journal of Mathematics - Quart. J. Math. {\bf 3}, no. 1, pp. 33-42 (1932).

\bibitem{Zurek2002} W. H. Zurek, {\it Decoherence and the Transition from Quantum to Classical-Revisited}, Los Alamos Science {\bf 27} 2002.


\end{thebibliography}
\end{document}